\PassOptionsToPackage{unicode}{hyperref}
\PassOptionsToPackage{hyphens}{url}
\PassOptionsToPackage{dvipsnames,svgnames,x11names}{xcolor}

\pdfoutput=1
\RequirePackage{fix-cm}
\documentclass[letterpaper,12pt]{article}

\usepackage[utf8]{inputenc}
\usepackage[T1]{fontenc}
\usepackage{newtxtext}
\usepackage[margin=1in]{geometry}
\usepackage{setspace}

\IfFileExists{upquote.sty}{\usepackage{upquote}}{}
\IfFileExists{microtype.sty}{
  \usepackage{microtype}
  \UseMicrotypeSet[protrusion]{basicmath}
}{}

\usepackage{xcolor}

\usepackage{amsmath,amssymb,amsthm,bm}

\usepackage{float}
\usepackage{graphicx}
\usepackage{subcaption}
\usepackage{longtable,booktabs,array,multirow,rotating,makecell,diagbox}
\usepackage{calc}
\usepackage{placeins}

\makeatletter
\def\maxwidth{\ifdim\Gin@nat@width>\linewidth\linewidth\else\Gin@nat@width\fi}
\def\maxheight{\ifdim\Gin@nat@height>\textheight\textheight\else\Gin@nat@height\fi}
\def\fps@figure{htbp}
\makeatother

\setkeys{Gin}{width=\maxwidth,height=\maxheight,keepaspectratio}

\usepackage{algorithm}
\usepackage{algpseudocode}
\usepackage{soul}

\usepackage[
  backend=biber,
  style=authoryear,
  natbib=true,
  sorting=nyt,
  maxcitenames=2,
  maxbibnames=999,
  giveninits=true,
  uniquename=false,
  doi=false,
  isbn=false
]{biblatex}
\usepackage{bookmark}

\IfFileExists{xurl.sty}{\usepackage{xurl}}{}
\hypersetup{
  pdftitle={Inference for high-dimensional regression with a possible threshold},
  pdfauthor={Jiatong Li and Hongqiang Yan},
  pdfkeywords={Debiased Lasso, Hypothesis testing, Residential segregation, Data-rich environment},
  colorlinks=true,
  linkcolor={Blue},
  filecolor={Maroon},
  citecolor={Blue},
  urlcolor={Blue},
  pdfcreator={LaTeX}
}

\DeclareMathOperator*{\argmin}{arg\,min}

\newtheorem{thm}{Theorem}
\newtheorem{cor}{Corollary}
\newtheorem{lem}{Lemma}
\newtheorem{prop}{Proposition}

\newtheorem{assm}{Assumption}

\newtheorem{rem}{Remark}

\numberwithin{equation}{section}

\let\appendixfootnotesize\footnotesize
\renewcommand{\footnotesize}{\normalsize\doublespacing}
\title{Inference for high-dimensional regression with a possible threshold}

\author{
Jiatong Li\thanks{
Jiatong Li: jiatongli@hust.edu.cn.
School of Economics, Huazhong University of Science and Technology,
No. 1037 Luoyu Road, Wuhan, Hubei 430074, China.
}
\qquad
Hongqiang Yan\thanks{
Hongqiang Yan: hongqiang.yan@asu.edu.
Morrison School of Agribusiness, Arizona State University,
7231 E Sonoran Arroyo Mall, Santan Hall Suite 230,
Mesa, AZ 85212-6414, USA.
Corresponding author: Hongqiang Yan.
}
}

\date{}

\begin{document}

\begin{refsection}

\doublespacing

\begingroup
\setstretch{1}
\renewcommand{\footnotesize}{\normalsize\setstretch{1}}

\maketitle

\begingroup
\renewcommand{\thefootnote}{}
\begin{NoHyper}
\footnotetext{%
We thank Mehmet Caner, Harold Chiang, Barry Goodwin, Vincent Han,
Atsushi Inoue, Ilze Kalnina, Tong Li, Zheng Li, Yuya Sasaki,
Yulong Wang, and Ping Yu, as well as participants at the
2022 and 2024 Midwest Econometrics Group Conferences,
the Bristol Econometrics Study Group, the 19th Economic
Graduate Student Conference, the 39th Canadian Econometrics
Study Group, and the 19th International Symposium on
Econometric Theory and Applications, for helpful comments.
}
\end{NoHyper}
\endgroup

\endgroup

\begin{abstract}
\normalsize
\doublespacing
\noindent This paper develops uniform inference for slope coefficients in
high-dimensional threshold regressions without pretesting for a
threshold effect. Confidence intervals based on a debiased threshold
Lasso estimator are uniformly asymptotically valid, allowing for
heteroskedasticity and non-sub-Gaussian covariates and errors.
We also propose a threshold-effect test and extend inference to
stationary time series. Applying the method to racial tipping in
U.S. neighborhoods, we find that the relationships between white
population change and residential turnover, housing age, commuting
schedules, and utility payments differ across minority-share regimes.
Threshold effects are detected in fewer metropolitan areas in
2000--2010 than in 1990--2000.
\end{abstract}

\noindent
\textit{Keywords:} Debiased Lasso, Hypothesis testing,
Residential segregation, Data-rich environment.

\noindent
\textit{JEL classification:} C12, C24, C55.

\vfill

\section{Introduction}

Many empirical questions concern whether economic relationships vary across different economic environments. Threshold regression provides a parsimonious framework for studying such heterogeneity without specifying the regime boundary in advance \citep{hansen2000}. Applications include business-cycle asymmetries \citep{Potter1995}, the propagation of macroeconomic shocks under different credit conditions \citep{Balke2000Credit}, and the response
of private transfers to household income \citep{CoxHansenJimenez2004}.

These applications may require researchers to account for many potentially relevant covariates whose coefficients can vary across regimes. This flexibility increases the number of parameters to be estimated and can make inference particularly challenging when one regime contains relatively few observations. Meanwhile, the data may be consistent with a common relationship across regimes, and the regime boundary is unidentified in the absence of a threshold effect. These features motivate an inference procedure that accommodates many covariates and threshold estimation without requiring a preliminary decision about whether a threshold effect exists.

\citet{lee2016} establish non-asymptotic oracle inequalities for prediction and $\ell_1$ estimation errors of the slope coefficients in high-dimensional regression with a possible change point. This paper instead focuses on inference for the slope coefficients from the following model, which accommodates cases with and without a threshold effect:
\begin{equation}\label{model}
Y_i=X_i'\beta_0+X_i'\delta_0\mathbf{1}\{Q_i<\tau_0\}+U_i,
\qquad i=1,\ldots,n,
\end{equation}
where $X_i$ is a $p$-dimensional vector of covariates, $Q_i$ determines regime membership, and $\delta_0$ captures coefficient differences across regimes separated by the unknown threshold $\tau_0$. We allow $p$ to exceed $n$, subject to sparsity restrictions. When $\delta_0=0$, the model reduces to a regression without a threshold effect and the location of the threshold is not identified. Valid inference on the slope coefficients must therefore account for the distinct asymptotic behavior of the threshold estimator in models with and without a threshold effect.

Our principal result establishes that inference on prespecified slope parameters does not require a preliminary test for a threshold effect. The debiased statistic admits a common standard normal approximation uniformly over the no-threshold and identified fixed-threshold classes.\footnote{Uniformity
is over distributions satisfying the stated moment, design, and identification conditions. The slope inference results do not cover arbitrary sequences of nonzero threshold effects shrinking toward zero.} Confidence intervals for individual coefficients attain uniformly asymptotically exact coverage, with lengths of order $n^{-1/2}$. The results also accommodate scalar contrasts involving an increasing number of coefficients. Researchers can therefore use a common procedure for coefficient inference without knowing whether a threshold effect exists.

Unlike the Lasso objective function used in \citet{lee2016} for estimation, ours incorporates a vanishing anchoring penalty to facilitate inference on the slope coefficients without pretesting for a threshold effect. In the absence of a threshold effect, this penalty stabilizes threshold selection around a prespecified reference point. Under a fixed, identified threshold effect, it preserves the coefficient estimation rates and consistency of threshold estimation. We then correct the bias induced by Lasso regularization using an approximate inverse of the Gram matrix constructed through nodewise regression. Establishing the normal approximation requires controlling the score at a threshold selected from the same observations, with regime membership changing discontinuously in the threshold. A central limit theorem at each fixed threshold alone is therefore insufficient. We combine threshold estimation rates with local stochastic equicontinuity to show that the leading score at the estimated threshold is asymptotically equivalent to its value at a deterministic reference point, the prespecified anchor, in the absence of a threshold effect and the true threshold under a fixed, identified effect. The analysis permits dependence between the threshold variable and the covariates, as well as non-sub-Gaussian covariates and heteroskedastic, non-sub-Gaussian errors.

Moreover, we develop a test for a threshold effect. The test statistic takes the maximum over candidate thresholds
and coefficient coordinates, with Gaussian multiplier
bootstrap critical values accounting for dependence across both dimensions. We show that the test has asymptotically correct size uniformly over the null class and is consistent against the specified alternatives. It assesses whether coefficients differ across regimes but is not a prerequisite for coefficient inference.

On one hand, this paper contributes to the literature on high-dimensional threshold models by establishing uniformly valid inference based on a debiased threshold Lasso estimator. \citet{lee2016} establish oracle inequalities
without prior knowledge of whether a threshold effect exists, while \citet{canerkock2017} derive sup-norm bounds and variable selection guarantees. We focus on uniformly valid inference on slope coefficients. Our construction builds on the debiasing approach\footnote{The literature refers to the
method as either the ``debiased'' or the ``desparsified'' Lasso. For clarity, we use the term ``debiased Lasso'' throughout the remainder of this paper.} of \citet{geer2014} and the uniform inference results for high-dimensional regression without a threshold effect by \citet{canerkock2018}. We extend these ideas to our threshold regression, where the design depends nonsmoothly on an estimated threshold and the threshold parameter is unidentified in the absence of a threshold effect. Addressing these features requires controlling the score evaluated at the selected threshold, beyond a central limit theorem at each fixed threshold.

Since the foundational work of \citet{tong1980threshold}, threshold models have been widely used to study nonlinear economic relationships. Applications include unemployment forecasting
\citep{MontgomeryEtAl1998Forecasting}, the effects of uncertainty shocks during periods of financial distress \citep{AlessandriMumtaz2019}, and the relationship between public debt and growth at different levels of indebtedness \citep{ChudikEtAl2017Debt}. Financial applications
analyze the joint dynamics of stock-market turnover, volatility, and returns \citep{GriffinEtAl2007Trade}, as well as adjustments in the term structure of interest rates
\citep{HansenSeo2002Threshold,AbreuRodrigues2026Threshold}. Other applications examine the persistence of deviations from purchasing-power parity \citep{Taylor2001PPP}, agricultural price transmission and spatial market integration in the presence of transaction costs \citep{GoodwinPiggott2001,GrebEtAl2013}, and international price convergence in relation to trade barriers and differences in economic development \citep{GlushenkovaEtAl2018}. At the household level, threshold models have also been used to investigate the role of family background in intergenerational income mobility \citep{DurlaufEtAl2017StatusTraps}.

On the empirical side, we apply our methods to racial tipping in U.S. neighborhoods, revisiting \citet{card2008tipping} with a richer set of neighborhood characteristics. We find that the relationships between white population change and neighborhood characteristics, including residential turnover, housing age, commuting schedules, and utility payments, vary across minority-share regimes. Estimated threshold locations differ across metropolitan areas, while evidence of threshold effects is less widespread in the later decade. Specifically, at the 5\% significance level, our proposed threshold-effect test rejects the null of no threshold effect in 64.4\% of metropolitan areas in 1990--2000, compared with 28.1\% in 2000--2010.

The appendix extends the inference framework to stationary time series under near-epoch dependence, and establishes local-projection inference over a fixed set of horizons for contrasts. A fiscal illustration using the military spending news of \citet{ramey2018government} estimates an unemployment cutoff of 9.27\%. Relative to a fixed 6.50\% cutoff, the estimated threshold produces larger differences in cumulative fiscal multipliers between high- and low-unemployment states over both two- and four-year windows. The high-state multiplier confidence intervals nevertheless include unity in both windows. This comparison illustrates how estimating the state boundary affects the empirical characterization of state dependence in fiscal transmission.

The remainder of the paper is organized as follows. Section \ref{recallest} introduces the Lasso estimator and establishes oracle inequalities. Section \ref{dbLASSO} develops uniform inference and a test for threshold effects. Section \ref{simemp} presents simulation results and the application to racial tipping, and Section \ref{conclusion} concludes. The proofs and the extension to time series are provided in the appendix.

\subsection*{Notation} 

For a vector $a$, let $|a|_q$ denote its $\ell_q$ norm.
For $a\in\mathbb R^n$, define the empirical norm by
$\|a\|_n :=\left(n^{-1}\sum_{i=1}^n a_i^2\right)^{1/2}.$
For an $m\times n$ matrix $A$, define the induced
$\ell_1$ and $\ell_\infty$ norms as
$\|A\|_{\ell_1} :=\max_{1\leq j\leq n}\sum_{i=1}^m|A_{ij}|,$ 
$ \|A\|_{\ell_\infty}
:=\max_{1\leq i\leq m}\sum_{j=1}^n|A_{ij}|,$
respectively, and define the elementwise maximum norm by
$\|A\|_\infty
:=\max_{\substack{1\leq i\leq m\\1\leq j\leq n}}|A_{ij}|.$

For $a\in\mathbb R^n$, let
$J(a):=\{j\in\{1,\ldots,n\}:a_j\ne0\}$ denote its support
and $|J(a)|$ its cardinality. For
$M\subseteq\{1,\ldots,n\}$, let $a_M\in\mathbb R^n$
keep the coordinates of $a$ indexed by $M$ and equal
zero on $M^c$. The superscript $^{(j)}$ denotes the
$j$th element of a vector or the $j$th column of a matrix.

For $\alpha=(\beta',\delta')'$, define
$ f_{(\alpha,\tau)}(x,q)
:=x'\beta+x'\delta\,\bm{1}\{q<\tau\},$
$f_0(x,q)
:=x'\beta_0+x'\delta_0\,\bm{1}\{q<\tau_0\},$ and
$\widehat f(x,q)
:=x'\widehat\beta
+x'\widehat\delta\,\bm{1}\{q<\widehat\tau\}.$
The empirical prediction norm is
\[
\|\widehat f-f_0\|_n
:=\left[
\frac{1}{n}\sum_{i=1}^n
\{\widehat f(X_i,Q_i)-f_0(X_i,Q_i)\}^2
\right]^{1/2}.
\]

Unless otherwise specified, all constants are finite and
independent of $n$ and $p$. The symbols $C$ and $c$ denote
generic positive constants whose values may vary across
occurrences. Named constants and constants fixed within
an argument keep their definitions and values whenever
subsequently invoked, including in the lemmas and proofs.
Quantities defined as sequences may depend on $n$, even
when this dependence is suppressed in the notation.
For nonnegative sequences, $a_n\lesssim b_n$ means that
$a_n\leq Cb_n$ for all sufficiently large $n$, and
$a_n\asymp b_n$ means that both $a_n\lesssim b_n$ and
$b_n\lesssim a_n$ hold.

%Throughout the paper, we use the superscript zero to signify the true
%parameter value.

\section{The Anchored Threshold Lasso Estimator and Oracle Inequalities }\label{recallest}

\subsection{Anchored Threshold Lasso Estimation}\label{sec:est}

The threshold regression model in \eqref{model} can be rewritten as
\begin{equation}\label{model2}
Y_i=\left\{\begin{aligned}& X_{i}^{\prime }\beta _{0}+U_{i},\quad &\text{if  $Q_{i} \geq \tau_{0}$}, \\
& X_{i}^{\prime }(\beta _{0}+\delta_0)+U_{i},\quad&\text{if  $Q_{i}<\tau_{0}.$}
\end{aligned}\right.
\end{equation}%
The observable threshold variable $Q_{i}$ determines regime membership, while $\delta_0$ captures the threshold effect between the two regimes. The unknown threshold parameter $\tau_0$ lies in the compact parameter space $\mathbb T=[t_0,t_1].$ When $\delta_0 = 0,$ there is no threshold effect, and the model reduces to a linear regression. %It is equivalent to variable selection in the linear model if $\widehat{\delta} = 0.$ 

Define the $(2p\times 1)$ augmented regressor vector $\bm{X}_{i}(\tau )=(X_{i}^{\prime }, X_{i}^{\prime }\bm{1}\{Q_i <\tau \})^{\prime }$ and let $\bm{{X}}(\tau)$ denote the corresponding $(n\times 2p)$ design matrix, whose $i$-th row is $\bm{X}_{i}(\tau )^{\prime }.$ Let ${X}$ and ${X}(\tau)$ denote its first and last $p$ columns, respectively. With $\alpha_{0}=(\beta _{0}^{\prime },\delta _{0}^{\prime })^{\prime },$ model \eqref{model} can be written as
\begin{equation} 
Y_{i}=\bm{X}_{i}(\tau _{0})^{\prime }\alpha _{0}+U_{i},\ \ \ i=1,\ldots
,n. \label{model1}
\end{equation}
We study uniform inference for the high-dimensional slope parameter $\alpha_0,$ allowing for $p>n.$ Let $\bm{Y}= (Y_{1},\ldots ,Y_{n})^{\prime }.$ The residual sum of squares is
\begin{equation}\label{defs}
S_{n}(\alpha ,\tau ) =\frac{1}{n}\sum_{i=1}^{n}\left( Y_{i}-X_{i}^{\prime }\beta
-X_{i}^{\prime }\delta \bm{1}\{Q_i <\tau \}\right) ^{2} 
 =\left\Vert \bm{Y}-\bm{X}(\tau)\alpha \right\Vert _{n}^{2}.
\end{equation}
%where $\alpha =(\beta ^{\prime },\delta ^{\prime })^{\prime }$.

Let $\tau_\ast$ be a fixed, deterministic reference point in the interior of $\mathbb{T}$, such as $\tau_\ast=(t_0+t_1)/2$. Define the threshold search grid as $\mathbb{T}_n = \bigl(\{Q_1,\ldots,Q_n\}\cap\mathbb{T}\bigr)
\cup\{\tau_\ast,t_1\}$.
This grid explicitly includes the reference point $\tau_\ast$. The endpoint $t_1$ ensures that the last sample partition is represented under the indicator convention $1\{Q_i<\tau\}$. Define the anchored threshold Lasso estimator as a solution to the joint minimization problem
\begin{equation}\label{joint-max}
(\widehat{\alpha }, \widehat{\tau }) :=\text{argmin}_{\alpha \in \mathcal{B}, \tau \in \mathbb T_n}\left\{
S_{n}(\alpha ,\tau )+\lambda_n \left| \bm{D}(\tau )\alpha \right|_{1} +\rho_n|\tau-\tau_\ast|\right\},
\end{equation}
where $\mathcal{B}$ is the parameter space for $\alpha_0,$ and $\lambda_n>0$ and $\rho_n>0$ govern coefficient regularization and threshold anchoring, respectively. The $(2p \times 2p)$ diagonal weighting matrix is
$\bm{D}(\tau ):=\operatorname{diag}\{\|\bm{X}^{(j)}(\tau)\|_n: j=1,\ldots,2p\}$. These data-dependent weights account for differences in covariate scales in the $\ell_1$ penalty. See Remark B.1 of \citet{chiang_rodrigue_sasaki_2023} for alternative weighting choices. Unlike the Lasso objective function used for estimation, ours incorporates a vanishing anchoring penalty $\rho_n|\tau-\tau_\ast|$ to facilitate inference on the slope coefficients without pretesting for a threshold effect. The anchoring penalty  penalizes deviations from the prespecified reference point. It stabilizes threshold selection when the threshold location is unidentified in the absence of a threshold effect. 

The joint minimization problem in \eqref{joint-max} %with the threshold restricted to $\mathbb{T}_n$, 
also admits the following equivalent two-step formulation.

\noindent
(i) For each $\tau\in\mathbb T_n$, define
\begin{equation}\label{LASSO-fixed-tau}
\widehat{\alpha}(\tau)
=
\operatorname*{arg\,min}_{\alpha\in\mathcal{B}}
\left\{
S_n(\alpha,\tau)
+\lambda_n\left|\bm{D}(\tau)\alpha\right|_1
+\rho_n|\tau-\tau_\ast|
\right\}.
\end{equation}
For each fixed $\tau$, the anchoring penalty is constant with respect to $\alpha$ and therefore leaves the usual weighted Lasso coefficient estimator unchanged.

\noindent
(ii) Select the threshold by minimizing the profile criterion:
\begin{equation}\label{tau-max}
\widehat{\tau}
=
\operatorname*{arg\,min}_{\tau\in\mathbb T_n}
\left\{
S_n\bigl(\widehat{\alpha}(\tau),\tau\bigr)
+\lambda_n
 \left|\bm{D}(\tau)\widehat{\alpha}(\tau)\right|_1
+\rho_n|\tau-\tau_\ast|
\right\},
\end{equation}
and set $\widehat{\alpha}=\widehat{\alpha}(\widehat{\tau})$. If multiple grid points attain the minimum, we select the smallest one.

Although the loss and coefficient penalty are unchanged across thresholds that yield the same sample partition, the anchoring penalty generally varies with $\tau$. Consequently, candidates generating identical threshold designs may have different values of the augmented criterion. Threshold selection therefore evaluates the full criterion in \eqref{tau-max} at each grid point.

%These are weighted Lasso estimators with a data-dependent $\ell_1$ penalty that balances the covariates. \citet{chiang_rodrigue_sasaki_2023} summarize several weighting choices in Remark B.1. In practice, $\widehat \tau$ is selected from the observed values of the threshold variable, $\{Q_1,\dots, Q_n\}.$ \footnote{If $n$ is very large, $\mathbb{T}$ can be approximated by a grid of $N$ evaluation points. See p.4 in \citet{hansen2000}.} This selection gives an interval, and we take the upper endpoint as $\widehat{\tau}.$ 

\subsection{Oracle Inequalities}\label{them1}

We begin by specifying the tuning parameters in \eqref{joint-max}.
Let $\ell_n=\log(2pn)$ and $b_n=\sqrt{\ell_n/n}$, and set
\begin{equation}\label{eq:tuning}
\lambda_n=A b_n,
\qquad
\rho_n=c_\rho\lambda_n,
\end{equation}
where $A>0$ is a sufficiently large fixed constant and
$c_\rho\in(0,\infty)$ is fixed. If $\log n\leq C\log(2p)$, as holds in particular when $p\geq n$, then $\lambda_n\asymp\sqrt{\log(2p)/n}$. The constants $A$ and $c_\rho$ are specified for a fixed normalization of the variables. We impose the following conditions to establish a preliminary bound on the prediction error.

\begin{assm}\label{as1}
Let $\{(X_i, Q_i, U_i)\}_{i=1}^n$ be independent observations that need not be identically distributed. All constants in this assumption are independent of $n$ and $p$. %and $\left\{ U_i\right\}_{i=1}^n$ be independently distributed.

\noindent (i) The true coefficient vector satisfies $\alpha_0\in\mathcal B=[-C_1,C_1]^{2p}$, for some constant $C_1>0$, and $|\alpha_0|_0\le s_0$,
where $s_0\ge1$. Furthermore, $s_0b_n\to0$, and
$s_0^2|\delta_0|_1^2\ell_n/n=o(1)$.

\noindent (ii) For each $i$, $Q_i\sim\mathrm{Unif}[0,1]$, and $\tau_0\in\mathbb{T}=[t_0,t_1]\subset(0,1)$.

\noindent (iii) For some constants $C_2,C_3>0$, $\max_{i,j}E\left|X_i^{(j)}\right|^4\le C_2$ and
$\min_j n^{-1}\sum_iE\left[\left(X_i^{(j)}\right)^2 1\{Q_i<t_0\}\right]\\ \ge C_3$. In addition, the conditional second-moment functions 
$E\!\left[X_i^{(j)}X_i^{(l)}\mid Q_i=\tau\right]$ are continuous and uniformly bounded on a neighborhood of
$\tau_0$ over $i,j,l$.

\noindent  (iv) For each $i$, $E[U_i\mid X_i,Q_i]=0$. For some constant $C_U>0$, $E[U_i^2\mid X_i,Q_i]\leq C_U$ almost surely, and $\max_{1\leq i\leq n}E[|U_i|^4]\leq C_U.$ In addition, the conditional moment functions $ E\left[X_i^{(j)} X_i^{(l)} U_i^{2} \vert Q_i=\tau\right]$
are continuous and uniformly bounded on a neighborhood of $\tau_0$ over $i,j,l$.

\noindent  (v)  ${\sqrt{EM_{UX}^2}\sqrt{\ell_n}}/{\sqrt{n}}=o(1),$ where $M_{UX}= \max_{i,j} \left|U_{i}X_{i}^{(j)}\right|$. 

\noindent  (vi) ${\sqrt{EM_{XX}^2}\sqrt{\ell_n}}/{\sqrt{n}}=o(1),$ where $M_{XX}= \max_{i,j,l} \left| X_i^{(j)} X_i^{(l)}\right|$.

\end{assm}

Assumption \ref{as1}(i) bounds the individual slope coefficients and restricts the growth of the sparsity level and the aggregate magnitude of the threshold effect. Both $s_0$ and $|\delta_0|_1$ may increase with $n$, subject to the stated rate conditions. Since $\lambda_n=Ab_n$, the final restriction implies $s_0\lambda_n|\delta_0|_1=o(1)$, ensuring that Assumption 6 in Appendix E of \citet{lee2016} holds for sufficiently large $n$. % Assumption \ref{as1} (ii) ensures that there are no ties among the $Q_i$s. Without loss of generality, we normalize the population threshold variable using the probability integral transform, $ Q_i=F_{\widetilde Q}(\widetilde Q_i) $, so that $Q_i\sim U[0,1]$ under the continuity assumption. In empirical implementation, the unknown population CDF may be replaced by the empirical CDF, $\widehat Q_i=\widehat F_{\widetilde Q}(\widetilde Q_i) =\operatorname{rank}(\widetilde Q_i)/n. $ The empirical transform takes values on a discrete rank grid and should therefore be distinguished from the population probability integral transform. Because $\widetilde Q$ is continuous and thus has no ties with probability one, the rank transformation preserves the ordering of the threshold variable and therefore preserves all sample partitions induced by candidate threshold values.% Suppose that the threshold variable $\tilde{Q}$ has a continuous distribution with cumulative distribution function $F_{\tilde{Q}}$. The random variable $Q$ thus has a standard uniform distribution with $Q=F_{\tilde{Q}}(\tilde{Q})$ by the probability integral transform. 
%\footnote{We obtain a uniform distributed variable by computing $Q_i=\widehat{F}_{\tilde{Q}}(\tilde{ Q}_i)=$ the rank of $\tilde{Q}_i$ among $\left\{\tilde{ Q}_i\right\}_{i=1}^n/n,$ where $\widehat{F}_{\tilde{Q}}$ denotes the empirical distribution function of the data $\left\{\tilde{ Q}_i\right\}_{i=1}^n.$}
Together with independence across observations, Assumption \ref{as1}(ii) ensures that the threshold variables $Q_1,\ldots, Q_n$ are distinct with probability one.
It also restricts the threshold parameter to a fixed interval in the interior of the support of $Q_i$. Independence between $Q_i$ and $X_i$ is not required. \footnote{The proofs allow dependence between $Q_i$ and $X_i$. Section~\ref{simmo} includes a simulation design in which the threshold variable depends on the covariates.} Assumptions \ref{as1}(iii) and (iv) impose moment bounds and local regularity conditions on the covariates and errors. By the Cauchy--Schwarz inequality,
$\max_{\substack{1\leq i\leq n\\1\leq j,l\leq p}}
\left|E\left[X_i^{(j)}X_i^{(l)}\right]\right|
\leq C_2^{1/2},$
$\max_{\substack{1\leq i\leq n\\1\leq j,l\leq p}}
E\!\left[\left(X_i^{(j)}X_i^{(l)}\right)^2\right]
\leq C_2,$ and $\max_{\substack{1\leq i\leq n\\1\leq j\leq p}}
E\!\left[U_i^2\left(X_i^{(j)}\right)^2\right]
\leq (C_UC_2)^{1/2}.$
It follows that $\max_{i,j}\operatorname{Var}\left(U_iX_i^{(j)}\right)$,
$\max_{i,j,l}\operatorname{Var}\left(X_i^{(j)}X_i^{(l)}\right)$,
$\sup_{\tau\in\mathbb T}\max_{i,j,l}
  \operatorname{Var}\!\left(
  X_i^{(j)}X_i^{(l)}\bm{1}\{Q_i<\tau\}\right)$,
$\max_{i,j}\operatorname{Var}\!\left(
  \left(X_i^{(j)}\right)^2\bm{1}\{Q_i<t_0\}\right)$ are
  uniformly bounded. Assumptions \ref{as1}(v) and (vi)
further control the growth of the maximum score and design products as the sample size and covariate dimension increase.

\begin{rem}
Assumption \ref{as1} accommodates non-sub-Gaussian covariates and conditionally heteroskedastic, non-sub-Gaussian errors, subject to the stated moment and maximum envelope conditions. The analysis does not require the Gaussian error specification used in the fixed-design framework of \citet{lee2016} or the sub-Gaussian covariate and error conditions imposed by \citet{canerkock2017}. This relaxation concerns distributional tail behavior while the moment bounds, dimensionality restrictions, and design conditions remain necessary to the theoretical results.
\end{rem}

\begin{lem}\label{lemmaf}
Suppose that Assumption~\ref{as1} holds. Fix $\mu_1>0$ and define $\mu_2=2\mu_1/\sqrt{C_3}$. Let the tuning parameters satisfy \eqref{eq:tuning}, with $A$ sufficiently large for these choices and the constants in Assumption~\ref{as1}.
Then, with probability tending to one,
\begin{equation}\label{lem2-conc}
\|\widehat f-f_0\|_n^2
 \le (6+2\mu_2)C_1\sqrt{C_2^{1/2}+\mu_1\lambda_n}\,
                   s_0\lambda_n+\rho_n(t_1-t_0).
\end{equation}

\noindent Consequently, $
\|\widehat f-f_0\|_n
=O_P\!\left(\sqrt{s_0\lambda_n}\right).$
\end{lem}

Lemma \ref{lemmaf} establishes an upper bound on the prediction error that applies whether or not a threshold effect is present. The anchoring penalty contributes at most $\rho_n(t_1-t_0)$ to the bound on the squared prediction error. Since $\rho_n=c_\rho\lambda_n$ and $s_0\ge1$, this additional term does not change the prediction rate. Moreover, $s_0\lambda_n\to0$ under Assumption~\ref{as1}(i), so $|\widehat f-f_0|_n\to_P0$. This preliminary bound provides the basis for the sharper oracle inequalities in Theorem~\ref{main-thm-case1} for the no-threshold case and Theorem~\ref{thmftau} for the fixed-threshold case.

To derive sharper estimation bounds, we impose conditions on
the population design. Let
$\bm{\Sigma}(\tau)
=E\!\left[n^{-1}\sum_{i=1}^n
\bm X_i(\tau)\bm X_i(\tau)'\right]$,
$\bm M
=E\!\left[n^{-1}\sum_{i=1}^n X_iX_i'\right]$,
and
$\bm M(\tau)
=E\!\left[n^{-1}\sum_{i=1}^n X_i(\tau)X_i(\tau)'\right]$
denote the population second-moment matrices, and define
$\bm N(\tau)=\bm M-\bm M(\tau)$.
Here, $X_i(\tau)=X_i\bm{1}\{Q_i<\tau\}$ and
$\bm X_i(\tau)=(X_i',X_i(\tau)')'$.
The identity $\bm{1}\{Q_i<\tau\}^2=\bm{1}\{Q_i<\tau\}$
implies that the augmented second-moment matrix
$\bm{\Sigma}(\tau)$ has the block representation
\begin{equation*}
\bm{\Sigma}(\tau)
=
\begin{bmatrix}
\bm M & \bm M(\tau)\\
\bm M(\tau) & \bm M(\tau)
\end{bmatrix}.
\end{equation*}

For a nonempty set $\mathbb{S}\subseteq\mathbb{T}$ and a
constant $c_0>0$, define the population uniform adaptive
restricted eigenvalue as
\[
\kappa(s_0,c_0,\mathbb{S},\bm{\Sigma})
=
\inf_{\tau\in\mathbb{S}}
\min_{\substack{
J_0\subseteq\{1,\ldots,2p\}\\
1\le |J_0|\le s_0}}
\inf_{\substack{
\gamma\in\mathbb{R}^{2p}\setminus\{0\}\\
\|\gamma_{J_0^c}\|_1
\le c_0\sqrt{s_0}\|\gamma_{J_0}\|_2}}
\frac{
\{\gamma'\bm{\Sigma}(\tau)\gamma\}^{1/2}
}{
\|\gamma_{J_0}\|_2
}.
\]

%\scalebox{0.85}{\parbox{0.1\linewidth}{\begin{equation*}
%\kappa(s_0,c_0, \mathbb{S},\bm {M}) =  \min_{\tau \in \mathbb{S}} \quad \min_{J_0 \subset \{1,..., 2p \}, |J_0 | \le s_0} \quad \min_{\gamma \neq 0, \Vert\gamma_{J_0^c} \Vert_1 \le c_0\sqrt{s_0} \Vert\gamma_{J_0}\Vert_2} \frac{ (\gamma'E\left[1/n \sum_{i=1}^n {X_i}(\tau){X_i}(\tau)' \right]\gamma)^{1/2}}{|\gamma_{J_0}|_2};
%\end{equation*}}}

%\scalebox{0.85}{\parbox{0.1\linewidth}{\begin{equation*}
%\kappa(s_0,c_0, \bm {M}) =   \min_{J_0 \subset \{1,..., 2p \}, |J_0 | \le s_0} \quad \min_{\gamma \neq 0, \Vert\gamma_{J_0^c} \Vert_1 \le c_0\sqrt{s_0} \Vert\gamma_{J_0}\Vert_2} \frac{ (\gamma'E\left[1/n \sum_{i=1}^n {X_i}{X_i}'  \right]\gamma)^{1/2}}{|\gamma_{J_0}|_2}.
%\end{equation*}}}

\begin{assm}\label{as2}
(i) The matrices $\bm{M}(\tau)$ and $\bm{N}(\tau)$ are nonsingular for every $\tau\in\mathbb{T}$.

\noindent
(ii) (Uniform Adaptive Restricted Eigenvalue Condition)
For a fixed constant $c_0>0$ and the nonempty set
$\mathbb{S}\subseteq\mathbb{T}$ specified below, there exists
a constant $\underline{\kappa}>0$, independent of $n$ and $p$,
such that $\kappa(s_0,c_0,\mathbb{S},\bm{\Sigma}) \ge \underline{\kappa}.$
\end{assm}
The restricted eigenvalue may depend on $n$ and $p$ through the
population matrices, sparsity level, and threshold index set. The constant $\underline{\kappa}$ is its fixed positive lower bound.

Assumption~\ref{as2}(i) requires the population second-moment matrices associated with both regimes to be nonsingular. Assumption \ref{as2}(ii) requires the population design to be uniformly well conditioned over the specified restricted directions and threshold values. This requirement is weaker than imposing a uniform positive lower bound on the smallest eigenvalue of $\bm{\Sigma}(\tau)$ over $\tau\in\mathbb S$.

\begin{equation}\label{invformlu}
\bm{\Theta}(\tau)
:=
\bm{\Sigma}(\tau)^{-1}
=
\begin{bmatrix}
\bm{N}(\tau)^{-1}
&
-\bm{N}(\tau)^{-1}
\\
-\bm{N}(\tau)^{-1}
&
\bm{M}(\tau)^{-1}+\bm{N}(\tau)^{-1}
\end{bmatrix}
\end{equation}
as can be verified by direct block multiplication. Nonsingularity alone does not imply the uniform lower bound in Assumption \ref{as2}(ii). Along a high-dimensional asymptotic sequence, the smallest eigenvalues may approach zero even when the matrices remain nonsingular for every finite $n$. We therefore impose the uniform adaptive restricted eigenvalue condition separately.

To establish the corresponding empirical condition, we combine the uniform Gram-matrix bound in Lemma \ref{tau22} of the appendix with the restriction $s_0b_n\to0$ in Assumption \ref{as1}. Specifically,
$
\sup_{\tau\in\mathbb{S}}
\left\|
n^{-1}\bm{X}(\tau)'\bm{X}(\tau)
-\bm{\Sigma}(\tau)
\right\|_\infty
= O_P(b_n).$ Together with Assumption \ref{as2}(ii), this bound yields the empirical adaptive restricted eigenvalue condition stated in Lemma \ref{lemeg}.

The choice of $\mathbb S$ depends on the case under consideration. When $\delta_0=0$, the threshold parameter is unidentified, and we impose Assumption~\ref{as2}(ii) over the entire threshold parameter space, setting $\mathbb S=\mathbb T$. In the identified threshold case, we use a neighborhood of $\tau_0$ after establishing that $\widehat{\tau}$ lies in that neighborhood with probability tending to one. Identification of $\tau_0$ requires additional conditions beyond the restricted eigenvalue condition.

The uniform adaptive restricted eigenvalue condition allows us to sharpen the preliminary bound in Lemma~\ref{lemmaf} and derive oracle inequalities for the prediction error and the $\ell_1$ estimation error of the slope coefficients. Although the theoretical analysis treats the two cases separately, the same estimation procedure applies in both, without pretesting for a threshold effect.

\subsubsection{Case I. No Threshold} \label{nothreshold}

When $\delta_0=0$, the true model reduces to the linear regression $Y_i=X_i'\beta_0+U_i$. Although the threshold specification includes the additional parameters $\delta$ and $\tau$, the slope vector $\alpha_0=(\beta_0',0')'$ can still be estimated at the usual Lasso rates, as the following theorem establishes.

\begin{thm}\label{main-thm-case1}
Fix a constant $C_{\ell_0}\in(0,C_1]$, and define
\[
\mathcal A_{\ell_0}^{(1)}(s_0)
=\left\{
\alpha_0=(\beta_0',\delta_0')'\in\mathbb R^{2p}:
|\beta_0|_\infty\le C_{\ell_0},\quad
|\beta_0|_0\le s_0,\quad \delta_0=0
\right\}.
\]
Suppose Assumptions~\ref{as1}--\ref{as2} hold with
$\mathbb S=\mathbb T$. Let $\mu$, $\mu_1$, and $\mu_2$ be
constants independent of $n$ and $p$ satisfying
$0<\mu_2={2\mu_1}/{\sqrt{C_3}}<\mu<1.$
Take $A$ in \eqref{eq:tuning} sufficiently large, and choose a
fixed $c_0$ satisfying
\[
c_0>
\frac{2+\mu}{1-\mu}\frac{C_2^{1/4}}{\sqrt{C_3}}.
\]
Suppose that $\kappa=\kappa(s_0,c_0,\mathbb S,\bm\Sigma)$ is bounded away from zero uniformly in $n$ and $p$, and let
$(\widehat\alpha,\widehat\tau)$ solve \eqref{joint-max} with
tuning parameters satisfying \eqref{eq:tuning}.

Then, with probability tending to one uniformly over
$\mathcal A_{\ell_0}^{(1)}(s_0)$, 
\begin{align*}
\|\widehat f-f_0\|_n
&\le
\frac{2\sqrt{2}}{\kappa}
\sqrt{C_2^{1/2}+\mu_1\lambda_n}\,
\sqrt{s_0}\lambda_n,\\
|\widehat\alpha-\alpha_0|_1
&\le
\frac{4\sqrt{2}}{(1-\mu)\kappa^2}
\frac{C_2^{1/2}+\mu_1\lambda_n}
     {\sqrt{C_3-\mu_1\lambda_n}}\,
s_0\lambda_n,\\
|\widehat\tau-\tau_\ast|
&\le
\frac{4(C_2^{1/2}+\mu_1\lambda_n)}{\kappa^2}
\frac{s_0\lambda_n^2}{\rho_n}.
\end{align*}
Since $\rho_n=c_\rho\lambda_n$ and $s_0\lambda_n\to0$ under
Assumption~\ref{as1}(i), it follows that
\[
|\widehat\tau-\tau_\ast|
=O_P(s_0\lambda_n)=o_P(1)
\]
uniformly over $\mathcal A_{\ell_0}^{(1)}(s_0)$.
\end{thm}

Theorem~\ref{main-thm-case1} sharpens the prediction bound in Lemma~\ref{lemmaf} from $O_P(\sqrt{s_0\lambda_n})$ to
$O_P(\sqrt{s_0}\lambda_n)$ and establishes the coefficient bound $|\widehat\alpha-\alpha_0|_1=O_P(s_0\lambda_n)$. When $\log n=O(\log(2p))$, these rates coincide with the standard Lasso rates for high-dimensional linear models; see, for example, \citet{Bickel2009,geer2014}. Thus, the additional parameters in the threshold specification do not affect the orders of the prediction and coefficient estimation bounds. The anchoring penalty also stabilizes threshold selection by ensuring that $\widehat\tau$ converges to the prespecified reference point $\tau_\ast$ when the threshold location is unidentified. These results are obtained using the same estimation procedure, without pretesting for a threshold effect.

\subsubsection{Case II. Fixed Threshold}\label{fixedthreshold}

When $\delta_0\ne0$, we impose the following identification and local regularity conditions. These conditions adapt the framework of \citet{lee2016,canerkock2017} to the anchored criterion and the present moment assumptions.

Choose $\mu$, $\mu_1$, and $\mu_2$ as in
Theorem~\ref{thmftau}. For sufficiently large $n$, let
\[
 \overline D_n:=\sqrt{C_2^{1/2}+\mu_1\lambda_n},
 \qquad
 \underline D_n:=\sqrt{C_3-\mu_1\lambda_n},
\]
and define the deterministic preliminary bounds
\[
 B_{0,n}:=(6+2\mu_2)C_1\overline D_n s_0\lambda_n
                  +\rho_n(t_1-t_0),
 \qquad
 c_\alpha^{(0)}:=\frac{B_{0,n}}
                  {(1-\mu)\lambda_n\underline D_n}.
\]

\begin{assm}[Identifiability] \label{A-discontinuity}
There exist constants $C_4>0$ and $C_I>0$ such that, with probability tending to one, 
\begin{equation*}
 \|f_{(\alpha,\tau)}-f_0\|_n^2
 \ge C_4|\tau-\tau_0|-C_I\lambda_n|\alpha-\alpha_0|_1-C_4b_n^2
\end{equation*}
holds simultaneously for every $\tau\in\mathbb T$ and every $\alpha\in\mathcal B$ satisfying $|\alpha-\alpha_0|_1\leq c_\alpha^{(0)}$.
\end{assm}

The coefficient bound in Lemma~\ref{lem:iter-localization} ensures that the Lasso estimate lies in this identification domain with probability tending to one. The condition does not require the estimator to select at most $s_0$ coefficients.
The linear allowance accommodates deviations of $\alpha$ from $\alpha_0$, while the term $C_4b_n^2$ allows for gaps between observed threshold values.
This is an explicit empirical design condition because it does not follow from identification restricted to $s_0$-sparse coefficient vectors
or from Assumption~\ref{as1} alone.

Define
\begin{equation*}
 \eta^*:=b_n^2+
       \frac{B_{0,n}+C_I\lambda_n c_\alpha^{(0)}}{C_4},
 \qquad
 \mathbb S:=\{\tau\in\mathbb T:|\tau-\tau_0|\le\eta^*\}.
\end{equation*}
Under Assumption~\ref{as1} and \eqref{eq:tuning},
$c_\alpha^{(0)}=O(s_0)$ and $\eta^*=O(s_0\lambda_n)=o(1)$. The restricted eigenvalue condition in Assumption~\ref{as2} is imposed on this neighborhood, with $c_0$ chosen as in Theorem~\ref{thmftau}.

\begin{assm}[Local smoothness of the design]
\label{A-smoothness}
There exists a finite constant $C_5>0$ such that,
with probability tending to one, 
\begin{equation}\label{eq:S21}
 \sup_{1\le j,l\le p}
 \sup_{\substack{\tau\in\mathbb T\\|\tau-\tau_0|\le\eta}}
 \frac1n\sum_{i=1}^n
 |X_i^{(j)}X_i^{(l)}|
 |1\{Q_i<\tau_0\}-1\{Q_i<\tau\}|
 \le C_5\eta
\end{equation}
holds simultaneously for all $b_n^2\leq\eta\leq\eta^*$.
\end{assm}

\begin{assm}[Bounded local second moments]
\label{A-secondmoments}
There exists a finite constant $C_h>0$ such that,
with probability tending to one,
\begin{equation}\label{eq:S22}
 \sup_{b_n^2\le\eta\le\eta^*}h_n^2(\eta)\le C_h,
 \qquad
 h_n^2(\eta):=\frac{1}{2n\eta}\sum_{i=1}^n
       (X_i'\delta_0)^2 1\{|Q_i-\tau_0|\le\eta\}.
\end{equation}
\end{assm}

Assumptions~\ref{A-smoothness} and \ref{A-secondmoments} control local changes in the empirical design moments and the local second moment of the threshold signal, respectively, uniformly over the relevant radii. In particular, the coordinatewise fourth-moment bounds in Assumption~\ref{as1} alone do not imply \eqref{eq:S21}.

\begin{thm}\label{thmftau}
Using the same fixed coefficient bound $C_{\ell_0}\in(0,C_1]$
as in Theorem~\ref{main-thm-case1}, define
\[
\mathcal A_{\ell_0}^{(2)}(s_0)
=\left\{
\alpha_0=(\beta_0',\delta_0')'\in\mathbb R^{2p}:
|\alpha_0|_\infty\le C_{\ell_0},\quad
|\alpha_0|_0\le s_0,\quad \delta_0\ne0
\right\}.
\]
Suppose that Assumptions~\ref{as1}--\ref{as2} and
\ref{A-discontinuity}--\ref{A-secondmoments} hold.
Let $\mu$, $\mu_1$, and $\mu_2$ be constants independent
of $n$ and $p$ satisfying
$0<\mu_2={2\mu_1}/{\sqrt{C_3}}<\mu<1$.

Take $A$ in \eqref{eq:tuning} sufficiently large, and choose
a fixed $c_0$ satisfying
\[
c_0>
\frac{2+\mu}{1-\mu}\frac{C_2^{1/4}}{\sqrt{C_3}}.
\]
Let $\mathbb S$ be the neighborhood defined above, and write $\kappa=\kappa(s_0,c_0,\mathbb S,\bm\Sigma).$ Suppose that $\kappa$ is bounded away from zero uniformly in $n$ and $p$. Let $C_4$ and $C_I$ denote the constants in
Assumption~\ref{A-discontinuity}, and let
$(\widehat\alpha,\widehat\tau)$ solve \eqref{joint-max}
with tuning parameters satisfying \eqref{eq:tuning}.

Then $C_4-\rho_n\ge C_4/2>0$ for all sufficiently large $n$.
Moreover, with probability tending to one uniformly over
$\mathcal A_{\ell_0}^{(2)}(s_0)$,
\begin{align*}
\|\widehat f-f_0\|_n
&\le
\frac{6}{\kappa}
\sqrt{C_2^{1/2}+\mu_1\lambda_n}\,
\sqrt{s_0}\lambda_n,\\
|\widehat\alpha-\alpha_0|_1
&\le
\frac{36(C_2^{1/2}+\mu_1\lambda_n)}
     {\kappa^2(1-\mu)\sqrt{C_3-\mu_1\lambda_n}}\,
s_0\lambda_n,\\
|\widehat\tau-\tau_0|
&\le
\left\{
\frac{3\bigl[(1+\mu)\sqrt{C_2^{1/2}+\mu_1\lambda_n}+C_I\bigr]}
     {(1-\mu)\sqrt{C_3-\mu_1\lambda_n}}
+2
\right\}
\frac{12(C_2^{1/2}+\mu_1\lambda_n)}
     {\kappa^2(C_4-\rho_n)}\,
s_0\lambda_n^2.
\end{align*}
\end{thm}

Theorem~\ref{thmftau} establishes the same prediction and
coefficient estimation rates as Theorem~\ref{main-thm-case1},
together with the threshold estimation rate
$|\widehat\tau-\tau_0|=O_P(s_0\lambda_n^2)$.
Since $\rho_n=c_\rho\lambda_n\to0$, the anchoring penalty preserves these rate orders in the fixed-threshold case.
When $\log n=O(\log(2p))$, the threshold estimation rate is
$O_P\{s_0\log(2p)/n\}$ reflecting the locally linear
identification of a discontinuous threshold.
See \citet{lee2016}.

\begin{rem}
Our analysis focuses on uniform inference for the slope
coefficients. The convergence bounds for $\widehat{\tau}$
control the effect of threshold estimation on coefficient
inference, without requiring a characterization of its
limiting distribution. In high-dimensional quantile regression
with an identified threshold, \citet{Lee03072018} establish
an oracle property for a refined threshold estimator and
derive a nonstandard limiting distribution characterized by
the smallest minimizer of a compound Poisson process.
Their analysis provides a useful starting point for studying
the distribution of the threshold estimator in our setting.
Adapting such results to the anchored threshold Lasso
estimator and constructing confidence sets for $\tau_0$
require a separate analysis and are beyond the scope of
this paper. When $\delta_0=0$, $\tau_0$ is unidentified.
In that case, convergence of $\widehat\tau$ to the reference
point $\tau_\ast$ stabilizes coefficient inference rather
than estimating an identified population threshold.
\end{rem}

Theorems~\ref{main-thm-case1} and \ref{thmftau} establish
oracle inequalities for a common estimation procedure across
the no-threshold and fixed-threshold classes, without
pretesting for a threshold effect. Uniformity is over
parameter configurations satisfying the stated assumptions
with common design and identification constants, and it
does not cover arbitrary sequences of nonzero threshold
effects shrinking toward zero. The results permit
non-sub-Gaussian covariates and conditionally heteroskedastic,
non-sub-Gaussian errors. These estimation bounds provide
the basis for the inference results developed next.

\section{The Debiased Threshold Lasso Estimator and Uniform Inference}\label{dbLASSO}

\subsection{The Debiased Threshold Lasso Estimator} 
\label{sec:debiasedest}

We adapt the debiased Lasso construction of \citet{geer2014} to inference on the slope parameters in threshold regression. Let $\widehat\alpha=\widehat\alpha(\widehat\tau)$ denote the Lasso estimator in \eqref{joint-max}. The debiased estimator is defined as \footnote{The bias correction is motivated
by an approximate inversion of the Karush--Kuhn--Tucker (KKT)
conditions for coefficient estimation.} 
\begin{equation}\label{eq:debiased}
 \widetilde a(\widehat\tau)
 =\widehat\alpha(\widehat\tau)
  +\widehat{\bm\Theta}(\widehat\tau)
       \frac{\bm X(\widehat\tau)'
          \{Y-\bm X(\widehat\tau)\widehat\alpha(\widehat\tau)\}}{n},
\end{equation}where $\widehat{\bm\Theta}(\tau)$ is an approximate inverse of the empirical Gram matrix $\widehat{\bm\Sigma}(\tau)=\bm X(\tau)'\bm X(\tau)/n$.
Because the empirical Gram matrix may be singular in a high-dimensional model, we construct an approximate inverse using weighted nodewise Lasso, building on the approach of \citet{Meinshausen2006}. We apply these regressions separately to the two regimes and combine the resulting matrices through the block representation in \eqref{invformlu}. %The basic idea is that $\widehat{\alpha }(\widehat{\tau})$ satisfies the Karush-Kuhn-Tucker (KKT) conditions \begin{align}\label{kt1} -\bm {X}(\widehat{\tau})'(Y- \bm {X}(\widehat{\tau}) \widehat{\alpha }(\widehat{\tau}) )/n + \lambda \bm{D}(\widehat{\tau}) \widehat{\rho} =0,\end{align}
%where $\widehat{\rho} $ is a $2p\times 1$ vector, arising from the subdifferential of $|\alpha(\tau) |_1,$ $|\widehat{\rho}|_{\infty}\le1 $ and $\widehat{\rho}_j = \text{sign}(\widehat{\alpha}^{(j)}(\widehat{\tau}))$ if $ \widehat{\alpha}^{(j)}(\widehat{\tau})\ne0.$ 

For fixed $\tau$, the anchoring penalty
$\rho_n|\tau-\tau_\ast|$ is constant with respect to $\alpha$. It therefore leaves the fixed-threshold coefficient optimization unchanged and introduces no additional correction term into \eqref{eq:debiased}. Its effect on the debiased estimator is through the selected threshold $\widehat\tau$. We next derive separate decompositions for the no-threshold and fixed-threshold cases.

\subsubsection{Case I. No Threshold}\label{noshtrsholdbias}

When $\delta_0=0$, the model reduces to $\bm Y=X\beta_0+U$, where $U=(U_1,\ldots,U_n)'$, and $\bm X(\tau)\alpha_0=X\beta_0$ for every $\tau\in\mathbb T$. Define the debiasing remainder as
$\Delta(\tau)=\sqrt n\{\widehat{\bm\Theta}(\tau)
\widehat{\bm\Sigma}(\tau)-I_{2p}\}
\{\widehat\alpha(\tau)-\alpha_0\}$. Let $g\in\mathbb R^{2p}$ be a deterministic contrast satisfying $|g|_2=1$, with support $H=\{j:g_j\ne0\}$ and cardinality $h=|H|$. Both $g$ and $h$ may vary with $n$, subject to the restrictions on the growth of $h$ imposed in Section~\ref{uniforminference}. The Cauchy-Schwarz inequality gives $|g|_1\le\sqrt h$.  Inference on an individual coefficient corresponds to $g=e_j$, where $e_j$ is the $j$th coordinate vector in $\mathbb R^{2p}$. Substitution into \eqref{eq:debiased} gives
\begin{equation}\label{eq:debias-null}
\sqrt{n}g'(\widetilde{a}(\widehat{\tau}) -\alpha_{0})=g'\widehat{\bm{\Theta}}(\widehat{\tau})\bm {X}(\widehat{\tau})'U/n^{1/2} -g'\Delta(\widehat{\tau}).
\end{equation}
Theorem~\ref{main-thm-case1} establishes that $\widehat\tau$
converges to $\tau_\ast$ in the no-threshold case. Establishing
the limiting distribution of the leading term at the
data dependent threshold $\widehat\tau$ nevertheless requires
the random-index argument developed in the inference analysis while a central limit theorem at each fixed threshold alone is
insufficient.

\subsubsection{Case II. Fixed Threshold}\label{shtrsholdbias}

In the identified fixed threshold case, the model takes the form $\bm Y=\bm X(\tau_0)\alpha_0+U$. The debiased estimator admits the exact decomposition

\begin{equation}
\begin{aligned}\label{stat}
\sqrt{n}g'(\widetilde{a}(\widehat{\tau}) -\alpha_{0})&=g'\widehat{\bm{\Theta}}(\widehat{\tau})\bm {X}(\widehat{\tau})'U/n^{1/2}-g'\Delta(\widehat{\tau})\\
&+g'\widehat{\bm\Theta}(\widehat\tau)\bm X(\widehat\tau)'
\{\bm X(\tau_0)-\bm X(\widehat\tau)\}\alpha_0/n^{1/2}.
\end{aligned}
\end{equation}
The additional term captures the difference between the estimated and true threshold designs. The nodewise regression
bounds developed below control the estimation error in the
approximate inverse. Establishing asymptotic normality and
showing that the additional term is asymptotically negligible
require the further conditions and arguments presented in
Section~\ref{uniforminference}.

\subsubsection{Constructing the Approximate Inverse
\texorpdfstring{$\widehat{\bm{\Theta}}(\tau)$}
{widehat{Theta}(tau)}}
\label{constructprecison}

Let $X_{<}(\tau)$ and $X_{\geq}(\tau)$ denote the $n\times p$ regime design matrices with $i$th rows
$X_i'\mathbf 1\{Q_i<\tau\}$ and $X_i'\mathbf 1\{Q_i\geq\tau\}$, respectively. Define
$\widehat{\bm M}(\tau)=n^{-1}X_{<}(\tau)'X_{<}(\tau)$ and
$\widehat{\bm N}(\tau)=n^{-1}X_{\geq}(\tau)'X_{\geq}(\tau)$.
For each $j$, we regress the $j$th column on the remaining $p-1$ columns separately within each regime design. The lower-regime nodewise Lasso estimator is
\begin{equation}\label{eq:nodewise-lasso}
\widehat\gamma_j(\tau)
=
\argmin_{\gamma\in\mathbb R^{p-1}}
\left\{\left\|X_{<}^{(j)}(\tau)-X_{<}^{(-j)}(\tau)\gamma\right\|_n^2
+\lambda_{\mathrm{node}}
\left|\widehat{\bm\Gamma}_j(\tau)\gamma\right|_1\right\},
\end{equation}
where $\widehat{\bm\Gamma}_j(\tau)$ is the diagonal matrix with entries $\|X_<^{(l)}(\tau)\|_n$ for $l\ne j$. %These weights are empirical root mean squares computed without centering, and all empirical norms use the full-sample denominator $n$.
The upper-regime estimator is constructed analogously. We set $\lambda_{\mathrm{node}}=A_{\mathrm{node}}b_n$,
where $b_n=\sqrt{\ell_n/n}$, $\ell_n=\log(2pn)$, and
$A_{\mathrm{node}}>0$ is a sufficiently large fixed constant.

For the lower regime, define the normalization factor
\begin{equation*}
 \widehat z_j^2(\tau)
 =\|X_<^{(j)}(\tau)-X_<^{(-j)}(\tau)\widehat\gamma_j(\tau)\|_n^2
       +\frac{\lambda_{\mathrm{node}}}{2}
                   |\widehat\Gamma_j(\tau)\widehat\gamma_j(\tau)|_1,
\end{equation*}
and define the upper-regime factors $\widehat{\widetilde z}_j^2(\tau)$ analogously. The factor $1/2$ is consistent with the loss and penalty normalization in \eqref{eq:nodewise-lasso}. Let $\widehat{\bm A}(\tau)=\widehat{\bm Z}(\tau)^{-2}
\widehat{\bm C}(\tau)$ and
$\widehat{\bm B}(\tau)=\widehat{\widetilde{\bm Z}}(\tau)^{-2}
\widehat{\widetilde{\bm C}}(\tau)$ denote the resulting approximate inverses of $\widehat{\bm M}(\tau)$ and $\widehat{\bm N}(\tau)$, respectively. Complete definitions of these matrices are provided in Appendix Section~\ref{supp:nodewise-construction}. Combining them according to \eqref{invformlu}, we obtain
\begin{equation}\label{eq:theta-block}
 \widehat{\bm\Theta}(\tau)=
 \begin{pmatrix}
  \widehat{\bm B}(\tau)&-\widehat{\bm B}(\tau)\\
  -\widehat{\bm B}(\tau)&\widehat{\bm A}(\tau)+\widehat{\bm B}(\tau)
 \end{pmatrix}.
\end{equation}
Throughout, $\Theta_j(\tau)'$ and $\widehat\Theta_j(\tau)'$
denote the $j$th rows of $\bm\Theta(\tau)$ and
$\widehat{\bm\Theta}(\tau)$, respectively.

Let $\gamma_{0,j}(\tau)$ and $\widetilde\gamma_{0,j}(\tau)$ denote the population nodewise coefficients, with corresponding residuals $\upsilon_i^{(j)}(\tau)$ and $\widetilde\upsilon_i^{(j)}(\tau)$. For independent observations that need not be identically
distributed, these coefficients minimize the average
population squared loss. The corresponding orthogonality
conditions therefore hold after averaging over observations.
Define $H_A:=\{j\le p:j+p\in H\},$ $ H_B:=\{j\le p:j\in H\text{ or }j+p\in H\}, $ and $ \bar s:=\sup_{\tau\in\mathbb T}\max_{j\in H}|\Theta_j(\tau)|_0.$ To construct the rows of $\widehat{\bm\Theta}(\tau)$ indexed by $H$, only the rows indexed by $H_A$ and $H_B$ are needed from $\widehat{\bm A}(\tau)$ and $\widehat{\bm B}(\tau)$,
respectively.

\begin{assm}\label{as6-nodewise}
\noindent (i) $\bar s\ge1$ and $\bar s b_n\to0$.
Fix $\mu\in(0,1)$, using the same value as in the oracle inequalities when combining the results. There exist constants
$\kappa_{\mathrm{node}}>0$ and
\[
 c_{\mathrm{node}}>
 \max\left\{1,\frac{1+\mu}{1-\mu}
                    \frac{C_2^{1/4}}{\kappa_{\mathrm{node}}}\right\}
\]
such that
 $\kappa(\bar s,c_{\mathrm{node}},\mathbb T,\bm M(\cdot))
 \wedge\kappa(\bar s,c_{\mathrm{node}},\mathbb T,\bm N(\cdot))
                       \ge\kappa_{\mathrm{node}}.$
The adaptive restricted eigenvalues are defined as in Assumption~\ref{as2}, in dimension $p$, with the regime matrices replacing $\bm\Sigma$ and index sets contained in $\{1,\ldots,p\}$.

\noindent (ii) The residual second moments and the second moments of their products with individual covariates satisfy
\begin{align*}
\sup_{\tau\in\mathbb T}
\max_{\substack{
1\leq i\leq n,\;j\in H_A\\
1\leq l\leq p}}
E\!\left[
(\upsilon_i^{(j)}(\tau))^2
+
(X_i^{(l)}\upsilon_i^{(j)}(\tau))^2
\right]
&\leq C_{\mathrm{node}},\\
\sup_{\tau\in\mathbb T}
\max_{\substack{
1\leq i\leq n,\;j\in H_B\\
1\leq l\leq p}}
E\!\left[
(\widetilde\upsilon_i^{(j)}(\tau))^2
+
(X_i^{(l)}\widetilde\upsilon_i^{(j)}(\tau))^2
\right]
&\leq C_{\mathrm{node}},
\end{align*}
for a common finite constant $C_{\mathrm{node}}$.

\noindent (iii) For a common finite constant $C_E$ and all sufficiently large $n$,
$b_n\{E M_{\upsilon X}^2\}^{1/2}\le C_E,$ and $b_n\{E M_{\widetilde\upsilon X}^2\}^{1/2}\le C_E, $ where $M_{\upsilon X} =\max_{i\le n,\,j\in H_A,\,l\le p}
       \sup_{\tau\in\mathbb T}|X_i^{(l)}\upsilon_i^{(j)}(\tau)|$, $
 M_{\widetilde\upsilon X}
  =\max_{i\le n,\,j\in H_B,\,l\le p}
       \sup_{\tau\in\mathbb T}|X_i^{(l)}\widetilde\upsilon_i^{(j)}(\tau)|.$
Maxima over empty row sets are defined as zero.
\end{assm}

Assumption \ref{as6-nodewise} (i) imposes the adaptive restricted eigenvalue condition at the nodewise sparsity level uniformly over the threshold parameter space. Together with Assumptions~\ref{as1} and \ref{as2}(i), it implies that the population nodewise coefficients are uniformly bounded in the $\ell_2$ norm and
that their average residual second moments are uniformly bounded away from zero for the required rows. Assumptions \ref{as6-nodewise} (ii)-(iii) impose the moment and maximum envelope conditions used to establish uniform concentration for the nodewise regressions.

\begin{lem}\label{lem:nodewise-main}
Suppose that Assumptions~\ref{as1}, \ref{as2}(i), and \ref{as6-nodewise} hold. Let $\lambda_{\mathrm{node}}=A_{\mathrm{node}}b_n$, where $A_{\mathrm{node}}>0$ is a sufficiently large fixed constant. Then
\begin{align}
 \max_{j\in H}\sup_{\tau\in\mathbb T}
       |\widehat\Theta_j(\tau)-\Theta_j(\tau)|_1
     &=O_P(\bar s b_n),\label{eq:node-l1}\\
 \max_{j\in H}\sup_{\tau\in\mathbb T}
       |\widehat\Theta_j(\tau)-\Theta_j(\tau)|_2
     &=O_P(\sqrt{\bar s}\,b_n),\label{eq:node-l2}\\
 \max_{j\in H}\sup_{\tau\in\mathbb T}|\widehat\Theta_j(\tau)|_1
     &=O_P(\sqrt{\bar s}),\label{eq:node-row}\\
 \max_{j\in H}\sup_{\tau\in\mathbb T}
       |\widehat\Theta_j(\tau)'\widehat{\bm\Sigma}(\tau)-e_j'|_\infty
     &=O_P(b_n).\label{eq:node-kkt}
\end{align}
\end{lem}

Because \eqref{eq:node-kkt} holds uniformly over $\mathbb T$, the same bound applies at the estimated threshold $\widehat\tau$. Combining this result with $|g|_1\leq\sqrt h$ and the coefficient estimation bound in Theorem~\ref{main-thm-case1} or Theorem~\ref{thmftau} yields
\[
 |g'\Delta(\widehat\tau)|
 =O_P\!\left(\sqrt{nh}\,s_0\lambda_n b_n\right)
 =O_P\!\left(\frac{\sqrt h\,s_0\ell_n}{\sqrt n}\right).
\]
Thus $\sqrt h\,s_0\ell_n/\sqrt n\to0$ is sufficient for the debiasing remainder to be asymptotically negligible. This bound does not, by itself, establish a central limit theorem for the leading score evaluated at the data-dependent threshold $\widehat\tau$.

\subsection{Uniform Inference for the Debiased Threshold Lasso Estimator} \label{uniforminference}

We establish asymptotically Gaussian inference for a deterministic linear contrast of the slope coefficients, both in the absence of a threshold
effect and under a uniformly identified fixed-threshold effect. The same estimator applies in both cases, without pretesting for a threshold effect. The anchoring penalty ensures that $\widehat\tau$ converges
to the prespecified reference point $\tau_\ast$ when $\delta_0=0$, while $\widehat\tau$ converges to $\tau_0$ in the identified case. The proof explicitly accounts for the dependence arising from using the same observations to estimate the threshold and construct the score.

We keep the notation for the deterministic contrast $g$, its support $H$, the support cardinality $h$, and the inverse row sparsity $\bar s$ introduced in Section~ \ref{sec:debiasedest}. For each deterministic $\tau\in\mathbb T$, define
$
 \bm\Sigma_{xu}(\tau)
 =1/n\sum_{i=1}^n E[\bm X_i(\tau)\bm X_i(\tau)'U_i^2],$ and $
 \mathcal V(\tau)=\bm\Theta(\tau)\bm\Sigma_{xu}(\tau)\bm\Theta(\tau)'.$
The notation $\bm\Sigma_{xu}(\widehat\tau)$ denotes evaluation of the population matrix function $\Sigma_{xu}(\tau)$ at $\widehat\tau$. Thus, the expectations are taken at a fixed threshold before the resulting function is evaluated at the estimated threshold.

\begin{assm}\label{asnd}
\noindent (i) The moments $\max_{i,j}E\left|X_i^{(j)}\right|^{12}$ and
$\max_i E|U_i|^8$ are uniformly bounded. Moreover,
$b_n\left\{EM_{X^6}^2\right\}^{1/2}=o(1),$ $
 b_n\left\{EM_{X^2U^2}^2\right\}^{1/2}=o(1),$ $
 b_n\left\{EM_{X^4U^2}^2\right\}^{1/2}=o(1), $
where $ M_{X^6}=\max_{i,k,l,j}\left|\left(X_i^{(k)}X_i^{(l)}X_i^{(j)}\right)^2\right|,$
$M_{X^2U^2}=\max_{i,j,l}\left|X_i^{(j)}X_i^{(l)}U_i^2\right|, $
$M_{X^4U^2}=\max_{i,j,l}\left|\left(X_i^{(j)}X_i^{(l)}U_i\right)^2\right|.$

\noindent (ii) The dimensionality and sparsity sequences satisfy
\[
 \frac{h^{3/2}s_0^2\bar s^2\ell_n}{\sqrt n}=o(1),
 \qquad \frac{(h\bar s)^2}{n}=o(1).
\]

\noindent (iii) Let $t=\tau_\ast$ when $\delta_0=0$
and $t=\tau_0$ when $\delta_0\neq0$. The largest eigenvalues of $\bm\Sigma(\tau)$ and
$\bm\Sigma_{xu}(\tau)$ are uniformly bounded over
$\tau\in\mathbb T$. For the contrast under consideration, $g'\mathcal V(t)g\ge c_{\mathcal V}>0$. The local continuity and boundedness conditions in Assumption~\ref{as1}(iii)--(iv) hold on $\{q\in[0,1]:|q-t|\le\varepsilon_0\}$.
\end{assm}

Assumption~\ref{asnd}(i) imposes finite-moment conditions that permit non-sub-Gaussian covariates and errors, while its maximum envelope restrictions control the empirical product moments used in variance estimation. Assumption~\ref{asnd}(ii) restricts the joint growth of the covariates dimension $p$, the number of nonzero weights $h$, the inverse-row sparsity $\bar s$, and the slope sparsity $s_0$. Together with the stated moment
and design conditions, its second restriction ensures the Lyapunov condition for the leading score at the relevant deterministic reference point $t$. Assumption~\ref{asnd}(iii) bounds the largest eigenvalues of the population design and score second-moment matrices and ensures that the variance of the contrast remains bounded away from zero. It also imposes local continuity and boundedness around $\tau_\ast$ in the no-threshold case and around $\tau_0$ in the identified case. These local conditions are used, together with threshold localization, to control the effect of evaluating the score at $\widehat\tau$ rather than at $t$.

 %\[\mathcal{A}^{({1})}_{\ell_0}(s_0)=\left\{\alpha_0\in	\mathbb{R}^{2p} \mid|\alpha_0|_{\infty}\le C,\mathcal{M}(\alpha_0)\le s_0, \delta_0=0\right\},\] and
%\[\mathcal{A}^{({2})}_{\ell_0}(s_0)=\left\{\alpha_0\in	\mathbb{R}^{2p} \mid|\alpha_0|_{\infty}\le C,\mathcal{M}(\alpha_0)\le s_0, \delta_0\ne0\right\}.\]

Let $ \widehat U_i(\tau)=Y_i-\bm X_i(\tau)'\widehat\alpha(\tau),$ $\widehat{\bm\Sigma}_{xu}(\tau) =1/n\sum_{i=1}^n\bm X_i(\tau)\bm X_i(\tau)'\widehat U_i(\tau)^2,$ and $\widehat{\mathcal V}_n(\tau) =\widehat{\bm\Theta}(\tau)\widehat{\bm\Sigma}_{xu}(\tau)\widehat{\bm\Theta}(\tau)'.$ As in the preceding oracle inequalities, uniformity is over
distributions satisfying the stated assumptions with common constants and uniform rate bounds. Uniformity over $\mathcal B_{\ell_0}(s_0)=\mathcal A_{\ell_0}^{(1)}(s_0) \cup\mathcal A_{\ell_0}^{(2)}(s_0)$ is understood subject to these restrictions. Write $P_{\alpha_0}$ for a qualifying data generating distribution with coefficient vector $\alpha_0$.

%%%%%%%%%%%%%%%%%%%%%%%%%%%%%%%%%%%%%
\begin{thm}\label{thm3}
Suppose that Assumptions~\ref{as1}, \ref{as6-nodewise}, and \ref{asnd} hold, with tuning parameters specified in \eqref{eq:tuning} and
Section~\ref{constructprecison}. Assume that Assumption~\ref{as2}(ii) holds over $\mathbb T$ when
$\delta_0=0$ and over the neighborhood of $\tau_0$ specified in Section~\ref{fixedthreshold} when $\delta_0\ne0$. In the latter case, suppose additionally that Assumptions~\ref{A-discontinuity}, \ref{A-smoothness}, and \ref{A-secondmoments} hold.
Let $g\in\mathbb R^{2p}$ be a deterministic contrast satisfying $|g|_2=1$, as defined in Section~\ref{sec:debiasedest}.

Then the debiased threshold Lasso estimator satisfies
\begin{equation}\label{asymdist}
\frac{\sqrt n\,g'\{\widetilde a(\widehat\tau)-\alpha_0\}}
     {\{g'\widehat{\mathcal V}_n(\widehat\tau)g\}^{1/2}}
\overset{d}{\longrightarrow}N(0,1)
\end{equation}
uniformly over
$\alpha_0\in \mathcal B_{\ell_0}(s_0)$. Moreover, for every $\varepsilon>0$,
\begin{align}
\sup_{\alpha_0\in\mathcal A_{\ell_0}^{(1)}(s_0)}
P_{\alpha_0}\!\left(
\left|
g'\{\widehat{\mathcal V}_n(\widehat\tau)
     -\mathcal V(\tau_\ast)\}g
\right|>\varepsilon
\right)
&\longrightarrow0,
\label{asymcov}\\
\sup_{\alpha_0\in\mathcal A_{\ell_0}^{(2)}(s_0)}
P_{\alpha_0}\!\left(
\left|
g'\{\widehat{\mathcal V}_n(\widehat\tau)
     -\mathcal V(\tau_0)\}g
\right|>\varepsilon
\right)
&\longrightarrow0
\label{asymcov2}.
\end{align}
The same consistency statements hold when the population variance at the deterministic reference point is replaced by $\mathcal V(\widehat\tau)$ in each case.
\end{thm}

Theorem~\ref{thm3} establishes a common standard normal approximation for linear contrasts in the no-threshold and fixed-threshold cases. The number of nonzero contrast weights may increase with $n$, subject to Assumption~\ref{asnd}(ii). The respective variance targets are $g'\mathcal V(\tau_\ast)g$ and $g'\mathcal V(\tau_0)g$.
These targets may differ, and the normal approximation does not require either variance sequence to converge to a fixed numerical limit. In the no-threshold case, the argument combines the localization rate $|\widehat\tau-\tau_\ast|=O_P(s_0\lambda_n)$ from
Theorem~\ref{main-thm-case1} with local stochastic equicontinuity, while the exact selection of $\tau_\ast$ is not required. The standard normal approximation holds uniformly within each subclass and therefore over their union, allowing coefficient inference without pretesting for a threshold effect.

%%%%%%%%%%%%%%%%%%%%%%%%%%%%%%%%%%%%%

%Furthermore, the asymptotic result can be applied to test for a threshold effect based on the active parameters. Define the active set as $\widehat J=\{j:\widehat{\delta}_j \neq 0\}$ and $\widehat s = |\widehat J|,$ where $\widehat s$ can grow with $n$ and is of the same magnitude as $s_0.$ Recall that $H$ is the index set of $g;$ let $H = \widehat J,$ and assign equal weights to the non-zero elements of $g.$ The null hypothesis is $H_0: \delta_{\widehat{J},0}=0.$ Subsequently, a t-test can be performed on $g'\widetilde{a}(\widehat{\tau})$ to test for the existence of a threshold effect.

The preceding scalar result also provides a basis for joint inference on a fixed number of prespecified coefficients. Specifically, let $H$ be a deterministic index set of fixed cardinality $h$, and impose the conditions of Theorem~\ref{thm3} for every deterministic unit contrast supported on $H$, with constants common to these contrasts. In particular, the smallest eigenvalue of $\mathcal V(\tau_\ast)_{H,H}$ in the no-threshold case, or of $\mathcal V(\tau_0)_{H,H}$ in the identified case, must be bounded away from zero. The corresponding vector central limit theorem and covariance
consistency imply
\begin{equation}\label{eq:fixed-h-wald}
 n\left\{\widetilde a(\widehat\tau)_H-\alpha_{0,H}\right\}'
 \left[\widehat{\mathcal V}_n(\widehat\tau)_{H,H}\right]^{-1}
 \left\{\widetilde a(\widehat\tau)_H-\alpha_{0,H}\right\}
 \overset{d}{\longrightarrow}\chi_h^2.
\end{equation}
This statistic provides a Wald test of a joint hypothesis. %A scalar contrast with growing support does not by itself establish a growing dimensional joint test.

We next consider confidence intervals for individual coefficients. Let $\Phi$ denote the standard normal distribution function and $z_{1-\alpha/2}$ its $(1-\alpha/2)$ quantile, where $\alpha\in(0,1)$ is the nominal significance level. For $j\in\{1,\ldots,2p\}$, define $\widehat\sigma_j(\tau)=\{e_j'\widehat{\mathcal V}_n(\tau)e_j\}^{1/2}, $ and $ I_{j,n}(1-\alpha)
 =\left[\widetilde a^{(j)}(\widehat\tau)
\pm z_{1-\alpha/2}{\widehat\sigma_j(\widehat\tau)}/{\sqrt n}\right].$

\begin{thm}\label{thm4}
Suppose that the conditions of Theorem~\ref{thm3} hold for a deterministic contrast $g$ satisfying $|g|_2=1$ in the respective no-threshold and fixed-threshold cases. Then
\begin{equation}\label{eq:thm4-normal}
\sup_{t\in\mathbb R}
\sup_{\alpha_0\in\mathcal B_{\ell_0}(s_0)}
\left|
P_{\alpha_0}\!\left\{
\frac{\sqrt n\,g'\{\widetilde a(\widehat\tau)-\alpha_0\}}
     {\{g'\widehat{\mathcal V}_n(\widehat\tau)g\}^{1/2}}
\le t
\right\}
-\Phi(t)
\right|
\longrightarrow0.
\end{equation}

Fix $\alpha\in(0,1)$ and a deterministic coordinate
$j\in\{1,\ldots,2p\}$. If the conditions of Theorem~\ref{thm3} hold with $g=e_j$, $h=1$, and the
corresponding inverse row sparsity $\bar s$, then
\begin{equation}\label{eq:thm4-coverage}
\lim_{n\to\infty}
\inf_{\alpha_0\in\mathcal B_{\ell_0}(s_0)}
P_{\alpha_0}\!\left\{
\alpha_0^{(j)}\in I_{j,n}(1-\alpha)
\right\}
=1-\alpha.
\end{equation}
Moreover, the interval length is $O_P(n^{-1/2})$ uniformly over $\mathcal B_{\ell_0}(s_0)$, in the sense that 
\begin{equation}\label{eq:thm4-length}
\lim_{K\to\infty}\limsup_{n\to\infty}
\sup_{\alpha_0\in\mathcal B_{\ell_0}(s_0)}
P_{\alpha_0}\!\left\{
\sqrt n\,\operatorname{diam}\{I_{j,n}(1-\alpha)\}>K
\right\}
=0.
\end{equation}
\end{thm}

Theorem~\ref{thm4} develops the uniform Gaussian approximation and establishes uniformly asymptotically
exact coverage for each prespecified coefficient. The same confidence interval construction applies across the no-threshold and uniformly identified fixed-threshold classes, without pretesting for a threshold effect. These intervals are asymptotically honest over the qualifying model class and their lengths contract at the usual parametric rate $n^{-1/2}$.

\begin{rem}
Although the theory is developed for independent cross-sectional data, the estimation procedure can be adapted to static panel models with individual fixed effects and strictly exogenous regressors. Individual fixed effects can be removed by applying the within transformation to the outcome and threshold-expanded regressors. Heteroskedasticity can be accommodated by allowing error variances to differ across individuals and periods. Extending the inference results requires suitable conditions on the transformed design, accounting for within-individual dependence, including that induced by the within transformation, and establishing consistency of the corresponding covariance estimator.
\end{rem}
%%%%%%%%%%%%%%%%%%%%%%%%%%%%%%%%%%%%%%%%%%%%%%%%%%%%%%%%%%%%%%%%%%%%%%%%%%%%%%%%%%%%%

%%%%%%%%%%%%%%%%%%%%%%%%%%%%%%%%%%%%%%%%%%%%%%%%%%%%%%%%%%%%%%%%%%
\subsection{Testing for a Threshold Effect}
\label{sec:threshold_test}

We test whether the regression coefficients change across regimes: $H_0:\delta_0=0$ against $H_1:\delta_0\neq0$.
Under $H_0$, model \eqref{model} reduces to the high-dimensional linear regression $Y_i=X_i'\beta_0+U_i$, and the threshold parameter is unidentified. We therefore construct a statistic that maximizes over candidate thresholds and coefficient coordinates. Taking absolute values allows nonzero threshold coefficients to have different signs.

Let $\mathcal G_n\subset\mathbb T=[t_0,t_1]$ denote the testing grid, with $K_n=|\mathcal G_n|$, and $d_n=pK_n$. We consider either a prespecified deterministic grid of at most polynomial cardinality in $n$, or the observed value grid
$\mathcal G_n=(\{Q_1,\ldots,Q_n\}\cap\mathbb T)\cup\{t_1\}.$
Under the indicator convention $\bm{1}\{Q_i<\tau\}$, including
$t_1$ ensures that this observed-value grid represents every
sample partition generated by $\tau\in\mathbb T$, almost surely. To estimate the score variances, first fit the linear model
restricted by the null hypothesis:
\begin{equation}\label{eq:test-null-fit}
\widehat\beta^{\mathrm{null}}
\in\operatorname*{arg\,min}_{\beta\in\mathbb R^p}
\left\{
\frac{1}{n}\sum_{i=1}^n(Y_i-X_i'\beta)^2
+\lambda_{\mathrm{null}}|\bm D_{\mathrm{null}}\beta|_1
\right\},
\end{equation}
where $\bm D_{\mathrm{null}} :=\operatorname{diag}
\{\|X^{(j)}\|_n:j=1,\ldots,p\}, $ $\lambda_{\mathrm{null}}=A_{\mathrm{null}}b_n,$ and $A_{\mathrm{null}}>0$ is a sufficiently large fixed constant.
These weights are empirical root mean squares, consistent with
the coefficient and nodewise Lasso constructions. Let
$\widehat U_i^{\mathrm{null}} :=Y_i-X_i'\widehat\beta^{\mathrm{null}}$ denote the resulting
null-restricted residuals.

For the test, we extend the fixed-threshold construction in
\eqref{LASSO-fixed-tau} to every $\tau\in\mathbb T$ by solving the same coefficient minimization problem. At each $\tau\in\mathcal G_n$, compute
\begin{equation}\label{eq:test-debiased}
\widetilde a(\tau)
=\widehat\alpha(\tau)
+\widehat{\bm\Theta}(\tau)
\frac{\bm X(\tau)'\{\bm Y-\bm X(\tau)\widehat\alpha(\tau)\}}{n},
\end{equation}
and set $\widetilde a_{\delta,j}(\tau)
:=e_{p+j}'\widetilde a(\tau)$ for $j=1,\ldots,p$.
The anchoring penalty is constant with respect to $\alpha$ at
a fixed $\tau$, so it does not change these coefficient fits
or their debiasing formula. The test evaluates the entire testing
grid rather than the selected threshold $\widehat\tau$.

For $j=1,\ldots,p$ and $\tau\in\mathcal G_n$, define
$\widehat m_{i,j}(\tau)
  =\widehat\Theta_{p+j}(\tau)'\bm X_i(\tau)\widehat U_i^{\mathrm{null}}, $ $m_j(\tau)=n^{-1}\sum_{i=1}^n\widehat m_{i,j}(\tau), $ and estimate the score variance by 
 \begin{equation}\label{eq:test-null-variance}
\widehat\sigma_{\delta,j}^2(\tau)
:=\frac{1}{n}\sum_{i=1}^n
\{\widehat m_{i,j}(\tau)-m_j(\tau)\}^2.
\end{equation}
Here $\widehat\Theta_r(\tau)'$ denotes row $r$ of
$\widehat{\bm\Theta}(\tau)$, as in Section~\ref{sec:debiasedest}. The test statistic is
\begin{equation}\label{eq:test-statistic}
 T_n=\max_{\tau\in\mathcal G_n}\max_{1\le j\le p}
       \left|\frac{\sqrt n\,\widetilde a_{\delta,j}(\tau)}
                         {\widehat\sigma_{\delta,j}(\tau)}\right|.
\end{equation}

We obtain critical values using a Gaussian multiplier bootstrap. Define the centered and standardized scores
$\widehat\psi_{i,j}(\tau)=
\{\widehat m_{i,j}(\tau)- m_j(\tau)\}/
\widehat\sigma_{\delta,j}(\tau)$. Conditional on the observed sample, draw independent $\iota_1,\ldots,\iota_n\sim N(0,1)$ and compute
\begin{equation}\label{eq:test-bootstrap}
    T_n^*
    =
    \max_{\tau\in\mathcal G_n}
    \max_{1\le j\le p}
    \left|
    \frac{1}{\sqrt n}\sum_{i=1}^n
    \iota_i\widehat\psi_{i,j}(\tau)
    \right|.
\end{equation}
Let $P^*(\cdot)$ denote probability conditional on the observed sample and define $c_{1-\zeta}^*=\inf\{t:P^*(T_n^*\le t)\ge1-\zeta\}$, for a fixed $\zeta\in(0,1)$. The test rejects $H_0$ when $T_n>c_{1-\zeta}^*$. Appendix Section~\ref{sec:nulltest} describes the Monte Carlo
implementation and the convention for zero estimated variances.

The nodewise conditions are imposed for $H=H_\delta$, with
$\bar s=\bar s_\delta$. The additional conditions in Assumption~\ref{as:max-test} require nondegenerate score variances, uniform control of coefficient, nodewise, and covariance estimation errors, and finite-moment bounds supporting a Gaussian approximation for the maximum statistic. Their rate restrictions are specific to this test and do not require setting $h=p$ in the scalar-contrast conditions of Assumption~\ref{asnd}(ii).

\begin{prop}\label{prop:threshold-size}
Suppose that Assumptions~\ref{as1}, \ref{asnd}(i), and Assumption~\ref{as:max-test} hold. Assume that
Assumption~\ref{as2} holds with $\mathbb S=\mathbb T$ and the
constant $c_0$ specified in Theorem~\ref{main-thm-case1},
and that Assumption~\ref{as6-nodewise} holds for $H=H_\delta$.
Let the tuning parameters satisfy \eqref{eq:tuning} and $
\lambda_{\mathrm{node}}=A_{\mathrm{node}}b_n,$ $
\lambda_{\mathrm{null}}=A_{\mathrm{null}}b_n,$
with $A$, $A_{\mathrm{node}}$, and $A_{\mathrm{null}}$
sufficiently large and fixed. Then
\begin{equation}\label{eq:test-size}
\sup_{\alpha_0\in\mathcal A_{\ell_0}^{(1)}(s_0)}
\left|
P_{\alpha_0}(T_n>c_{1-\zeta}^*)-\zeta
\right|
\longrightarrow0.
\end{equation}
\end{prop}

\begin{rem}
Proposition~\ref{prop:threshold-size} establishes uniform asymptotic size $\zeta$ despite the lack of threshold identification under $H_0$. Maximization over $\tau$ accounts for the unidentified threshold, while maximization over $j$ accommodates uncertainty about which coefficients differ across regimes. The multiplier bootstrap captures dependence across both dimensions. Its validity requires a Gaussian approximation for the entire maximum because the scalar central limit theorem in Theorem~\ref{thm3} alone is insufficient.
\end{rem}

%%%%%%%%%%%%%%%%%%%%%%%%%%%%%%%%%%%%%%%%%
\begin{prop}\label{prop:threshold-power}
Suppose that Assumptions~\ref{as1}--\ref{as2},
Assumption~\ref{as6-nodewise} for $H=H_\delta$, Assumption~\ref{as:max-test}(i), and \eqref{eq:max-linear-rate}
hold along a sequence of alternatives with $\delta_0\ne0$. Let the tuning parameters be as specified in Proposition~\ref{prop:threshold-size}. In Assumption~\ref{as2}(ii), take $c_0$ to be the constant in Theorem~\ref{main-thm-case1},
and let $\mathbb S$ contain $\tau_0$ and, under condition~(ii) below, $\tau_n$.

Suppose that there exist a deterministic coordinate
$j_n\in\{1,\ldots,p\}$, a candidate $\tau_n\in\mathcal G_n$,
and a deterministic sequence $\bar\sigma_n$ satisfying
$0<c_\sigma\le\bar\sigma_n<\infty$ for a fixed constant $c_\sigma>0$,
such that one of the following conditions holds:

\noindent (i) With probability approaching one,
\begin{equation}\label{eq:test-same-partition}
\bm 1\{Q_i<\tau_n\}
=\bm 1\{Q_i<\tau_0\},
\qquad i=1,\ldots,n.
\end{equation}

\noindent (ii) The candidate $\tau_n$ is deterministic, and the empirical local design inequality in Assumption~\ref{A-smoothness} holds at $\eta_n=|\tau_n-\tau_0|\vee b_n^2$, with $\eta_n$ in the range of radii covered by that assumption. In addition,
\begin{equation}\label{eq:test-grid-rate}
\sqrt n\sqrt{\bar s_\delta}\,|\delta_0|_1
\bigl(|\tau_n-\tau_0|\vee b_n^2\bigr)
\longrightarrow0.
\end{equation}
If
\[
P_{\alpha_0}\!\left\{
0<\widehat\sigma_{\delta,j_n}(\tau_n)\le\bar\sigma_n
\right\}
\longrightarrow1,
\qquad
\frac{\sqrt n\,|\delta_{0,j_n}|}{\bar\sigma_n}
-\sqrt{2\log(2d_n/\zeta)}
\xrightarrow{P}+\infty,
\]
then $P_{\alpha_0}(T_n>c_{1-\zeta}^*)\longrightarrow1.$
\end{prop}
%%%%%%%%%%%%%%%%%%%%%%%%%%%%%%%%%%%%%%%%
\begin{rem}
If $\bar\sigma_n=O(1)$ and the variance and grid conditions hold for a deterministic coordinate attaining $|\delta_0|_\infty$, a sufficient signal condition is
\[
 \frac{\sqrt n\,|\delta_0|_\infty}{\sqrt{\log(2d_n)}}
                       \xrightarrow{P}+\infty.
\]
The explicit signal gap condition in Proposition~\ref{prop:threshold-power} also permits fixed multiples of $\sqrt{\log d_n/n}$ when $d_n\to\infty$, provided that the multiple exceeds the variance dependent critical value constant by a fixed positive margin. If $p$ and $K_n$ grow at most polynomially in $n$, a signal larger in order than
$\sqrt{\log n/n}$ is sufficient, subject to the remaining
conditions. These are sufficient conditions for consistency,
not a claim of minimax optimality. Nonzero threshold coefficients
may have different signs, and the proposition does not impose
a common positive lower bound on $|\delta_0|_2$.
\end{rem}

\begin{rem}
For the observed value grid including $t_1$, the candidate $\tau_n=\min\bigl(\{Q_i:Q_i\in[\tau_0,t_1]\}\cup\{t_1\}\bigr)$
generates the same sample partition as $\tau_0$, almost surely. Thus \eqref{eq:test-same-partition} holds without requiring \eqref{eq:test-grid-rate}. For a coarser deterministic grid, \eqref{eq:test-grid-rate} controls the bias from misclassified observations. The term $b_n^2$ is the smallest radius covered by the local design bound in Assumption~\ref{A-smoothness}.
\end{rem}

\begin{rem}
The threshold effect test and the slope inference procedure in Sections \ref{sec:debiasedest} and \ref{uniforminference} address different hypotheses. Theorems~\ref{thm3} and \ref{thm4} provide
inference for prespecified slope contrasts or coefficients over the model classes, using a common anchored threshold Lasso estimator. The present test addresses the joint hypothesis $H_0:\delta_0=0$ by maximizing over thresholds and coefficient coordinates. It is not a prerequisite for constructing the confidence intervals in Section~\ref{uniforminference}.
\end{rem}

%%%%%%%%%%%%%%%%%%%%%%%%%%%%%%%%%%%%%%%%%
\section{Monte Carlo Simulation and Applications}\label{simemp}
\subsection{Monte Carlo Simulation}\label{simmo}

We examine the finite-sample performance of the debiased threshold
Lasso estimator under model \eqref{model}. The covariate
vectors are drawn independently from $N(0,\bm{\Sigma})$, where
$\Sigma_{jk}=0.5^{|j-k|}$. The regression errors follow Student's
$t$ distribution with 10 degrees of freedom and are independent
across observations and independent of $(X_i,Q_i)$.
In the independent design, $Q_i$ is drawn from the uniform distribution on $[0,1]$, independently of $X_i$. In the dependent
design, we generate
\[
Q_i=\Phi\!\left(
0.5\,\widetilde X_{i1}+\sqrt{0.75}\,\varepsilon_i
\right),
\]
where $\widetilde X_{i1}$ denotes the first covariate standardized
using its sample mean and standard deviation, $\varepsilon_i$
are independent standard normal draws independent of the
covariates and regression errors, and $\Phi$ is the standard
normal distribution function. This construction introduces
dependence between the threshold variable and the covariates.

The first 15 entries of $\beta_0$ equal $b_1=2$, entries 16--30
of $\delta_0$ equal $b_2$, and all remaining entries are zero.
The baseline design sets $n=800$, $p=500$, $b_2=1$, and
$\tau_0=0.5$, with $Q_i$ independent of $X_i$. The six additional designs change the settings indicated in Table~\ref{tbl:confidence}, while keeping the remaining parameters at their baseline values. The threshold search grid is
$\{0.15,0.16,\ldots,0.85\}$.

We use the generalized information criterion (GIC) of \citet{GIC}
to select an initial Lasso penalty that is common across candidate thresholds. For each candidate penalty $\lambda$, we obtain $\widetilde\tau(\lambda)$ by minimizing the criterion in \eqref{tau-max}, with $\tau_\ast=0.3$ and
$\rho_n=\lambda$ ($c_\rho=1$). We then select
\[
\widehat\lambda
=\arg\min_{\lambda}
\left[
\log\!\left\{
\frac{\operatorname{RSS}(\widetilde\tau(\lambda);\lambda)}{n}
\right\}
+\frac{\log(2p)\log\log n}{n}\,
df\{\widetilde\tau(\lambda),\lambda\}
\right],
\]
where $df\{\widetilde\tau(\lambda),\lambda\}$ denotes the number of nonzero fitted coefficients. The final threshold estimate is
$\widehat\tau=\widetilde\tau(\widehat\lambda)$, with
$\rho_n=\widehat\lambda$. We construct an approximate inverse of the Gram matrix using nodewise Lasso and conduct coefficient
inference using heteroskedasticity-robust standard errors. GIC selects a separate penalty for each nodewise regression, and separately determines the initial Lasso penalties for the linear
benchmark and the null-restricted fit in the threshold-effect test. The unrestricted fits for that test use the penalty selected
for the threshold model.

For each configuration, we conduct 500 Monte Carlo replications.
Let $S=\{j:\beta_{0,j}\ne0\}$, with $S^c$ defined within
$\{1,\ldots,p\}$, and let $J_j^{(r)}(\beta)$ denote the reported
95\% confidence interval for $\beta_{0,j}$ in replication
$r$.\footnote{The intervals include an adjustment to the Wald
half-width adapted from the \texttt{addlength} step in the R code
accompanying \citet{javanmard2014}
(\url{https://web.stanford.edu/~montanar/sslasso/code.html}).
The adjustment uses the selected penalty and the number of nonzero fitted coefficients from each nodewise regression.}
For $A\in\{S,S^c\}$, we report average interval length and empirical coverage, defined as
\[
\ell_A=\frac{1}{500|A|}\sum_{r=1}^{500}\sum_{j\in A}
\operatorname{length}\left(J_j^{(r)}(\beta)\right),
\qquad
\widehat{\operatorname{Cov}}_A
=\frac{1}{500|A|}\sum_{r=1}^{500}\sum_{j\in A}
\mathbf{1}\left\{\beta_{0,j}\in J_j^{(r)}(\beta)\right\}.
\]
The unsubscripted quantities $\ell$ and $\widehat{\operatorname{Cov}}$ denote the corresponding averages over all $p$ coefficients. Table~\ref{tbl:confidence} also reports the mean absolute threshold estimation error.

\FloatBarrier
\begin{table*}[t]
\centering
\caption{Anchored Threshold Estimation and Debiased Threshold Lasso Confidence Intervals}
\label{tbl:confidence}
\small
\setlength{\tabcolsep}{3.4pt}
\begin{tabular}{lccccccc}
\toprule
Changed designs
& $|\widehat\tau-\tau_0|$
& $\ell$ & $\ell_S$ & $\ell_{S^c}$
& $\widehat{\operatorname{Cov}}$
& $\widehat{\operatorname{Cov}}_S$
& $\widehat{\operatorname{Cov}}_{S^c}$ \\
\midrule
Baseline
& 0.00000 & 0.266 & 0.266 & 0.266 & 0.979 & 0.955 & 0.980 \\
$p=300$
& 0.00006 & 0.263 & 0.263 & 0.263 & 0.976 & 0.954 & 0.977 \\
$b_2=0.5$
& 0.00016 & 0.265 & 0.265 & 0.265 & 0.979 & 0.959 & 0.979 \\
Dependent $Q$
& 0.00010 & 0.265 & 0.265 & 0.265 & 0.978 & 0.951 & 0.979 \\
$\tau_0=0.4$
& 0.00008 & 0.252 & 0.252 & 0.252 & 0.980 & 0.953 & 0.981 \\
$b_2=0$
& n.a. & 0.242 & 0.242 & 0.242 & 0.982 & 0.950 & 0.983 \\
$n=1000,\ p=600$
& 0.00004 & 0.238 & 0.238 & 0.238 & 0.978 & 0.956 & 0.979 \\
\bottomrule
\end{tabular}

\begin{flushleft}
\footnotesize
Notes: The baseline sets $n=800$, $p=500$, $b_1=2$, $b_2=1$, and $\tau_0=0.5$, with $Q_i$ independent of $X_i$. The remaining rows indicate departures from the baseline, and all other parameters are unchanged. The second column reports mean absolute threshold estimation error, which is not reported when $b_2=0$ because the threshold is unidentified. Interval results concern $\beta_0$. Subscripts $S$ and $S^c$ refer to nonzero and zero coefficients, respectively.
\end{flushleft}
\end{table*}

Table~\ref{tbl:confidence} shows that threshold estimation is
accurate across the designs with a threshold effect. Reducing the magnitude of the threshold effect slightly increases the mean absolute estimation error, although the true threshold is still selected in nearly all replications. This reduction has little effect on interval length or coverage for $\beta_0$. In the no-threshold design, the estimator selects the reference point $\tau_\ast$ in almost all replications.

Average intervals are shorter in the $(n,p)=(1000,600)$ design
than in the baseline, where the sample size and covariate
dimension are both smaller. The $p=300$ design also produces
slightly shorter intervals. Interval lengths are lower in the
designs with $\tau_0=0.4$ or no threshold effect, whereas
dependence between $Q_i$ and $X_i$ has little effect on the
reported results. Within each design, average interval lengths
are similar for zero and nonzero coefficients.

Coverage for nonzero coefficients in $\beta_0$ remains close
to the nominal level across all designs, including the no-threshold case.\footnote{Without the half-width adjustment,
coverage for nonzero coefficients is about 90\% for $\beta_0$
and 87.4--90.2\% for $\delta_0$. Coverage for zero coefficients
in both vectors is close to 95\%.} Coverage for zero coefficients is consistently above the nominal level. Because most coefficients of $\beta_0$ are zero, their conservative coverage largely determines the overall average.
%%%%%%%%%%%%%%%%%%%%%%%%%%%%%%%%%%%%%%%%%%%%%%%%%%%%%%%%%%%%%%

Table~\ref{tbl:FWER} compares coefficient inference based on the
debiased threshold Lasso with a linear benchmark fitted to the same simulated samples. The threshold method estimates all $2p$
coefficients in $\alpha_0=(\beta_0',\delta_0')'$. The benchmark fits a linear Lasso regression and debiases its $p$ coefficient estimates, omitting the threshold component. We evaluate the benchmark estimates against $\beta_0$ to assess the consequence of ignoring threshold effects.

Following \citet{javanmard2014}, we report false positive rates
(FP) and true positive rates (TP), together with the familywise
error rate (FWER) of Bonferroni-adjusted tests. FP and TP are
the average proportions of zero and nonzero coefficients,
respectively, rejected by individual tests at the 5\% level. FWER is the proportion of replications with at least one false rejection after the Bonferroni correction. Zero and nonzero coefficients are defined by $\alpha_0$ for the threshold method and by $\beta_0$ for the linear benchmark.

The final column reports the rejection rate of the threshold-effect test of $H_0:\delta_0=0$ in Section~\ref{sec:threshold_test}, using the same threshold grid and 499 Gaussian multiplier bootstrap draws. This rejection rate measures the size of the test when $b_2=0$ and its power otherwise.

\begin{table*}[t]
\centering
\caption{Coefficient Tests and the Threshold-Effect Test}
\label{tbl:FWER}
\small
\setlength{\tabcolsep}{3.8pt}
\begin{tabular}{lccccccc}
\toprule
& \multicolumn{3}{c}{Debiased threshold Lasso}
& \multicolumn{3}{c}{Linear benchmark}
& Threshold-effect test \\
\cmidrule(lr){2-4}\cmidrule(lr){5-7}\cmidrule(lr){8-8}
Changed parameters
& FP & TP & FWER
& FP & TP & FWER
& Rej. \\
\midrule
Baseline
& 0.049 & 1.000 & 0.048
& 0.077 & 1.000 & 1.000
& 0.960 \\
$p=300$
& 0.050 & 1.000 & 0.044
& 0.097 & 1.000 & 1.000
& 0.984 \\
$b_2=0.5$
& 0.050 & 1.000 & 0.030
& 0.076 & 1.000 & 1.000
& 0.874 \\
Dependent $Q$
& 0.050 & 1.000 & 0.040
& 0.076 & 1.000 & 1.000
& 0.962 \\
$\tau_0=0.4$
& 0.050 & 1.000 & 0.036
& 0.075 & 1.000 & 0.986
& 0.720 \\
$b_2=0$
& 0.049 & 1.000 & 0.026
& 0.049 & 1.000 & 0.050
& 0.046 \\
$n=1000,\ p=600$
& 0.049 & 1.000 & 0.022
& 0.074 & 1.000 & 1.000
& 1.000 \\
\bottomrule
\end{tabular}

\begin{flushleft}
\footnotesize
Notes: The designs are defined in Table~\ref{tbl:confidence}.
All measures are averaged over 500 Monte Carlo replications.
FP and TP are based on individual two-sided Wald tests at
the 5\% significance level. FWER is based on Bonferroni-adjusted
tests, with individual significance levels of $0.05/(2p)$ for the threshold method and $0.05/p$ for the linear benchmark. Rej.\ denotes the rejection rate of the threshold-effect test at the 5\% level.
\end{flushleft}
\end{table*}

Table~\ref{tbl:FWER} shows that the debiased threshold Lasso
maintains FP close to the nominal level and detects nearly
all nonzero coefficients across the designs. Its Bonferroni
FWER remains below 5\%. These results remain similar when the
threshold effect is weaker, the threshold location changes,
or $Q_i$ and $X_i$ are dependent.

The linear benchmark also detects the nonzero coefficients
in $\beta_0$, but produces more false positives when a
threshold effect is present. Even after the Bonferroni
correction, it makes at least one false rejection in nearly
every replication under these designs. When $b_2=0$, the
two methods have similar false and true positive rates, and both have Bonferroni FWER at or below the nominal level.

The power of the threshold-effect test varies more substantially across designs. Its power falls when the threshold effect is weaker and when $\tau_0$ moves from 0.5 to 0.4. Power is higher in the $p=300$ design compared to the baseline. Power reaches 100\% in the $(n,p)=(1000,600)$ design, where both the sample size and covariate dimension exceed their baseline values. The dependent-$Q_i$ design yields power similar to the baseline. Under the no-threshold model, the rejection rate is close to the nominal 5\% level.

\FloatBarrier
\subsection{Racial Tipping and Neighborhood Composition}
\label{empiricalthree}

\citet{card2008tipping} (hereafter CMR) find a decline in white population growth when a neighborhood's minority share exceeds a critical level, a pattern known as racial tipping. This critical share, referred to as the ``tipping point,'' varies across metropolitan areas. Threshold regression has since been applied in related settings, including \citet{yu2021threshold} and \citet{lee2023threshold}.

CMR identify candidate tipping points using part of the tract sample and estimate discontinuities in white population change using a separate sample, mainly pooling across metropolitan
areas. Their baseline specification focuses on an intercept shift at the tipping point. We also allow the coefficients on neighborhood characteristics to differ across regimes. Using the debiased threshold Lasso, we estimate separate models for each metropolitan area and decade with a rich set of controls. The method provides inference on individual slope differences
without requiring a preliminary choice between linear and threshold specifications. %We also use the threshold-effect test to assess evidence of regime heterogeneity.

We construct tract-level data from National Historical Geographic Information System (NHGIS) Census extracts and Census Bureau tract-relationship files, mapping observations
from both periods to 2000 Census tract boundaries. Metropolitan areas follow CMR's 1999 metropolitan statistical area (MSA) and primary metropolitan statistical area (PMSA) classifications. The outcome $Y_i$ is the ten-year change in the non-Hispanic white population, divided by the tract's base-year total population. The threshold variable $Q_i$ is the base-year minority share, defined as one minus the non-Hispanic white population share.

CMR include six neighborhood controls: the unemployment rate, mean family income, and the shares of single-unit housing, vacant housing, renter-occupied housing, and public-transit
commuters. We expand this set to 50 candidate controls covering demographic characteristics, income composition, housing, and commuting. Additional controls include the share of recent movers, the poverty rate, income sources, housing age, and commuting times.

We restrict the sample to tracts with $0\leq Q_i\leq0.60$ based on CMR, and require at least 100 complete observations in each metropolitan area and decade. Within each sample, we keep a control and its threshold interaction only if the control is nonzero in at least 30 tracts. The resulting sample covers 101 metropolitan areas in 1990--2000 and 96 in 2000--2010. Each specification includes 44--50 controls, their threshold interactions, a baseline intercept, and a threshold intercept shift.

For each metropolitan area $m$ and decade
$t\in\{1990\text{--}2000,2000\text{--}2010\}$, we estimate

\begin{equation}\label{eq:emp03}
Y_i
=
\mu+X_i'\beta
+\mathbf{1}\{Q_i>\tau\}(\gamma+X_i'\delta)+U_i,
\qquad i=1,\ldots,n_{m,t},
\end{equation}
where the metropolitan area and decade subscripts on the parameters are suppressed. The slopes are $\beta$ when $Q_i\leq\tau$ and $\beta+\delta$ when $Q_i>\tau$. Thus, $\delta_j$ captures the difference across regimes in the conditional relationship between neighborhood characteristic $j$ and white population change.

We search for tipping points in $[0,0.50]$, requiring at least 15\% of tracts on each side of the threshold and excluding candidates that yield constant regressors. All eligible tracts are included in each regression. We set $\tau_\ast$ to the midpoint between the smallest and largest thresholds in the grid and add it as an additional candidate. We demean the outcome and regressors, recovering the baseline intercept $\mu$ from the sample means after estimation. We penalize and report inference on $\beta$, $\delta$, and $\gamma$.

Ten-fold cross-validation selects a Lasso penalty for each metropolitan area and decade. For each candidate penalty, we select a tipping point within each training fold using the criterion in \eqref{tau-max}, with $\rho_n=\lambda$.
We choose the penalty with the smallest prediction error on the held-out observations. We then fit the model on the full sample using this common $\widehat\lambda$ and select $\widehat\tau$ by minimizing the same criterion. Nodewise penalties are selected by GIC.

We report debiased threshold Lasso estimates and
heteroskedasticity-robust standard errors. Because \eqref{eq:emp03} includes a threshold intercept, the threshold-effect test evaluates the joint null hypothesis $H_0:\gamma=0,\ \delta=0$. We use the statistic in Section~\ref{sec:threshold_test} over the threshold grid and 499 Gaussian multiplier bootstrap draws. The null-restricted linear model uses a separately cross-validated penalty.

Table~\ref{tb:tipping-tests} indicates threshold
effects are detected less frequently across metropolitan areas in 2000--2010 than in 1990--2000. At the 5\% significance level, the threshold-effect test rejects in 64.4\% of metropolitan areas in 1990--2000, compared with
28.1\% in 2000--2010. The same pattern holds among the 96 metropolitan areas observed in both decades, with 62 and 27 rejections, respectively. Figures~\ref{fig:tipping-map-1990} and~\ref{fig:tipping-map-2000} display the spatial distributions and estimated tipping points of the metropolitan areas in which the null is rejected. The coefficient results in Table~\ref{tb:tipping-tests} show that regime heterogeneity extends beyond an intercept shift. Both CMR's original controls and the additional controls have different relationships with white population change across minority-share regimes, even after Bonferroni correction within each metropolitan area and decade. Such differences are detected in fewer metropolitan areas in 2000--2010 than in 1990--2000.

\begin{table}[htbp]
\centering
\small
\caption{Threshold Effects Across Metropolitan Areas}
\label{tb:tipping-tests}
\setlength{\tabcolsep}{6pt}
\begin{tabular}{lrr}
\toprule
 & 1990--2000 & 2000--2010 \\
\midrule
Metropolitan areas
& 101 & 96 \\
Threshold-effect test rejections
& 65 (64.4\%) & 27 (28.1\%) \\
\addlinespace
\multicolumn{3}{l}{\textit{At least one significant slope difference}} \\
CMR's original controls
& 21 & 6 \\
Additional controls
& 56 & 23 \\
\bottomrule
\end{tabular}

\begin{flushleft}
\footnotesize
Notes: Entries report the number of metropolitan areas. The threshold-effect test uses a 5\% significance level, with rejection percentages in parentheses. The final two rows count metropolitan areas with at least one significant slope difference $\delta_j$ at 5\% after Bonferroni correction across all available coefficient tests within the metropolitan area and decade.
\end{flushleft}
\end{table}

\begin{figure}[htbp]
\centering
\includegraphics[width=\linewidth]{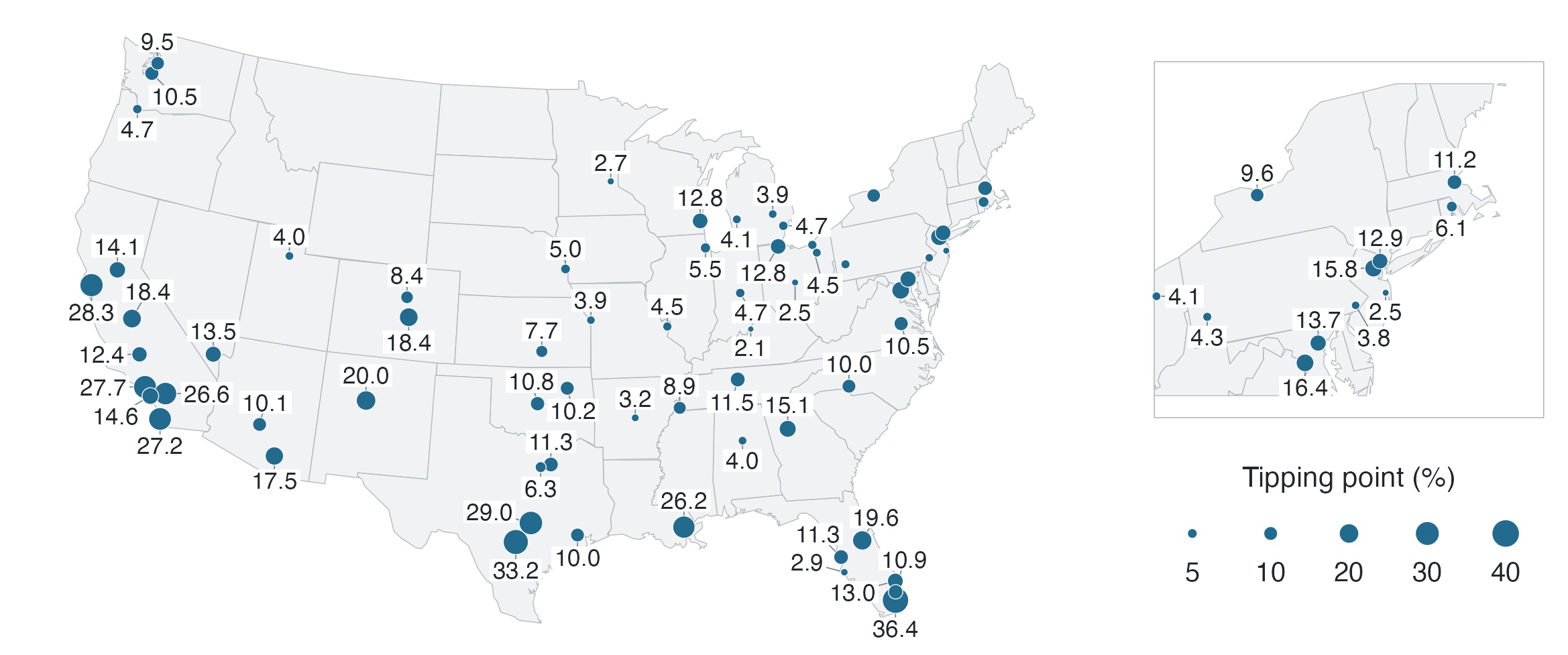}
\caption{Estimated tipping points in metropolitan areas rejecting the null hypothesis of no threshold effect, 1990--2000.}
\label{fig:tipping-map-1990}
\end{figure}

\begin{figure}[htbp]
\centering
\includegraphics[width=\linewidth]{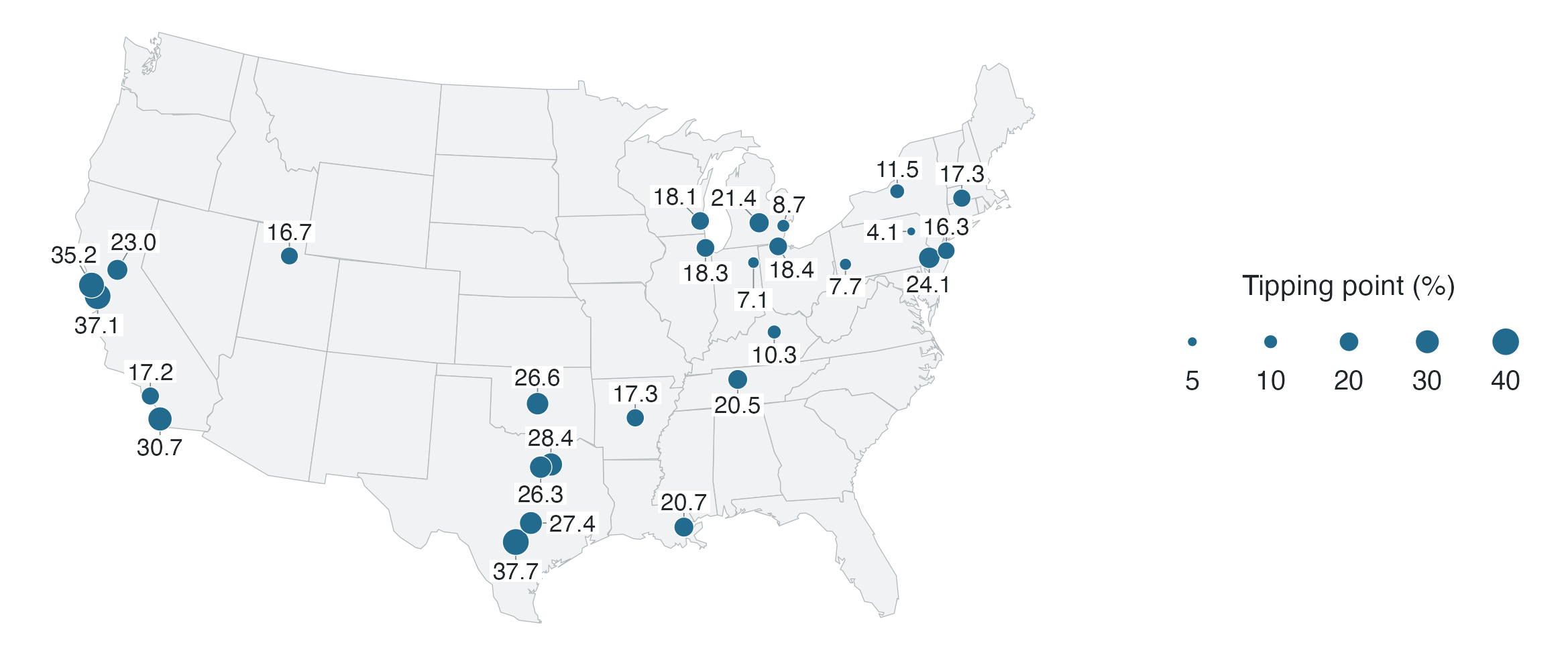}
\caption{Estimated tipping points in metropolitan areas rejecting the null hypothesis of no threshold effect, 2000--2010.}
\label{fig:tipping-map-2000}
\end{figure}

The maps display metropolitan areas in which the threshold-effect
test rejects at the 5\% level. Marker areas are proportional to the estimated tipping points, and labels report the corresponding minority shares as percentages. Both maps use the same scale.

We use Indianapolis and Chicago, illustrated in CMR's
Figures~IV and~I, respectively, to examine how neighborhood
characteristics relate to white population change and how
these relationships differ across the tipping point.
Tables~\ref{tb:indianapolis-coefficients}
and~\ref{tb:chicago-coefficients} report the baseline slopes
$\beta_j$ and slope differences $\delta_j$ for both decades.
All significance statements for individual coefficients
refer to Bonferroni-adjusted two-sided Wald
$p$-values. The correction is applied across all
coefficient tests within each metropolitan area and decade.

In Indianapolis, the threshold-effect test rejects the
linear specification in 1990--2000 but not in 2000--2010.
Among CMR's original controls, higher mean family income
and a larger renter-occupied housing share are associated
with smaller increases in white population below the tipping
point in 1990--2000. Neither association differs significantly
across the tipping point. For the single-unit housing share,
the estimated slope is small and insignificant below the
tipping point ($0.099$) and becomes negative above it
($-0.150$). This slope difference is significant at the
1\% level.

The additional controls provide further evidence of
differences across regimes. In 1990--2000, tracts with a
larger share of households whose householder moved into
the current unit within the previous five years experienced
larger increases in white population below the tipping point. This association
weakens above it, where the estimated slope falls from
$1.286$ to $0.463$. For the share of renter-occupied units with utility payments separate from rent, the slope rises from
$-0.508$ to $0.033$, so the negative association nearly
disappears. The association for workers leaving for work
at 9 a.m.\ or later changes sign, with the estimated slope
rising from $-0.551$ to $0.170$. For housing units built
more than 40 years earlier, the slope is small and
insignificant below the tipping point ($-0.109$) and
becomes positive above it ($0.177$). All four slope
differences are significant at the 1\% level.

Moreover, three other additional controls have significant baseline
slopes without significant differences across regimes.
Below the tipping point, tracts with a larger share of
workers commuting at least 60 minutes experienced smaller
increases in white population ($-2.283$), while tracts
with higher mean wage income and larger land area
experienced larger increases ($0.803$ and $0.069$).

For Indianapolis in 2000--2010, none of the slope
differences are significant at the 10\% level, consistent
with the threshold-effect test. Only the recent-mover
share ($0.898$) and tract land area ($0.081$) have
significant baseline effects.

The results for Chicago present a different pattern
over time. The threshold-effect test rejects the linear
specification in both decades, although the individual
slope differences are significant only in 1990--2000.
The unemployment rate, one of CMR's original controls,
is positively associated with white population change
below the tipping point in the initial decade. Its baseline slope is significant at the 5\% level, but the difference across regimes is not significant.

Income sources account for two of Chicago's significant
slope differences. For the share of households receiving wage income,
the estimated relationship with white population change
is nearly flat below the tipping point and negative
above it, where the slope is $-0.634$. The difference
of $-0.660$ is significant at the 1\% level. Social
Security receipt shows a shift in the opposite direction.
Its estimated association with white population change
becomes less negative above the tipping point. The difference
of $1.868$ is significant at the 5\% level.

%Housing conditions and commuting also have different relationships with white population change across regimes. Below the tipping point, a larger share of housing lacking complete plumbing is associated with lower white population change. The estimated slope rises from $-13.517$ below the tipping point to $-0.148$ above it, leaving little of the negative association. For bicycle commuting, the estimated relationship changes from negative to positive, with slopes of $-20.756$ and $4.920$. Both slope differences are significant at the 1\% level. A larger share of households headed by someone aged at least 65 and greater tract land area are associated with higher white population change below the tipping point, at the 10\% and 5\% levels, respectively. Their slope differences are not significant.

In Chicago's later decade, the differences involving
income sources, housing conditions, and bicycle commuting
are no longer individually significant. None of the slope
differences reaches the 10\% level after Bonferroni
correction. Among the baseline slopes, mean family income
is negative and significant at the 10\% level, while
tract land area remains positive and significant at
the 1\% level.
\begin{table*}[htbp]
\centering
\small
\caption{Selected Coefficient Estimates for Indianapolis}
\label{tb:indianapolis-coefficients}
\setlength{\tabcolsep}{5pt}
\renewcommand{\arraystretch}{1.05}
\begin{tabular}{lcccc}
\toprule
& \multicolumn{2}{c}{1990--2000}
& \multicolumn{2}{c}{2000--2010} \\
\cmidrule(lr){2-3}\cmidrule(lr){4-5}
& $\widehat\beta_j$ & $\widehat\delta_j$
& $\widehat\beta_j$ & $\widehat\delta_j$ \\
\midrule
\multicolumn{5}{l}{\textit{CMR's original controls}} \\
\addlinespace
Mean family income
& $-0.405^{***}$ & $-0.0001$ & $0.100$ & $0.0005$ \\
& $(0.097)$ & $(0.002)$ & $(0.053)$ & $(0.002)$ \\
\addlinespace
Renter-occupied housing share
& $-0.695^{**}$ & $0.037$ & $0.359$ & $-0.347$ \\
& $(0.182)$ & $(0.269)$ & $(0.229)$ & $(0.158)$ \\
\addlinespace
Single-unit housing share
& $0.099$ & $-0.249^{***}$ & $0.263$ & $-0.216$ \\
& $(0.125)$ & $(0.056)$ & $(0.127)$ & $(0.170)$ \\
\addlinespace
\multicolumn{5}{l}{\textit{Additional controls}} \\
\addlinespace
Recent movers share
& $1.286^{***}$ & $-0.823^{***}$ & $0.898^{**}$ & $-0.329$ \\
& $(0.203)$ & $(0.157)$ & $(0.246)$ & $(0.223)$ \\
\addlinespace
Housing aged at least 41 years share
& $-0.109$ & $0.285^{***}$ & $0.199$ & $-0.020$ \\
& $(0.085)$ & $(0.068)$ & $(0.125)$ & $(0.084)$ \\
\addlinespace
Extra utility payments share
& $-0.508^{**}$ & $0.541^{***}$ & $-0.016$ & $0.007$ \\
& $(0.138)$ & $(0.098)$ & $(0.098)$ & $(0.078)$ \\
\addlinespace
Departures at or after 9 a.m.\ share
& $-0.551$ & $0.721^{***}$ & $-0.438$ & $0.431$ \\
& $(0.204)$ & $(0.170)$ & $(0.258)$ & $(0.224)$ \\
\addlinespace
Commutes of at least 60 minutes share
& $-2.283^{***}$ & $1.620$ & $-1.040$ & $1.462$ \\
& $(0.528)$ & $(0.664)$ & $(0.533)$ & $(0.565)$ \\
\addlinespace
Mean wage income
& $0.803^{***}$ & $0.00003$ & $0.039$ & $-0.001$ \\
& $(0.113)$ & $(0.002)$ & $(0.070)$ & $(0.002)$ \\
\addlinespace
Tract land area
& $0.069^{***}$ & $0.031$ & $0.081^{***}$ & $-0.037$ \\
& $(0.013)$ & $(0.021)$ & $(0.017)$ & $(0.022)$ \\
\midrule
Tracts
& \multicolumn{2}{c}{302}
& \multicolumn{2}{c}{292} \\
Tipping point (\%)
& \multicolumn{2}{c}{4.7}
& \multicolumn{2}{c}{9.0} \\
Threshold-effect test $p$-value
& \multicolumn{2}{c}{0.002}
& \multicolumn{2}{c}{0.202} \\
\bottomrule
\end{tabular}

\begin{flushleft}
\footnotesize
Notes: Debiased threshold Lasso estimates with
heteroskedasticity-robust standard errors in parentheses.
$\beta_j$ is the slope for $Q_i\leq\widehat\tau$, and
$\delta_j$ is its change for $Q_i>\widehat\tau$.
The table includes all controls with an adjusted $p<0.10$
for $\beta_j$ or $\delta_j$ in at least one decade.
Mean family income and tract land area enter in logs.
Mean wage income among recipients enters as
$\log(1+\text{mean wage income})$.
Stars use two-sided Wald $p$-values adjusted by Bonferroni
across all coefficient tests within each metropolitan
area and decade:
$^{*}p<0.10$, $^{**}p<0.05$, and $^{***}p<0.01$.
\end{flushleft}
\end{table*}
\begin{table*}[htbp]
\centering
\small
\caption{Selected Coefficient Estimates for Chicago}
\label{tb:chicago-coefficients}
\setlength{\tabcolsep}{5pt}
\renewcommand{\arraystretch}{1.05}
\begin{tabular}{lcccc}
\toprule
& \multicolumn{2}{c}{1990--2000}
& \multicolumn{2}{c}{2000--2010} \\
\cmidrule(lr){2-3}\cmidrule(lr){4-5}
& $\widehat\beta_j$ & $\widehat\delta_j$
& $\widehat\beta_j$ & $\widehat\delta_j$ \\
\midrule
\multicolumn{5}{l}{\textit{CMR's original controls}} \\
\addlinespace
Unemployment rate
& $4.207^{**}$ & $-3.469$ & $0.118$ & $0.464$ \\
& $(1.192)$ & $(1.161)$ & $(0.799)$ & $(0.957)$ \\
\addlinespace
Mean family income
& $-0.193$ & $-0.0007$ & $-0.257^{*}$ & $0.0001$ \\
& $(0.094)$ & $(0.002)$ & $(0.076)$ & $(0.002)$ \\
\addlinespace
\multicolumn{5}{l}{\textit{Additional controls}} \\
\addlinespace
Wage income receipt share
& $0.026$ & $-0.660^{***}$ & $0.0004$ & $-0.0010$ \\
& $(1.364)$ & $(0.152)$ & $(1.111)$ & $(0.026)$ \\
\addlinespace
Social Security receipt share
& $-2.305$ & $1.868^{**}$ & $0.470$ & $-0.403$ \\
& $(0.704)$ & $(0.512)$ & $(0.358)$ & $(0.313)$ \\
\addlinespace
Householder aged at least 65 share
& $2.602^{*}$ & $-1.551$ & $-0.096$ & $0.514$ \\
& $(0.756)$ & $(0.557)$ & $(0.224)$ & $(0.272)$ \\
\addlinespace
Housing lacking complete plumbing share
& $-13.517^{**}$ & $13.370^{***}$ & $-2.108$ & $5.498$ \\
& $(3.548)$ & $(2.579)$ & $(1.355)$ & $(2.412)$ \\
\addlinespace
Bicycle commuting share
& $-20.756^{***}$ & $25.675^{***}$ & $2.787$ & $-1.962$ \\
& $(5.023)$ & $(5.152)$ & $(2.376)$ & $(2.079)$ \\
\addlinespace
Tract land area
& $0.092^{**}$ & $-0.031$ & $0.088^{***}$ & $-0.054$ \\
& $(0.024)$ & $(0.026)$ & $(0.015)$ & $(0.017)$ \\
\midrule
Tracts
& \multicolumn{2}{c}{1,284}
& \multicolumn{2}{c}{1,154} \\
Tipping point (\%)
& \multicolumn{2}{c}{5.5}
& \multicolumn{2}{c}{18.3} \\
Threshold-effect test $p$-value
& \multicolumn{2}{c}{0.002}
& \multicolumn{2}{c}{0.004} \\
\bottomrule
\end{tabular}

\begin{flushleft}
\footnotesize
Notes: Debiased threshold Lasso estimates with
heteroskedasticity-robust standard errors in parentheses.
$\beta_j$ is the slope for $Q_i\leq\widehat\tau$, and
$\delta_j$ is its change for $Q_i>\widehat\tau$.
The table includes all controls with an adjusted $p<0.10$
for $\beta_j$ or $\delta_j$ in at least one decade.
Stars use two-sided Wald $p$-values adjusted by Bonferroni
across all coefficient tests within each metropolitan
area and decade:
$^{*}p<0.10$, $^{**}p<0.05$, and $^{***}p<0.01$.
\end{flushleft}
\end{table*}

The appendix reports the initial and debiased threshold Lasso estimates for all coefficients in Chicago, Indianapolis, and Nashville in both decades. These metropolitan areas appear among CMR's examples in Figures~I and~IV. The online repository and replication package provide the complete coefficient estimates, confidence intervals, and test results for all metropolitan areas in both decades. 
\FloatBarrier

\section{Conclusion}\label{conclusion}

This paper develops uniform inference for prespecified slope parameters in high-dimensional threshold regressions without pretesting for a threshold effect. Confidence intervals based
on the debiased threshold Lasso estimator are asymptotically
valid uniformly over the no-threshold and identified fixed-threshold model classes. The analysis accounts for the
dependence arising from estimating the threshold and constructing
the score using the same observations. It permits heteroskedasticity and non-sub-Gaussian covariates and errors. We also develop a threshold-effect test to assess whether regression coefficients differ across regimes.

Across the simulation designs, the proposed procedure approaches nominal familywise error rates while achieving high power for detecting nonzero coefficients. When threshold effects are present, ignoring them in a linear specification substantially increases the familywise error rate for inference on the baseline coefficients. The racial tipping application shows differences across minority share regimes in the conditional relationships between white population change and residential turnover, housing age, commuting schedules, and utility payments. Threshold effects are detected in a smaller proportion of metropolitan areas in 2000--2010 than in 1990--2000.

Future research could develop inference methods under weak threshold identification and extend the theory to models with multiple thresholds in dynamic panels or panels with cross-sectional dependence.

\clearpage
\printbibliography[title={References}]
\end{refsection}

\clearpage
\begin{refsection}
% Included by main.tex after the main-text reference list.
\let\footnotesize\appendixfootnotesize
\appendix
\section*{Appendix}

% Preserve the original S numbering while making PDF destinations unique.
\renewcommand{\theHsection}{appendix.\arabic{section}}
\renewcommand{\theHsubsection}{\theHsection.\arabic{subsection}}
\renewcommand{\theHsubsubsection}{\theHsubsection.\arabic{subsubsection}}
\renewcommand{\theHequation}{\theHsection.\arabic{equation}}
\renewcommand{\theHtable}{appendix.\arabic{table}}
\renewcommand{\theHfigure}{appendix.\arabic{figure}}
\renewcommand{\theHthm}{appendix.\arabic{thm}}
\renewcommand{\theHcor}{appendix.\arabic{cor}}
\renewcommand{\theHlem}{appendix.\arabic{lem}}
\renewcommand{\theHprop}{appendix.\arabic{prop}}
\renewcommand{\theHdefn}{\theHsection.\arabic{defn}}
\renewcommand{\theHassm}{appendix.\arabic{assm}}
\renewcommand{\theHrem}{appendix.\arabic{rem}}

\doublespacing
\renewcommand\thesection{S\arabic{section}}
\setcounter{subsection}{0}
\renewcommand\thesubsection{S\arabic{section}.\arabic{subsection}}

\setcounter{equation}{0}
\renewcommand\theequation{S\arabic{section}.\arabic{equation}}

\setcounter{table}{0}
\renewcommand\thetable{S\arabic{table}}

\setcounter{figure}{0}
\renewcommand\thefigure{S\arabic{figure}}

\setcounter{thm}{0}
\renewcommand\thethm{S.\arabic{thm}}

\setcounter{cor}{0}
\renewcommand\thecor{S.\arabic{cor}}

\setcounter{lem}{0}
\renewcommand\thelem{S.\arabic{lem}}

\setcounter{prop}{0}
\renewcommand\theprop{S.\arabic{prop}}

\setcounter{defn}{0}
\renewcommand\thedefn{S.\arabic{defn}}

\setcounter{assm}{0}
\renewcommand\theassm{S.\arabic{assm}}

\setcounter{rem}{0}
\renewcommand\therem{S.\arabic{rem}}

\noindent The appendix is organized as follows. Section
\ref{supp:tables} reports detailed coefficient estimates for
selected metropolitan areas in the racial-tipping application. Section \ref{supp:technical} provides the nodewise approximate-inverse construction, and implementation and approximation details for the threshold-effect test. Section \ref{supp:proofs} collects the auxiliary lemmas and proofs of the main
results. Section \ref{timeseries} extends the framework to time series data, including local-projection threshold regressions, an empirical illustration, and the corresponding proofs.

\section{Additional Tables for the Racial-Tipping Application}
\label{supp:tables}
The baseline slopes $\beta$ apply when $Q_i\leq\widehat\tau$,
and the slopes are $\beta+\delta$ when $Q_i>\widehat\tau$.
Shares and the outcome are expressed as fractions.
Conditional mean-income controls use $\log(1+\text{income})$,
with income measured in 1999 dollars.

\begin{sidewaystable*}[p]
\centering
\setlength{\tabcolsep}{2.2pt}
\renewcommand{\arraystretch}{0.70}
\caption{All Candidate Controls: Lasso and Debiased Estimates for Indianapolis, 1990--2000}
\label{tb:citycoef-3480-1990-2000}\label{tb:citycoef-3480}
{\fontsize{5.6}{6.2}\selectfont
\begin{minipage}[t]{0.49\textwidth}
\centering
\textit{Panel A}\par\vspace{0pt}
\begin{tabular}{@{}p{0.245\textwidth}rrrr@{}}
\toprule
Control & $\beta_L$ & $\delta_L$ & $\beta_D$ & $\delta_D$ \\
\midrule
At least 175\% poverty & $0$ & $0$ & $0.0038$ {\tiny ($0.15$)} & $-0.046$ {\tiny ($0.16$)} \\
Below poverty line & $0.49$ & $0$ & $0.77^{**}$ {\tiny ($0.34$)} & $-0.48$ {\tiny ($0.41$)} \\
Bicycle commute & $6$ & $-4.4$ & $7^{*}$ {\tiny ($3.7$)} & $-5.9$ {\tiny ($5$)} \\
Bottled-gas heat & $-1$ & $-0.47$ & $-1.5^{***}$ {\tiny ($0.49$)} & $-0.81$ {\tiny ($0.74$)} \\
Car commute & $0$ & $0$ & $0.58$ {\tiny ($0.43$)} & $0.016$ {\tiny ($0.06$)} \\
Coal/coke heat & -- & -- & -- & -- \\
Commute 60+ min & $-2.2$ & $1.4$ & $-2.3^{***}$ {\tiny ($0.53$)} & $1.6^{**}$ {\tiny ($0.66$)} \\
Commute under 30 min & $-0.31$ & $0$ & $-0.33^{**}$ {\tiny ($0.13$)} & $0.045$ {\tiny ($0.096$)} \\
Depart 9 a.m. or later & $-0.48$ & $0.62$ & $-0.55^{***}$ {\tiny ($0.2$)} & $0.72^{***}$ {\tiny ($0.17$)} \\
Depart before 7 a.m. & $0.18$ & $-0.028$ & $0.22$ {\tiny ($0.25$)} & $-0.096$ {\tiny ($0.13$)} \\
Earnings receipt & $0.097$ & $-0.0006$ & $0.21$ {\tiny ($0.69$)} & $-0.0056$ {\tiny ($0.023$)} \\
Electric heat & $0.0069$ & $-0.067$ & $-0.14$ {\tiny ($0.19$)} & $-0.44$ {\tiny ($0.28$)} \\
English-only persons age 5+ & $0.44$ & $0$ & $0.47$ {\tiny ($0.47$)} & $-0.0047$ {\tiny ($0.033$)} \\
Four+ vehicles & $-0.15$ & $1.2$ & $-0.29$ {\tiny ($0.63$)} & $1.6$ {\tiny ($1.2$)} \\
Fuel-oil heat & $0.054$ & $0.26$ & $-0.35$ {\tiny ($0.51$)} & $-0.92$ {\tiny ($0.89$)} \\
High-density housing & $0$ & $0.16$ & $0.69^{**}$ {\tiny ($0.31$)} & $0.13$ {\tiny ($0.5$)} \\
Householder age 15--34 & $0.98$ & $-0.36$ & $1.1^{***}$ {\tiny ($0.34$)} & $-0.49$ {\tiny ($0.38$)} \\
Householder age 65+ & $0.42$ & $0.44$ & $0.46^{**}$ {\tiny ($0.23$)} & $0.5$ {\tiny ($0.35$)} \\
Housing built 41+ yrs ago & $-0.054$ & $0.24$ & $-0.11$ {\tiny ($0.085$)} & $0.29^{***}$ {\tiny ($0.068$)} \\
Housing built prev. 20 yrs & $0.2$ & $-0.066$ & $0.21^{*}$ {\tiny ($0.12$)} & $-0.078$ {\tiny ($0.1$)} \\
Intercept & $-6.2$ & $0$ & $-5.8$ {\tiny (n/a)} & $-0.0014$ {\tiny ($0.02$)} \\
Lacking kitchen & $0.052$ & $-0.13$ & $0.3$ {\tiny ($1.3$)} & $-0.62$ {\tiny ($1.6$)} \\
Lacking plumbing & $1.9$ & $-2.2$ & $2.2$ {\tiny ($1.6$)} & $-2.6$ {\tiny ($1.8$)} \\
Log mean fam. income & $-0.33$ & $0$ & $-0.41^{***}$ {\tiny ($0.097$)} & $-0.00012$ {\tiny ($0.0017$)} \\
Log tract land area & $0.068$ & $0.021$ & $0.069^{***}$ {\tiny ($0.013$)} & $0.031$ {\tiny ($0.021$)} \\
Low-density housing & $0.007$ & $-0.33$ & $0.26^{**}$ {\tiny ($0.11$)} & $-0.23^{**}$ {\tiny ($0.091$)} \\
\bottomrule
\end{tabular}
\end{minipage}\hfill
\begin{minipage}[t]{0.49\textwidth}
\centering
\textit{Panel B}\par\vspace{0pt}
\begin{tabular}{@{}p{0.245\textwidth}rrrr@{}}
\toprule
Control & $\beta_L$ & $\delta_L$ & $\beta_D$ & $\delta_D$ \\
\midrule
Log mean public-assist. income & $-0.003$ & $0$ & $-0.0097$ {\tiny ($0.021$)} & $0.007$ {\tiny ($0.012$)} \\
Log mean retirement income & $-0.0027$ & $0$ & $-0.014$ {\tiny ($0.029$)} & $0.0029$ {\tiny ($0.01$)} \\
Log mean Social Security income & $0.026$ & $0$ & $-0.078$ {\tiny ($0.11$)} & $-0.00062$ {\tiny ($0.0061$)} \\
Log mean wage income & $0.75$ & $0$ & $0.8^{***}$ {\tiny ($0.11$)} & $3.2\!\times\!10^{-5}$ {\tiny ($0.002$)} \\
Middle-density housing & $-0.24$ & $0.086$ & $0.88^{*}$ {\tiny ($0.47$)} & $0.13$ {\tiny ($0.86$)} \\
Motorcycle commute & $0.5$ & $0$ & $0.81$ {\tiny ($7.2$)} & $0.66$ {\tiny ($8$)} \\
Moved in 21+ yrs ago & $0.32$ & $0$ & $0.37$ {\tiny ($0.28$)} & $-0.074$ {\tiny ($0.33$)} \\
Moved in prev. 5 yrs & $1.2$ & $-0.69$ & $1.3^{***}$ {\tiny ($0.2$)} & $-0.82^{***}$ {\tiny ($0.16$)} \\
No-fuel heat & $0$ & $0$ & $-0.48$ {\tiny ($2.2$)} & $0.71$ {\tiny ($2.6$)} \\
Pays extra utilities & $-0.45$ & $0.42$ & $-0.51^{***}$ {\tiny ($0.14$)} & $0.54^{***}$ {\tiny ($0.098$)} \\
Public assist. receipt & $-0.054$ & $0$ & $-0.74$ {\tiny ($0.83$)} & $0.49$ {\tiny ($0.77$)} \\
Public transit share & $0.68$ & $0$ & $2.6^{*}$ {\tiny ($1.3$)} & $-0.5$ {\tiny ($0.81$)} \\
Renter-occupied share & $-0.67$ & $0$ & $-0.69^{***}$ {\tiny ($0.18$)} & $0.037$ {\tiny ($0.27$)} \\
Retirement receipt & $0.16$ & $-0.2$ & $0.32$ {\tiny ($0.3$)} & $-0.58$ {\tiny ($0.36$)} \\
Single-unit share & $0.09$ & $-0.23$ & $0.099$ {\tiny ($0.13$)} & $-0.25^{***}$ {\tiny ($0.056$)} \\
Social Security receipt & $0.77$ & $0$ & $0.81^{**}$ {\tiny ($0.41$)} & $0.09$ {\tiny ($0.42$)} \\
Solar heat & -- & -- & -- & -- \\
Unemployment rate & $0.2$ & $0.44$ & $0.16$ {\tiny ($0.63$)} & $0.72$ {\tiny ($0.45$)} \\
Utility-gas heat & $-0.19$ & $0$ & $-0.34^{**}$ {\tiny ($0.14$)} & $-0.082$ {\tiny ($0.086$)} \\
Vacancy share & $1$ & $-0.91$ & $1.2^{***}$ {\tiny ($0.48$)} & $-1.3^{**}$ {\tiny ($0.54$)} \\
Wage receipt & $0.65$ & $-0.0047$ & $0.66$ {\tiny ($0.68$)} & $-0.0079$ {\tiny ($0.025$)} \\
Walked to work & $0.18$ & $0$ & $1.2$ {\tiny ($0.71$)} & $-1.5$ {\tiny ($1.5$)} \\
Wood heat & $-0.094$ & $0.56$ & $-0.43$ {\tiny ($0.58$)} & $1.2$ {\tiny ($1.5$)} \\
Worked at home & $-0.13$ & $0.15$ & $0.53$ {\tiny ($1.1$)} & $0.58$ {\tiny ($0.95$)} \\
Zero/one vehicle & $-0.028$ & $0.081$ & $-0.098$ {\tiny ($0.19$)} & $0.19$ {\tiny ($0.12$)} \\
\bottomrule
\end{tabular}
\end{minipage}
}
\begin{flushleft}
\fontsize{6.2}{6.8}\selectfont
Notes: Subscripts $L$ and $D$ denote initial Lasso and
debiased threshold Lasso estimates, respectively.
Heteroskedasticity-robust standard errors are in parentheses.
The Intercept row reports $\mu$ in the $\beta$ columns
and $\gamma$ in the $\delta$ columns.
Dashes indicate excluded controls, and n/a indicates
that no standard error is reported.
Stars use unadjusted two-sided Wald $p$-values:
$^{*}p<0.10$, $^{**}p<0.05$, and $^{***}p<0.01$.
$N=302$, $\widehat\tau=0.0470$, $\tau_\ast=0.1040$.
Threshold-effect test: $T_n=12.92$, bootstrap $p=0.002$.
\end{flushleft}
\end{sidewaystable*}

\begin{sidewaystable*}[p]
\centering
\setlength{\tabcolsep}{2.2pt}
\renewcommand{\arraystretch}{0.70}
\caption{All Candidate Controls: Lasso and Debiased Estimates for Indianapolis, 2000--2010}
\label{tb:citycoef-3480-2000-2010}
{\fontsize{5.6}{6.2}\selectfont
\begin{minipage}[t]{0.49\textwidth}
\centering
\textit{Panel A}\par\vspace{0pt}
\begin{tabular}{@{}p{0.245\textwidth}rrrr@{}}
\toprule
Control & $\beta_L$ & $\delta_L$ & $\beta_D$ & $\delta_D$ \\
\midrule
At least 175\% poverty & $0$ & $0$ & $0.11$ {\tiny ($0.15$)} & $-0.3$ {\tiny ($0.23$)} \\
Below poverty line & $0$ & $0$ & $-0.57$ {\tiny ($0.43$)} & $0.9^{***}$ {\tiny ($0.29$)} \\
Bicycle commute & $3.5$ & $0$ & $4.3$ {\tiny ($4.2$)} & $0.88$ {\tiny ($3.7$)} \\
Bottled-gas heat & $0$ & $2$ & $-0.091$ {\tiny ($0.52$)} & $3.5^{*}$ {\tiny ($1.8$)} \\
Car commute & $0$ & $0$ & $-0.082$ {\tiny ($0.49$)} & $-0.064$ {\tiny ($0.082$)} \\
Coal/coke heat & -- & -- & -- & -- \\
Commute 60+ min & $-0.019$ & $0$ & $-1^{*}$ {\tiny ($0.53$)} & $1.5^{***}$ {\tiny ($0.57$)} \\
Commute under 30 min & $0$ & $0$ & $-0.045$ {\tiny ($0.15$)} & $0.016$ {\tiny ($0.094$)} \\
Depart 9 a.m. or later & $0$ & $0$ & $-0.44^{*}$ {\tiny ($0.26$)} & $0.43^{*}$ {\tiny ($0.22$)} \\
Depart before 7 a.m. & $0$ & $-0.065$ & $0.042$ {\tiny ($0.22$)} & $-0.43^{***}$ {\tiny ($0.16$)} \\
Earnings receipt & $0$ & $0$ & $-0.48$ {\tiny ($0.36$)} & $-0.018$ {\tiny ($0.027$)} \\
Electric heat & $0$ & $0$ & $0.19$ {\tiny ($0.15$)} & $0.076$ {\tiny ($0.38$)} \\
English-only persons age 5+ & $0.32$ & $0$ & $1$ {\tiny ($0.67$)} & $0.025$ {\tiny ($0.06$)} \\
Four+ vehicles & $0$ & $0$ & $-0.32$ {\tiny ($0.73$)} & $0.43$ {\tiny ($0.53$)} \\
Fuel-oil heat & $0$ & $0.15$ & $1.2$ {\tiny ($0.78$)} & $1.5$ {\tiny ($1.2$)} \\
High-density housing & $0$ & $0$ & $-0.036$ {\tiny ($0.11$)} & $0.43$ {\tiny ($0.51$)} \\
Householder age 15--34 & $0$ & $0$ & $0.037$ {\tiny ($0.28$)} & $-0.013$ {\tiny ($0.25$)} \\
Householder age 65+ & $0$ & $0$ & $0.19$ {\tiny ($0.4$)} & $0.3$ {\tiny ($0.39$)} \\
Housing built 41+ yrs ago & $0$ & $0$ & $0.2$ {\tiny ($0.12$)} & $-0.02$ {\tiny ($0.084$)} \\
Housing built prev. 20 yrs & $0.28$ & $-0.41$ & $0.39^{***}$ {\tiny ($0.13$)} & $-0.42^{***}$ {\tiny ($0.13$)} \\
Intercept & $-2$ & $0$ & $-2.4$ {\tiny (n/a)} & $-0.011$ {\tiny ($0.021$)} \\
Lacking kitchen & $-0.42$ & $0$ & $-1.2$ {\tiny ($1.1$)} & $-0.75$ {\tiny ($1.3$)} \\
Lacking plumbing & $0$ & $0.065$ & $-2$ {\tiny ($1.4$)} & $3.6^{**}$ {\tiny ($1.5$)} \\
Log mean fam. income & $0.09$ & $0$ & $0.1^{*}$ {\tiny ($0.053$)} & $0.00052$ {\tiny ($0.002$)} \\
Log tract land area & $0.062$ & $0$ & $0.081^{***}$ {\tiny ($0.017$)} & $-0.037^{*}$ {\tiny ($0.022$)} \\
Low-density housing & $0$ & $0$ & $0.016$ {\tiny ($0.049$)} & $-0.085$ {\tiny ($0.086$)} \\
\bottomrule
\end{tabular}
\end{minipage}\hfill
\begin{minipage}[t]{0.49\textwidth}
\centering
\textit{Panel B}\par\vspace{0pt}
\begin{tabular}{@{}p{0.245\textwidth}rrrr@{}}
\toprule
Control & $\beta_L$ & $\delta_L$ & $\beta_D$ & $\delta_D$ \\
\midrule
Log mean public-assist. income & $0.00027$ & $0$ & $0.0044$ {\tiny ($0.0056$)} & $-0.0026$ {\tiny ($0.0038$)} \\
Log mean retirement income & $0$ & $0$ & $-0.01$ {\tiny ($0.027$)} & $0.015$ {\tiny ($0.01$)} \\
Log mean Social Security income & $-0.0055$ & $0$ & $-0.12^{*}$ {\tiny ($0.071$)} & $-0.0033$ {\tiny ($0.0033$)} \\
Log mean wage income & $0.022$ & $0$ & $0.039$ {\tiny ($0.07$)} & $-0.0014$ {\tiny ($0.002$)} \\
Middle-density housing & $-0.12$ & $0$ & $-0.31^{*}$ {\tiny ($0.17$)} & $1.7^{***}$ {\tiny ($0.51$)} \\
Motorcycle commute & $0$ & $0$ & $-2.5$ {\tiny ($5.9$)} & $3.3$ {\tiny ($5.4$)} \\
Moved in 21+ yrs ago & $0$ & $0$ & $0.051$ {\tiny ($0.14$)} & $0.4$ {\tiny ($0.27$)} \\
Moved in prev. 5 yrs & $0.62$ & $0$ & $0.9^{***}$ {\tiny ($0.25$)} & $-0.33$ {\tiny ($0.22$)} \\
No-fuel heat & $0$ & $0$ & $-4.6^{*}$ {\tiny ($2.7$)} & $4.6^{*}$ {\tiny ($2.5$)} \\
Pays extra utilities & $0$ & $0$ & $-0.016$ {\tiny ($0.098$)} & $0.0071$ {\tiny ($0.078$)} \\
Public assist. receipt & $0$ & $0$ & $0.59$ {\tiny ($0.85$)} & $0.56$ {\tiny ($0.88$)} \\
Public transit share & $0$ & $0$ & $0.42$ {\tiny ($0.78$)} & $0.099$ {\tiny ($0.69$)} \\
Renter-occupied share & $0$ & $0$ & $0.36$ {\tiny ($0.23$)} & $-0.35^{**}$ {\tiny ($0.16$)} \\
Retirement receipt & $0$ & $0$ & $0.25$ {\tiny ($0.28$)} & $0.22$ {\tiny ($0.21$)} \\
Single-unit share & $0.14$ & $0$ & $0.26^{**}$ {\tiny ($0.13$)} & $-0.22$ {\tiny ($0.17$)} \\
Social Security receipt & $0$ & $0$ & $-0.033$ {\tiny ($0.22$)} & $-0.024$ {\tiny ($0.32$)} \\
Solar heat & -- & -- & -- & -- \\
Unemployment rate & $0$ & $0$ & $-0.19$ {\tiny ($0.57$)} & $0.33$ {\tiny ($0.45$)} \\
Utility-gas heat & $0$ & $0$ & $0.091$ {\tiny ($0.24$)} & $0.04$ {\tiny ($0.087$)} \\
Vacancy share & $0$ & $0$ & $-1.1^{**}$ {\tiny ($0.56$)} & $1.2^{**}$ {\tiny ($0.61$)} \\
Wage receipt & $0$ & $0$ & $0.69$ {\tiny ($0.44$)} & $-0.024$ {\tiny ($0.03$)} \\
Walked to work & $0$ & $0.39$ & $-0.75$ {\tiny ($1.1$)} & $1.6^{**}$ {\tiny ($0.69$)} \\
Wood heat & $-1.6$ & $1.1$ & $-2.5^{***}$ {\tiny ($0.94$)} & $4.9^{*}$ {\tiny ($3$)} \\
Worked at home & $0$ & $0$ & $0.23$ {\tiny ($0.98$)} & $-0.94$ {\tiny ($1.8$)} \\
Zero/one vehicle & $-0.049$ & $0$ & $-0.16$ {\tiny ($0.13$)} & $0.04$ {\tiny ($0.15$)} \\
\bottomrule
\end{tabular}
\end{minipage}
}
\begin{flushleft}
\fontsize{6.2}{6.8}\selectfont
Notes: Subscripts $L$ and $D$ denote initial Lasso and
debiased threshold Lasso estimates, respectively.
Heteroskedasticity-robust standard errors are in parentheses.
The Intercept row reports $\mu$ in the $\beta$ columns
and $\gamma$ in the $\delta$ columns.
Dashes indicate excluded controls, and n/a indicates
that no standard error is reported.
Stars use unadjusted two-sided Wald $p$-values:
$^{*}p<0.10$, $^{**}p<0.05$, and $^{***}p<0.01$.
$N=292$, $\widehat\tau=0.0900$, $\tau_\ast=0.1820$.
Threshold-effect test: $T_n=3.54$, bootstrap $p=0.202$.
\end{flushleft}
\end{sidewaystable*}

\begin{sidewaystable*}[p]
\centering
\setlength{\tabcolsep}{2.2pt}
\renewcommand{\arraystretch}{0.70}
\caption{All Candidate Controls: Lasso and Debiased Estimates for Nashville, 1990--2000}
\label{tb:citycoef-5360-1990-2000}\label{tb:citycoef-5360}
{\fontsize{5.6}{6.2}\selectfont
\begin{minipage}[t]{0.49\textwidth}
\centering
\textit{Panel A}\par\vspace{0pt}
\begin{tabular}{@{}p{0.245\textwidth}rrrr@{}}
\toprule
Control & $\beta_L$ & $\delta_L$ & $\beta_D$ & $\delta_D$ \\
\midrule
At least 175\% poverty & $0$ & $-0.067$ & $-0.13$ {\tiny ($0.23$)} & $-0.25$ {\tiny ($0.38$)} \\
Below poverty line & $0$ & $0$ & $-0.3$ {\tiny ($0.53$)} & $1.3^{***}$ {\tiny ($0.35$)} \\
Bicycle commute & $0$ & $0$ & $-1.1$ {\tiny ($2.4$)} & $20^{***}$ {\tiny ($7.8$)} \\
Bottled-gas heat & $0.029$ & $0$ & $0.21$ {\tiny ($1.8$)} & $-1.1$ {\tiny ($2.9$)} \\
Car commute & $0$ & $0$ & $-0.035$ {\tiny ($0.44$)} & $-0.18^{**}$ {\tiny ($0.087$)} \\
Coal/coke heat & $0$ & $0$ & $-2.1$ {\tiny ($2.8$)} & $2$ {\tiny ($2.6$)} \\
Commute 60+ min & $0$ & $0$ & $-0.26$ {\tiny ($0.63$)} & $0.021$ {\tiny ($0.79$)} \\
Commute under 30 min & $0$ & $0$ & $-0.042$ {\tiny ($0.13$)} & $-0.15$ {\tiny ($0.15$)} \\
Depart 9 a.m. or later & $0$ & $0$ & $-0.28$ {\tiny ($0.36$)} & $0.046$ {\tiny ($0.43$)} \\
Depart before 7 a.m. & $0$ & $0$ & $-0.12$ {\tiny ($0.24$)} & $-0.43^{**}$ {\tiny ($0.19$)} \\
Earnings receipt & $0$ & $0$ & $0.42$ {\tiny ($0.87$)} & $-0.029$ {\tiny ($0.036$)} \\
Electric heat & $0$ & $-0.027$ & $-0.24^{**}$ {\tiny ($0.12$)} & $-0.29^{**}$ {\tiny ($0.12$)} \\
English-only persons age 5+ & $0$ & $0$ & $0.26$ {\tiny ($0.57$)} & $-0.083^{*}$ {\tiny ($0.047$)} \\
Four+ vehicles & $0.061$ & $-0.29$ & $1$ {\tiny ($0.88$)} & $-2.9^{**}$ {\tiny ($1.4$)} \\
Fuel-oil heat & $0$ & $0$ & $-2.1$ {\tiny ($2.4$)} & $-0.12$ {\tiny ($1.1$)} \\
High-density housing & $0$ & $0$ & $-0.12$ {\tiny ($0.4$)} & $0.35$ {\tiny ($0.47$)} \\
Householder age 15--34 & $0$ & $0$ & $0.063$ {\tiny ($0.19$)} & $-0.53$ {\tiny ($0.38$)} \\
Householder age 65+ & $0$ & $0$ & $-0.48$ {\tiny ($0.41$)} & $0.84^{**}$ {\tiny ($0.34$)} \\
Housing built 41+ yrs ago & $0$ & $0$ & $0.34$ {\tiny ($0.21$)} & $-0.056$ {\tiny ($0.28$)} \\
Housing built prev. 20 yrs & $0.24$ & $0$ & $0.31^{***}$ {\tiny ($0.12$)} & $-0.19$ {\tiny ($0.19$)} \\
Intercept & $-0.46$ & $0$ & $-0.73$ {\tiny (n/a)} & $-0.043$ {\tiny ($0.029$)} \\
Lacking kitchen & $0$ & $0$ & $-4.6^{*}$ {\tiny ($2.5$)} & $2.5$ {\tiny ($3.3$)} \\
Lacking plumbing & $0$ & $0$ & $-2.8$ {\tiny ($2.9$)} & $2.8$ {\tiny ($3.4$)} \\
Log mean fam. income & $0$ & $0$ & $0.14$ {\tiny ($0.14$)} & $-0.004$ {\tiny ($0.0032$)} \\
Log tract land area & $0.053$ & $0$ & $0.055^{**}$ {\tiny ($0.021$)} & $-0.005$ {\tiny ($0.026$)} \\
Low-density housing & $0.13$ & $0$ & $0.3^{*}$ {\tiny ($0.17$)} & $0.031$ {\tiny ($0.38$)} \\
\bottomrule
\end{tabular}
\end{minipage}\hfill
\begin{minipage}[t]{0.49\textwidth}
\centering
\textit{Panel B}\par\vspace{0pt}
\begin{tabular}{@{}p{0.245\textwidth}rrrr@{}}
\toprule
Control & $\beta_L$ & $\delta_L$ & $\beta_D$ & $\delta_D$ \\
\midrule
Log mean public-assist. income & $0$ & $0$ & $0.025$ {\tiny ($0.03$)} & $0.012$ {\tiny ($0.031$)} \\
Log mean retirement income & $0$ & $0$ & $0.083$ {\tiny ($0.061$)} & $-0.055^{***}$ {\tiny ($0.02$)} \\
Log mean Social Security income & $0$ & $0$ & $-0.075$ {\tiny ($0.064$)} & $0.018$ {\tiny ($0.013$)} \\
Log mean wage income & $0.028$ & $0$ & $-0.03$ {\tiny ($0.17$)} & $-0.0047$ {\tiny ($0.003$)} \\
Middle-density housing & $0$ & $0$ & $0.56$ {\tiny ($0.82$)} & $-0.73$ {\tiny ($1$)} \\
Motorcycle commute & $0$ & $0$ & $-9.3$ {\tiny ($8.2$)} & $7.6$ {\tiny ($8.6$)} \\
Moved in 21+ yrs ago & $-0.24$ & $0$ & $-0.42$ {\tiny ($0.3$)} & $-0.58$ {\tiny ($0.45$)} \\
Moved in prev. 5 yrs & $0$ & $0$ & $0.24$ {\tiny ($0.17$)} & $0.27$ {\tiny ($0.22$)} \\
No-fuel heat & $0$ & $0$ & $-7.2^{*}$ {\tiny ($3.7$)} & $11^{**}$ {\tiny ($4.3$)} \\
Pays extra utilities & $0$ & $0$ & $-0.3$ {\tiny ($0.27$)} & $0.14$ {\tiny ($0.17$)} \\
Public assist. receipt & $0$ & $0$ & $1.6$ {\tiny ($1.2$)} & $-1.1$ {\tiny ($0.87$)} \\
Public transit share & $0$ & $0$ & $-6.8^{***}$ {\tiny ($1.7$)} & $4.4^{***}$ {\tiny ($1.2$)} \\
Renter-occupied share & $0$ & $0$ & $0.37$ {\tiny ($0.37$)} & $0.3$ {\tiny ($0.33$)} \\
Retirement receipt & $-0.02$ & $0$ & $-0.83$ {\tiny ($0.58$)} & $0.88$ {\tiny ($0.67$)} \\
Single-unit share & $0$ & $0$ & $-0.19$ {\tiny ($0.23$)} & $-0.034$ {\tiny ($0.16$)} \\
Social Security receipt & $0$ & $0$ & $-0.16$ {\tiny ($0.31$)} & $-0.27$ {\tiny ($0.26$)} \\
Solar heat & -- & -- & -- & -- \\
Unemployment rate & $0$ & $0$ & $-1.5$ {\tiny ($0.96$)} & $0.044$ {\tiny ($0.92$)} \\
Utility-gas heat & $0$ & $0$ & $-0.03$ {\tiny ($0.13$)} & $-0.24$ {\tiny ($0.27$)} \\
Vacancy share & $0$ & $0$ & $1.3^{*}$ {\tiny ($0.73$)} & $0.29$ {\tiny ($0.6$)} \\
Wage receipt & $0$ & $0$ & $-0.76$ {\tiny ($0.94$)} & $-0.11^{**}$ {\tiny ($0.042$)} \\
Walked to work & $0$ & $0$ & $0.74$ {\tiny ($1$)} & $-2.3$ {\tiny ($1.9$)} \\
Wood heat & $0$ & $0$ & $-0.85$ {\tiny ($1.1$)} & $3.3^{*}$ {\tiny ($1.7$)} \\
Worked at home & $0$ & $0$ & $1.1$ {\tiny ($1.4$)} & $2.4$ {\tiny ($2.3$)} \\
Zero/one vehicle & $0$ & $0$ & $0.13$ {\tiny ($0.21$)} & $0.031$ {\tiny ($0.17$)} \\
\bottomrule
\end{tabular}
\end{minipage}
}
\begin{flushleft}
\fontsize{6.2}{6.8}\selectfont
Notes: Subscripts $L$ and $D$ denote initial Lasso and
debiased threshold Lasso estimates, respectively.
Heteroskedasticity-robust standard errors are in parentheses.
The Intercept row reports $\mu$ in the $\beta$ columns
and $\gamma$ in the $\delta$ columns.
Dashes indicate excluded controls, and n/a indicates
that no standard error is reported.
Stars use unadjusted two-sided Wald $p$-values:
$^{*}p<0.10$, $^{**}p<0.05$, and $^{***}p<0.01$.
$N=218$, $\widehat\tau=0.1150$, $\tau_\ast=0.1095$.
Threshold-effect test: $T_n=4.22$, bootstrap $p=0.022$.
\end{flushleft}
\end{sidewaystable*}

\begin{sidewaystable*}[p]
\centering
\setlength{\tabcolsep}{2.2pt}
\renewcommand{\arraystretch}{0.70}
\caption{All Candidate Controls: Lasso and Debiased Estimates for Nashville, 2000--2010}
\label{tb:citycoef-5360-2000-2010}
{\fontsize{5.6}{6.2}\selectfont
\begin{minipage}[t]{0.49\textwidth}
\centering
\textit{Panel A}\par\vspace{0pt}
\begin{tabular}{@{}p{0.245\textwidth}rrrr@{}}
\toprule
Control & $\beta_L$ & $\delta_L$ & $\beta_D$ & $\delta_D$ \\
\midrule
At least 175\% poverty & $0$ & $0$ & $0.034$ {\tiny ($0.19$)} & $-0.3^{***}$ {\tiny ($0.11$)} \\
Below poverty line & $0$ & $0$ & $0.19$ {\tiny ($0.36$)} & $0.74^{*}$ {\tiny ($0.4$)} \\
Bicycle commute & $0$ & $0$ & $-2.6$ {\tiny ($2.3$)} & $-9.8^{*}$ {\tiny ($5.7$)} \\
Bottled-gas heat & $0$ & $0$ & $1.1^{**}$ {\tiny ($0.55$)} & $-1.9$ {\tiny ($2.8$)} \\
Car commute & $-0.17$ & $0$ & $-0.43$ {\tiny ($0.45$)} & $-0.08$ {\tiny ($0.099$)} \\
Coal/coke heat & -- & -- & -- & -- \\
Commute 60+ min & $0$ & $0$ & $-0.46$ {\tiny ($0.35$)} & $0.078$ {\tiny ($0.49$)} \\
Commute under 30 min & $0$ & $0$ & $0.1$ {\tiny ($0.093$)} & $0.033$ {\tiny ($0.11$)} \\
Depart 9 a.m. or later & $0$ & $0$ & $0.45^{*}$ {\tiny ($0.24$)} & $-0.45$ {\tiny ($0.34$)} \\
Depart before 7 a.m. & $0$ & $0$ & $-0.1$ {\tiny ($0.16$)} & $-0.18$ {\tiny ($0.21$)} \\
Earnings receipt & $0$ & $0$ & $0.019$ {\tiny ($0.35$)} & $-0.032$ {\tiny ($0.029$)} \\
Electric heat & $-0.11$ & $0$ & $-0.16^{*}$ {\tiny ($0.099$)} & $-0.16$ {\tiny ($0.12$)} \\
English-only persons age 5+ & $0$ & $0$ & $0.33$ {\tiny ($0.41$)} & $-0.069$ {\tiny ($0.15$)} \\
Four+ vehicles & $0$ & $0$ & $-0.24$ {\tiny ($0.61$)} & $1.5$ {\tiny ($1.5$)} \\
Fuel-oil heat & $-0.38$ & $0$ & $-2.1$ {\tiny ($3.1$)} & $0.95$ {\tiny ($1.6$)} \\
High-density housing & $0$ & $0$ & $0.014$ {\tiny ($0.11$)} & $-0.11$ {\tiny ($0.45$)} \\
Householder age 15--34 & $0$ & $0$ & $0.14$ {\tiny ($0.26$)} & $-4.9\!\times\!10^{-6}$ {\tiny ($0.43$)} \\
Householder age 65+ & $-0.042$ & $0$ & $-0.34$ {\tiny ($0.23$)} & $0.34$ {\tiny ($0.31$)} \\
Housing built 41+ yrs ago & $0$ & $0$ & $0.051$ {\tiny ($0.054$)} & $0.13^{**}$ {\tiny ($0.062$)} \\
Housing built prev. 20 yrs & $0.19$ & $-0.14$ & $0.23^{**}$ {\tiny ($0.11$)} & $0.017$ {\tiny ($0.16$)} \\
Intercept & $0.11$ & $0$ & $-1.9$ {\tiny (n/a)} & $-0.029$ {\tiny ($0.023$)} \\
Lacking kitchen & $0$ & $0$ & $-0.39$ {\tiny ($1.2$)} & $0.64$ {\tiny ($1.2$)} \\
Lacking plumbing & $0$ & $0$ & $0.74$ {\tiny ($1.7$)} & $0.28$ {\tiny ($2$)} \\
Log mean fam. income & $0$ & $0$ & $0.091$ {\tiny ($0.092$)} & $-0.00083$ {\tiny ($0.0026$)} \\
Log tract land area & $0.056$ & $0$ & $0.1^{***}$ {\tiny ($0.024$)} & $-0.014$ {\tiny ($0.021$)} \\
Low-density housing & $0$ & $0$ & $-0.036$ {\tiny ($0.065$)} & $-0.12$ {\tiny ($0.28$)} \\
\bottomrule
\end{tabular}
\end{minipage}\hfill
\begin{minipage}[t]{0.49\textwidth}
\centering
\textit{Panel B}\par\vspace{0pt}
\begin{tabular}{@{}p{0.245\textwidth}rrrr@{}}
\toprule
Control & $\beta_L$ & $\delta_L$ & $\beta_D$ & $\delta_D$ \\
\midrule
Log mean public-assist. income & $0$ & $0$ & $0.00033$ {\tiny ($0.006$)} & $-0.003$ {\tiny ($0.0048$)} \\
Log mean retirement income & $0$ & $0$ & $0.015$ {\tiny ($0.031$)} & $0.007$ {\tiny ($0.012$)} \\
Log mean Social Security income & $0$ & $0$ & $0.074$ {\tiny ($0.071$)} & $0.0027$ {\tiny ($0.0036$)} \\
Log mean wage income & $0$ & $0$ & $-0.069$ {\tiny ($0.091$)} & $-0.0027$ {\tiny ($0.002$)} \\
Middle-density housing & $0$ & $0$ & $0.017$ {\tiny ($0.39$)} & $-0.23$ {\tiny ($0.65$)} \\
Motorcycle commute & $4.4$ & $0$ & $14$ {\tiny ($9.2$)} & $-1.6$ {\tiny ($10$)} \\
Moved in 21+ yrs ago & $-0.066$ & $0$ & $-0.12$ {\tiny ($0.21$)} & $0.96^{***}$ {\tiny ($0.37$)} \\
Moved in prev. 5 yrs & $0$ & $0$ & $0.26$ {\tiny ($0.23$)} & $-0.16$ {\tiny ($0.11$)} \\
No-fuel heat & $0$ & $0$ & $-0.68$ {\tiny ($2.8$)} & $2.2$ {\tiny ($3$)} \\
Pays extra utilities & $0$ & $0$ & $-0.11$ {\tiny ($0.084$)} & $-0.11$ {\tiny ($0.088$)} \\
Public assist. receipt & $0$ & $0$ & $0.64$ {\tiny ($1.4$)} & $0.26$ {\tiny ($0.83$)} \\
Public transit share & $0$ & $0$ & $0.92$ {\tiny ($0.99$)} & $2$ {\tiny ($1.2$)} \\
Renter-occupied share & $0$ & $0$ & $-0.15$ {\tiny ($0.12$)} & $0.053$ {\tiny ($0.17$)} \\
Retirement receipt & $-0.24$ & $0$ & $-0.47$ {\tiny ($0.4$)} & $-0.35$ {\tiny ($0.34$)} \\
Single-unit share & $0$ & $0$ & $0.34^{**}$ {\tiny ($0.13$)} & $0.032$ {\tiny ($0.16$)} \\
Social Security receipt & $-0.075$ & $0$ & $0.19$ {\tiny ($0.4$)} & $-0.075$ {\tiny ($0.28$)} \\
Solar heat & -- & -- & -- & -- \\
Unemployment rate & $0$ & $0$ & $0.41$ {\tiny ($0.44$)} & $0.21$ {\tiny ($0.47$)} \\
Utility-gas heat & $0$ & $0$ & $0.1^{*}$ {\tiny ($0.061$)} & $-0.1$ {\tiny ($0.14$)} \\
Vacancy share & $0.47$ & $0$ & $1.4^{*}$ {\tiny ($0.83$)} & $-0.17$ {\tiny ($0.89$)} \\
Wage receipt & $0$ & $0$ & $0.48$ {\tiny ($0.36$)} & $-0.027$ {\tiny ($0.029$)} \\
Walked to work & $0$ & $0$ & $0.68$ {\tiny ($0.71$)} & $-2.4^{**}$ {\tiny ($1.1$)} \\
Wood heat & $0$ & $-0.83$ & $1.8$ {\tiny ($2.6$)} & $-7^{**}$ {\tiny ($2.9$)} \\
Worked at home & $0.13$ & $0$ & $0.62$ {\tiny ($1$)} & $1.2$ {\tiny ($0.93$)} \\
Zero/one vehicle & $0$ & $0$ & $-0.15$ {\tiny ($0.11$)} & $-0.034$ {\tiny ($0.067$)} \\
\bottomrule
\end{tabular}
\end{minipage}
}
\begin{flushleft}
\fontsize{6.2}{6.8}\selectfont
Notes: Subscripts $L$ and $D$ denote initial Lasso and
debiased threshold Lasso estimates, respectively.
Heteroskedasticity-robust standard errors are in parentheses.
The Intercept row reports $\mu$ in the $\beta$ columns
and $\gamma$ in the $\delta$ columns.
Dashes indicate excluded controls, and n/a indicates
that no standard error is reported.
Stars use unadjusted two-sided Wald $p$-values:
$^{*}p<0.10$, $^{**}p<0.05$, and $^{***}p<0.01$.
$N=216$, $\widehat\tau=0.2050$, $\tau_\ast=0.2050$.
Threshold-effect test: $T_n=4.64$, bootstrap $p=0.006$.
\end{flushleft}
\end{sidewaystable*}

\begin{sidewaystable*}[p]
\centering
\setlength{\tabcolsep}{2.2pt}
\renewcommand{\arraystretch}{0.70}
\caption{All Candidate Controls: Lasso and Debiased Estimates for Chicago, 1990--2000}
\label{tb:citycoef-1600-1990-2000}\label{tb:citycoef-1600}
{\fontsize{5.6}{6.2}\selectfont
\begin{minipage}[t]{0.49\textwidth}
\centering
\textit{Panel A}\par\vspace{0pt}
\begin{tabular}{@{}p{0.245\textwidth}rrrr@{}}
\toprule
Control & $\beta_L$ & $\delta_L$ & $\beta_D$ & $\delta_D$ \\
\midrule
At least 175\% poverty & $0.83$ & $0$ & $1$ {\tiny ($0.65$)} & $-0.062$ {\tiny ($0.21$)} \\
Below poverty line & $-0.42$ & $1.2$ & $-2.1$ {\tiny ($1.5$)} & $3.1^{**}$ {\tiny ($1.6$)} \\
Bicycle commute & $-17$ & $22$ & $-21^{***}$ {\tiny ($5$)} & $26^{***}$ {\tiny ($5.2$)} \\
Bottled-gas heat & $-0.45$ & $-0.00078$ & $-0.8$ {\tiny ($0.86$)} & $-0.81$ {\tiny ($2$)} \\
Car commute & $-0.2$ & $0$ & $-0.64^{*}$ {\tiny ($0.34$)} & $-0.026$ {\tiny ($0.075$)} \\
Coal/coke heat & $0$ & $0$ & $-4.4$ {\tiny ($7.6$)} & $4.4$ {\tiny ($7.5$)} \\
Commute 60+ min & $-1.2$ & $0.16$ & $-1.4^{*}$ {\tiny ($0.86$)} & $0.43$ {\tiny ($0.9$)} \\
Commute under 30 min & $-0.054$ & $0$ & $-0.12$ {\tiny ($0.29$)} & $0.065$ {\tiny ($0.27$)} \\
Depart 9 a.m. or later & $-1.5$ & $1.5$ & $-1.7^{***}$ {\tiny ($0.56$)} & $1.6^{***}$ {\tiny ($0.6$)} \\
Depart before 7 a.m. & $0.2$ & $-0.56$ & $0.54$ {\tiny ($0.44$)} & $-0.68^{*}$ {\tiny ($0.37$)} \\
Earnings receipt & $1.4$ & $0$ & $1.8$ {\tiny ($1.2$)} & $0.12$ {\tiny ($0.11$)} \\
Electric heat & $-0.19$ & $0.19$ & $-0.44$ {\tiny ($0.37$)} & $0.14$ {\tiny ($0.34$)} \\
English-only persons age 5+ & $-0.069$ & $0$ & $-0.083$ {\tiny ($0.46$)} & $0.012$ {\tiny ($0.38$)} \\
Four+ vehicles & $0$ & $1.2$ & $-0.0051$ {\tiny ($1.1$)} & $1.3$ {\tiny ($1.2$)} \\
Fuel-oil heat & $-1.6$ & $2$ & $-2.3^{**}$ {\tiny ($1.2$)} & $2.2^{*}$ {\tiny ($1.2$)} \\
High-density housing & $-0.21$ & $0$ & $-0.48^{*}$ {\tiny ($0.26$)} & $-0.18$ {\tiny ($0.37$)} \\
Householder age 15--34 & $0.63$ & $-0.39$ & $0.64$ {\tiny ($0.52$)} & $-0.5$ {\tiny ($0.42$)} \\
Householder age 65+ & $1.7$ & $-0.74$ & $2.6^{***}$ {\tiny ($0.76$)} & $-1.6^{***}$ {\tiny ($0.56$)} \\
Housing built 41+ yrs ago & $-0.0057$ & $0.11$ & $-0.13$ {\tiny ($0.12$)} & $0.24^{**}$ {\tiny ($0.12$)} \\
Housing built prev. 20 yrs & $0.31$ & $0$ & $0.28^{**}$ {\tiny ($0.13$)} & $0.04$ {\tiny ($0.14$)} \\
Intercept & $1.3$ & $0$ & $4.6$ {\tiny (n/a)} & $-0.0065$ {\tiny ($0.023$)} \\
Lacking kitchen & $-0.37$ & $0$ & $-1.2$ {\tiny ($1.4$)} & $0.88$ {\tiny ($1.6$)} \\
Lacking plumbing & $-12$ & $12$ & $-14^{***}$ {\tiny ($3.5$)} & $13^{***}$ {\tiny ($2.6$)} \\
Log mean fam. income & $-0.17$ & $0$ & $-0.19^{**}$ {\tiny ($0.094$)} & $-0.00071$ {\tiny ($0.0021$)} \\
Log tract land area & $0.082$ & $-0.021$ & $0.092^{***}$ {\tiny ($0.024$)} & $-0.031$ {\tiny ($0.026$)} \\
Low-density housing & $0$ & $0$ & $-0.3$ {\tiny ($0.26$)} & $0.0028$ {\tiny ($0.17$)} \\
\bottomrule
\end{tabular}
\end{minipage}\hfill
\begin{minipage}[t]{0.49\textwidth}
\centering
\textit{Panel B}\par\vspace{0pt}
\begin{tabular}{@{}p{0.245\textwidth}rrrr@{}}
\toprule
Control & $\beta_L$ & $\delta_L$ & $\beta_D$ & $\delta_D$ \\
\midrule
Log mean public-assist. income & $-0.0091$ & $0$ & $-0.012$ {\tiny ($0.013$)} & $0.0023$ {\tiny ($0.011$)} \\
Log mean retirement income & $-0.00094$ & $0$ & $0.0058$ {\tiny ($0.051$)} & $-0.0035$ {\tiny ($0.0086$)} \\
Log mean Social Security income & $-0.2$ & $0$ & $-0.45^{**}$ {\tiny ($0.19$)} & $0.003$ {\tiny ($0.002$)} \\
Log mean wage income & $0.086$ & $0$ & $0.037$ {\tiny ($0.15$)} & $0.00052$ {\tiny ($0.002$)} \\
Middle-density housing & $0$ & $-0.044$ & $-0.21$ {\tiny ($0.29$)} & $-0.45$ {\tiny ($0.4$)} \\
Motorcycle commute & $0$ & $-1.7$ & $2.9$ {\tiny ($11$)} & $-4.6$ {\tiny ($11$)} \\
Moved in 21+ yrs ago & $0.14$ & $0$ & $0.074$ {\tiny ($0.61$)} & $0.15$ {\tiny ($0.63$)} \\
Moved in prev. 5 yrs & $0.33$ & $0$ & $0.44$ {\tiny ($0.63$)} & $-0.12$ {\tiny ($0.54$)} \\
No-fuel heat & $-0.83$ & $1.4$ & $-4$ {\tiny ($3.8$)} & $4.6$ {\tiny ($3.8$)} \\
Pays extra utilities & $0.057$ & $0$ & $0.032$ {\tiny ($0.14$)} & $0.011$ {\tiny ($0.081$)} \\
Public assist. receipt & $-2$ & $2$ & $-2.7^{**}$ {\tiny ($1.3$)} & $2.9^{**}$ {\tiny ($1.3$)} \\
Public transit share & $0$ & $0.55$ & $-0.68$ {\tiny ($0.42$)} & $1$ {\tiny ($0.64$)} \\
Renter-occupied share & $0$ & $-0.071$ & $0.13$ {\tiny ($0.35$)} & $-0.18$ {\tiny ($0.33$)} \\
Retirement receipt & $0.38$ & $-0.34$ & $1.1$ {\tiny ($0.83$)} & $-0.72$ {\tiny ($0.53$)} \\
Single-unit share & $0.16$ & $-0.099$ & $0.46^{***}$ {\tiny ($0.15$)} & $-0.53^{***}$ {\tiny ($0.2$)} \\
Social Security receipt & $-1.4$ & $1.1$ & $-2.3^{***}$ {\tiny ($0.7$)} & $1.9^{***}$ {\tiny ($0.51$)} \\
Solar heat & $-4.7$ & $0$ & $-7$ {\tiny ($8.9$)} & $1.8$ {\tiny ($7.5$)} \\
Unemployment rate & $3.3$ & $-2.7$ & $4.2^{***}$ {\tiny ($1.2$)} & $-3.5^{***}$ {\tiny ($1.2$)} \\
Utility-gas heat & $0$ & $0$ & $-0.2$ {\tiny ($0.3$)} & $-0.02$ {\tiny ($0.043$)} \\
Vacancy share & $0.97$ & $0.42$ & $0.93$ {\tiny ($0.7$)} & $0.29$ {\tiny ($0.85$)} \\
Wage receipt & $-0.45$ & $-0.47$ & $0.026$ {\tiny ($1.4$)} & $-0.66^{***}$ {\tiny ($0.15$)} \\
Walked to work & $0.89$ & $-1.2$ & $1.6$ {\tiny ($0.98$)} & $-2.9^{***}$ {\tiny ($1.1$)} \\
Wood heat & $0$ & $0$ & $-0.11$ {\tiny ($7.7$)} & $0.26$ {\tiny ($6.1$)} \\
Worked at home & $0.2$ & $2.7$ & $-1.8$ {\tiny ($2.8$)} & $2.6$ {\tiny ($2.4$)} \\
Zero/one vehicle & $0$ & $-0.083$ & $0.13$ {\tiny ($0.46$)} & $-0.29$ {\tiny ($0.35$)} \\
\bottomrule
\end{tabular}
\end{minipage}
}
\begin{flushleft}
\fontsize{6.2}{6.8}\selectfont
Notes: Subscripts $L$ and $D$ denote initial Lasso and
debiased threshold Lasso estimates, respectively.
Heteroskedasticity-robust standard errors are in parentheses.
The Intercept row reports $\mu$ in the $\beta$ columns
and $\gamma$ in the $\delta$ columns.
Dashes indicate excluded controls, and n/a indicates
that no standard error is reported.
Stars use unadjusted two-sided Wald $p$-values:
$^{*}p<0.10$, $^{**}p<0.05$, and $^{***}p<0.01$.
$N=1284$, $\widehat\tau=0.0550$, $\tau_\ast=0.1865$.
Threshold-effect test: $T_n=10.17$, bootstrap $p=0.002$.
\end{flushleft}
\end{sidewaystable*}

\begin{sidewaystable*}[p]
\centering
\setlength{\tabcolsep}{2.2pt}
\renewcommand{\arraystretch}{0.70}
\caption{All Candidate Controls: Lasso and Debiased Estimates for Chicago, 2000--2010}
\label{tb:citycoef-1600-2000-2010}
{\fontsize{5.6}{6.2}\selectfont
\begin{minipage}[t]{0.49\textwidth}
\centering
\textit{Panel A}\par\vspace{0pt}
\begin{tabular}{@{}p{0.245\textwidth}rrrr@{}}
\toprule
Control & $\beta_L$ & $\delta_L$ & $\beta_D$ & $\delta_D$ \\
\midrule
At least 175\% poverty & $0.03$ & $0$ & $1.1^{*}$ {\tiny ($0.58$)} & $-0.48^{**}$ {\tiny ($0.23$)} \\
Below poverty line & $0$ & $0$ & $-0.25$ {\tiny ($0.58$)} & $0.52$ {\tiny ($0.58$)} \\
Bicycle commute & $0.13$ & $0$ & $2.8$ {\tiny ($2.4$)} & $-2$ {\tiny ($2.1$)} \\
Bottled-gas heat & $0.68$ & $0$ & $1.6^{*}$ {\tiny ($0.89$)} & $2^{*}$ {\tiny ($1.1$)} \\
Car commute & $-0.045$ & $0$ & $0.49$ {\tiny ($0.87$)} & $-0.042$ {\tiny ($0.16$)} \\
Coal/coke heat & $0$ & $-0.18$ & $2.9$ {\tiny ($5.9$)} & $-10$ {\tiny ($7.5$)} \\
Commute 60+ min & $-0.034$ & $0$ & $-0.68$ {\tiny ($0.48$)} & $0.35$ {\tiny ($0.49$)} \\
Commute under 30 min & $0$ & $0$ & $-0.18^{*}$ {\tiny ($0.1$)} & $0.067$ {\tiny ($0.17$)} \\
Depart 9 a.m. or later & $0$ & $0$ & $-0.37^{*}$ {\tiny ($0.2$)} & $0.27$ {\tiny ($0.34$)} \\
Depart before 7 a.m. & $0$ & $-0.19$ & $0.61$ {\tiny ($0.44$)} & $-0.91^{***}$ {\tiny ($0.33$)} \\
Earnings receipt & $0$ & $0$ & $0.8$ {\tiny ($1.1$)} & $0.0053$ {\tiny ($0.024$)} \\
Electric heat & $0$ & $0$ & $0.94^{**}$ {\tiny ($0.42$)} & $0.48$ {\tiny ($0.54$)} \\
English-only persons age 5+ & $0.0035$ & $0$ & $0.24^{**}$ {\tiny ($0.098$)} & $-0.24^{*}$ {\tiny ($0.13$)} \\
Four+ vehicles & $0.17$ & $0$ & $0.49$ {\tiny ($0.56$)} & $0.58$ {\tiny ($0.58$)} \\
Fuel-oil heat & $-0.61$ & $-1.4$ & $-1.6$ {\tiny ($2.9$)} & $0.67$ {\tiny ($3.2$)} \\
High-density housing & $0$ & $0.094$ & $-0.036$ {\tiny ($0.041$)} & $0.11$ {\tiny ($0.14$)} \\
Householder age 15--34 & $0$ & $0$ & $-0.046$ {\tiny ($0.24$)} & $0.074$ {\tiny ($0.37$)} \\
Householder age 65+ & $0$ & $0$ & $-0.096$ {\tiny ($0.22$)} & $0.51^{*}$ {\tiny ($0.27$)} \\
Housing built 41+ yrs ago & $0.08$ & $0$ & $0.003$ {\tiny ($0.069$)} & $0.075$ {\tiny ($0.07$)} \\
Housing built prev. 20 yrs & $0.11$ & $0$ & $0.16^{**}$ {\tiny ($0.08$)} & $0.13$ {\tiny ($0.1$)} \\
Intercept & $0.43$ & $0$ & $-2.3$ {\tiny (n/a)} & $-0.023$ {\tiny ($0.025$)} \\
Lacking kitchen & $0$ & $0$ & $-0.78$ {\tiny ($0.5$)} & $0.23$ {\tiny ($0.68$)} \\
Lacking plumbing & $0$ & $2$ & $-2.1$ {\tiny ($1.4$)} & $5.5^{**}$ {\tiny ($2.4$)} \\
Log mean fam. income & $0$ & $0$ & $-0.26^{***}$ {\tiny ($0.076$)} & $0.0001$ {\tiny ($0.0016$)} \\
Log tract land area & $0.057$ & $-0.035$ & $0.088^{***}$ {\tiny ($0.015$)} & $-0.054^{***}$ {\tiny ($0.017$)} \\
Low-density housing & $0$ & $0$ & $0.016$ {\tiny ($0.056$)} & $-0.36^{**}$ {\tiny ($0.18$)} \\
\bottomrule
\end{tabular}
\end{minipage}\hfill
\begin{minipage}[t]{0.49\textwidth}
\centering
\textit{Panel B}\par\vspace{0pt}
\begin{tabular}{@{}p{0.245\textwidth}rrrr@{}}
\toprule
Control & $\beta_L$ & $\delta_L$ & $\beta_D$ & $\delta_D$ \\
\midrule
Log mean public-assist. income & $0$ & $-0.0063$ & $0.0027$ {\tiny ($0.0023$)} & $-0.017^{***}$ {\tiny ($0.0063$)} \\
Log mean retirement income & $-0.00069$ & $0$ & $-0.012$ {\tiny ($0.028$)} & $0.0067$ {\tiny ($0.016$)} \\
Log mean Social Security income & $-0.079$ & $0$ & $-0.061^{***}$ {\tiny ($0.021$)} & $-0.021$ {\tiny ($0.036$)} \\
Log mean wage income & $0$ & $0$ & $0.25^{***}$ {\tiny ($0.091$)} & $0.0018$ {\tiny ($0.0016$)} \\
Middle-density housing & $0$ & $0$ & $0.011$ {\tiny ($0.033$)} & $-0.18$ {\tiny ($0.17$)} \\
Motorcycle commute & $0$ & $-2.7$ & $2.5$ {\tiny ($4.2$)} & $-8.9^{*}$ {\tiny ($4.6$)} \\
Moved in 21+ yrs ago & $0$ & $0$ & $0.32$ {\tiny ($0.27$)} & $0.36$ {\tiny ($0.31$)} \\
Moved in prev. 5 yrs & $0.63$ & $0$ & $0.92^{***}$ {\tiny ($0.34$)} & $-0.12$ {\tiny ($0.37$)} \\
No-fuel heat & $0$ & $0$ & $2.1^{**}$ {\tiny ($0.94$)} & $-3.6^{*}$ {\tiny ($2.1$)} \\
Pays extra utilities & $-0.082$ & $0$ & $-0.11$ {\tiny ($0.073$)} & $-0.034$ {\tiny ($0.11$)} \\
Public assist. receipt & $0$ & $0$ & $0.52$ {\tiny ($0.77$)} & $-0.32$ {\tiny ($0.84$)} \\
Public transit share & $0$ & $0.061$ & $0.9$ {\tiny ($1.3$)} & $-1.7^{**}$ {\tiny ($0.76$)} \\
Renter-occupied share & $-0.23$ & $0$ & $-0.32^{***}$ {\tiny ($0.11$)} & $-0.15$ {\tiny ($0.24$)} \\
Retirement receipt & $0$ & $-0.14$ & $0.16$ {\tiny ($0.27$)} & $-0.88^{***}$ {\tiny ($0.3$)} \\
Single-unit share & $0$ & $0$ & $-0.23$ {\tiny ($0.21$)} & $0.28$ {\tiny ($0.2$)} \\
Social Security receipt & $0$ & $0$ & $0.47$ {\tiny ($0.36$)} & $-0.4$ {\tiny ($0.31$)} \\
Solar heat & -- & -- & -- & -- \\
Unemployment rate & $0.17$ & $0.069$ & $0.12$ {\tiny ($0.8$)} & $0.46$ {\tiny ($0.96$)} \\
Utility-gas heat & $0$ & $0$ & $0.67^{**}$ {\tiny ($0.27$)} & $0.046$ {\tiny ($0.059$)} \\
Vacancy share & $0.27$ & $0$ & $0.42$ {\tiny ($0.32$)} & $-0.11$ {\tiny ($0.48$)} \\
Wage receipt & $0$ & $0$ & $0.00039$ {\tiny ($1.1$)} & $-0.00097$ {\tiny ($0.026$)} \\
Walked to work & $0$ & $0$ & $1.1$ {\tiny ($1.3$)} & $-2.8^{**}$ {\tiny ($1.1$)} \\
Wood heat & $-5.1$ & $0$ & $-9.8$ {\tiny ($6.6$)} & $0.85$ {\tiny ($7.1$)} \\
Worked at home & $0.45$ & $1.6$ & $2.1$ {\tiny ($1.4$)} & $-2.3$ {\tiny ($1.8$)} \\
Zero/one vehicle & $0$ & $0.048$ & $0.27$ {\tiny ($0.2$)} & $0.4$ {\tiny ($0.32$)} \\
\bottomrule
\end{tabular}
\end{minipage}
}
\begin{flushleft}
\fontsize{6.2}{6.8}\selectfont
Notes: Subscripts $L$ and $D$ denote initial Lasso and
debiased threshold Lasso estimates, respectively.
Heteroskedasticity-robust standard errors are in parentheses.
The Intercept row reports $\mu$ in the $\beta$ columns
and $\gamma$ in the $\delta$ columns.
Dashes indicate excluded controls, and n/a indicates
that no standard error is reported.
Stars use unadjusted two-sided Wald $p$-values:
$^{*}p<0.10$, $^{**}p<0.05$, and $^{***}p<0.01$.
$N=1154$, $\widehat\tau=0.1830$, $\tau_\ast=0.2480$.
Threshold-effect test: $T_n=4.52$, bootstrap $p=0.004$.
\end{flushleft}
\end{sidewaystable*}
\FloatBarrier

\section{Technical Details}
\label{supp:technical}

This section provides additional definitions and conditions
for the estimation and inference procedures developed in
the main paper. It details the nodewise approximate inverse
construction and specifies the grid conventions, bootstrap
implementation, and additional regularity conditions for
the threshold-effect test.

\subsection{Nodewise Approximate Inverse Construction}
\label{supp:nodewise-construction}

This subsection details the weighted nodewise Lasso
construction introduced in Section~\ref{constructprecison}
of the main paper. We construct approximate inverses of
the two regime-specific Gram matrices and combine them
to obtain $\widehat{\bm\Theta}(\tau)$, with a normalization
that yields the KKT bounds used in the debiasing analysis. We keep $X_<(\tau)$ and $X_\ge(\tau)$ for the $n\times p$ regime designs, and $\bm X(\tau)$ for the augmented $n\times2p$ design. Define
\[
 \widehat{\bm M}(\tau)=X_<(\tau)'X_<(\tau)/n,
 \quad \widehat{\bm N}(\tau)=X_\ge(\tau)'X_\ge(\tau)/n,
 \quad \bm A(\tau)=\bm M(\tau)^{-1},
 \quad \bm B(\tau)=\bm N(\tau)^{-1}.
\]
For $j=1,\ldots,p$, define
\[
 \gamma_{0,j}(\tau)=\bm M_{-j,-j}(\tau)^{-1}\bm M_{-j,j}(\tau),
 \qquad
 \widetilde\gamma_{0,j}(\tau)
     =\bm N_{-j,-j}(\tau)^{-1}\bm N_{-j,j}(\tau).
\]
Their residual vectors satisfy
\begin{align}
 X_<^{(j)}(\tau)
    &=X_<^{(-j)}(\tau)\gamma_{0,j}(\tau)
                                   +\upsilon^{(j)}(\tau),
       \label{eq:S23}\\
 X_\ge^{(j)}(\tau)
    &=X_\ge^{(-j)}(\tau)\widetilde\gamma_{0,j}(\tau)
                                   +\widetilde\upsilon^{(j)}(\tau).
       \label{eq:S24}
\end{align}
These are projections under the average population inner product. For example,
\[
 \frac1n\sum_i E[X_{i,-j}1\{Q_i<\tau\}\upsilon_i^{(j)}(\tau)]=0,
 \qquad z_j^2(\tau):=\frac1n\sum_i E[(\upsilon_i^{(j)}(\tau))^2],
\]
with analogous definitions for the upper regime and
$\widetilde z_j^2(\tau)$. Orthogonality need not hold separately for each observation when their distributions differ, nor conditionally on $Q_i$. The required moment and design conditions are those of Assumption~\ref{as6-nodewise} in the main paper.

For the lower regime, let
\begin{equation}\label{eq:S25}
 \widehat\gamma_j(\tau)\in\arg\min_{\gamma\in\mathbb R^{p-1}}
 \left\{\|X_<^{(j)}(\tau)-X_<^{(-j)}(\tau)\gamma\|_n^2
       +\lambda_{\mathrm{node}}|\widehat\Gamma_j(\tau)\gamma|_1\right\},
\end{equation}
where
$\widehat\Gamma_j(\tau)
  =\operatorname{diag}\{\|X_<^{(l)}(\tau)\|_n:l\ne j\}$.
%The norm uses denominator $n$ and does not center the columns. 
Let $\widehat{\bm C}_{jj}(\tau)=1$ and
$\widehat{\bm C}_{jk}(\tau)=-\widehat\gamma_j^{(k)}(\tau)$ for $k\ne j$, where the coordinates of $\gamma_j$ are the original covariate indices. Define
\begin{align*}
 \widehat z_j^2(\tau)
   &:=\|X_<^{(j)}(\tau)-X_<^{(-j)}(\tau)\widehat\gamma_j(\tau)\|_n^2
       +\frac{\lambda_{\mathrm{node}}}{2}
             |\widehat\Gamma_j(\tau)\widehat\gamma_j(\tau)|_1,\\
 \widehat{\bm Z}(\tau)^2
   &:=\operatorname{diag}\{\widehat z_1^2(\tau),\ldots,
                                              \widehat z_p^2(\tau)\}.
\end{align*}
Set
\begin{equation}\label{eq:S26}
 \widehat{\bm A}(\tau)
        :=\widehat{\bm Z}(\tau)^{-2}\widehat{\bm C}(\tau).
\end{equation}
Write $\widehat A_j(\tau)'$ for its $j$th row, and let
$\widetilde e_j$ be the $j$th coordinate vector in $\mathbb R^p$.
The KKT conditions give
\begin{equation}\label{eq:S27}
 ||\widehat A_j(\tau)'\widehat{\bm M}(\tau)-\widetilde e_j'||_\infty
 \le\frac{\lambda_{\mathrm{node}}|\widehat\Gamma_j(\tau)|_\infty}
                {2\widehat z_j^2(\tau)}.
\end{equation}
In particular, the diagonal entry of $\widehat A_j(\tau)'\widehat{\bm M}(\tau)$ is exactly one
under this scale convention.

For the upper regime, define
\begin{equation}\label{eq:S28}
 \widehat{\widetilde\gamma}_j(\tau)
   \in\arg\min_{\gamma\in\mathbb R^{p-1}}
 \left\{\|X_\ge^{(j)}(\tau)-X_\ge^{(-j)}(\tau)\gamma\|_n^2
    +\lambda_{\mathrm{node}}
                   |\widehat{\widetilde\Gamma}_j(\tau)\gamma|_1\right\},
\end{equation}
where
$\widehat{\widetilde\Gamma}_j(\tau)
  =\operatorname{diag}\{\|X_\ge^{(l)}(\tau)\|_n:l\ne j\}$.
Set $\widehat{\widetilde{\bm C}}_{jj}(\tau)=1$ and
$\widehat{\widetilde{\bm C}}_{jk}(\tau)
       =-\widehat{\widetilde\gamma}_j^{(k)}(\tau)$ for $k\ne j$,
and define
\begin{align*}
 \widehat{\widetilde z}_j^2(\tau)
   &:=\|X_\ge^{(j)}(\tau)
           -X_\ge^{(-j)}(\tau)\widehat{\widetilde\gamma}_j(\tau)\|_n^2
       +\frac{\lambda_{\mathrm{node}}}{2}
        |\widehat{\widetilde\Gamma}_j(\tau)
                               \widehat{\widetilde\gamma}_j(\tau)|_1,\\
 \widehat{\widetilde{\bm Z}}(\tau)^2
   &:=\operatorname{diag}\{\widehat{\widetilde z}_1^2(\tau),\ldots,
                                  \widehat{\widetilde z}_p^2(\tau)\}.
\end{align*}
Then
\begin{equation}\label{eq:S29}
 \widehat{\bm B}(\tau)
 :=\widehat{\widetilde{\bm Z}}(\tau)^{-2}
                              \widehat{\widetilde{\bm C}}(\tau).
\end{equation}
Combining these matrices yields \eqref{eq:theta-block}. They need not be symmetric. The scales for the required rows are bounded away from zero with probability tending to one by Lemma~\ref{lem:nodewise-main}. %If a required empirical column norm or residual scale is zero, define the corresponding row estimate arbitrarily. The events used in its proof exclude this exceptional case with probability tending to one.

\subsection{Details for the Threshold-Effect Test}
\label{supp:threshold_test}

This section complements Section~\ref{sec:threshold_test} by
specifying the grid conventions and bootstrap implementation
and stating the additional conditions used to approximate the
maximum statistic.

\subsubsection{Threshold Grid}\label{sec:null-grid}
Let $\mathcal G_n\subset\mathbb T=[t_0,t_1]$ be the testing grid, $K_n=|\mathcal G_n|$, and $d_n=pK_n$. For a prespecified
deterministic grid, assume $1\le K_n\le Cn^{C_G}$ for fixed positive constants $C$ and $C_G$. Alternatively, use the observed value grid
\[
 \mathcal G_n=(\{Q_1,\ldots,Q_n\}\cap\mathbb T)\cup\{t_1\},
 \qquad K_n\le n+1.
\]
With the indicator convention $\bm1\{Q_i<\tau\}$, this grid
represents every sample partition generated by $\tau\in\mathbb T$, almost surely. The testing grid need not contain
$\tau_\ast$, whereas the estimation grid $\mathbb T_n$ includes
the reference point by definition.

%Fix measurable rules for selecting Lasso and nodewise solutions such that identical design matrices yield identical solutions. For example, select the minimizer with the smallest Euclidean norm whenever the minimizer is not unique. Use tuning constants that do not vary with the numerical value of $\tau$ within a sample partition. Then all fitted quantities used in $T_n$ and $T_n^*$ depend on $\tau$ only through $\bm X(\tau)$.

To justify the observed value implementation, we introduce an auxiliary deterministic grid used only in the proof:
\[
 N_n=\lceil n^3\rceil,\qquad
 \mathcal G_n^\circ=
 \left\{t_0+\frac{k(t_1-t_0)}{N_n}:k=0,\ldots,N_n\right\}.
\]
For a deterministic implementation, set $\mathcal G_n^\circ=\mathcal G_n$. Write $K_n^\circ=|\mathcal G_n^\circ|$,
$d_n^\circ=pK_n^\circ$, and $\mathcal I_n^\circ=\{(j,\tau):1\le j\le p,\tau\in\mathcal G_n^\circ\}$. In either case, $\log(2d_n^\circ n)=O(\ell_n)$. The fixed-threshold fitting rules in the main paper also define the fitted quantities on $\mathcal G_n^\circ$. For the observed-value implementation, the proof of
Proposition~\ref{prop:threshold-size} establishes that the two
grids yield identical test statistics and bootstrap statistics
with probability tending to one. This allows the Gaussian
approximation to be performed on a deterministic index set.
The auxiliary grid need not be evaluated in computation.

\subsubsection{Bootstrap Implementation}\label{sec:nulltest}

Use the null-restricted residuals, score estimates, and centered
variance estimates defined in \eqref{eq:test-null-fit}--\eqref{eq:test-null-variance} of the main paper. In particular, the score mean is denoted by $m_j(\tau)=n^{-1}\sum_{i=1}^n\widehat m_{i,j}(\tau)$, so the standardized score is
\[
\widehat\psi_{i,j}(\tau)
=\frac{\widehat m_{i,j}(\tau)-m_j(\tau)}
{\widehat\sigma_{\delta,j}(\tau)}.
\]

Each bootstrap replication uses the same multipliers across all thresholds and coefficient coordinates. If an estimated variance is zero, omit that threshold-coordinate from both maxima. If no pairs remain, set both maxima to zero and do not reject. Under the null conditions, all estimated variances on the testing grid are positive with probability tending to one.

For computation, draw $B$ independent sets of Gaussian multipliers and obtain $T_n^{*(1)},\ldots,T_n^{*(B)}$. We then compute the Monte Carlo bootstrap $p$-value
\begin{equation}\label{eq:test-mc-pvalue}
    \widehat p_B
    =
    \frac{
        1+
        \sum_{b=1}^B
        \mathbf 1\{T_n^{*(b)}\ge T_n\}
    }{
        B+1
    },
\end{equation}
and reject when $\widehat p_B\le\zeta$.

\subsubsection{Conditions for Uniform Gaussian Approximation}

The maximum over $pK_n$ coordinates requires a high-dimensional Gaussian approximation in addition to the scalar normal approximation of
Theorem~\ref{thm3}. For deterministic $\tau\in\mathbb T$, define
\[
 m_{i,j}(\tau)=\Theta_{p+j}(\tau)'\bm X_i(\tau)U_i,
 \qquad
 \sigma_{\delta,j}^2(\tau)=\frac1n\sum_{i=1}^nE[m_{i,j}(\tau)^2],
 \qquad \psi_{i,j}(\tau)=\frac{m_{i,j}(\tau)}{\sigma_{\delta,j}(\tau)}.
\]
All population functions are defined at deterministic arguments. Let $H_\delta=\{p+1,\ldots,2p\}$ and
\[
 \bar s_\delta=\sup_{\tau\in\mathbb T}\max_{r\in H_\delta}s_r(\tau),
 \qquad s_r(\tau)=|S_r(\tau)|,
 \qquad S_r(\tau)=\{k:\Theta_{r,k}(\tau)\ne0\}.
\]
The nodewise assumptions apply with $H=H_\delta$ and $\bar s=\bar s_\delta$, including the corresponding $A$ and $B$ row
conditions in Section~\ref{constructprecison}.
Write $\mathcal I_n=\{(j,\tau):j\le p,\tau\in\mathcal G_n\}$ for the
implemented index set, and use $\mathcal I_n^\circ$ for the deterministic
proof index set. For $a=(j,\tau)$, use $\psi_{i,a}=\psi_{i,j}(\tau)$ and the analogous notation for estimated scores and variances. Finally, let $M_{XX,n}=(EM_{XX}^2)^{1/2}$, with $M_{XX}$ as in
Assumption~\ref{as1}(vi).

\begin{assm}[Uniform max-test conditions]
\label{as:max-test} 
The following conditions hold uniformly over
$\alpha_0\in\mathcal A_{\ell_0}^{(1)}(s_0)$.
(i) There exist constants $0<c_{\mathrm{test}}<C_{\mathrm{test}}<\infty$ such that
\[
    c_{\mathrm{test}}
    \le
    \inf_{\tau\in\mathcal G_n^\circ}
    \min_{1\le j\le p}\sigma_{\delta,j}^2(\tau)
    \le
    \sup_{\tau\in\mathcal G_n^\circ}
    \max_{1\le j\le p}\sigma_{\delta,j}^2(\tau)
    \le
    C_{\mathrm{test}}.
\] 

\noindent(ii) The sparsity and envelope rates satisfy
\begin{align}
 \frac{(s_0+\bar s_\delta)\ell_n^{3/2}}{\sqrt n}
       &\longrightarrow0,\label{eq:max-linear-rate}\\
 \ell_n^2\{\bar s_\delta^{3/2}b_n
                  +\bar s_\delta\sqrt{s_0}\,b_n\}
       &\longrightarrow0,\label{eq:max-cov-rate}\\
 \sqrt{s_0}\,M_{XX,n}b_n&=O(1).\label{eq:max-envelope-rate}
\end{align}

\noindent(iii) For a deterministic sequence $B_n\ge1$,
\[
 \max_{a\in\mathcal I_n^\circ}\frac1n\sum_{i=1}^n
       E|\psi_{i,a}|^{2+k}\le B_n^k,\quad k=1,2,
 \qquad
 \max_{i\le n}\left(E\max_{a\in\mathcal I_n^\circ}
                          |\psi_{i,a}|^4\right)^{1/4}\le B_n,
\]
and
\begin{equation}\label{eq:max-cck-rate}
                \frac{B_n^4\ell_n^7}{n}\longrightarrow0.
\end{equation}
\end{assm}

Assumption \ref{as:max-test} (i) ensures that the population score variances are uniformly bounded and bounded away from zero on the deterministic grid used in the Gaussian approximation. Assumption \ref{as:max-test} (ii) controls the uniform Lasso, nodewise, residual-estimation, and covariance errors. Assumption \ref{as:max-test} (iii) is a finite moment condition on scores indexed by a deterministic grid and does not impose independence between $X_i$ and $Q_i$. For the observed value implementation, its index set is the auxiliary grid $\mathcal G_n^\circ$, rather than the random set $\mathcal G_n$. A bound with a supremum over all $\tau\in\mathbb T$ is sufficient, but is not required here. We use Proposition~2.1 of \citet{da801549-3c84-30a7-aaa4-ce3f2dd2650d}, under its polynomial moment condition with exponent four, for the Gaussian approximation. The Gaussian comparison and anti-concentration steps use \citet{chernozhukov2015comparison}. The rates in part (ii) replace the scalar contrast rate restrictions in Assumption~\ref{asnd}(ii) for this test, and one need not set $h=p$ in those
scalar restrictions.

\section{Auxiliary Results and Proofs}\label{supp:proofs}

\subsection{Auxiliary Lemmas}\label{appendixA}

This section collects the concentration inequalities used to derive the oracle inequalities. We use the tuning parameters in the main text: $\ell_n=\log(2pn)$, $b_n=\sqrt{\ell_n/n}$,
$\lambda_n=Ab_n$, and $\rho_n=c_\rho\lambda_n$.

\subsubsection{Concentration Inequalities}

We recall the concentration inequality of \citet{chernozhukov2014gaussian} and \citet{chernozhukov2015comparison}, as formulated in Lemma B.2 of the supplement to \citet{chiang_rodrigue_sasaki_2023}, to derive oracle inequalities. We then extend it to partial sums in a deterministic observation order.

\begin{lem}[A Concentration Inequality]
\label{concenine}
Let $\{X_i\}_{i=1}^n$ be $p$-dimensional independent random vectors, let $B = \sqrt{E\left[\max_{1\le i\le n}\max_{1\le j\le p}\left|X_{i}^{(j)}\right|^2\right]}$ and $\sigma^2 =  \max_{1\le j\le p}1/n\sum_{i=1}^nE\left[\left(X_{i}^{(j)}\right)^2\right].$ Suppose $B<\infty$. For some constant $C>0$, with probability at least $1-C(\log n)^{-1},$
    $$\max_{1\le j\le p}\left|\frac{1}{n}\sum_{i=1}^n\left(X_{i}^{(j)}-E\left[X_{i}^{(j)}\right]\right)\right| \lesssim \sqrt{\frac{\sigma^2\log(p\vee n)}{n}} + \frac{B\log(p\vee n)}{n}.$$
\end{lem} 

\noindent The logarithmic factors in the following results are bounded by constant multiples of $\ell_n$. Thus, no restriction $p>n$ is needed.

\begin{lem}[A Concentration Inequality for Partial Sums of Random Variables]\label{conpart}
Let $\{X_i\}_{i=1}^n$, $B$, and $\sigma^2$ be as in
Lemma~\ref{concenine}. With probability at least $1-C(\log n)^{-1}$,
the following inequalities hold: 

$$(i) \, \max_{1\le j\le p}\max_{1\le k\le n}\left|\frac{1}{n}\sum_{i=1}^k\left(X_{i}^{(j)}-E\left[X_{i}^{(j)}\right]\right)\right| \lesssim \sqrt{\frac{\sigma^2\log(pn)}{n}} + \frac{B\log(pn)}{n},$$ 

$$(ii) \, \max_{1\le j\le p}\max_{1\le q \le k\le n}\left|\frac{1}{n}\sum_{i=q}^k\left(X_{i}^{(j)}-E\left[X_{i}^{(j)}\right]\right)\right| \lesssim \sqrt{\frac{\sigma^2\log(pn^2)}{n}} + \frac{B\log(pn^2)}{n}.$$

\end{lem}

\begin{proof}[Proof of Lemma \ref{conpart}]
Let $\Xi_{n,n}$ be the deterministic upper triangular matrix
\begin{equation}\label{simmat}\Xi_{n,n} = 
\begin{pmatrix}
1 &1 & \cdots & 1\\
0 & 1 & \cdots & 1 \\
\vdots  & \vdots  & \ddots & \vdots  \\
0 & 0 & \cdots & 1
\end{pmatrix}.
\end{equation}
Denote its $(i,k)$ entry by $\xi_i^{(k)}=1\{i\le k\}$, and the $(i,q)$ entry of $\Xi_{n,n}'$ by
$\widetilde\xi_i^{(q)}=1\{q\le i\}$.

\noindent For (i), write
\begin{equation*}
\sum_{i=1}^k\{X_i^{(j)}-E[X_i^{(j)}]\}
 =\sum_{i=1}^n\{X_i^{(j)}-E[X_i^{(j)}]\}\xi_i^{(k)}.
\end{equation*} 
The $pn$-dimensional vectors with coordinates
$X_i^{(j)}\xi_i^{(k)}$ are independent across $i$.
Their maximum average second moment is $\sigma^2$, because $k=n$ is included, and their maximum-envelope quantity is $B$. Applying Lemma~\ref{concenine} with dimension $pn$ proves (i).

%and they are bounded as $\max_{1\le j \le p}\max_{1\le k\le n} 1/n\sum_{i=1}^n Var\left( X_i^{(j)}\xi_i^{(k)}\right)=\max_{1\le j \le p}\max_{1\le k\le n}\\ 1/n\sum_{i=1}^kVar\left(X_i^{(j)}\right) \le\max_{1\le j \le p} 1/n\sum_{i=1}^n Var\left(X_i^{(j)}\right) \le \max_{1\le i \le n}\max_{1\le j \le p} Var\left( X_i^{(j)}\right)$ and  $ E\left[\max_{1\le i\le n}\max_{1\le k\le n}\max_{1\le j\le p}\left|X_{i}^{(j)}\xi_i^{(k)}\right|^2\right]$ are bounded under Assumption \ref{as1} (iii), 

\noindent For (ii), similarly write
\[
 \sum_{i=q}^k\{X_i^{(j)}-E[X_i^{(j)}]\}
 =\sum_{i=1}^n\{X_i^{(j)}-E[X_i^{(j)}]\}
                   \xi_i^{(k)}\widetilde\xi_i^{(q)}.
\]
The vectors with coordinates
$X_i^{(j)}\xi_i^{(k)}\widetilde\xi_i^{(q)}$, indexed by $j$ and $1\le q\le k\le n$, are independent across $i$. Their dimension is at most $pn^2$, and their corresponding quantities remain $\sigma^2$ and $B$, since $(q,k)=(1,n)$ is included. Lemma~\ref{concenine} proves (ii). A union bound gives the joint probability statement.
\end{proof}

\subsection{Proofs of Main Results}

This section proves the paper's main results. We first give the auxiliary lemmas for the oracle inequalities, the empirical uniform adaptive restricted eigenvalue condition, and the node-wise lemmas for constructing the debiased estimator. We then give the moment bounds required for uniform inference and prove Theorem \ref{thm3} and Theorem \ref{thm4}, and the propositions on the power of the proposed test.

\subsubsection{Proofs for Section \ref{them1}}

\noindent Fix $0<\mu<1$, choose a fixed $\mu_1>0$ such that $\mu_2:=2\mu_1/\sqrt{C_3}\le\mu$, and take $A$ sufficiently large. These constants do not depend on $n$ or $p$.

\begin{lem}[Regularized events $\mathbb{A}_1$ and $\mathbb{A}_2$]
\label{lemmaprobA3}Suppose that Assumption \ref{as1} holds. Define 
\begin{align*}
 \mathbb{A}_1&=\left\{\max_{1\le j\le p}\frac1n\sum_{i=1}^n
          (X_i^{(j)})^2\le C_2^{1/2}+\mu_1\lambda_n\right\},\\
 \mathbb{A}_2&=\left\{\min_{1\le j\le p}\frac1n\sum_{i=1}^n
          (X_i^{(j)}(t_0))^2\ge C_3-\mu_1\lambda_n\right\},\\
 \mathbb{A}_2'&=\left\{\min_{1\le j\le p}\frac1n\sum_{i=1}^n
          (X_i^{(j)})^2\ge C_3-\mu_1\lambda_n\right\}.
\end{align*}
Then $\mathbb{A}_2\subseteq \mathbb{A}_2'$ and, for sufficiently large $n$,
\[
 P(\mathbb{A}_1)\ge1-C(\log n)^{-1},\quad
 P(\mathbb{A}_2)\ge1-C(\log n)^{-1},\quad
 P(\mathbb{A}_2')\ge1-C(\log n)^{-1},
\]
and $P(\mathbb{A}_1\cap \mathbb{A}_2)\ge1-C(\log n)^{-1}$.
%\text{$\mathbb{A}_1$ and $\mathbb{A}_2$ hold with probability at least $1-C(\log n)^{-1}.$}
\end{lem}

\begin{proof}[Proof of Lemma \ref{lemmaprobA3}]

\noindent Apply Lemma~\ref{concenine} to the coordinate squares. Its second-moment
quantity is
\[
 \sigma^2=\max_{1\le j\le p}\frac1n\sum_{i=1}^n
                         E[(X_i^{(j)})^4]\le C_2,
\]
and its envelope quantity is $B=\sqrt{E M_{XX}^2}$.
Assumption~\ref{as1}(vi), $\log(p\vee n)\le\ell_n$, and $\lambda_n=Ab_n$ imply, for sufficiently large $A$ and $n$,
\[
 \max_{1\le j\le p}\left|\frac1n\sum_{i=1}^n
          \bigl\{(X_i^{(j)})^2-E[(X_i^{(j)})^2]\bigr\}\right|
 \le\mu_1\lambda_n
\]
with probability at least $1-C(\log n)^{-1}$.
Since $\max_j n^{-1}\sum_iE[(X_i^{(j)})^2]\le C_2^{1/2}$, this proves the assertion for $\mathbb{A}_1$.

\noindent Apply the same argument to $(X_i^{(j)}(t_0))^2$. The second-moment quantity is again at most $C_2$, and the envelope quantity is at most $\sqrt{E M_{XX}^2}$. Thus
\[
 \max_{1\le j\le p}\left|\frac1n\sum_{i=1}^n
 \bigl\{(X_i^{(j)}(t_0))^2-E[(X_i^{(j)}(t_0))^2]\bigr\}\right|
 \le\mu_1\lambda_n
\]
with the same probability bound. The population lower bound in Assumption~\ref{as1}(iii) proves the assertion for $A_2$.
Since $\mathbb{A}_2\subseteq \mathbb{A}_2'$, the bound for $\mathbb{A}_2'$ follows.
Finally, apply the union bound to $\mathbb{A}_1\cap \mathbb{A}_2$.
\end{proof}

%%%%%%%%%%%%%%%%%%%%%%%%%%%5
\begin{lem}[Regularized events $\mathbb{A}_3$ and $\mathbb{A}_4$]
\label{lemmaprobAB}Suppose that Assumption \ref{as1} holds. Let $\mu_2 = 2\mu_1/\sqrt{C_3}$ and define 
\begin{equation*}
\begin{aligned}
&\mathbb{A}_3 =\left\{\max_{1\le j \le p} \frac{1}{\left\Vert X^{(j)}\right\Vert _{n}}\left|\frac{1}{n}\sum_{i=1}^{n}U_{i}X_{i}^{(j)}\right|\le\frac{\mu_2\lambda_n}{2}
\right\} , \\
&\mathbb{A}_4 =\left\{\max_{1\le j \le p}  \sup_{\tau \in \mathbb{T}%
}\frac{1}{\left\Vert X^{(j)}(\tau)\right\Vert _{n}}\left|\frac{1}{n}\sum_{i=1}^{n}U_{i}X_{i}^{(j)}\bm{1}\{Q_{i}<\tau \}\right|\le\frac{\mu_2\lambda_n}{2}\right\}. \\
\end{aligned} 
\end{equation*}
For sufficiently large $A$ and $n$,
\[
 P(\mathbb{A}_3)\ge1-C(\log n)^{-1},\qquad
 P(\mathbb{A}_4)=1-o(1),\qquad P(\mathbb{A}_3\cap \mathbb{A}_4)=1-o(1).
\]
\end{lem}

\begin{proof}[Proof of Lemma \ref{lemmaprobAB}]

\noindent By $E[U_i\mid X_i,Q_i]=0$, the products $U_iX_i^{(j)}$ have mean zero. In Lemma~\ref{concenine}, their second-moment quantity
$\sigma^2=\max_j n^{-1}\sum_i E[(U_iX_i^{(j)})^2]$ is bounded, and $B=\sqrt{E M_{UX}^2}$. Assumption~\ref{as1}(v) therefore gives
\[
 \max_{1\le j\le p}\left|\frac1n\sum_{i=1}^nU_iX_i^{(j)}\right|
 \le Cb_n
\]
with probability at least $1-C(\log n)^{-1}$ for sufficiently large $n$. On $\mathbb{A}_2'$, the minimum column norm is at least $\sqrt{C_3-\mu_1\lambda_n}\ge\sqrt{C_3}/2$ for large $n$.
For sufficiently large $A$, division by this lower bound gives $2Cb_n/\sqrt{C_3}\le\mu_1\lambda_n/\sqrt{C_3}
=\mu_2\lambda_n/2$. Lemma~\ref{lemmaprobA3} proves the claim for $\mathbb{A}_3$.

\noindent For $\mathbb{A}_4$, condition on the entire design $\{X_i,Q_i\}_{i=1}^n$. Denote observations sorted in increasing order of $Q_i$ by
$(X_{(i)},U_{(i)},Q_{(i)})$. Conditional on the design, this ordering is fixed, and the errors $U_{(i)}$ are independent and mean zero.
Thus Lemma~\ref{conpart}(i) applies conditionally to the vectors
$U_{(i)}X_{(i)}$.
The conditional quantities in that lemma satisfy
\begin{align*}
 \sigma^2
 &=\max_{1\le j\le p}\frac1n\sum_{i=1}^n
        \left(X_i^{(j)}\right)^2 E[U_i^2\mid X_i,Q_i]
 \le C_U\max_{1\le j\le p}\frac1n\sum_{i=1}^n\left(X_i^{(j)}\right)^2,\\
 B^2&=E\left[M_{UX}^2\mid\{X_i,Q_i\}_{i=1}^n\right].
\end{align*}
Here the conditional variance restriction in Assumption~\ref{as1}(iv) is used in the form $E[U_i^2\mid X_i,Q_i]\le C_U$ almost surely,
with a common finite constant $C_U$.
Lemma~\ref{conpart}(i) gives, with conditional probability at least $1-C(\log n)^{-1}$,
\[
 \max_{1\le j\le p}\max_{1\le k\le n}
 \left|\frac1n\sum_{i=1}^k U_{(i)}X_{(i)}^{(j)}\right|
 \le C\{\sigma b_n+Bb_n^2\}.
\]
On $\mathbb{A}_1$, $\sigma$ is bounded by a fixed constant. Moreover,
\[
 P(Bb_n>1)\le b_n^2 E B^2
             =b_n^2 E M_{UX}^2=o(1).
\]
Consequently, the maximum partial sum above is bounded by $Cb_n$ with probability at least
$1-C(\log n)^{-1}-C b_n^2 E M_{UX}^2$ for sufficiently large $n$. Since
\begin{align*}
 &\max_{1\le j\le p}\sup_{\tau\in\mathbb T}
 \frac{\left|n^{-1}\sum_{i=1}^nU_iX_i^{(j)}1\{Q_i<\tau\}\right|}
      {\|X^{(j)}(\tau)\|_n}\\
 &\qquad\le
 \frac{\displaystyle\max_{1\le j\le p}\max_{1\le k\le n}
       \left|n^{-1}\sum_{i=1}^kU_{(i)}X_{(i)}^{(j)}\right|}
      {\displaystyle\min_{1\le j\le p}\|X^{(j)}(t_0)\|_n},
\end{align*}
the lower column-norm bound on $\mathbb{A}_2$ and the same choice of $A$ prove $P(\mathbb{A}_4)=1-o(1)$. A union bound proves the joint assertion.
The score events may be defined arbitrarily when a denominator is zero, since $\mathbb{A}_2$ excludes this possibility for sufficiently large $n$.
\end{proof}

\noindent Define $J_{0}=J(\alpha _{0})$, $\widehat{\bm{D}}=\widehat{\bm{D}}(\widehat{\tau})$, $\bm{D}=\bm{D}(\tau_0)$  and $%
R_{n}=R_{n}(\alpha _{0},\tau _{0}),$ where%
\begin{equation*}
R_{n}(\alpha ,\tau ):=2n^{-1}\sum_{i=1}^{n}U_{i}X_{i}^{\prime }\delta
\left\{ 1(Q_{i}<\widehat{\tau })-1(Q_{i}<\tau )\right\} .
\end{equation*}
For a search over observed threshold values, use
\[
 \mathbb T_n=(\{Q_1,\ldots,Q_n\}\cap\mathbb T)
                  \cup\{\tau_\ast,t_1\}.
\] Define
\[
 \tau_{0,n}
 :=\min\Bigl[\bigl(\{Q_1,\ldots,Q_n\}\cap[\tau_0,t_1]\bigr)
                      \cup\{t_1\}\Bigr].
\]
Then $\tau_{0,n}\in\mathbb T_n$,
$\bm X(\tau_{0,n})=\bm X(\tau_0)$, and
$\bm D(\tau_{0,n})=\bm D$. The comparison point $\tau_{0,n}$
is used only in the proof. If an attained minimum over the whole
interval $\mathbb T$ is used, replace $\tau_{0,n}$ by $\tau_0$ below.

\begin{lem}
\label{lemma1} On the events $\mathbb{A}_1$, $\mathbb{A}_2$, $\mathbb{A}_3$ and $\mathbb{A}_4$, for $0<\mu_2 \leq \mu <1,$ we
have
\begin{equation}\label{l5basic}
\begin{aligned}
&\|\widehat f-f_0\|_n^2
 +(1-\mu)\lambda_n
       |\widehat{\bm D}(\widehat\alpha-\alpha_0)|_1\\
 &\quad\le
 2\lambda_n
       |[\widehat{\bm D}(\widehat\alpha-\alpha_0)]_{J_0}|_1
 +\lambda_n\bigl||\widehat{\bm D}\alpha_0|_1
                         -|\bm D\alpha_0|_1\bigr|+R_n\\
 &\qquad\quad
 +\rho_n\bigl\{|\tau_{0,n}-\tau_\ast|
                          -|\widehat\tau-\tau_\ast|\bigr\},
\end{aligned}%
\end{equation}%

\begin{equation}
\left\Vert \widehat{f}-f_{0}\right\Vert _{n}^{2}+(1-\mu)\lambda_n\left| \widehat{\bm{D}}(\widehat{\alpha }-\alpha _{0})\right|_{1}\leq 2\lambda_n \left| \widehat{\bm{D}}(\widehat{%
\alpha }-\alpha _{0})_{J_{0}}\right| _{1}+\left\Vert f_{(\alpha _{0},%
\widehat{\tau })}-f_{0}\right\Vert _{n}^{2},  \label{imp-ineq-2-lem}
\end{equation}
where $f_{(\alpha_0,\widehat\tau)}=\bm X(\widehat\tau)\alpha_0$.
\end{lem}
\begin{proof}[Proof of Lemma~\ref{lemma1}]
Compare the objective at $(\widehat\alpha,\widehat\tau)$
and the feasible pair $(\alpha_0,\tau_{0,n})$. Expanding
$Y=f_0+U$ gives
\begin{align*}
 \|\widehat f-f_0\|_n^2
 +\lambda_n|\widehat{\bm D}\widehat\alpha|_1
 \le{}&2n^{-1}U'(\widehat f-f_0)+\lambda_n|\bm D\alpha_0|_1\\
 &+\rho_n\bigl\{|\tau_{0,n}-\tau_\ast|
                         -|\widehat\tau-\tau_\ast|\bigr\}.
\end{align*}
By $\mathbb{A}_3$, $\mathbb{A}_4$, and $\mu_2\le\mu$,
\[
 2n^{-1}U'(\widehat f-f_0)
 \le\mu\lambda_n
       |\widehat{\bm D}(\widehat\alpha-\alpha_0)|_1+R_n.
\]
The decomposition of the weighted $\ell_1$ norm over $J_0$ and
$J_0^c$ gives
\begin{align*}
 |\widehat{\bm D}\alpha_0|_1
 -|\widehat{\bm D}\widehat\alpha|_1
 \le{}&2|[\widehat{\bm D}(\widehat\alpha-\alpha_0)]_{J_0}|_1\\
 &-|\widehat{\bm D}(\widehat\alpha-\alpha_0)|_1.
\end{align*}
Combining these inequalities proves \eqref{l5basic}.
For \eqref{imp-ineq-2-lem}, compare $(\widehat\alpha,\widehat\tau)$
with $(\alpha_0,\widehat\tau)$. The anchoring penalties cancel,
and expansion of the losses yields
\begin{align*}
 \|\widehat f-f_0\|_n^2
 \le{}&\|f_{(\alpha_0,\widehat\tau)}-f_0\|_n^2
 +2n^{-1}U'\bm X(\widehat\tau)(\widehat\alpha-\alpha_0)\\
 &+\lambda_n\bigl\{|\widehat{\bm D}\alpha_0|_1
                         -|\widehat{\bm D}\widehat\alpha|_1\bigr\}.
\end{align*}
Apply the same score and weighted-norm bounds.
\end{proof}

\begin{proof}[Proof of Lemma \ref{lemmaf}]
On $\mathbb{A}_1 \cap \mathbb{A}_2 \cap \mathbb{A}_3 \cap \mathbb{A}_4$, the first three terms on the right-hand side of \eqref{l5basic} can be bounded as follows:
\begin{align}
\begin{split}
\label{rough-ineq3}|R_{n}| \leq 2\mu_2 \lambda_n \sum_{j=1}^{p}\left\Vert X^{(j)}\right\Vert
_{n}\left|\delta _{0}^{(j)}\right| \leq 2\mu_2  \left|\delta_0 \right|_{1} \lambda_n\sqrt{C_2^{1/2}+\mu_1\lambda_n}
\end{split},\\
\begin{split}
\label{dc1}\left| \widehat{\bm{D}}(\widehat{\alpha }-\alpha
_{0})_{J_{0}}\right|_{1}\leq\left\Vert \widehat{\bm{D}}\right\Vert _{\infty} \left| (\widehat{\alpha }-\alpha
_{0})_{J_{0}}\right|_{1}\le\left| (\widehat{\alpha }-\alpha
_{0})_{J_{0}}\right|_{1} \sqrt{C_2^{1/2}+\mu_1\lambda_n}, \\
\end{split}\\
\begin{split}
 \label{dc2}\left| \left|\widehat{\bm{D}}{\alpha _{0}}%
\right|_{1}-\left| \bm{D}{\alpha _{0}}\right|
_{1}\right|\leq \left| (\widehat{\bm{D}}-\bm{D}){\alpha _{0}}\right|_{1}\leq \left\Vert \widehat{\bm{D}}-\bm{D}\right\Vert _{\infty}\left|{\alpha _{0}}\right|_{1}\leq 2\left|\alpha _{0}\right| _{1}\sqrt{C_2^{1/2}+ \mu_1\lambda_n}
\end{split}\end{align}
In \eqref{rough-ineq3}, apply $\mathbb{A}_4$ at both $\widehat\tau$ and $\tau_0$.
The remaining inequalities follow from H\"older's inequality and $\mathbb{A}_1$.
The additional anchoring term satisfies
\[
 \rho_n\bigl\{|\tau_{0,n}-\tau_\ast|
                         -|\widehat\tau-\tau_\ast|\bigr\}
 \le\rho_n|\tau_{0,n}-\widehat\tau|
 \le\rho_n(t_1-t_0).
\]
Combining \eqref{rough-ineq3}-\eqref{dc2} with \eqref{l5basic},
and dropping its nonnegative penalty term on the left, yields
\begin{align*}
 \|\widehat f-f_0\|_n^2
 \le{}&\bigl(2|(\widehat\alpha-\alpha_0)_{J_0}|_1
                  +2|\alpha_0|_1+2\mu_2|\delta_0|_1\bigr)
        \lambda_n\sqrt{C_2^{1/2}+\mu_1\lambda_n}\\
 &+\rho_n(t_1-t_0)\\
 \le{}&(6+2\mu_2)C_1\sqrt{C_2^{1/2}+\mu_1\lambda_n}\,
                   s_0\lambda_n+\rho_n(t_1-t_0).
\end{align*}
The last step uses $|J_0|\le s_0$ and
$|\widehat\alpha|_\infty,|\alpha_0|_\infty\le C_1$.
Lemmas~\ref{lemmaprobA3} and \ref{lemmaprobAB} imply
$P(\mathbb{A}_1\cap \mathbb{A}_2\cap \mathbb{A}_3\cap \mathbb{A}_4)=1-o(1)$.
Since $\lambda_n\to0$, $s_0\ge1$, and $\rho_n=c_\rho\lambda_n$,
the claimed bound and
$\|\widehat f-f_0\|_n=O_P(\sqrt{s_0\lambda_n})$ follow.
\end{proof}

\subsubsection{Proofs for Section \ref{nothreshold}}

Our first result gives the uniform Gram-matrix bound used to establish the empirical adaptive restricted eigenvalue condition.

\begin{lem}\label{tau22}
Suppose that Assumption~\ref{as1} holds. For sufficiently large $n$,
with probability at least $1-C(\log n)^{-1}$,
\begin{align*}
 \left\|\frac1n\sum_{i=1}^nX_iX_i'
              -\frac1n\sum_{i=1}^nE[X_iX_i']\right\|_\infty
 &\le Cb_n,\\
 \sup_{\tau\in\mathbb T}
 \left\|\frac1n\bm X(\tau)'\bm X(\tau)
                         -\bm\Sigma(\tau)\right\|_\infty
 &\le Cb_n.
\end{align*}
\end{lem}

\begin{proof}[Proof of Lemma \ref{tau22}]
Under Assumption~\ref{as1}, the quantities in Lemma \ref{concenine} satisfy
\[
 \sigma^2=\max_{1\le j,l\le p}\frac1n\sum_{i=1}^n
                E[(X_i^{(j)}X_i^{(l)})^2]\le C_2,
 \qquad B=\sqrt{E M_{XX}^2}.
\]
Applying Lemma \ref{concenine} with dimension $p^2$ and using $Bb_n=o(1)$,
we obtain, with probability at least $1-C(\log n)^{-1}$,
\begin{equation}\label{eq:S37}
 \max_{1\le j,l\le p}
 \left|\frac1n\sum_{i=1}^n
       \{X_i^{(j)}X_i^{(l)}-E[X_i^{(j)}X_i^{(l)}]\}\right|
 \le Cb_n.
\end{equation}
This proves the first assertion.

For the threshold indexed assertion, partition $[0,1]$ into $n^3$ equal subintervals with endpoints $q_k=k/n^3$, $k=0,\ldots,n^3$. This deterministic grid is used only in the proof.
Apply Lemma \ref{concenine} to the independent vectors with coordinates $X_i^{(j)}X_i^{(l)}1\{Q_i<q_k\}$, indexed by $(j,l,k)$.
Their dimension is at most $p^2(n^3+1)$, their second-moment quantity is at most $C_2$, and their envelope quantity is at most $B$. Since $\log\{p^2(n^3+1)\}=O(\ell_n)$, with probability at
least $1-C(\log n)^{-1}$,
\begin{align}
 &\max_{1\le j,l\le p}\max_{0\le k\le n^3}
 \left|\frac1n\sum_{i=1}^n
   \left\{X_i^{(j)}X_i^{(l)}1\{Q_i<q_k\}
       -E[X_i^{(j)}X_i^{(l)}1\{Q_i<q_k\}]\right\}\right|
 \le Cb_n.
\label{eq:S38}
\end{align}

Independence and $Q_i\sim\mathrm{Unif}[0,1]$ imply that the probability that any grid subinterval contains two or more observations is at most $\binom n2 n^{-3}=O(n^{-1})$. Outside this event, for $q_k\le\tau<q_{k+1}$, the empirical increment between $q_k$ and $\tau$ is bounded by $M_{XX}/n$,
uniformly in $j,l,k$. The corresponding population increment is bounded by
\[
 \frac1n\sum_{i=1}^n
 E\bigl[|X_i^{(j)}X_i^{(l)}|1\{q_k\le Q_i<q_{k+1}\}\bigr]
 \le C_2^{1/2}n^{-3/2},
\]
by Cauchy--Schwarz. This step does not require independence between $X_i$ and $Q_i$.
Moreover,
\[
 P(M_{XX}/n>b_n)
 \le\frac{E M_{XX}^2}{n^2b_n^2}
 =\frac{B^2b_n^2}{\ell_n^2}
 =o(\ell_n^{-2}).
\]
Combining these bounds with \eqref{eq:S38} gives
\[
 \sup_{\tau\in\mathbb T}\max_{1\le j,l\le p}
 \left|\frac1n\sum_{i=1}^n
 \left\{X_i^{(j)}X_i^{(l)}1\{Q_i<\tau\}
       -E[X_i^{(j)}X_i^{(l)}1\{Q_i<\tau\}]\right\}\right|
 \le Cb_n
\]
with probability at least $1-C(\log n)^{-1}$ for sufficiently large $n$. Every entry of the augmented Gram matrix is either an unthresholded product or a thresholded product of this form.
Together with \eqref{eq:S37}, this proves the second assertion.
\end{proof}

\noindent For the empirical uniform adaptive restricted eigenvalue condition, define $\widehat{\bm\Sigma}(\tau)
=n^{-1}\bm X(\tau)'\bm X(\tau)$ and
\[
 \widehat\kappa(s_0,c_0,\mathbb T,\widehat{\bm\Sigma})
 =\inf_{\tau\in\mathbb T}
  \min_{\substack{J_0\subseteq\{1,\ldots,2p\}\\1\le|J_0|\le s_0}}
  \inf_{\substack{\gamma\ne0\\
       |\gamma_{J_0^c}|_1\le c_0\sqrt{s_0}|\gamma_{J_0}|_2}}
 \frac{\{\gamma'\widehat{\bm\Sigma}(\tau)\gamma\}^{1/2}}
      {|\gamma_{J_0}|_2}.
\]
The population quantity $\kappa(s_0,c_0,\mathbb T,\bm\Sigma)$ uses the same adaptive cone, as in Assumption~\ref{as2}(ii).
%%%%%%%%%%%%%%%%%%%%%%%%%%%%%%%%%%%%%%%%%%%%%%%%
\begin{lem}\label{lemeg}
Suppose that Assumptions~\ref{as1}-\ref{as2} hold with $\mathbb S=\mathbb T$, a fixed $c_0>0$, and a common positive lower bound on $\kappa(s_0,c_0,\mathbb T,\bm\Sigma)$. Let
\begin{equation*}
\mathbb{A}_5 :=\left\{\frac{\kappa(s_0, c_0, \mathbb{T}, \bm{\Sigma})^2}{2}<{\widehat{\kappa}(s_0,  c_0, \mathbb{T},\widehat{\bm{\Sigma}})}^2
\right\}.
\end{equation*}
Then $P(\mathbb{A}_5)\ge1-C(\log n)^{-1}$ for sufficiently large $n$.

\end{lem}

\begin{proof}[Proof of Lemma \ref{lemeg}]
For every $\tau\in\mathbb T$, H\"older's inequality gives
\begin{equation}\label{eq:S39}
 \left|\gamma'\left\{
      \frac1n\bm X(\tau)'\bm X(\tau)-\bm\Sigma(\tau)
       \right\}\gamma\right|
 \le |\gamma|_1^2
 \left\|\frac1n\bm X(\tau)'\bm X(\tau)
                              -\bm\Sigma(\tau)\right\|_\infty.
\end{equation}
For a vector in the adaptive cone,
\[
 |\gamma|_1
 \le |\gamma_{J_0}|_1+c_0\sqrt{s_0}|\gamma_{J_0}|_2
 \le (1+c_0)\sqrt{s_0}|\gamma_{J_0}|_2.
\]
Consequently,
\begin{align*}
 \frac{\gamma'\widehat{\bm\Sigma}(\tau)\gamma}
      {|\gamma_{J_0}|_2^2}
 \ge{}\frac{\gamma'\bm\Sigma(\tau)\gamma}{|\gamma_{J_0}|_2^2}-(1+c_0)^2s_0
 \sup_{\tau\in\mathbb T}
       \|\widehat{\bm\Sigma}(\tau)-\bm\Sigma(\tau)\|_\infty.
\end{align*}
Taking the infimum over the same thresholds, index sets, and restricted directions on both sides yields
\begin{align}
 \widehat\kappa(s_0,c_0,\mathbb T,\widehat{\bm\Sigma})^2
 \ge{}\kappa(s_0,c_0,\mathbb T,\bm\Sigma)^2
 -(1+c_0)^2s_0\sup_{\tau\in\mathbb T}
       \|\widehat{\bm\Sigma}(\tau)-\bm\Sigma(\tau)\|_\infty.
\label{eq:S310}
\end{align}
By Lemma~\ref{tau22}, with probability at least $1-C(\log n)^{-1}$,
\begin{equation}\label{eq:S311}
 (1+c_0)^2s_0\sup_{\tau\in\mathbb T}
       \|\widehat{\bm\Sigma}(\tau)-\bm\Sigma(\tau)\|_\infty
 \le C(1+c_0)^2s_0b_n=o(1).
\end{equation}
The last step uses $s_0b_n\to0$ from Assumption~\ref{as1}. Since the population restricted eigenvalue is bounded away from
zero, \eqref{eq:S310}-\eqref{eq:S311} prove the claim.
\end{proof}

\begin{lem}\label{lem-case1-supp}
Suppose that $\delta_0=0$ and Assumptions~\ref{as1}--\ref{as2}
hold with $\mathbb S=\mathbb T$. Let
\[
 c_0>\frac{1+\mu}{1-\mu}\frac{C_2^{1/4}}{\sqrt{C_3}},
 \qquad\kappa=\kappa(s_0,c_0,\mathbb T,\bm\Sigma).
\]
Let $(\widehat\alpha,\widehat\tau)$ be the Lasso estimator in \eqref{joint-max}, with tuning parameters satisfying \eqref{eq:tuning}. Then, for sufficiently large $n$, on
$\mathbb{A}_1\cap \mathbb{A}_2\cap \mathbb{A}_3\cap \mathbb{A}_4\cap \mathbb{A}_5$,
\begin{align*}
 \|\widehat f-f_0\|_n
 &\le\frac{2\sqrt{2}}{\kappa}
          \sqrt{C_2^{1/2}+\mu_1\lambda_n}\,
          \sqrt{s_0}\lambda_n,\\
 |\widehat\alpha-\alpha_0|_1
 &\le\frac{4\sqrt{2}}{(1-\mu)\kappa^2}
          \frac{C_2^{1/2}+\mu_1\lambda_n}
               {\sqrt{C_3-\mu_1\lambda_n}}\,
          s_0\lambda_n.
\end{align*}
\end{lem}

\begin{proof}[Proof of Lemma~\ref{lem-case1-supp}]
Work on $\mathbb{A}_1\cap \mathbb{A}_2\cap \mathbb{A}_3\cap \mathbb{A}_4\cap \mathbb{A}_5$ and keep $J_0=J(\alpha_0)$ and $\widehat{\bm D}=\bm D(\widehat\tau)$.
Because $\delta_0=0$, we have $f_{(\alpha_0,\tau_\ast)}=f_0$
and $|\bm D(\tau_\ast)\alpha_0|_1
=|\widehat{\bm D}\alpha_0|_1$.
Compare the joint objective at $(\widehat\alpha,\widehat\tau)$
with the feasible pair $(\alpha_0,\tau_\ast)$.
The score and weighted-norm arguments used in Lemma \ref{lemma1} give
\begin{align}
 &\|\widehat f-f_0\|_n^2
 +(1-\mu)\lambda_n
           |\widehat{\bm D}(\widehat\alpha-\alpha_0)|_1
 +\rho_n|\widehat\tau-\tau_\ast|\notag\\
 &\qquad\le
 2\lambda_n|[\widehat{\bm D}(\widehat\alpha-\alpha_0)]_{J_0}|_1.
\label{eq:S312}
\end{align}
In particular,
\begin{equation}\label{eq:S313}
 |[\widehat{\bm D}(\widehat\alpha-\alpha_0)]_{J_0^c}|_1
 \le\frac{1+\mu}{1-\mu}
       |[\widehat{\bm D}(\widehat\alpha-\alpha_0)]_{J_0}|_1.
\end{equation}
Let $\max(\widehat{\bm D})$ and
$\min(\widehat{\bm D})$ denote the largest and smallest diagonal
entries. On $\mathbb{A}_1\cap \mathbb{A}_2$,
\[
 \sqrt{C_3-\mu_1\lambda_n}
 \le\min(\widehat{\bm D})
 \le\max(\widehat{\bm D})
 \le\sqrt{C_2^{1/2}+\mu_1\lambda_n}.
\]
Therefore \eqref{eq:S313} implies
\begin{align*}
 |(\widehat\alpha-\alpha_0)_{J_0^c}|_1
 &\le\frac{1+\mu}{1-\mu}
       \sqrt{\frac{C_2^{1/2}+\mu_1\lambda_n}
                    {C_3-\mu_1\lambda_n}}
       |(\widehat\alpha-\alpha_0)_{J_0}|_1\\
 &\le c_0\sqrt{s_0}|(\widehat\alpha-\alpha_0)_{J_0}|_2
\end{align*}
for sufficiently large $n$, by the strict choice of $c_0$. If $\widehat\alpha=\alpha_0$, the coefficient and prediction bounds follow immediately. Otherwise, the cone above has a nonzero $J_0$ component, and the empirical restricted eigenvalue condition
applies to $\widehat\alpha-\alpha_0$ itself. Since $\delta_0=0$,
$\widehat f-f_0=\bm X(\widehat\tau)(\widehat\alpha-\alpha_0)$.
On $\mathbb{A}_5$, it follows that
\begin{align}
 \kappa^2
 |[\widehat{\bm D}(\widehat\alpha-\alpha_0)]_{J_0}|_2^2
 &\le\kappa^2\max(\widehat{\bm D})^2
                  |(\widehat\alpha-\alpha_0)_{J_0}|_2^2\notag\\
 &\le2\max(\widehat{\bm D})^2\|\widehat f-f_0\|_n^2.
\label{eq:S314}
\end{align}
Combining \eqref{eq:S312} and \eqref{eq:S314},
\begin{align*}
 \|\widehat f-f_0\|_n^2
 &\le2\lambda_n\sqrt{s_0}
          |[\widehat{\bm D}(\widehat\alpha-\alpha_0)]_{J_0}|_2\\
 &\le\frac{2\sqrt{2}\lambda_n}{\kappa}\sqrt{s_0}
          \max(\widehat{\bm D})\|\widehat f-f_0\|_n.
\end{align*}
Dividing by the prediction norm when it is nonzero, and using $\mathbb{A}_1$, proves the prediction bound.

For the coefficient bound, apply Young's inequality to the same right-hand side while keeping the left-hand side of \eqref{eq:S312}. We obtain
\begin{align*}
 &\frac12\|\widehat f-f_0\|_n^2
 +(1-\mu)\lambda_n
          |\widehat{\bm D}(\widehat\alpha-\alpha_0)|_1
 +\rho_n|\widehat\tau-\tau_\ast|\\
 &\qquad\le\frac{4}{\kappa^2}
            \max(\widehat{\bm D})^2s_0\lambda_n^2.
\end{align*}
In particular,
\begin{align}
 |\widehat{\bm D}(\widehat\alpha-\alpha_0)|_1
 &\le\frac{4}{(1-\mu)\kappa^2}
          \max(\widehat{\bm D})^2s_0\lambda_n\notag\\
 &\le\frac{4\sqrt{2}}{(1-\mu)\kappa^2}
          (C_2^{1/2}+\mu_1\lambda_n)s_0\lambda_n.
\label{eq:S315}
\end{align}
The last inequality keeps the stated, slightly looser numerical constant. On $\mathbb{A}_2$, because $\widehat\tau\ge t_0$,
\begin{align}
 |\widehat{\bm D}(\widehat\alpha-\alpha_0)|_1
 &\ge\min(\widehat{\bm D})|\widehat\alpha-\alpha_0|_1\notag\\
 &\ge\sqrt{C_3-\mu_1\lambda_n}|\widehat\alpha-\alpha_0|_1.
\label{eq:S316}
\end{align}
Combining \eqref{eq:S315} and \eqref{eq:S316} yields
\begin{equation}\label{eq:S317}
 |\widehat\alpha-\alpha_0|_1
 \le\frac{4\sqrt{2}}{(1-\mu)\kappa^2}
       \frac{C_2^{1/2}+\mu_1\lambda_n}
            {\sqrt{C_3-\mu_1\lambda_n}}\,
       s_0\lambda_n.
\end{equation}
Finally, keeping the anchoring term in Young's inequality above and dropping the other nonnegative terms gives
\[
 \rho_n|\widehat\tau-\tau_\ast|
 \le\frac{4}{\kappa^2}\max(\widehat{\bm D})^2s_0\lambda_n^2
 \le\frac{4(C_2^{1/2}+\mu_1\lambda_n)}{\kappa^2}
           s_0\lambda_n^2,
\]
where the last inequality uses $\mathbb{A}_1$. Dividing by $\rho_n>0$ establishes the threshold bound on the same event and completes
the proof.
\end{proof}

\begin{proof}[Proof of Theorem~\ref{main-thm-case1}]
The result follows from Lemma~\ref{lem-case1-supp} and \[
 P(\mathbb{A}_1\cap \mathbb{A}_2\cap \mathbb{A}_3\cap \mathbb{A}_4\cap \mathbb{A}_5)=1-o(1),
\]
by Lemmas \ref{lemmaprobA3}, \ref{lemmaprobAB} and \ref{lemeg}. For the threshold convergence statement, substitute $\rho_n=c_\rho\lambda_n$ into the threshold bound established in the proof of Lemma~\ref{lem-case1-supp}. Since $\kappa$ is bounded away from zero, $c_\rho>0$ is fixed, and $s_0\lambda_n=A s_0b_n\to0$,
\[
 |\widehat\tau-\tau_\ast|=O_P(s_0\lambda_n)=o_P(1).
\]
The concentration events can be chosen uniformly over the stated
parameter class when the assumptions hold with common constants
and uniform rates. The remaining inequalities depend on
$\alpha_0$ only through its support size and feasibility in
$\mathcal B$. This proves the asserted uniformity.
\end{proof}%under Assumption \ref{A-discontinuity}, conditional on the events $\mathbb{A}_1$, $\mathbb{A}_2$, $\mathbb{A}_3$ and $\mathbb{A}_4$.

%%%%%%%%%%%%%%%%%%%%%%%%%%%%%%%%%%%%%%%%%%%%%%%%%%%%%

\subsubsection{Proofs for Section \ref{fixedthreshold}}

For a search over observed threshold values, use
\[
 \mathbb T_n=(\{Q_1,\ldots,Q_n\}\cap\mathbb T)
                  \cup\{\tau_\ast,t_1\}.
\]
Let
$\tau_{0,n}:=\min\{\tau\in\mathbb T_n:\tau\ge\tau_0\}$. Then
$\bm X(\tau_{0,n})=\bm X(\tau_0)$ almost surely. Write $e_n=\tau_{0,n}-\tau_0$. If the estimator is defined by an attained minimum over all of $\mathbb T$, use $\tau_{0,n}=\tau_0$ and $e_n=0$ instead.

\noindent Write $\widehat{\bm D}=\bm D(\widehat\tau)$, $\bm D=\bm D(\tau_0)$, and $J_0=J(\alpha_0)$. We keep the events $\mathbb{A}_1,\ldots,\mathbb{A}_4$ in Lemmas \ref{lemmaprobA3} and \ref{lemmaprobAB}.
On $\mathbb{A}_1\cap \mathbb{A}_2$, all diagonal entries of $\bm D(\tau)$, $\tau\in\mathbb T$, lie between $\underline D_n$ and $\overline D_n$. The event $\mathbb{A}_5$ is the empirical restricted
eigenvalue event from Lemma \ref{lemeg}, on $\mathbb S$ and with the constant $c_0$ in Theorem~\ref{thmftau}. The uniform Gram-matrix bound of Lemma \ref{tau22} gives this event on the deterministic neighborhood $\mathbb S$ defined above. In particular, on $\mathbb{A}_5$, the squared empirical restricted eigenvalue is at least $\kappa^2/2$.

%%%%%%%%%%%%%%%%%%%%%%%%%%%%%%%%%%%%%%%%%%
\begin{lem}\label{lem:iter-localization}
Suppose that Assumptions~\ref{as1} and \ref{A-discontinuity}
hold. On $\mathbb{A}_1\cap \mathbb{A}_2\cap \mathbb{A}_3\cap \mathbb{A}_4$ and the event in
Assumption~\ref{A-discontinuity},
\[
 \|\widehat f-f_0\|_n^2\le B_{0,n},\qquad
 |\widehat\alpha-\alpha_0|_1\le c_\alpha^{(0)},\qquad
 |\widehat\tau-\tau_0|\le\eta^*.
\]
\end{lem}

\begin{proof}[Proof of Lemma  \ref{lem:iter-localization}]
Comparing the joint objective with the feasible pair
$(\alpha_0,\tau_{0,n})$ and applying Lemma \ref{lemma1}, we obtain
\begin{align}
 &\|\widehat f-f_0\|_n^2
 +(1-\mu)\lambda_n
       |\widehat{\bm D}(\widehat\alpha-\alpha_0)|_1\notag\\
 &\quad\le2\lambda_n
       |[\widehat{\bm D}(\widehat\alpha-\alpha_0)]_{J_0}|_1
       +\lambda_n\bigl||\bm D\alpha_0|_1
                          -|\widehat{\bm D}\alpha_0|_1\bigr|
       +|R_n|\notag\\
 &\qquad\quad+\rho_n\{|\tau_{0,n}-\tau_\ast|
                         -|\widehat\tau-\tau_\ast|\},
 \label{eq:S318}
\end{align}
where
$R_n=2n^{-1}\sum_i U_iX_i'\delta_0
 \{1(Q_i<\widehat\tau)-1(Q_i<\tau_0)\}$.
The anchoring difference in braces is at most $t_1-t_0$. As in the proof of Lemma~\ref{lemmaf}, the coefficient box, $|J_0|\le s_0$, and the column norm and score events imply
\begin{align*}
 2|[\widehat{\bm D}(\widehat\alpha-\alpha_0)]_{J_0}|_1
   &\le4C_1\overline D_n s_0,\\
 \bigl||\bm D\alpha_0|_1-|\widehat{\bm D}\alpha_0|_1\bigr|
   &\le2C_1\overline D_n s_0,\\
 |R_n|&\le2\mu_2 C_1\overline D_n s_0\lambda_n.
\end{align*}
Thus the right-hand side of \eqref{eq:S318} is bounded by $B_{0,n}$. Dropping either nonnegative term on the left proves the first two assertions. In particular, the estimated coefficients belong to the domain of Assumption~\ref{A-discontinuity}, irrespective of their selected support. Applying that assumption gives
\begin{equation}\label{eq:S319}
 C_4|\widehat\tau-\tau_0|
 \le\|\widehat f-f_0\|_n^2
      +C_I\lambda_n|\widehat\alpha-\alpha_0|_1+C_4b_n^2
 \le C_4\eta^*.
\end{equation}
Since $B_{0,n}=O(s_0\lambda_n)$,
$c_\alpha^{(0)}=O(s_0)$, and $s_0\lambda_n\to0$,
we also have $\eta^*=O(s_0\lambda_n)=o(1)$.
\end{proof}

\begin{lem}\label{lem:iter-noise}
Suppose that Assumptions~\ref{as1}, \ref{A-smoothness}, and \ref{A-secondmoments} hold. On the event in Assumption~\ref{A-smoothness}, for $b_n^2\le\eta\le\eta^*$,
\begin{equation}\label{eq:S320}
 \sup_{j,l}\sup_{\substack{\tau\in\mathbb T\\|\tau-\tau_0|\le\eta}}
 \frac1n\sum_{i=1}^n\left|X_i^{(j)}X_i^{(l)}\right|
 |1\{Q_i<\tau\}-1\{Q_i<\tau_0\}|
 \le C_5\eta.
\end{equation}
Fix $\nu>0$. For a deterministic radius
$\eta\in[b_n^2,\eta^*]$, define $\mathcal C_n(\eta)$ by
\begin{equation}\label{eq:S321}
 \sup_{\substack{\tau\in\mathbb T\\|\tau-\tau_0|\le\eta}}
 \left|\frac1n\sum_{i=1}^n U_iX_i'\delta_0
         \{1(Q_i<\tau)-1(Q_i<\tau_0)\}\right|
 \le\frac{\nu\lambda_n\sqrt\eta}{2}.
\end{equation}
Conditional on the entire design and on the design event in Assumption~\ref{A-secondmoments}, its failure probability is at most $C/(\nu^2A^2\ell_n)$. For any deterministic finite collection of $N_n$ such radii, the conditional probability
that one of these events fails is at most
$CN_n/(\nu^2A^2\ell_n)$. The constant $C$ in these probability bounds may depend on $C_U$
and $C_h$, but is independent of $n$, $p$, $A$, $\nu$, $\eta$,
and the number $N_n$ of radii. Moreover,
$P(e_n>b_n^2)\le\exp(-\ell_n)$ for the stated grid.
\end{lem}

\begin{proof}[Proof of Lemma \ref{lem:iter-noise}]

The design bound is precisely Assumption~\ref{A-smoothness}. For the noise bound, condition on $\{(X_i,Q_i)\}_{i=1}^n$. Independence across observations and $E[U_i\mid X_i,Q_i]=0$ imply that the errors are conditionally independent and centered. On each side of $\tau_0$, order
the observations by increasing distance from $\tau_0$. The threshold-indexed sums are then partial sums with fixed weights $X_i'\delta_0$. On the event in Assumption~\ref{A-secondmoments},
\begin{equation}\label{eq:S322}
 \frac1n\sum_{i=1}^n(X_i'\delta_0)^2
                    1\{|Q_i-\tau_0|\le\eta\}
 =2\eta h_n^2(\eta)\le2C_h\eta.
\end{equation}
The conditional $L^2$ maximal inequality for independent centered partial sums, applied on both sides, and the bound
$E[U_i^2\mid X_i,Q_i]\le C_U$ imply
\[
 E\!\left[\left.
  \sup_{\substack{\tau\in\mathbb T\\|\tau-\tau_0|\le\eta}}
  \left|\frac2n\sum_i U_iX_i'\delta_0
       \{1(Q_i<\tau)-1(Q_i<\tau_0)\}\right|^2
  \,\right|\,\{X_i,Q_i\}_{i=1}^n\right]
 \le\frac{C\eta}{n}.
\]
Consequently, Markov's inequality gives, on that design event,
\begin{equation}\label{eq:S323}
 P\{\mathcal C_n(\eta)^c\mid\{X_i,Q_i\}_{i=1}^n\}
 \le\frac{C}{\nu^2n\lambda_n^2}
 =\frac{C}{\nu^2A^2\ell_n}.
\end{equation}
A conditional union bound proves the assertion for finitely many deterministic radii. The design event is included only once when passing to unconditional probabilities. This argument does not assert \eqref{eq:S321} simultaneously at every radius without a corresponding probability calculation.

If $\tau_0+b_n^2\le t_1$, the event $e_n>b_n^2$ requires that none of the independent uniform $Q_i$ lie in $[\tau_0,\tau_0+b_n^2]$. Its probability is at most $(1-b_n^2)^n\le\exp(-nb_n^2)=\exp(-\ell_n)$.
If $\tau_0+b_n^2>t_1$, the inclusion of $t_1$ in the grid makes this event impossible.
\end{proof}

\noindent Let $\mathcal E_n^0$ be the intersection of $\mathbb{A}_1,\ldots,\mathbb{A}_5$,
the design events in Assumptions~\ref{A-discontinuity}--\ref{A-secondmoments},
and $\{e_n\le b_n^2\}$. Under the stated assumptions,
$P(\mathcal E_n^0)=1-o(1)$. For deterministic
$0\le c_\alpha\le c_\alpha^{(0)}$ and
$b_n^2\le c_\tau\le\eta^*$, define
\begin{align*}
 L_n&:=\frac{C_5|\delta_0|_1}{2\underline D_n},\\
 G_1(c_\tau)&:=\nu\sqrt{c_\tau}+L_nc_\tau
                    +\frac{\rho_n}{\lambda_n}(c_\tau+b_n^2),\\
 G_2&:=\frac{12\overline D_n^2}{\kappa^2},\qquad
 G_3(c_\alpha,c_\tau):=
       \frac{2\sqrt2\,\overline D_n\sqrt{C_5}}{\kappa}
                                  \sqrt{c_\alpha c_\tau},\\
 H(c_\alpha,c_\tau)&:=G_1(c_\tau)\vee G_2s_0\lambda_n
                \vee G_3(c_\alpha,c_\tau)\sqrt{s_0|\delta_0|_1}.
\end{align*}
The quantities $c_\alpha$ and $c_\tau$ are deterministic error radii,
not fixed constants. The quantities $L_n$, $G_1$, $G_2$, $G_3$, and $H$
may depend on $n$ through their definitions.
The constant $\nu>0$ is used only in the proof and is chosen below.
It is not an additional tuning parameter of the estimator.

\begin{lem}[Prediction and coefficient refinement]
\label{lem:iter-coefficients}
Under the conditions of Theorem~\ref{thmftau}, focus on $\mathcal E_n^0\cap\mathcal C_n(c_\tau)$ and suppose that
$|\widehat\alpha-\alpha_0|_1\le c_\alpha$ and
$|\widehat\tau-\tau_0|\le c_\tau$. Then, for sufficiently large $n$,
\begin{align*}
 \|\widehat f-f_0\|_n^2&\le3\lambda_nH(c_\alpha,c_\tau),\\
 |\widehat\alpha-\alpha_0|_1
 &\le\frac{3H(c_\alpha,c_\tau)}{(1-\mu)\underline D_n}.
\end{align*}
\end{lem}

\begin{proof}[Proof of Lemma \ref{lem:iter-coefficients}]
Write $v=\widehat\alpha-\alpha_0$,
$q=\|\widehat f-f_0\|_n$, $p_v=|\widehat{\bm D}v|_1$,
and $z=|[\widehat{\bm D}v]_{J_0}|_1$.
By \eqref{eq:S321},
\begin{equation}\label{eq:S324}
 |R_n|\le\nu\lambda_n\sqrt{c_\tau}.
\end{equation}
The difference-of-squares identity, the lower column-norm bound, and \eqref{eq:S320} imply
\begin{align}
 \bigl||\widehat{\bm D}\alpha_0|_1-|\bm D\alpha_0|_1\bigr|
 &\le\sum_{j=1}^p|\delta_0^{(j)}|
   \frac{\bigl|\|X^{(j)}(\widehat\tau)\|_n^2
                    -\|X^{(j)}(\tau_0)\|_n^2\bigr|}
        {\|X^{(j)}(\widehat\tau)\|_n+
                    \|X^{(j)}(\tau_0)\|_n}\notag\\
 &\le L_nc_\tau.
\label{eq:S325}
\end{align}
Also, the anchoring difference in \eqref{eq:S318} is at most
$\rho_n(|\widehat\tau-\tau_0|+e_n)
 \le\rho_n(c_\tau+b_n^2)$.
Thus \eqref{eq:S318} gives
\[
 q^2+(1-\mu)\lambda_np_v
       \le2\lambda_nz+\lambda_nG_1(c_\tau).
\]
We consider two cases.
\emph{Case (i): $z>G_1(c_\tau)$.}
The preceding inequality yields
\begin{equation}\label{eq:S326}
 q^2+(1-\mu)\lambda_np_v\le3\lambda_nz,
\end{equation}
and hence
\begin{equation}\label{eq:S327}
 |[\widehat{\bm D}v]_{J_0^c}|_1
 \le\frac{2+\mu}{1-\mu}|[\widehat{\bm D}v]_{J_0}|_1.
\end{equation}
For sufficiently large $n$, the column-norm bounds and the choice of $c_0$ therefore imply
\[
 |v_{J_0^c}|_1
 \le\frac{2+\mu}{1-\mu}\frac{\overline D_n}{\underline D_n}
                                      |v_{J_0}|_1
 \le c_0\sqrt{s_0}|v_{J_0}|_2.
\]
The empirical restricted eigenvalue on $\mathbb{A}_5$ applies to $v$ itself. Put
$h=\{\bm X(\widehat\tau)-\bm X(\tau_0)\}\alpha_0$.
Since $\widehat f-f_0=\bm X(\widehat\tau)v+h$,
\[
 \|\bm X(\widehat\tau)v\|_n^2
 =q^2-\frac2n v'\bm X(\widehat\tau)'h-\|h\|_n^2
 \le q^2+2C_5|\delta_0|_1c_\alpha c_\tau.
\]
Consequently,
\begin{align*}
 z^2&\le s_0\overline D_n^2|v_{J_0}|_2^2
 \le\frac{2s_0\overline D_n^2}{\kappa^2}
           \{q^2+2C_5|\delta_0|_1c_\alpha c_\tau\}\\
 &\le\frac{6s_0\overline D_n^2\lambda_n}{\kappa^2}z
       +\frac{4s_0\overline D_n^2C_5}{\kappa^2}
                            |\delta_0|_1c_\alpha c_\tau.
\end{align*}
If $z>G_2s_0\lambda_n$, the first term on the last line is
less than $z^2/2$. It follows that
\begin{equation}\label{eq:S328}
 z\le G_2s_0\lambda_n
          \vee G_3(c_\alpha,c_\tau)\sqrt{s_0|\delta_0|_1}.
\end{equation}
Combining this with \eqref{eq:S326} and
$p_v\ge\underline D_n|v|_1$ proves both assertions in case (i).

\emph{Case (ii): $z\le G_1(c_\tau)$.}
Now the basic inequality gives
\[
 q^2+(1-\mu)\lambda_np_v\le3\lambda_nG_1(c_\tau).
\]
Dropping each nonnegative term in turn proves the assertions.
\end{proof}

\begin{lem}[Threshold refinement]\label{lem:iter-threshold}
Under the conditions and on the events in
Lemma~\ref{lem:iter-coefficients}, for sufficiently large $n$,
\[
 |\widehat\tau-\tau_0|
 \le\frac{\lambda_n\{[(1+\mu)\overline D_n+C_I]c_\alpha
                          +G_1(c_\tau)\}+C_4b_n^2}
              {C_4-\rho_n}.
\]
\end{lem}

\begin{proof}[Proof of Lemma \ref{lem:iter-threshold}]
Write $a=|\widehat\alpha-\alpha_0|_1$,
$d=|\widehat\tau-\tau_0|$, and $q=\|\widehat f-f_0\|_n$. Since $z\le p_v$,
$2z-(1-\mu)p_v\le(1+\mu)p_v\le(1+\mu)\overline D_na$.
Using the sharper anchoring bound $\rho_n(d+b_n^2)$ in \eqref{eq:S318}, rather than replacing $d$ by $c_\tau$, yields
\[
 q^2\le(1+\mu)\lambda_n\overline D_na
       +\lambda_nL_nc_\tau+\nu\lambda_n\sqrt{c_\tau}
       +\rho_nd+\rho_nb_n^2.
\]
Assumption~\ref{A-discontinuity} gives
$C_4d\le q^2+C_I\lambda_na+C_4b_n^2$.
Move $\rho_nd$ to the left and use $a\le c_\alpha$ and
\[
 \lambda_nL_nc_\tau+\nu\lambda_n\sqrt{c_\tau}
                     +\rho_nb_n^2\le\lambda_nG_1(c_\tau).
\]
Since $\rho_n\to0$, the denominator is positive.
\end{proof}

\begin{lem}[Iteration and oracle inequalities]\label{lem:iter-oracle}
Under the conditions of Theorem~\ref{thmftau}, with probability tending to one,
\begin{align*}
 \|\widehat f-f_0\|_n
  &\le\frac{6\overline D_n}{\kappa}\sqrt{s_0}\lambda_n,\\
 |\widehat\alpha-\alpha_0|_1
  &\le\frac{36\overline D_n^2}
             {\kappa^2(1-\mu)\underline D_n}s_0\lambda_n,\\
 |\widehat\tau-\tau_0|
  &\le\left\{\frac{3[(1+\mu)\overline D_n+C_I]}
                         {(1-\mu)\underline D_n}+2\right\}
       \frac{12\overline D_n^2}{\kappa^2(C_4-\rho_n)}
                                         s_0\lambda_n^2.
\end{align*}
\end{lem}

\begin{proof}[Proof of Lemma \ref{lem:iter-oracle}]

We provide a deterministic iteration and bound the number of local-noise
events it uses. Fix $A$ sufficiently large as required above, also
ensuring $A^2\ge C_4/12$, before making the following choices.
Set $c_\tau^{(0)}=\eta^*$ and keep $c_\alpha^{(0)}$ from Assumption~\ref{A-discontinuity}.
For this proof, write
\[
 K_\alpha:=\frac3{(1-\mu)\underline D_n},\qquad
 K_0:=(1+\mu)\overline D_n+C_I,\qquad
 H_\ast:=G_2s_0\lambda_n.
\]
Starting from these preliminary bounds, define
\begin{align*}
 H_m&:=H(c_\alpha^{(m)},c_\tau^{(m)}),\\
 c_\alpha^{(m+1)}
   &:=\min\{c_\alpha^{(m)},K_\alpha H_m\},\\
 c_\tau^{(m+1)}
   &:=\min\left\{c_\tau^{(m)},
       \frac{\lambda_n\{K_0c_\alpha^{(m+1)}
                          +G_1(c_\tau^{(m)})\}+C_4b_n^2}
            {C_4-\rho_n}\right\}.
\end{align*}
These sequences are deterministic, given the model, sample size, and tuning constants. Both are nonincreasing and $c_\tau^{(m)}\ge b_n^2$, because the second entry in its minimum is at least $C_4b_n^2/(C_4-\rho_n)$. Consequently, $H_m$ is nonincreasing and $H_m\ge H_\ast$.
On $\mathcal E_n^0\cap\bigcap_{j=0}^{m}\mathcal C_n(c_\tau^{(j)})$, Lemmas~\ref{lem:iter-localization}--\ref{lem:iter-threshold} give inductively
\[
 |\widehat\alpha-\alpha_0|_1\le c_\alpha^{(m+1)},\qquad
 |\widehat\tau-\tau_0|\le c_\tau^{(m+1)},\qquad
 \|\widehat f-f_0\|_n^2\le3\lambda_nH_m.
\]
In the threshold update, we use the newly obtained coefficient bound and the preceding radius, whose noise event is available.

We next show that $H_m=H_\ast$ after $o(\ell_n)$ iterations. The population restricted-eigenvalue infimum evaluated at a single coordinate gives
$\kappa^2\le C_2^{1/2}\le\overline D_n^2$, so $G_2\ge12$.
Since $H_m\ge G_2s_0\lambda_n$ and
$\lambda_n^2=A^2b_n^2$, there exists a fixed finite
$B_*\ge1$, independent of $\nu$, such that 
\[
 c_\alpha^{(m+1)}\le K_\alpha H_m,
 \qquad c_\tau^{(m+1)}\le B_*\lambda_nH_m
 \quad\text{for all }m.
\]
Indeed, it suffices to choose $B_*$ as a common eventual upper bound for
\[
 \frac{K_0K_\alpha+1+C_4/(A^2G_2s_0)}{C_4-\rho_n},
\]
whose numerator is bounded and whose denominator is bounded away from zero. Define
\begin{align*}
 K_3&:=\frac{2\sqrt2\,\overline D_n\sqrt{C_5}}{\kappa},\\
 \gamma_{1,n}&:=B_*(\lambda_nL_n+\rho_n)
                               +\frac{\rho_n}{A^2G_2s_0},\\
 \gamma_{2,n}&:=K_3\sqrt{K_\alpha B_*s_0|\delta_0|_1\lambda_n},
 \qquad
 \gamma_n:=\max\{\gamma_{1,n},\gamma_{2,n},1/\ell_n\}.
\end{align*}
The quantities $G_2$, $K_\alpha$, $K_0$, and $K_3$ are deterministic
and may depend on $n$ through their displayed definitions. This
dependence is suppressed in the notation.
Assumption~\ref{as1}(i) implies
$s_0|\delta_0|_1\lambda_n\to0$, and hence $\gamma_n\to0$. Choose a fixed $\nu>0$ small enough that $\nu\sqrt{B_*/12}\le1/4$.
For $u_m:=H_m/H_\ast\ge1$, the preceding bounds give
\begin{align*}
 \frac{G_1(c_\tau^{(m+1)})}{H_\ast}
  &\le\frac14\sqrt{u_m}+\gamma_{1,n}u_m,\\
 \frac{G_3(c_\alpha^{(m+1)},c_\tau^{(m+1)})
                          \sqrt{s_0|\delta_0|_1}}{H_\ast}
  &\le\gamma_{2,n}u_m.
\end{align*}
Thus
\[
 u_{m+1}\le\max\{1,\tfrac14\sqrt{u_m}+\gamma_nu_m\}
          \le\max\{1,\tfrac12\sqrt{u_m},2\gamma_nu_m\}.
\]
For large $n$, $\gamma_n\le1/8$. Let
$v_m=\log u_m\ge0$ and
$d_n=\log\{1/(2\gamma_n)\}\to\infty$. Then
\begin{equation}\label{eq:S329}
 v_{m+1}\le\max\{0,v_m/2-\log2,v_m-d_n\}.
\end{equation}
While $v_m>2d_n$, it decreases by at least $d_n$.
Once $v_m\le2d_n$, it decreases at least by a factor of two until $v_m\le\log4$. From $u_m\le4$, the next step has $u_{m+1}=1$. In particular, a sufficient deterministic bound on the stopping index is
\[
 M_n:=2+\left\lceil\frac{v_0}{d_n}\right\rceil
       +\left\lceil\log_2\left(1+\frac{v_0}{\log4}\right)
         \right\rceil.
\]
To bound its order, note that $c_\alpha^{(0)}=O(s_0)$,
$c_\tau^{(0)}=O(s_0\lambda_n)$,
$|\delta_0|_1\le C_1s_0$, and
$s_0|\delta_0|_1\lambda_n=o(1)$ imply
$H_0/H_\ast\le C/\lambda_n$.
Then $v_0=O(\log n)$, because
$\lambda_n=A\sqrt{\ell_n/n}$ and $\ell_n\ge\log n$.
Since $d_n\to\infty$, it follows that $M_n=o(\ell_n)$.

All required local-noise events now hold jointly with
probability tending to one. More precisely, define
\[
 \mathcal E_n:=\mathcal E_n^0
            \cap\bigcap_{m=0}^{M_n}\mathcal C_n(c_\tau^{(m)}).
\]
Lemma~\ref{lem:iter-noise}, applied conditionally on the design, gives
\[
 P(\mathcal E_n^c)
 \le P((\mathcal E_n^0)^c)
       +\frac{C(M_n+1)}{\nu^2A^2\ell_n}=o(1).
\]
This bounds the failure probability for the entire iteration, not just for a fixed step.

On $\mathcal E_n$, some $m_\ast\le M_n$ satisfies
$H_{m_\ast}=H_\ast$. Lemma~\ref{lem:iter-coefficients} then yields
\[
 \|\widehat f-f_0\|_n^2\le3\lambda_nH_\ast,
 \qquad |\widehat\alpha-\alpha_0|_1\le K_\alpha H_\ast.
\]
Apply Lemma~\ref{lem:iter-threshold} using
$c_\alpha^{(m_\ast+1)}\le K_\alpha H_\ast$
and radius $c_\tau^{(m_\ast)}$. Since
$G_1(c_\tau^{(m_\ast)})\le H_\ast$,
\[
 |\widehat\tau-\tau_0|
 \le\frac{(K_0K_\alpha+1)\lambda_nH_\ast+C_4b_n^2}
                {C_4-\rho_n}.
\]
The choice $A^2\ge C_4/12$ ensures $C_4b_n^2\le\lambda_nH_\ast$,
because $G_2s_0\ge12$.
Substituting $H_\ast=12\overline D_n^2s_0\lambda_n/\kappa^2$
and the definitions of $K_0$ and $K_\alpha$ proves all three bounds.
\end{proof}

\begin{proof}[Proof of Theorem~\ref{thmftau}]
Lemma~\ref{lem:iter-oracle} proves the inequalities on an event of probability $1-o(1)$. Substitution of
$\overline D_n=\sqrt{C_2^{1/2}+\mu_1\lambda_n}$ and
$\underline D_n=\sqrt{C_3-\mu_1\lambda_n}$ gives the 
constants in the main paper. Since $C_4-\rho_n\ge C_4/2$, the threshold rate is $O_P(s_0\lambda_n^2)$.
When the assumptions have common constants and their
convergences hold uniformly, the initial concentration and
design events, the bounds $\gamma_n\to0$ and
$M_n/\ell_n\to0$, and the final failure probability are
uniform over the stated class. This proves the asserted
uniformity on that class.
\end{proof}

%%%%%%%%%%%%%%%%%%

\subsubsection{Proofs for Section \ref{constructprecison}}

We use $H_A,H_B,\bar s,c_{\mathrm{node}}$, and
$\kappa_{\mathrm{node}}$ from the main text. The symbol $\kappa$ below denotes the same adaptive restricted eigenvalue functional
as in Assumption~\ref{as2}, applied to the corresponding matrix and dimension. Define the weighted score events
\begin{align*}
 \mathbb{A}_{\mathrm{node}}
  &:=\left\{\max_{j\in H_A}\sup_{\tau\in\mathbb T}
   \left|\widehat\Gamma_j(\tau)^{-1}
           \frac{X_<^{(-j)}(\tau)'\upsilon^{(j)}(\tau)}n\right|_\infty
                         \le\frac{\mu\lambda_{\mathrm{node}}}{2}\right\},\\
 \mathbb{B}_{\mathrm{node}}
  &:=\left\{\max_{j\in H_B}\sup_{\tau\in\mathbb T}
   \left|\widehat{\widetilde\Gamma}_j(\tau)^{-1}
          \frac{X_\ge^{(-j)}(\tau)'\widetilde\upsilon^{(j)}(\tau)}n
          \right|_\infty
                         \le\frac{\mu\lambda_{\mathrm{node}}}{2}\right\},\\
 \mathbb{A}_{\mathrm{EV}}^{(j)}
  &:=\left\{\widehat\kappa(\bar s,c_{\mathrm{node}},\mathbb T,
                  \widehat{\bm M}_{-j,-j})^2
       \ge\tfrac12\kappa(\bar s,c_{\mathrm{node}},\mathbb T,
                              \bm M_{-j,-j})^2\right\},\\
 \mathbb{B}_{\mathrm{EV}}^{(j)}
  &:=\left\{\widehat\kappa(\bar s,c_{\mathrm{node}},\mathbb T,
                  \widehat{\bm N}_{-j,-j})^2
       \ge\tfrac12\kappa(\bar s,c_{\mathrm{node}},\mathbb T,
                              \bm N_{-j,-j})^2\right\}.
\end{align*}
Events involving an inverse empirical weight may be defined arbitrarily when that weight is zero. The event established below
excludes zero weights.

\begin{lem}[Uniform nodewise events]\label{lem:S14-nodewise}
Suppose that Assumptions~\ref{as1}, \ref{as2}(i), and \ref{as6-nodewise} hold. Set
$\lambda_{\mathrm{node}}=A_{\mathrm{node}}b_n$ for a sufficiently large fixed $A_{\mathrm{node}}$. Then
\[
 P\!\left(\mathbb{A}_{\mathrm{node}}\cap
       \bigcap_{j\in H_A}\mathbb{A}_{\mathrm{EV}}^{(j)}\cap
       \mathbb{B}_{\mathrm{node}}\cap
       \bigcap_{j\in H_B}\mathbb{B}_{\mathrm{EV}}^{(j)}\right)
                  \ge1-C(\log n)^{-1}.
\]
On an event with the same probability bound, the regime column
norms are uniformly bounded and bounded away from zero, and
\[
 \max_{j\in H_A,\,l\le p}\sup_{\tau\in\mathbb T}
 \left|\frac1n\sum_i\left\{
      X_i^{(l)}\upsilon_i^{(j)}(\tau)
           -E[X_i^{(l)}\upsilon_i^{(j)}(\tau)]\right\}\right|
                    \le C b_n,
\]
with the analogous bound for $H_B$ and
$\widetilde\upsilon_i^{(j)}(\tau)$, including $l=j$.
No restriction on $\delta_0$ or $\widehat\delta$ is needed.
\end{lem}

\begin{proof}[Proof of Lemma \ref{lem:S14-nodewise}]
We first provide the population bounds needed for uniformity in $\tau$. Write $A_j(\tau)'$ and $B_j(\tau)'$ for the rows of $\bm A(\tau)$ and $\bm B(\tau)$.
The block identity \eqref{invformlu} implies
\[
 \max_{j\in H_A}\sup_\tau|A_j(\tau)|_0\le\bar s,
 \qquad \max_{j\in H_B}\sup_\tau|B_j(\tau)|_0\le\bar s.
\]
For the first inequality, use
$\operatorname{supp}(A_j)\subseteq
 \operatorname{supp}(B_j)\cup\operatorname{supp}(A_j+B_j)$. Thus
the sum of these two support sizes is the support size of the corresponding lower row of $\bm\Theta$. The population projection identities therefore give nodewise
supports of size at most $\bar s-1$ for the required rows.

For a lower-regime row, form $c_j(\tau)\in\mathbb R^p$
with $j$th entry one and other entries
$-\gamma_{0,j}(\tau)$. The normal equations give $A_j(\tau)=c_j(\tau)/z_j^2(\tau)$ and
$z_j^2(\tau)=c_j(\tau)'\bm M(\tau)c_j(\tau)$.
Assumption~\ref{as6-nodewise}(i), applied to this sparse vector, and the projection property imply
\[
 \kappa_{\mathrm{node}}^2\le z_j^2(\tau)
      \le M_{jj}(\tau)\le C_2^{1/2},
 \qquad
 |\gamma_{0,j}(\tau)|_2\le C,
 \qquad |\gamma_{0,j}(\tau)|_1\le C\sqrt{\bar s}.
\]
Indeed, $\gamma_{0,j}'\bm M_{-j,-j}\gamma_{0,j} \le M_{jj}$, and the same restricted eigenvalue applies to
$\gamma_{0,j}$ after inserting a zero at coordinate $j$. The corresponding upper-regime bounds hold as well.

By Assumption~\ref{as1}(iii) and $Q_i\sim\mathrm{Unif}[0,1]$, Cauchy-Schwarz gives
\[
 ||\bm M(\tau)-\bm M(t)||_\infty
       \vee||\bm N(\tau)-\bm N(t)||_\infty
                   \le C|\tau-t|^{1/2}.
\]
For $d_j=\gamma_{0,j}(\tau)-\gamma_{0,j}(t)$,
the normal equations imply
\begin{align*}
 \bm M_{-j,-j}(\tau)d_j
  ={}\bm M_{-j,j}(\tau)-\bm M_{-j,j}(t)
     -\{\bm M_{-j,-j}(\tau)-\bm M_{-j,-j}(t)\}
                                          \gamma_{0,j}(t).
\end{align*}
The support of $d_j$ has size at most $2\bar s$. Take its largest $\bar s$ coordinates (or all coordinates if there are fewer). Since $c_{\mathrm{node}}>1$, this vector
satisfies the adaptive cone condition for that set, giving
$d_j'\bm M_{-j,-j}(\tau)d_j
\ge(\kappa_{\mathrm{node}}^2/2)|d_j|_2^2$.
Multiplication of the last expression by $d_j'$ now yields
\[
 |d_j|_2\le C\bar s|\tau-t|^{1/2},
 \qquad |d_j|_1\le C\bar s^{3/2}|\tau-t|^{1/2}.
\]
The same bounds hold for the upper-regime coefficients. Thus no additional smoothness assumption on these coefficient functions is needed for the following fine grid argument.

Partition $\mathbb T$ into $K_{\mathrm{node},n}=\lceil n^6p^4\rceil$ deterministic intervals of length at most $a_n=K_{\mathrm{node},n}^{-1}$. At their endpoints, apply Lemma \ref{concenine} to the independent vectors with coordinates $X_i^{(l)}\upsilon_i^{(j)}(t)$ and
$X_i^{(l)}\widetilde\upsilon_i^{(j)}(t)$, over the required
rows, all $l$, and all endpoints. Their dimension is at most
$2p^2(K_{\mathrm{node},n}+1)$, whose logarithm is $O(\ell_n)$.
Their second moments are bounded by $C_{\mathrm{node}}$ from
Assumption~\ref{as6-nodewise}(ii), and their maximum envelope quantity is bounded by
$\{E(M_{\upsilon X}\vee M_{\widetilde\upsilon X})^2\}^{1/2}$.
Assumption~\ref{as6-nodewise}(iii) and Lemma \ref{concenine} therefore give an upper bound $Cb_n$ for all centered averages at the endpoints, with probability at least $1-C(\log n)^{-1}$.

To extend this bound between endpoints, the probability that some interval contains two observations is at most $Cn^2/K_{\mathrm{node},n}=o(n^{-1})$, by independence and uniformity of the
$Q_i$. Outside this event, changing the indicator within one interval contributes at most $(M_{\upsilon X}\vee M_{\widetilde\upsilon X})/n$ to each empirical average. Its probability of exceeding $b_n$ is at most
\[
 \frac{E(M_{\upsilon X}\vee M_{\widetilde\upsilon X})^2}
             {n^2b_n^2}\le\frac{C}{\ell_n^2}.
\]
For the common observations, the coefficient change contributes at most
\[
 C\bar s^{3/2}\sqrt{a_n}\,
               \max_l\frac1n\sum_i(X_i^{(l)})^2
                  =O(n^{-3})=o(b_n)
\]
on the column event of Lemma \ref{lemmaprobA3} and $\bar s\le2p$.
The corresponding change in the population average is bounded
by $C(1+\bar s^{3/2})\sqrt{a_n}=o(b_n)$.
For this last bound, use the product second moments in Assumption~\ref{as6-nodewise}(ii) on the interval where the
indicator changes, and the preceding coefficient modulus on the common region. This proves the centered bounds uniformly
over $\mathbb T$. For $l\ne j$, the average expectations vanish by the population normal equations. The conditional mean zero given $Q_i$ is not used.

Finally, Lemma \ref{tau22} and
$\widehat{\bm N}(\tau)=X'X/n-\widehat{\bm M}(\tau)$ give
\[
 \sup_{\tau\in\mathbb T}
 \{|\widehat{\bm M}(\tau)-\bm M(\tau)|_\infty
       \vee|\widehat{\bm N}(\tau)-\bm N(\tau)|_\infty\}
             \le Cb_n
\]
with probability at least $1-C(\log n)^{-1}$.
The population diagonal entries lie between
$\kappa_{\mathrm{node}}^2$ and $C_2^{1/2}$.
The empirical weights are therefore uniformly bounded away from zero and above. Division by the lower weight bound, followed by a sufficiently large choice of $A_{\mathrm{node}}$,
gives $\mathbb{A}_{\mathrm{node}}$ and $\mathbb{B}_{\mathrm{node}}$.

For any vector in an adaptive cone of sparsity $\bar s$,
the Gram perturbation in the restricted quotient is at most
\[
 (1+c_{\mathrm{node}})^2\bar s\,
      \sup_\tau|\widehat{\bm M}(\tau)-\bm M(\tau)|_\infty=o(1),
\]
and analogously for $\bm N$. Inserting a zero at coordinate $j$ shows that each population submatrix restricted eigenvalue
is at least $\kappa_{\mathrm{node}}$.
The same argument as Lemma \ref{lemeg} thus proves every required
$\mathbb{A}_{\mathrm{EV}}^{(j)}$ and $\mathbb{B}_{\mathrm{EV}}^{(j)}$
simultaneously, with no additional union bound over $j$. Combining these events proves the lemma and its accompanying
column and centered-product assertions.
\end{proof}

\begin{proof}[Proof of Lemma \ref{lem:nodewise-main}]
We establish the nodewise oracle bounds, the regime inverse bounds, and the block KKT calculation. We first work on the common event from Lemma~\ref{lem:S14-nodewise}.
For a lower regime regression, write
$d_j=\widehat\gamma_j(\tau)-\gamma_{0,j}(\tau)$ and
$J_j=\operatorname{supp}\{\gamma_{0,j}(\tau)\}$.
The weighted Lasso basic inequality gives
\[
 \|X_<^{(-j)}(\tau)d_j\|_n^2
 +(1-\mu)\lambda_{\mathrm{node}}|\widehat\Gamma_j(\tau)d_j|_1
 \le2\lambda_{\mathrm{node}}
                  |[\widehat\Gamma_j(\tau)d_j]_{J_j}|_1.
\]
Hence the weighted cone ratio is $(1+\mu)/(1-\mu)$. The weight bounds and the strict choice of $c_{\mathrm{node}}$ transfer this into the unweighted adaptive
cone for the coefficient error. The empirical restricted
eigenvalue then yields
\begin{align}
 \sup_{\tau\in\mathbb T}\max_{j\in H_A}
  \|X_<^{(-j)}(\tau)
         \{\widehat\gamma_j(\tau)-\gamma_{0,j}(\tau)\}\|_n
   &\le\frac{C}{\kappa_{\mathrm{node}}}
                           \sqrt{\bar s}\lambda_{\mathrm{node}},
       \label{eq:S330}\\
 \sup_{\tau\in\mathbb T}\max_{j\in H_A}
       |\widehat\gamma_j(\tau)-\gamma_{0,j}(\tau)|_1
   &\le\frac{C}{\kappa_{\mathrm{node}}^2}
                              \bar s\lambda_{\mathrm{node}}.
    \label{eq:S331}
\end{align}
If $J_j$ is empty, the basic inequality gives $d_j=0$ directly. Otherwise, the same adaptive cone applies to the set of the
largest $\bar s$ coordinates of $d_j$. Its Euclidean norm is bounded by the prediction bound divided by the restricted
eigenvalue, and the Euclidean norm of its complement is at most $|d_j|_1/\sqrt{\bar s}$. Thus
\[
 \sup_{\tau\in\mathbb T}\max_{j\in H_A}
       |\widehat\gamma_j(\tau)-\gamma_{0,j}(\tau)|_2
                      \le C\sqrt{\bar s}\lambda_{\mathrm{node}}.
\]
When fewer than $\bar s$ coordinates exist, take all coordinates and the complement is empty. The same three bounds hold for
the upper regime regressions over $H_B$.

We next bound the residual scales. Let
$\widehat r_j(\tau)=X_<^{(j)}(\tau)
             -X_<^{(-j)}(\tau)\widehat\gamma_j(\tau)$.
The KKT equations for \eqref{eq:nodewise-lasso} are
\[
 \frac1nX_<^{(-j)}(\tau)'\widehat r_j(\tau)
      =\frac{\lambda_{\mathrm{node}}}{2}
                          \widehat\Gamma_j(\tau)\widehat s_j(\tau),
 \quad |\widehat s_j(\tau)|_\infty\le1,
 \quad \widehat s_j(\tau)'\widehat\Gamma_j(\tau)
                  \widehat\gamma_j(\tau)
            =|\widehat\Gamma_j(\tau)\widehat\gamma_j(\tau)|_1.
\]
Consequently, the scale in Section~\ref{supp:nodewise-construction}
satisfies
\[
 \widehat z_j^2(\tau)
           =\frac1nX_<^{(j)}(\tau)'\widehat r_j(\tau).
\]
The population normal equations similarly give
$n^{-1}\sum_iE[X_i^{(j)}\upsilon_i^{(j)}(\tau)]=z_j^2(\tau)$.
It follows that
\begin{align*}
 \widehat z_j^2(\tau)-z_j^2(\tau)
   =\frac1n\sum_i\{X_i^{(j)}\upsilon_i^{(j)}(\tau)
                       -E[X_i^{(j)}\upsilon_i^{(j)}(\tau)]\}
      -\frac1nX_<^{(j)}(\tau)'X_<^{(-j)}(\tau)d_j.
\end{align*}
The first term is $O_P(b_n)$ uniformly by
Lemma~\ref{lem:S14-nodewise}, including $l=j$.
By Cauchy--Schwarz and \eqref{eq:S330}, the second is $O_P(\sqrt{\bar s}\lambda_{\mathrm{node}})$ uniformly.
Since $z_j^2(\tau)\ge\kappa_{\mathrm{node}}^2$ and
$\sqrt{\bar s}\lambda_{\mathrm{node}}=o(1)$, the required empirical scales are bounded away from zero with probability tending to one. The same reasoning applies to
$\widehat{\widetilde z}_j^2(\tau)$.

The population coefficient vectors and inverse rows have bounded Euclidean norm and $\ell_1$ norm $O(\sqrt{\bar s})$,
by the population bounds in Lemma~\ref{lem:S14-nodewise}.
Combining the coefficient and scale bounds in
$\widehat A_j=\widehat c_j/\widehat z_j^2$ and
$A_j=c_j/z_j^2$ therefore gives
\begin{align}
 \max_{j\in H_A}\sup_{\tau\in\mathbb T}|\widehat A_j(\tau)-A_j(\tau)|_1
       &=O_P(\bar s b_n),\label{eq:S332}\\
 \max_{j\in H_A}\sup_{\tau\in\mathbb T}|\widehat A_j(\tau)-A_j(\tau)|_2
       &=O_P(\sqrt{\bar s}\,b_n),\label{eq:S333}\\
 \max_{j\in H_A}\sup_{\tau\in\mathbb T}|\widehat A_j(\tau)|_1
       &=O_P(\sqrt{\bar s}),\label{eq:S334}\\
 \max_{j\in H_A}\sup_{\tau\in\mathbb T}\frac1{\widehat z_j^2(\tau)}
       &=O_P(1),\label{eq:S335}\\
 \max_{j\in H_B}\sup_{\tau\in\mathbb T}|\widehat B_j(\tau)-B_j(\tau)|_1
       &=O_P(\bar s b_n),\label{eq:S336}\\
 \max_{j\in H_B}\sup_{\tau\in\mathbb T}|\widehat B_j(\tau)-B_j(\tau)|_2
       &=O_P(\sqrt{\bar s}\,b_n),\label{eq:S337}\\
 \max_{j\in H_B}\sup_{\tau\in\mathbb T}|\widehat B_j(\tau)|_1
       &=O_P(\sqrt{\bar s}),\label{eq:S338}\\
 \max_{j\in H_B}\sup_{\tau\in\mathbb T}
                           \frac1{\widehat{\widetilde z}_j^2(\tau)}
       &=O_P(1).\label{eq:S339}
\end{align}
To combine these bounds, note that for $k\in H\cap\{1,\ldots,p\}$,
\[
 \widehat\Theta_k'=(\widehat B_k',-\widehat B_k'),
\]
whereas for $j\in H_A$,
\[
 \widehat\Theta_{j+p}'
                  =(-\widehat B_j',\widehat A_j'+\widehat B_j').
\]
The corresponding identities hold for the population rows. For either $q=1$ or $q=2$, the first row error is at most
$2|\widehat B_k-B_k|_q$, and the second is at most
$2|\widehat B_j-B_j|_q+|\widehat A_j-A_j|_q$.
The same inequalities hold for the row norms. This proves
\eqref{eq:node-l1}-\eqref{eq:node-row}, using $A$ rows only on $H_A$ and $B$ rows only on $H_B$.

For the final bound, the KKT equations and the matching scales give
\begin{align}
 |\widehat A_j(\tau)'\widehat{\bm M}(\tau)-\widetilde e_j'|_\infty
 &\le\frac{\lambda_{\mathrm{node}}|\widehat\Gamma_j(\tau)|_\infty}
                 {2\widehat z_j^2(\tau)},
       \label{eq:S340}\\
 |\widehat B_j(\tau)'\widehat{\bm N}(\tau)-\widetilde e_j'|_\infty
 &\le\frac{\lambda_{\mathrm{node}}
                          |\widehat{\widetilde\Gamma}_j(\tau)|_\infty}
                 {2\widehat{\widetilde z}_j^2(\tau)}.
       \label{eq:S341}
\end{align}
The empirical weights and inverse scales on the right-hand sides are uniformly bounded in probability. Since
\[
 \widehat{\bm\Sigma}(\tau)=
 \begin{pmatrix}
  \widehat{\bm M}(\tau)+\widehat{\bm N}(\tau)&\widehat{\bm M}(\tau)\\
  \widehat{\bm M}(\tau)&\widehat{\bm M}(\tau)
 \end{pmatrix},
\]
direct multiplication gives
\[
 \widehat{\bm\Theta}(\tau)\widehat{\bm\Sigma}(\tau)-I_{2p}
   =\begin{pmatrix}
       R_B(\tau)&0\\ R_A(\tau)-R_B(\tau)&R_A(\tau)
     \end{pmatrix},
\]
where $R_A=\widehat{\bm A}\widehat{\bm M}-I_p$ and $R_B=\widehat{\bm B}\widehat{\bm N}-I_p$. The uniform weight and inverse-scale bounds in Lemma~\ref{lem:S14-nodewise}, \eqref{eq:S335}, and \eqref{eq:S339} show that the required rows of $R_A$ and $R_B$
are $O_P(\lambda_{\mathrm{node}})=O_P(b_n)$ in maximum norm. The block identity proves \eqref{eq:node-kkt}. All arguments and probability bounds are uniform over the stated
model class when the constants and convergence conditions are uniform. None depends on whether a threshold effect is present.
\end{proof}

\subsubsection{Proofs for Section \ref{uniforminference}}

For a sequence $Z_i$, write $P_nZ_i=n^{-1}\sum_i Z_i$ and
$\overline E Z_i=n^{-1}\sum_i EZ_i$.  Define, for deterministic $\tau$, $\widetilde{\bm\Sigma}_{xu}(\tau)
   =P_n\bm X_i(\tau)\bm X_i(\tau)'U_i^2.$

\begin{lem}\label{lem:S15-inference}
Under Assumptions~\ref{as1} and \ref{asnd}(i),
\begin{align*}
 \max_{k,l,j\le p}\left|(P_n-\overline E)
             \left(X_i^{(k)}X_i^{(l)}X_i^{(j)}\right)^2\right|&=O_P(b_n),\\
 \max_{k,l\le p}\left|(P_n-\overline E)
             \left(X_i^{(k)}X_i^{(l)}U_i\right)^2\right|&=o_P(1),\\
 \sup_{\tau\in\mathbb T}
 \|\widetilde{\bm\Sigma}_{xu}(\tau)-\bm\Sigma_{xu}(\tau)\|_\infty
                  &=O_P(b_n).
\end{align*}
In particular,
\[
 \max_{k,l,j\le p}P_n\left(X_i^{(k)}X_i^{(l)}X_i^{(j)}\right)^2=O_P(1),
 \qquad
 \max_{k,l\le p}P_n\left(X_i^{(k)}X_i^{(l)}U_i\right)^2=O_P(1).
\]
\end{lem}
\begin{proof}
For the first assertion, H\"older's inequality and the twelfth covariate
moments give a common bound on the second moments of $\left(X_i^{(k)}X_i^{(l)}X_i^{(j)}\right)^2$. Apply Lemma~\ref{concenine} with dimension $p^3$ and envelope $\left(EM_{X^6}^2\right)^{1/2}$. Its logarithmic factor is $O(\ell_n)$,
and the envelope contribution is $o(b_n)$.

For the second assertion, let $B_n=\left(EM_{X^4U^2}^2\right)^{1/2}$. The second moment quantity in Lemma~\ref{concenine} is at most $B_n^2$.
The lemma therefore bounds the centered maximum by $C(B_nb_n+B_nb_n^2)=o(1)$ with probability tending to one.
Moreover, Assumption~\ref{as1}(iv) and the covariate fourth moments imply
\[
 E\left[\left(X_i^{(k)}X_i^{(l)}U_i\right)^2\right]
 \le C_U E\left[\left(X_i^{(k)}X_i^{(l)}\right)^2\right]\le C.
\]
This proves the second boundedness assertion. The corresponding boundedness for the triple products follows from the first assertion.
Only these boundedness conclusions, rather than an $O_P(b_n)$ rate for
the second centered maximum, are needed below.

For the last assertion, H\"older's inequality gives
\[
 E[(X_i^{(k)}X_i^{(l)}U_i^2)^2]
 \le \{E|X_i^{(k)}|^{12}\}^{1/6}
      \{E|X_i^{(l)}|^{12}\}^{1/6}
      \{E|U_i|^8\}^{1/2}\le C.
\]
Since the exponents on the right-hand side sum to $5/6$, the
remaining exponent may be assigned to a unit factor. Following the
deterministic grid argument in Lemma~\ref{tau22}, let
$q_r=r/n^3$, $r=0,\ldots,n^3$. Apply Lemma~\ref{concenine} to the
independent vectors with coordinates
$X_i^{(k)}X_i^{(l)}U_i^2\bm1\{Q_i<q_r\}$, indexed by $(k,l,r)$.
Their dimension is at most $p^2(n^3+1)$, their second moments are
uniformly bounded, and their envelope quantity is bounded by
$\{EM_{X^2U^2}^2\}^{1/2}$. Hence, with probability tending to one,
\[
\max_{k,l\le p}\max_{0\le r\le n^3}
\left|
(P_n-\overline E)
\left\{X_i^{(k)}X_i^{(l)}U_i^2\bm1\{Q_i<q_r\}\right\}
\right|
\le
C\left[b_n+\{EM_{X^2U^2}^2\}^{1/2}b_n^2\right]
=O(b_n),
\]
where the last equality follows from Assumption~\ref{asnd}(i).

Independence across observations and
$Q_i\sim\mathrm{Unif}[0,1]$ imply that the probability that any
grid interval contains two or more observations is at most
$\binom n2n^{-3}=O(n^{-1})$. Outside this event, for
$q_r\le\tau<q_{r+1}$, the empirical interpolation error is bounded
uniformly by $M_{X^2U^2}/n=o_P(b_n)$. Indeed, for every fixed
$\varepsilon>0$, Assumption~\ref{asnd}(i) gives
\[
P\!\left\{\frac{M_{X^2U^2}}{n}>\varepsilon b_n\right\}
\le
\frac{EM_{X^2U^2}^2}{\varepsilon^2n^2b_n^2}
=
\frac{b_n^2EM_{X^2U^2}^2}{\varepsilon^2\ell_n^2}
\longrightarrow0.
\]
The corresponding population interpolation error satisfies
\[
\begin{aligned}
&\frac1n\sum_{i=1}^n
E\!\left[
 \left|X_i^{(k)}X_i^{(l)}U_i^2\right|
 \bm1\{q_r\le Q_i<q_{r+1}\}
\right]\\
&\quad\le
\frac1n\sum_{i=1}^n
\left\{E\!\left[
 (X_i^{(k)}X_i^{(l)}U_i^2)^2
\right]\right\}^{1/2}
\left\{P(q_r\le Q_i<q_{r+1})\right\}^{1/2}\\
&\quad\le Cn^{-3/2}=o(b_n),
\end{aligned}
\]
uniformly in $k,l,r$, by Cauchy--Schwarz and
Assumption~\ref{as1}(ii). Combining the grid bound with these
interpolation bounds yields
\[
\sup_{\tau\in\mathbb T}\max_{k,l\le p}
\left|
(P_n-\overline E)
\left\{X_i^{(k)}X_i^{(l)}U_i^2\bm1\{Q_i<\tau\}\right\}
\right|
=O_P(b_n).
\]
The grid endpoint $q_{n^3}=1$ also covers the unthresholded
products, since $\bm1\{Q_i<1\}=1$ almost surely. Every entry of
the augmented covariance matrix is either an unthresholded
product or a thresholded product of this form. This proves
the last assertion.
\end{proof}

%%%%%%%%%%%%%%%%%%%%%%%%%
\subsubsection{Proof of Theorem \ref{thm3}}

\paragraph{No Threshold Effect.}
When $\delta_0=0$, define
\begin{equation}\label{eq:S342-inference}
 t=\frac{\sqrt n\,g'\{\widetilde a(\widehat\tau)-\alpha_0\}}
             {\{g'\widehat{\mathcal V}_n(\widehat\tau)g\}^{1/2}}.
\end{equation}
The exact no-threshold decomposition in Section~\ref{noshtrsholdbias} gives $t=t_1+t_2$, where
\[
 t_1=\frac{g'\widehat{\bm\Theta}(\widehat\tau)
                         \bm X(\widehat\tau)'U/\sqrt n}
                 {\{g'\widehat{\mathcal V}_n(\widehat\tau)g\}^{1/2}},
 \qquad
 t_2=-\frac{g'\Delta(\widehat\tau)}
                 {\{g'\widehat{\mathcal V}_n(\widehat\tau)g\}^{1/2}}.
\]
The minus sign in $t_2$ follows from the definition of $\Delta$ in
Section~\ref{noshtrsholdbias}. We show that $t_1\overset d\longrightarrow N(0,1)$ and
$t_2=o_P(1)$.

\begin{lem}\label{lem:S16-inference}
Under the no-threshold conditions of Theorem~\ref{thm3},
\[
 |g'\Delta(\widehat\tau)|
       =O_P\!\left(\frac{\sqrt h\,s_0\ell_n}{\sqrt n}\right)=o_P(1).
\]
\end{lem}
\begin{proof}
By H\"older's inequality, Lemma~\ref{lem:nodewise-main}, and Theorem~\ref{main-thm-case1},
\begin{align*}
 |g'\Delta(\widehat\tau)|
 &\le\sqrt n\,|g|_1
       \max_{j\in H}|\widehat{\bm\Theta}_j(\widehat\tau)'
                         \widehat{\bm\Sigma}(\widehat\tau)-e_j'|_\infty
                         |\widehat\alpha-\alpha_0|_1\\
 &=O_P(\sqrt n\sqrt h\,b_ns_0\lambda_n)
  =O_P(\sqrt h\,s_0\ell_n/\sqrt n)=o_P(1).
\end{align*}
All maxima used here already hold simultaneously over $\mathbb T$. Thus,
evaluation at $\widehat\tau$ requires no independence argument.
\end{proof}

%%%%%%%%%%%%%%%%%%%%%%%%%%%%
\begin{lem}\label{lem:S17-inference}
Under the no-threshold conditions of Theorem~\ref{thm3},
\[
 |g'\{\widehat{\mathcal V}_n(\widehat\tau)
                   -\mathcal V(\widehat\tau)\}g|
        =O_P(h\bar s s_0^{3/2}b_n)=o_P(1).
\]
\end{lem}
\begin{proof}
Consider the following three-term decomposition:
\begin{align}
 &|g'\widehat{\bm\Theta}(\widehat\tau)
   \{\widehat{\bm\Sigma}_{xu}(\widehat\tau)
                -\widetilde{\bm\Sigma}_{xu}(\widehat\tau)\}
                      \widehat{\bm\Theta}(\widehat\tau)'g|,
                                                        \label{eq:S343-inference}\\
 &|g'\widehat{\bm\Theta}(\widehat\tau)
   \{\widetilde{\bm\Sigma}_{xu}(\widehat\tau)
                -\bm\Sigma_{xu}(\widehat\tau)\}
                      \widehat{\bm\Theta}(\widehat\tau)'g|,
                                                        \label{eq:S344-inference}\\
 &|g'\widehat{\bm\Theta}(\widehat\tau)\bm\Sigma_{xu}(\widehat\tau)
                      \widehat{\bm\Theta}(\widehat\tau)'g
                  -g'\mathcal V(\widehat\tau)g|.
                                                        \label{eq:S345-inference}
\end{align}
Write $d_i=\widehat U_i(\widehat\tau)-U_i=f_{0,i}-\widehat f_i$.
Then $\widehat U_i^2-U_i^2=2U_id_i+d_i^2$.
Theorem~\ref{main-thm-case1} and the coefficient box in Assumption~\ref{as1}(i) imply
\[
 \|d\|_n=O_P(\sqrt{s_0}b_n),\qquad
 |\alpha_0|_1+|\widehat\alpha|_1=O_P(s_0).
\]
By Cauchy-Schwarz and Lemma~\ref{lem:S15-inference},
\[
 \max_{k,l\le p}P_n|X_i^{(k)}X_i^{(l)}U_id_i|
 \le\max_{k,l\le p}\{P_n(X_i^{(k)}X_i^{(l)}U_i)^2\}^{1/2}\|d\|_n
 =O_P(\sqrt{s_0}b_n).
\]
Also, bounding every threshold indicator by one, we derive
\begin{align*}
 \max_{k,l\le p}P_n|X_i^{(k)}X_i^{(l)}|d_i^2
 &\le\max_{k,l\le p}\|X^{(k)}X^{(l)}d\|_n\|d\|_n,\\
 \max_{k,l\le p}\|X^{(k)}X^{(l)}d\|_n
 &\le (|\alpha_0|_1+|\widehat\alpha|_1)
        \max_{k,l,j\le p}\{P_n(X_i^{(k)}X_i^{(l)}X_i^{(j)})^2\}^{1/2}
  =O_P(s_0).
\end{align*}
Consequently,
\[
 \|\widehat{\bm\Sigma}_{xu}(\widehat\tau)
                -\widetilde{\bm\Sigma}_{xu}(\widehat\tau)\|_\infty
                         =O_P(s_0^{3/2}b_n).
\]
Lemma~\ref{lem:nodewise-main} and the corresponding population row bounds give
$|\widehat{\bm\Theta}(\widehat\tau)'g|_1=O_P(\sqrt{h\bar s})$.
Thus \eqref{eq:S343-inference} is $O_P(h\bar s s_0^{3/2}b_n)$. Lemma~\ref{lem:S15-inference} makes \eqref{eq:S344-inference}
$O_P(h\bar s b_n)$.

For \eqref{eq:S345-inference}, write
$a=\bm\Theta(\widehat\tau)'g$ and
$e=\{\widehat{\bm\Theta}(\widehat\tau)-\bm\Theta(\widehat\tau)\}'g$.
The population row norms and Lemma~\ref{lem:nodewise-main} imply
$|a|_2\le C\sqrt h$ and $|e|_2=O_P(\sqrt{h\bar s}\,b_n)$.
The upper eigenvalue bound in Assumption~\ref{asnd}(iii) therefore gives
\[
 |2e'\bm\Sigma_{xu}(\widehat\tau)a
                  +e'\bm\Sigma_{xu}(\widehat\tau)e|
 \le C(2|e|_2|a|_2+|e|_2^2)
 =O_P(h\sqrt{\bar s}\,b_n+h\bar s b_n^2).
\]
Since $s_0,\bar s\ge1$ and $b_n=o(1)$, all three terms are $O_P(h\bar s s_0^{3/2}b_n)$. This tends to zero by Assumption~\ref{asnd}(ii). The argument uses no upper bound on the
operator norm of the full precision matrix.
\end{proof}

%%%%%%%%%%%%%%%%%%%%%%%%%%%%%%
\begin{proof}[Proof of Theorem~\ref{thm3} in the no-threshold case]
\emph{Step 1.1: a central limit theorem at the deterministic reference.}
Define
\[
 t_1'(\tau)=
 \frac{g'\bm\Theta(\tau)\bm X(\tau)'U/\sqrt n}
      {\{g'\mathcal V(\tau)g\}^{1/2}}
\]
whenever the denominator is positive. Apply the central limit theorem at $\tau=\tau_\ast$. The summands
$\xi_i=g'\bm\Theta(\tau_\ast)\bm X_i(\tau_\ast)U_i$ are independent and centered by $E[U_i\mid X_i,Q_i]=0$. Their average variance is
$g'\mathcal V(\tau_\ast)g\ge c_{\mathcal V}$.
The population row bounds imply
$|\bm\Theta(\tau_\ast)'g|_1\le C\sqrt{h\bar s}$.
Moreover, $\max_{i,j}E|X_i^{(j)}U_i|^4\le C$ by the eighth moments
and Cauchy--Schwarz. Minkowski's inequality yields
\[
 \frac{\sum_i E|\xi_i|^4}
      {n^2\{g'\mathcal V(\tau_\ast)g\}^2}
 \le C\frac{(h\bar s)^2}{n}
 \le C\frac{(h\bar s)^3}{n}\longrightarrow0.
\]
Lyapunov's theorem gives $t_1'(\tau_\ast)\overset d\longrightarrow N(0,1)$. All bounds hold uniformly over the class. Consequently, this conclusion holds along every sequence of qualifying distributions. A contradiction argument, together with the continuity of the normal distribution function, establishes uniform convergence of the distribution functions.

\emph{Step 1.2: comparison with the selected threshold.}
Theorem~\ref{main-thm-case1} gives
$|\widehat\tau-\tau_\ast|=O_P(s_0\lambda_n)$.
Lemma~\ref{lem:S20-inference}, proved below, gives, for a deterministic reference $t$ and sufficiently small deterministic $\eta>0$,
\begin{align*}
 &\sup_{\tau\in\mathbb T:\,|\tau-t|\le\eta}
 |g'\bm\Theta(\tau)\bm X(\tau)'U/\sqrt n
                -g'\bm\Theta(t)\bm X(t)'U/\sqrt n|\\
 &\hspace{15mm}
 =O_P\!\left(\sqrt h\,\bar s^{3/2}\eta\sqrt{\ell_n}
                         +\sqrt{h\bar s\eta}\right),
\end{align*}
and the population variance difference on this neighborhood is at most
$Ch\bar s\eta$. Use $t=\tau_\ast$ and
$\eta=M s_0\lambda_n$ for any fixed $M$. Assumption~\ref{asnd}(ii) implies
\[
 \sqrt h\,\bar s^{3/2}s_0\lambda_n\sqrt{\ell_n}=o(1),
 \qquad h\bar s s_0\lambda_n=o(1).
\]
First localize $\widehat\tau$ to this neighborhood, then let $n$ tend to
infinity for fixed $M$, and finally let $M$ increase. This proves
\[
 t_1'(\widehat\tau)=t_1'(\tau_\ast)+o_P(1).
\]
The ratio of the two population standard deviations tends to one, and
the reference standard deviation is bounded away from zero. This is a
random-index argument and does not assume that $\widehat\tau$ is
independent of the score.

\emph{Step 1.3: estimated inverse and studentization.}
By Lemma~\ref{lem:nodewise-main} and the uniform score event of Lemma~\ref{lemmaprobAB}, we obtain
\begin{align*}
 &|g'\{\widehat{\bm\Theta}(\widehat\tau)-\bm\Theta(\widehat\tau)\}
                                  \bm X(\widehat\tau)'U/\sqrt n|\\
 &\quad\le \sqrt h\max_{j\in H}\sup_{\tau\in\mathbb T}
           |\widehat{\bm\Theta}_j(\tau)-\bm\Theta_j(\tau)|_1
       \sup_{\tau\in\mathbb T}|\bm X(\tau)'U/\sqrt n|_\infty\\
 &\quad=O_P(\sqrt h\,\bar s b_n\sqrt{\ell_n})=o_P(1).
\end{align*}
Lemma~\ref{lem:S17-inference}, the local variance comparison, and
Assumption~\ref{asnd}(iii) imply
\[
 \frac{g'\widehat{\mathcal V}_n(\widehat\tau)g}
      {g'\mathcal V(\tau_\ast)g}\longrightarrow_P1.
\]
It follows that $t_1=t_1'(\tau_\ast)+o_P(1)$. This step uses the
tightness of the standardized reference score, which does not require the unstandardized variance to have a fixed numerical limit.

\emph{Step 2: the debiasing remainder.}
The denominator is bounded away from zero with probability tending to one, so Lemma~\ref{lem:S16-inference} gives $t_2=o_P(1)$. Slutsky's theorem establishes \eqref{asymdist}. The same uniform
probability bounds make this conclusion uniform over the qualifying no-threshold class. Lemmas~\ref{lem:S17-inference} and
\ref{lem:S22-inference} give \eqref{asymcov}.
\end{proof}

\paragraph{Fixed Threshold Effect.}
For $\delta_0\ne0$, define $t$ as in \eqref{eq:S342-inference}.
The exact decomposition in Section~\ref{shtrsholdbias} gives $t=t_1+t_2$, with the same
$t_1$ and
\[
 t_2=\frac{g'\widehat{\bm\Theta}(\widehat\tau)\bm X(\widehat\tau)'
           \{\bm X(\tau_0)-\bm X(\widehat\tau)\}\alpha_0/\sqrt n
                    -g'\Delta(\widehat\tau)}
                 {\{g'\widehat{\mathcal V}_n(\widehat\tau)g\}^{1/2}}.
\]

\begin{lem}\label{lem:S18-inference}
Under the fixed-threshold conditions of Theorem~\ref{thm3},
\[
 |g'\Delta(\widehat\tau)|
       =O_P(\sqrt h\,s_0\ell_n/\sqrt n)=o_P(1).
\]
\end{lem}
\begin{proof}
Use the proof of Lemma~\ref{lem:S16-inference}, with the coefficient bound from Theorem~\ref{thmftau} in place of Theorem~\ref{main-thm-case1}.
\end{proof}
%%%%%%%%%%%%%%%%%%%%%%%%%%%%%%%%%%%%%%%%%%%%
\begin{proof}[Proof of Theorem~\ref{thm3} in the fixed-threshold case]
\emph{Step 1.} Apply Step~1.1 of the no-threshold proof at the deterministic
reference $\tau_0$. Theorem~\ref{thmftau} gives
$|\widehat\tau-\tau_0|=O_P(s_0\lambda_n^2)$.
Apply the local comparison in Lemma~\ref{lem:S20-inference} on the
deterministic neighborhood $|\tau-\tau_0|\le M s_0\lambda_n^2$. Both modulus terms tend to zero, and this radius is smaller than that used in the no-threshold proof. Replace the estimated inverse by its population
counterpart as in Step~1.3, and apply Lemmas~\ref{lem:S21-inference} and
\ref{lem:S22-inference} for studentization. These steps prove
$t_1=t_1'(\tau_0)+o_P(1)\overset d\longrightarrow N(0,1)$.

\emph{Step 2.} Lemmas~\ref{lem:S18-inference} and \ref{lem:S19-inference} imply that the numerator of $t_2$ is $o_P(1)$.
Its denominator is bounded away from zero with probability tending to one. Hence $t_2=o_P(1)$, proving \eqref{asymdist} uniformly over the qualifying fixed threshold class. Lemma~\ref{lem:S22-inference} proves
\eqref{asymcov2}. Since the two subclasses have common bounds and the same standard normal limit, their union has the stated uniform limit.
\end{proof}
%%%%%%%%%%%%%%%%%%%%%%%%%%%%%%%%%%%%%%%%%%%%%%

\begin{lem}\label{lem:S19-inference}
Under the fixed-threshold conditions of Theorem~\ref{thm3},
\begin{align*}
 &|g'\widehat{\bm\Theta}(\widehat\tau)\bm X(\widehat\tau)'
           \{\bm X(\tau_0)-\bm X(\widehat\tau)\}\alpha_0/\sqrt n|\\
 &\qquad=O_P\!\left(\frac{s_0|\delta_0|_1\sqrt{h\bar s}\,\ell_n}{\sqrt n}\right)
 =O_P\!\left(\frac{s_0^2\sqrt{h\bar s}\,\ell_n}{\sqrt n}\right)=o_P(1).
\end{align*}
\end{lem}
\begin{proof}
Consider the block matrix calculation, using the exact cross Gram identity
\[
 \frac1n\bm X(t)'\bm X(s)=
 \begin{pmatrix}
   \widehat M&\widehat M(s)\\
   \widehat M(t)&\widehat M(\min\{t,s\})
 \end{pmatrix},
 \qquad \widehat M=P_nX_iX_i'.
\]
Thus
\[
 \frac1n\bm X(\widehat\tau)'
              \{\bm X(\tau_0)-\bm X(\widehat\tau)\}
 =\begin{pmatrix}
  0&\widehat M(\tau_0)-\widehat M(\widehat\tau)\\
  0&\widehat M(\min\{\tau_0,\widehat\tau\})-\widehat M(\widehat\tau)
 \end{pmatrix}.
\]
In general, the cross Gram matrix is not the same as the threshold Gram matrix. Write $\eta=|\widehat\tau-\tau_0|\vee b_n^2$.
With probability tending to one, $\eta$ lies in the simultaneous range of Assumption~\ref{A-smoothness}, by
Theorem~\ref{thmftau}. That assumption bounds each nonzero block by
$C_5\eta=O_P(s_0\lambda_n^2)$ in maximum entry norm. It applies even
when $|\widehat\tau-\tau_0|<b_n^2$ because of the maximum in the definition of $\eta$. Lemma~\ref{lem:nodewise-main} now bounds the left side in the lemma by
\[
 \sqrt n\,|g|_1
      \max_{j\in H}|\widehat{\bm\Theta}_j(\widehat\tau)|_1
      C_5\eta|\delta_0|_1
 =O_P(\sqrt n\sqrt{h\bar s}\,s_0\lambda_n^2|\delta_0|_1).
\]
Finally $|\delta_0|_1\le C_1s_0$, and
Assumption~\ref{asnd}(ii) makes this $o_P(1)$.
\end{proof}

\begin{lem}[Local comparison at a deterministic reference]\label{lem:S20-inference}
Under Assumptions~\ref{as1}, \ref{as2}(i), \ref{as6-nodewise}, and \ref{asnd}(iii), let $t=\tau_\ast$ in the no-threshold case and
$t=\tau_0$ in the identified case. For $\tau\in\mathbb T$ with
$|\tau-t|\le\varepsilon_0$, uniformly over the required rows,
\begin{align*}
 \max_{j\in H}|\bm\Theta_j(\tau)-\bm\Theta_j(t)|_2
          &\le C\bar s|\tau-t|,\\
 \max_{j\in H}|\bm\Theta_j(\tau)-\bm\Theta_j(t)|_1
          &\le C\bar s^{3/2}|\tau-t|,\\
 |g'\{\mathcal V(\tau)-\mathcal V(t)\}g|
          &\le Ch\bar s|\tau-t|.
\end{align*}
For every sufficiently small deterministic $\eta>0$,
\begin{align*}
 &\sup_{\tau\in\mathbb T:\,|\tau-t|\le\eta}
 |g'\bm\Theta(\tau)\bm X(\tau)'U/\sqrt n
                 -g'\bm\Theta(t)\bm X(t)'U/\sqrt n|\\
 &\hspace{12mm}
 =O_P\!\left(\sqrt h\,\bar s^{3/2}\eta\sqrt{\ell_n}
                         +\sqrt{h\bar s\eta}\right).
\end{align*}
In the identified case, under all the conditions of Theorem~\ref{thm3},
\[
 |g'\{\widehat{\bm\Theta}(\widehat\tau)
                    -\widehat{\bm\Theta}(\tau_0)\}|_1
                       =O_P(\sqrt h\,\bar s b_n).
\]
\end{lem}
\begin{proof}
\emph{Population matrix increments.}
The local conditional second moment bound $C_{\mathrm{loc}}$ in
Assumption~\ref{asnd}(iii), conditional Cauchy-Schwarz, and
$Q_i\sim\mathrm{Unif}[0,1]$ give
\[
 \|M(\tau)-M(t)\|_\infty+
 \|N(\tau)-N(t)\|_\infty\le C|\tau-t|.
\]
The same entrywise bound holds for the augmented matrix $\bm\Sigma(\tau)$. Further,
\[
 E[|X_i^{(j)}X_i^{(l)}|U_i^2\bm1\{Q_i\text{ lies between }t\text{ and }\tau\}]
 \le C_U E[|X_i^{(j)}X_i^{(l)}|
                 \bm1\{Q_i\text{ lies between }t\text{ and }\tau\}].
\]
It follows that
$\|\bm\Sigma_{xu}(\tau)-\bm\Sigma_{xu}(t)\|_\infty\le C|\tau-t|$.
No independence of $X_i$, $Q_i$, and $U_i^2$ is used.

\emph{Inverse row increments.}
Consider a required row $A_j(\tau)'$ of $A(\tau)=M(\tau)^{-1}$.
As in the proof of Lemma~\ref{lem:nodewise-main}, it has at most $\bar s$ nonzero entries, $|A_j(\tau)|_2\le C$, and $|A_j(\tau)|_1\le C\sqrt{\bar s}$.
For completeness, the first norm bound follows by applying the nodewise
restricted eigenvalue to $a=A_j(\tau)$ and using $a'M(\tau)a=e_j'a\le|a|_2$.
Write $v=A_j(\tau)-A_j(t)$. Its support has at most $2\bar s$ entries. The adaptive restricted eigenvalue in Assumption~\ref{as6-nodewise}(i)
also bounds its quadratic form below by $c|v|_2^2$, where we fix
$c:=\kappa_{\mathrm{node}}^2/2$ for this argument: choose the
$\bar s$ largest coordinates as $J$, so that
$|v_{J^c}|_1\le\sqrt{\bar s}|v_J|_2$ and
$|v_J|_2^2\ge|v|_2^2/2$.
The normal equations for the inverse rows give
\[
 M(\tau)v=\{M(t)-M(\tau)\}A_j(t).
\]
Therefore, we derive
\[
 c|v|_2^2\le |v|_1\|M(t)-M(\tau)\|_\infty|A_j(t)|_1
                  \le C\bar s|v|_2|\tau-t|.
\]
Thus $|v|_2\le C\bar s|\tau-t|$ and
$|v|_1\le C\bar s^{3/2}|\tau-t|$.
The same argument applies to the required rows of $B(\tau)=N(\tau)^{-1}$.
The block representation of $\bm\Theta(\tau)$ proves the stated inverse row bounds. This argument uses sparse restricted eigenvalues,
not a lower eigenvalue bound for the full matrix.

\emph{Population variance increments.}
Write $a(\tau)=\bm\Theta(\tau)'g$ within this proof. The row bounds give
\[
 |a(\tau)|_2\le C\sqrt h,\quad
 |a(\tau)|_1\le C\sqrt{h\bar s},\quad
 |a(\tau)-a(t)|_2\le C\sqrt h\,\bar s|\tau-t|.
\]
The identity
\begin{align*}
 g'\{\mathcal V(\tau)-\mathcal V(t)\}g
 ={}&\{a(\tau)-a(t)\}'\bm\Sigma_{xu}(\tau)\{a(\tau)+a(t)\}\\
   &+a(t)'\{\bm\Sigma_{xu}(\tau)-\bm\Sigma_{xu}(t)\}a(t)
\end{align*}
and the upper eigenvalue bound for $\bm\Sigma_{xu}$ show that the
absolute value is at most $Ch\bar s|\tau-t|$.

\emph{Local score increments.}
Write $S(\tau)=\bm X(\tau)'U/\sqrt n$ within this proof and decompose
\[
 a(\tau)'S(\tau)-a(t)'S(t)
  =\{a(\tau)-a(t)\}'S(\tau)+a(t)'\{S(\tau)-S(t)\}.
\]
Lemma~\ref{lemmaprobAB} gives $\sup_{\tau\in\mathbb T}|S(\tau)|_\infty
=O_P(\sqrt{\ell_n})$. The first local supremum is thus
$O_P(\sqrt h\,\bar s^{3/2}\eta\sqrt{\ell_n})$.
For the second term, write $a(t)=(a_\beta',a_\delta')'$.
Then
\[
 a(t)'\{S(\tau)-S(t)\}
 =\frac1{\sqrt n}\sum_i a_\delta'X_iU_i
                    \{\bm1(Q_i<\tau)-\bm1(Q_i<t)\}.
\]
Conditional on the entire design $\{X_i,Q_i\}_{i=1}^n$, the errors are independent and centered. Ordering the observations separately on either side of $t$ produces partial sum martingales for the
deterministic coefficient $a_\delta$. Their conditional $L^2$ maximal inequality, followed by integration over the design, yields
\begin{align*}
 &E\sup_{\tau\in\mathbb T:\,|\tau-t|\le\eta}
                  |a(t)'\{S(\tau)-S(t)\}|^2\\
 &\qquad\le\frac Cn\sum_i
       E[(a_\delta'X_i)^2U_i^2\bm1\{|Q_i-t|\le\eta\}]\\
 &\qquad\le C\eta|a_\delta|_1^2\le Ch\bar s\eta.
\end{align*}
The second inequality uses the local coordinate second moment bound
$C_{\mathrm{loc}}$ in Assumption~\ref{asnd}(iii) and
$E[U_i^2\mid X_i,Q_i]\le C_U$, not an operator bound on $E[X_iX_i'\mid Q_i]$. Markov's inequality proves the claimed modulus. More explicitly, if the uniform score event has failure probability
$\epsilon_n\to0$, then for every $x>0$ its local supremum exceeds
$C\sqrt h\,\bar s^{3/2}\eta\sqrt{\ell_n}+x$ with probability at most
$\epsilon_n+Ch\bar s\eta/x^2$.
This bound remains valid for deterministic references and radii that depend on the distribution and on $n$.

\emph{Estimated row comparison.}
Add and subtract the population inverse at each threshold. Lemma~\ref{lem:nodewise-main}, the population row increment, and Theorem~\ref{thmftau} imply
\begin{align*}
 |g'\{\widehat{\bm\Theta}(\widehat\tau)
              -\widehat{\bm\Theta}(\tau_0)\}|_1
 &=O_P(\sqrt h\,\bar s b_n)
       +O_P(\sqrt h\,\bar s^{3/2}s_0b_n^2)\\
 &=O_P(\sqrt h\,\bar s b_n),
\end{align*}
because $s_0\sqrt{\bar s}\,b_n=o(1)$ follows from Assumption~\ref{asnd}(ii).
\end{proof}

\begin{lem}\label{lem:S21-inference}
Under the fixed-threshold conditions of Theorem~\ref{thm3},
\[
 |g'\{\widehat{\mathcal V}_n(\widehat\tau)
                     -\mathcal V(\widehat\tau)\}g|
                         =O_P(h\bar s s_0^{3/2}b_n)=o_P(1).
\]
\end{lem}
\begin{proof}
Consider the three differences in
\eqref{eq:S343-inference}-\eqref{eq:S345-inference}.
The difference from Lemma~\ref{lem:S17-inference} is
\[
 d_i=\widehat U_i(\widehat\tau)-U_i
       =\bm X_i(\tau_0)'\alpha_0
                   -\bm X_i(\widehat\tau)'\widehat\alpha.
\]
Theorem~\ref{thmftau} still gives $\|d\|_n=O_P(\sqrt{s_0}b_n)$
and $|\alpha_0|_1+|\widehat\alpha|_1=O_P(s_0)$.
The proof of Lemma~\ref{lem:S17-inference} bounds both true and fitted threshold indicators by one. Its bounds for $2U_id_i$, $d_i^2$, and the other two differences therefore apply.
\end{proof}

\begin{lem}\label{lem:S22-inference}
Under the conditions of Theorem~\ref{thm3},
\begin{align*}
 |g'\{\widehat{\mathcal V}_n(\widehat\tau)
                       -\mathcal V(\tau_\ast)\}g|
          &=O_P(h\bar s s_0^{3/2}b_n)=o_P(1),&&\delta_0=0,\\
 |g'\{\widehat{\mathcal V}_n(\widehat\tau)
                       -\mathcal V(\tau_0)\}g|
          &=O_P(h\bar s s_0^{3/2}b_n)=o_P(1),&&\delta_0\ne0.
\end{align*}
\end{lem}
\begin{proof}
Add and subtract $g'\mathcal V(\widehat\tau)g$. Lemmas
\ref{lem:S17-inference} and \ref{lem:S21-inference} bound the variance estimation difference. Lemma~\ref{lem:S20-inference} and
Theorems~\ref{main-thm-case1}-\ref{thmftau} bound the population comparison by $O_P(h\bar s s_0\lambda_n)$ in the no-threshold case
and $O_P(h\bar s s_0\lambda_n^2)$ in the threshold case. In either case, the resulting order is bounded by the rate asserted in the lemma. Assumption~\ref{asnd}(ii)
makes that rate tend to zero. Together with the positive reference variance, these bounds also ensure a positive estimated variance with
probability tending to one uniformly over the model class.
\end{proof}
%%%%%%%%%%%%%%%%%%%%%%%%%%%%%%%%%%%
\paragraph{Fixed dimensional joint inference.}
The adjacent Wald statement in the main paper follows under its additional common contrast condition. Let $H$ be deterministic with fixed cardinality $h$. The componentwise versions of the expansions above give a vector remainder $o_P(1)$ and a deterministic-reference leading vector with covariance $\mathcal V(t)_{H,H}$, where
$t=\tau_\ast$ or $\tau_0$ in the respective case. This covariance has eigenvalues bounded above and away from zero, where the upper bound follows from $h$ fixed, bounded required inverse-row Euclidean norms, and the
upper bound for $\bm\Sigma_{xu}$, and the lower bound is the common contrast
condition. Apply the scalar Lyapunov argument to every deterministic projection of the whitened vector. The Cram\'er-Wold argument gives a standard $h$-dimensional normal limit. Polarization of the quadratic
variance consistency bounds gives consistency of all entries of the $H\times H$ covariance block. Studentization and the continuous mapping theorem prove \eqref{eq:fixed-h-wald}. This argument concerns fixed $h$.

%%%%%%%%%%%%%%%%%%%%%%%%%%%%%%%%%%%
\begin{proof}[Proof of Theorem~\ref{thm4}]
All probability bounds below are uniform over the qualifying distributions specified in Theorem~\ref{thm3}. For each such distribution $P_{\alpha_0}$, use the following deterministic reference solely as auxiliary notation in this proof:
\[
 \tau^\circ=
 \begin{cases}
  \tau_\ast,&\alpha_0\in\mathcal A_{\ell_0}^{(1)}(s_0),\\
  \tau_0,&\alpha_0\in\mathcal A_{\ell_0}^{(2)}(s_0).
 \end{cases}
\]
Write $\sigma_g^2(\tau)=g'\mathcal V(\tau)g$ and
$\widehat\sigma_g^2(\widehat\tau)
=g'\widehat{\mathcal V}_n(\widehat\tau)g$.
The lower bound in Assumption~\ref{asnd}(iii) gives
$\sigma_g^2(\tau^\circ)\ge c_{\mathcal V}>0$.
Consider the statistic
\[
 Z_n(\alpha_0)=
 \frac{\sqrt n\,g'\{\widetilde a(\widehat\tau)-\alpha_0\}}
                         {\widehat\sigma_g(\widehat\tau)},
 \qquad
 Z_n^\circ(\alpha_0)=
 \frac{g'\bm\Theta(\tau^\circ)\bm X(\tau^\circ)'U/\sqrt n}
                         {\sigma_g(\tau^\circ)}.
\]
Statistics involving a zero estimated variance may be assigned any fixed
value. Lemma~\ref{lem:S22-inference} makes this event uniformly negligible.
The deterministic reference Lyapunov argument in the proof of Theorem~\ref{thm3} gives
\begin{equation}\label{eq:thm4-S346}
 \sup_{\alpha_0\in\mathcal B_{\ell_0}(s_0)}\sup_{t\in\mathbb R}
 \left|P_{\alpha_0}\{Z_n^\circ(\alpha_0)\le t\}-\Phi(t)\right|
 \longrightarrow0.
\end{equation}
Its fourth moment Lyapunov bound is common to both classes and tends to zero along every sequence of qualifying distributions. The
continuity of the normal limit then yields the uniform distributional approximation. This uses the deterministic reference
part of that proof, rather than a substitution of $\widehat\tau$ into
a pointwise central limit theorem.

\emph{The seven events.}
Fix $0<\varepsilon<1/2$. Define
the following events under each $P_{\alpha_0}$. Their dependence on $\alpha_0$ and $\varepsilon$ is suppressed in the notation:
\begin{align*}
 \mathcal{F}_{1,n}&=\{|g'\Delta(\widehat\tau)|\le\varepsilon\},\\
 \mathcal{F}_{2,n}&=\left\{\sigma_g^2(\widehat\tau)>0,
          \left|\frac{\widehat\sigma_g^2(\widehat\tau)}
                         {\sigma_g^2(\widehat\tau)}-1\right|
                  \le\varepsilon\right\},\\
 \mathcal{F}_{3,n}&=\left\{
    \left|g'\{\widehat{\bm\Theta}(\widehat\tau)-\bm\Theta(\widehat\tau)\}
                 \bm X(\widehat\tau)'U/\sqrt n\right|\le\varepsilon
                 \right\},\\
 \mathcal{F}_{4,n}&=\left\{
    \left|g'\widehat{\bm\Theta}(\widehat\tau)\bm X(\widehat\tau)'
           \{\bm X(\tau_0)-\bm X(\widehat\tau)\}\alpha_0/\sqrt n\right|
            \le\varepsilon\right\},\qquad \delta_0\ne0,\\
 \mathcal{F}_{5,n}&=\left\{
    \left|g'\widehat{\bm\Theta}(\widehat\tau)\bm X(\widehat\tau)'U/\sqrt n
      -g'\widehat{\bm\Theta}(\tau^\circ)\bm X(\tau^\circ)'U/\sqrt n\right|
            \le\varepsilon\right\},\\
 \mathcal{F}_{6,n}&=\left\{
    \left|g'\{\widehat{\bm\Theta}(\tau^\circ)-\bm\Theta(\tau^\circ)\}
                 \bm X(\tau^\circ)'U/\sqrt n\right|\le\varepsilon
                 \right\},\\
 \mathcal{F}_{7,n}&=\left\{
          \left|\frac{\widehat\sigma_g^2(\widehat\tau)}
                         {\sigma_g^2(\tau^\circ)}-1\right|
                  \le\varepsilon\right\}.
\end{align*}
Set $\mathcal{F}_{4,n}$ equal to the whole sample space when $\delta_0=0$. 

Lemmas~\ref{lem:S16-inference} and \ref{lem:S18-inference} control
$\mathcal{F}_{1,n}$. Lemmas~\ref{lem:S17-inference} and \ref{lem:S21-inference}, together with the local variance comparison
in Lemma~\ref{lem:S20-inference}, control $\mathcal{F}_{2,n}$. Lemma~\ref{lem:nodewise-main} and the uniform score bound of
Lemma~\ref{lemmaprobAB} give
\[
 \sup_{\tau\in\mathbb T}
 \left|g'\{\widehat{\bm\Theta}(\tau)-\bm\Theta(\tau)\}
                       \bm X(\tau)'U/\sqrt n\right|
   =O_P(\sqrt h\,\bar s\ell_n/\sqrt n)=o_P(1),
\]
which controls both $\mathcal{F}_{3,n}$ and $\mathcal{F}_{6,n}$. Lemma~\ref{lem:S19-inference} controls $\mathcal{F}_{4,n}$ in the fixed-threshold case. To control $\mathcal{F}_{5,n}$, add and subtract the population-inverse score at both arguments. The two inverse estimation errors are bounded
by the last equality. Lemma~\ref{lem:S20-inference}, localized using
Theorem~\ref{main-thm-case1} or \ref{thmftau}, makes the intervening
population score difference $o_P(1)$ at $\tau_\ast$ or $\tau_0$, respectively. Finally, Lemma~\ref{lem:S22-inference} controls $\mathcal{F}_{7,n}$. Thus, for every fixed $\varepsilon>0$,
\begin{equation}\label{eq:thm4-S347}
 \sup_{\alpha_0\in\mathcal B_{\ell_0}(s_0)}
 P_{\alpha_0}\!\left(\bigcup_{k=1}^7 \mathcal{F}_{k,n}^{c}\right)
 \longrightarrow0.
\end{equation}
The same arguments show that the estimated variance and the population
variance at $\widehat\tau$ are bounded away from zero with probability
tending to one uniformly.

Because $\mathcal B_{\ell_0}(s_0)=\mathcal A_{\ell_0}^{(1)}(s_0)
\cup\mathcal A_{\ell_0}^{(2)}(s_0)$,
\begin{equation}\label{eq:thm4-S348}
\begin{split}
 &\sup_{\alpha_0\in\mathcal B_{\ell_0}(s_0)}\sup_{t\in\mathbb R}
       |P_{\alpha_0}\{Z_n(\alpha_0)\le t\}-\Phi(t)|\\
 &\quad=\max_{r\in\{1,2\}}
     \sup_{\alpha_0\in\mathcal A_{\ell_0}^{(r)}(s_0)}
     \sup_{t\in\mathbb R}
       |P_{\alpha_0}\{Z_n(\alpha_0)\le t\}-\Phi(t)|.
\end{split}
\end{equation}
We treat the two classes separately.

\emph{No-threshold case.}
For $\alpha_0\in\mathcal A_{\ell_0}^{(1)}(s_0)$, write
$\mathcal{G}_{1,n}=\mathcal{F}_{1,n}\cap \mathcal{F}_{2,n}\cap \mathcal{F}_{3,n}\cap \mathcal{F}_{5,n}\cap \mathcal{F}_{6,n}\cap \mathcal{F}_{7,n}$.
The exact decomposition in Section~\ref{noshtrsholdbias} and the definitions of $\mathcal{F}_{1,n}$, $\mathcal{F}_{5,n}$, and $\mathcal{F}_{6,n}$ imply, on $\mathcal{G}_{1,n}$,
\begin{equation}\label{eq:thm4-S349}
 \left|
  \frac{\sqrt n\,g'\{\widetilde a(\widehat\tau)-\alpha_0\}}
       {\sigma_g(\tau_\ast)}-Z_n^\circ(\alpha_0)
 \right|\le D_1\varepsilon,
 \qquad D_1=\frac3{\sqrt{c_{\mathcal V}}}.
\end{equation}
Also, $\mathcal{F}_{7,n}$ gives
\begin{equation}\label{eq:thm4-S350}
 \left|\frac{\widehat\sigma_g(\widehat\tau)}
                    {\sigma_g(\tau_\ast)}-1\right|
 \le\left|\frac{\widehat\sigma_g^2(\widehat\tau)}
                    {\sigma_g^2(\tau_\ast)}-1\right|
 \le\varepsilon.
\end{equation}
Consequently, the following lower and upper bounds hold for every $t\in\mathbb R$:
\begin{equation}\label{eq:thm4-S351}
\begin{split}
 &P_{\alpha_0}\{Z_n^\circ(\alpha_0)\le t-\varepsilon|t|-D_1\varepsilon\}
                      -P_{\alpha_0}(G_{1,n}^{c})\\
 &\quad\le P_{\alpha_0}\{Z_n(\alpha_0)\le t\}\\
 &\quad\le
 P_{\alpha_0}\{Z_n^\circ(\alpha_0)\le t+\varepsilon|t|+D_1\varepsilon\}
                      +P_{\alpha_0}(G_{1,n}^{c}).
\end{split}
\end{equation}
The terms $\varepsilon|t|$ account for either sign of $t$ when bounding the effect of studentization.

For each fixed $D>0$, there exists a finite constant $C_D>0$,
depending only on $D$, such that, for every $0<\varepsilon<1/2$,
\begin{equation}\label{eq:thm4-S352}
 \sup_{t\in\mathbb R}\max_{s\in\{-1,1\}}
 \left|\Phi(t+s\varepsilon|t|+sD\varepsilon)-\Phi(t)\right|
 \le C_D\varepsilon.
\end{equation}
To verify this bound, split at $t=0$. The argument of $\Phi$ then has the form $(1\pm\varepsilon)t\pm D\varepsilon$.
The translation changes $\Phi$ by at most
$D\varepsilon/\sqrt{2\pi}$. For the dilation, differentiation with respect to its scale and
$\sup_x |x|\phi(x)<\infty$, where $\phi=\Phi'$, give a bound $C\varepsilon$ uniformly in $t$ because the scale is at least $1/2$. Thus the argument controls the supremum over all $t$, not only each
fixed $t$.

Combining \eqref{eq:thm4-S351} and \eqref{eq:thm4-S352} yields
\begin{equation}\label{eq:thm4-S353}
\begin{split}
 &\sup_{\alpha_0\in\mathcal A_{\ell_0}^{(1)}(s_0)}\sup_{t\in\mathbb R}
        |P_{\alpha_0}\{Z_n(\alpha_0)\le t\}-\Phi(t)|\\
 &\quad\le
 \sup_{\alpha_0\in\mathcal A_{\ell_0}^{(1)}(s_0)}\sup_{u\in\mathbb R}
        |P_{\alpha_0}\{Z_n^\circ(\alpha_0)\le u\}-\Phi(u)|
       +C_{D_1}\varepsilon
       +\sup_{\alpha_0\in\mathcal A_{\ell_0}^{(1)}(s_0)}
                                  P_{\alpha_0}(G_{1,n}^{c}).
\end{split}
\end{equation}
For fixed $\varepsilon$, the first and last terms tend to zero by \eqref{eq:thm4-S346}-\eqref{eq:thm4-S347}. Letting
$\varepsilon\downarrow0$ proves the no-threshold conclusion.

\emph{Fixed-threshold case.}
For $\alpha_0\in\mathcal A_{\ell_0}^{(2)}(s_0)$, set
$\mathcal{G}_{2,n}=\bigcap_{k=1}^7 \mathcal{F}_{k,n}$.
The decomposition in Section~\ref{shtrsholdbias} includes the additional mismatch term controlled by $\mathcal{F}_{4,n}$. Hence, on $\mathcal{G}_{2,n}$,
\begin{equation}\label{eq:thm4-S354}
 \left|
  \frac{\sqrt n\,g'\{\widetilde a(\widehat\tau)-\alpha_0\}}
       {\sigma_g(\tau_0)}-Z_n^\circ(\alpha_0)
 \right|\le D_3\varepsilon,
 \qquad D_3=\frac4{\sqrt{c_{\mathcal V}}}.
\end{equation}
The scale bound \eqref{eq:thm4-S350} holds with $\tau_0$ in place of
$\tau_\ast$, by $\mathcal{F}_{7,n}$. Thus
\begin{equation}\label{eq:thm4-S355}
\begin{split}
 &P_{\alpha_0}\{Z_n^\circ(\alpha_0)\le t-\varepsilon|t|-D_3\varepsilon\}
                      -P_{\alpha_0}(G_{2,n}^{c})\\
 &\quad\le P_{\alpha_0}\{Z_n(\alpha_0)\le t\}\\
 &\quad\le
 P_{\alpha_0}\{Z_n^\circ(\alpha_0)\le t+\varepsilon|t|+D_3\varepsilon\}
                      +P_{\alpha_0}(G_{2,n}^{c}).
\end{split}
\end{equation}
Using \eqref{eq:thm4-S352}, we obtain
\begin{equation}\label{eq:thm4-S356}
\begin{split}
 &\sup_{\alpha_0\in\mathcal A_{\ell_0}^{(2)}(s_0)}\sup_{t\in\mathbb R}
        |P_{\alpha_0}\{Z_n(\alpha_0)\le t\}-\Phi(t)|\\
 &\quad\le
 \sup_{\alpha_0\in\mathcal A_{\ell_0}^{(2)}(s_0)}\sup_{u\in\mathbb R}
        |P_{\alpha_0}\{Z_n^\circ(\alpha_0)\le u\}-\Phi(u)|
       +C_{D_3}\varepsilon
       +\sup_{\alpha_0\in\mathcal A_{\ell_0}^{(2)}(s_0)}
                                  P_{\alpha_0}(G_{2,n}^{c}).
\end{split}
\end{equation}
First let $n\to\infty$ for fixed $\varepsilon$ and then let
$\varepsilon\downarrow0$. The fixed-threshold conclusion follows.
Combining the two classes through \eqref{eq:thm4-S348} gives
\begin{equation}\label{eq:thm4-S357}
 \sup_{\alpha_0\in\mathcal B_{\ell_0}(s_0)}\sup_{t\in\mathbb R}
        |P_{\alpha_0}\{Z_n(\alpha_0)\le t\}-\Phi(t)|\longrightarrow0,
\end{equation}
which proves \eqref{eq:thm4-normal}.

\emph{Coverage.}
Take $g=e_j$ for a deterministic coordinate satisfying the stated coordinate specific conditions. Denote the resulting statistic by
$Z_{n,j}(\alpha_0)$, and put $z=z_{1-\alpha/2}>0$.
On the event $\widehat\sigma_j(\widehat\tau)>0$, the interval misses $\alpha_0^{(j)}$ exactly when $Z_{n,j}>z$ or $Z_{n,j}<-z$.
Consequently,
\begin{equation}\label{eq:thm4-S358}
\begin{split}
 P_{\alpha_0}\{\alpha_0^{(j)}\notin I_{j,n}(1-\alpha)\}
  ={}&1-P_{\alpha_0}\{Z_{n,j}(\alpha_0)\le z\}\\
     &+P_{\alpha_0}\{Z_{n,j}(\alpha_0)<-z\}+o(1),
\end{split}
\end{equation}
where the $o(1)$ is uniform and accounts for the negligible zero variance event. By taking left limits, the uniform approximation in
\eqref{eq:thm4-S357} also applies to the strict inequality at $-z$, and no continuity of the finite sample statistic is needed. Since
$\Phi(z)-\Phi(-z)=1-\alpha$,
\begin{equation}\label{eq:thm4-S359}
\begin{split}
 &\sup_{\alpha_0\in\mathcal B_{\ell_0}(s_0)}
  \left|P_{\alpha_0}\{\alpha_0^{(j)}\in I_{j,n}(1-\alpha)\}-(1-\alpha)\right|\\
 &\quad\le
 2\sup_{\alpha_0\in\mathcal B_{\ell_0}(s_0)}\sup_{t\in\mathbb R}
       |P_{\alpha_0}\{Z_{n,j}(\alpha_0)\le t\}-\Phi(t)|+o(1)
 \longrightarrow0.
\end{split}
\end{equation}
This proves \eqref{eq:thm4-coverage}, including the equality to the nominal limiting coverage rather than only a lower bound.

\emph{Interval length.}
For the same coordinate, the population inverse row bound established in the proof of Lemma~\ref{lem:nodewise-main} gives
$\sup_{\tau\in\mathbb T}|\bm\Theta_j(\tau)|_2\le C$.
Together with Assumption~\ref{asnd}(iii), it implies a common constant
$C_\sigma<\infty$ such that
\[
 e_j'\mathcal V(\tau^\circ)e_j
 \le \phi_{\max}\{\bm\Sigma_{xu}(\tau^\circ)\}
                     |\bm\Theta_j(\tau^\circ)|_2^2
 \le C_\sigma.
\]
This uses the required inverse row, without imposing a lower bound on the smallest eigenvalue of the full design matrix.
Lemma~\ref{lem:S22-inference}, applied with $g=e_j$, now gives
\begin{equation}\label{eq:thm4-S360}
\begin{split}
 \sqrt n\,\operatorname{diam}\{I_{j,n}(1-\alpha)\}
   &=2z_{1-\alpha/2}\widehat\sigma_j(\widehat\tau)\\
   &=2z_{1-\alpha/2}
       \{e_j'\mathcal V(\tau^\circ)e_j+o_P(1)\}^{1/2}
     =O_P(1),
\end{split}
\end{equation}
uniformly over the qualifying model class. More explicitly, for any
$K>2z_{1-\alpha/2}\sqrt{C_\sigma+1}$,
\begin{align*}
 &\sup_{\alpha_0\in\mathcal B_{\ell_0}(s_0)}
 P_{\alpha_0}\{\sqrt n\,\operatorname{diam}\{I_{j,n}(1-\alpha)\}>K\}\\
 &\quad\le
 \sup_{\alpha_0\in\mathcal B_{\ell_0}(s_0)}
 P_{\alpha_0}\{|\widehat\sigma_j^2(\widehat\tau)
                         -e_j'\mathcal V(\tau^\circ)e_j|>1\}
 \longrightarrow0.
\end{align*}
This establishes \eqref{eq:thm4-length} and completes the proof.
\end{proof}

% Lemma 3
%===============================================================
\subsubsection{Proof of Proposition \ref{prop:threshold-size}}

We first justify the deterministic proof grid introduced in Section~\ref{sec:null-grid}. For the observed value implementation, divide $[t_0,t_1]$ into the $N_n$ cells with endpoints in $\mathcal G_n^\circ$. Let $\mathcal E_n$ be the event that no cell contains two observations and no $Q_i$ equals a cell endpoint. Independence across observations and their uniform marginal
distributions give
\[
 P_{\alpha_0}(\mathcal E_n^c)
 \le \binom{n}{2}\frac{(t_1-t_0)^2}{N_n}
 =O(n^{-1}).
\]
On $\mathcal E_n$, consecutive deterministic grid points either generate the same partition or differ by one observation. Thus every sample partition over $\mathbb T$ occurs on $\mathcal G_n^\circ$, as it does on
$\mathcal G_n$. Because identical design matrices yield identical fitted quantities, the maxima over the two grids are the same. Denoting the statistics on $\mathcal G_n^\circ$ by $T_n^\circ$ and $T_n^{*,\circ}$, we have,
on $\mathcal E_n$,
\[
 T_n=T_n^\circ,\qquad T_n^*=T_n^{*,\circ}
 \quad\text{for every multiplier realization},\qquad
 c_{1-\zeta}^*=c_{1-\zeta}^{*,\circ}.
\]
For a deterministic implemented grid, these identities hold exactly, so we can
take $\mathcal E_n$ to be the whole sample space. This argument concerns the fitted statistics. It does not assert that population scores at random grid points are independent across observations.

All stochastic bounds below are uniform over the qualifying null distributions. For compactness, $\max_a$ and $\max_{a,b}$ in
Lemmas~\ref{lem:test-variance}-\ref{lem:test-covariance} range over $\mathcal I_n^\circ$, a deterministic index set. Write
$Z_{n,j}(\tau)=\sqrt n\,\widetilde a_{\delta,j}(\tau)/
\widehat\sigma_{\delta,j}(\tau)$ when the denominator is positive, and set it to zero otherwise.

%%%%%%%%%%%%%%%%%%%%%%%%%
\begin{lem}[Uniform linear representation]\label{lem:test-linear}
Suppose that $\delta_0=0$, Assumptions~\ref{as1}, \ref{as2} with
$\mathbb S=\mathbb T$ and the cone constant from Theorem~\ref{main-thm-case1}, and Assumption \ref{as6-nodewise} for $H=H_\delta$ hold, and that \eqref{eq:max-linear-rate} holds. Then
\begin{equation}\label{eq:test-S361}
 \sup_{\tau\in\mathbb T}\max_{j\le p}
 \left|\sqrt n\,\widetilde a_{\delta,j}(\tau)
 -\frac1{\sqrt n}\sum_{i=1}^n
                    \Theta_{p+j}(\tau)'\bm X_i(\tau)U_i\right|
                  =o_P(\ell_n^{-1/2}).
\end{equation}
\end{lem}

\begin{proof}[Proof of Lemma~\ref{lem:test-linear}]
Under $H_0$, $\alpha_0=(\beta_0',0_p')'$ and
$Y=\bm X(\tau)\alpha_0+U$ for every $\tau\in\mathbb T$. Consequently,
\[
 \widetilde a(\tau)-\alpha_0
 =\widehat{\bm\Theta}(\tau)\frac{\bm X(\tau)'U}{n}
  +\{I_{2p}-\widehat{\bm\Theta}(\tau)\widehat{\bm\Sigma}(\tau)\}
                         \{\widehat\alpha(\tau)-\alpha_0\}.
\]
For $r=p+j$, add and subtract the population precision row to obtain
\begin{equation}\label{eq:test-S362}
 \sqrt n\,\widetilde a_{\delta,j}(\tau)
 =\Theta_r(\tau)'\frac{\bm X(\tau)'U}{\sqrt n}
                       +R_{1,r}(\tau)+R_{2,r}(\tau),
\end{equation}
where
\[
\begin{split}
 R_{1,r}(\tau)&=\{\widehat\Theta_r(\tau)-\Theta_r(\tau)\}'
                              \bm X(\tau)'U/\sqrt n,\\
 R_{2,r}(\tau)&=\sqrt n\,e_r'
       \{I_{2p}-\widehat{\bm\Theta}(\tau)\widehat{\bm\Sigma}(\tau)\}
                              \{\widehat\alpha(\tau)-\alpha_0\}.
\end{split}
\]
H\"older's inequality yields
\begin{equation}\label{eq:test-S363}
 \sup_{\tau,j}|R_{1,p+j}(\tau)|
 \le \sup_{\tau,j}|\widehat\Theta_{p+j}(\tau)-\Theta_{p+j}(\tau)|_1
           \sup_\tau\left|\frac{\bm X(\tau)'U}{\sqrt n}\right|_\infty,
\end{equation}
where suprema range over $\mathbb T$ and $1\le j\le p$. Lemma~\ref{lem:nodewise-main} gives the first factor as
$O_P(\bar s_\delta b_n)$ while the uniform score events in Lemma~S.4 give
$\sup_\tau|\bm X(\tau)'U/n|_\infty=O_P(b_n)$. We thus derive
\begin{equation}\label{eq:test-S364}
 \sup_{\tau,j}|R_{1,p+j}(\tau)|
             =O_P(\bar s_\delta\ell_n/\sqrt n).
\end{equation}
Similarly,
\begin{equation}\label{eq:test-S365}
\begin{split}
 \sup_{\tau,j}|R_{2,p+j}(\tau)|
 &\le\sqrt n\sup_{\tau,j}
       |\widehat\Theta_{p+j}(\tau)'\widehat{\bm\Sigma}(\tau)-e_{p+j}'|_\infty
       \sup_\tau|\widehat\alpha(\tau)-\alpha_0|_1.
\end{split}
\end{equation}
For each fixed $\tau$, the anchoring term is constant in $\alpha$ and cancels from the coefficient basic inequality. With
$v(\tau)=\widehat\alpha(\tau)-\alpha_0$, that inequality is
\[
 \|\bm X(\tau)v(\tau)\|_n^2
 +\lambda_n|\bm D(\tau)\widehat\alpha(\tau)|_1
 \le\frac2nU'\bm X(\tau)v(\tau)
       +\lambda_n|\bm D(\tau)\alpha_0|_1.
\]
On the simultaneous score, weight, and restricted eigenvalue events used for Theorem~\ref{main-thm-case1}, the usual support decomposition and weighted cone argument apply for every $\tau$. They give
\begin{equation}\label{eq:test-S366}
 \sup_{\tau\in\mathbb T}|\widehat\alpha(\tau)-\alpha_0|_1
            =O_P(s_0\lambda_n)=O_P(s_0b_n).
\end{equation}
This is a bound for the whole fixed-$\tau$ family, derived from its basic inequality, rather than a substitution into a bound at $\widehat\tau$. The nodewise KKT bound in Lemma~\ref{lem:nodewise-main} now yields
\begin{equation}\label{eq:test-S367}
 \sup_{\tau,j}|R_{2,p+j}(\tau)|
                       =O_P(s_0\ell_n/\sqrt n).
\end{equation}
Combining \eqref{eq:test-S362}, \eqref{eq:test-S364}, and
\eqref{eq:test-S367} bounds the left side of \eqref{eq:test-S361} by
$O_P\{(s_0+\bar s_\delta)\ell_n/\sqrt n\}$, which is $o_P(\ell_n^{-1/2})$ by \eqref{eq:max-linear-rate}.
\end{proof}

%===============================================================

\begin{lem}[Uniform variance estimation and studentization]\label{lem:test-variance}
Suppose that the conditions of Lemma~\ref{lem:test-linear},
Assumption~\ref{asnd}(i), and Assumption~\ref{as:max-test} hold. Then
\begin{align}
 \max_{\tau\in\mathcal G_n^\circ}\max_{j\le p}
 \left|\frac{\widehat\sigma_{\delta,j}(\tau)}{\sigma_{\delta,j}(\tau)}-1\right|
     &=o_P(\ell_n^{-1}),\label{eq:test-S368}\\
 \max_{\tau\in\mathcal G_n^\circ}\max_{j\le p}
 \left|Z_{n,j}(\tau)-\frac1{\sqrt n}\sum_{i=1}^n\psi_{i,j}(\tau)\right|
     &=o_P(\ell_n^{-1/2}).\label{eq:test-S369}
\end{align}
\end{lem}

\begin{proof}[Proof of Lemma~\ref{lem:test-variance}]
Define
\[
\begin{split}
 \widehat{\bm\Sigma}_{xu}^{\mathrm{null}}(\tau)
 &=\frac1n\sum_i\bm X_i(\tau)\bm X_i(\tau)'\left(\widehat U_i^{\mathrm{null}}\right)^2,\\
 \widetilde{\bm\Sigma}_{xu}(\tau)
 &=\frac1n\sum_i\bm X_i(\tau)\bm X_i(\tau)'U_i^2,\qquad
 \bm\Sigma_{xu}(\tau)=\frac1n\sum_iE[\bm X_i(\tau)\bm X_i(\tau)'U_i^2].
\end{split}
\]
For suprema in the following covariance bounds, $\tau$ ranges over
$\mathcal G_n^\circ$. Let
\begin{equation}\label{eq:test-S370}
\begin{split}
 r_{\Sigma,n}&=\sup_\tau
   |\widehat{\bm\Sigma}_{xu}^{\mathrm{null}}(\tau)-\bm\Sigma_{xu}(\tau)|_\infty
                    \le r_{\Sigma,n,1}+r_{\Sigma,n,2},\\
 r_{\Sigma,n,1}&=\sup_\tau
   |\widetilde{\bm\Sigma}_{xu}(\tau)-\bm\Sigma_{xu}(\tau)|_\infty,\\
 r_{\Sigma,n,2}&=\sup_\tau
   |\widehat{\bm\Sigma}_{xu}^{\mathrm{null}}(\tau)
                                      -\widetilde{\bm\Sigma}_{xu}(\tau)|_\infty.
\end{split}
\end{equation}
Each entry of the oracle covariance is an average of $X_i^{(j)}X_i^{(l)}U_i^2$
possibly multiplied by $\bm1\{Q_i<\tau\}$. Apply Lemma~\ref{concenine} to the
independent vectors containing these entries on the deterministic grid. Their dimension is at most $4p^2K_n^\circ$, with logarithm $O(\ell_n)$. Their average second moments are bounded by Assumption~\ref{asnd}(i), and their maximum envelope is bounded by $M_{X^2U^2}$. It follows that
\begin{equation}\label{eq:test-S371}
                         r_{\Sigma,n,1}=O_P(b_n).
\end{equation}
No reordering by random $Q_i$ is used in this application.

Let $v=\widehat\beta^{\mathrm{null}}-\beta_0$. The null-restricted Lasso
has the same score and weight bounds. Its restricted eigenvalue bound follows by embedding a $p$-vector as $(v',0_p')'$ in the augmented design. Its basic inequality therefore gives
\begin{equation}\label{eq:test-S372}
 |v|_1=O_P(s_0b_n),\qquad \|Xv\|_n=O_P(\sqrt{s_0}\,b_n).
\end{equation}
Under $H_0$, $\left(\widehat U_i^{\mathrm{null}}\right)^2-U_i^2
=-2U_iX_i'v+(X_i'v)^2$. Thus
$r_{\Sigma,n,2}\le2R_{n,21}+R_{n,22}$, where
\[
\begin{split}
 R_{n,21}&=\sup_\tau\left|\frac1n\sum_i
             \bm X_i(\tau)\bm X_i(\tau)'U_iX_i'v\right|_\infty,\\
 R_{n,22}&=\sup_\tau\left|\frac1n\sum_i
             \bm X_i(\tau)\bm X_i(\tau)'(X_i'v)^2\right|_\infty.
\end{split}
\]
The empirical moment bound
$\max_{j,l}n^{-1}\sum_i\left(X_i^{(j)}X_i^{(l)}U_i\right)^2=O_P(1)$
from Lemma~\ref{lem:S15-inference} suffices. A centered $O_P(b_n)$ rate for this moment is not needed. Cauchy-Schwarz gives
\begin{equation}\label{eq:test-S373}
 R_{n,21}\le
 \max_{j,l}\left\{\frac1n\sum_i(X_i^{(j)}X_i^{(l)}U_i)^2\right\}^{1/2}
 \|Xv\|_n=O_P(\sqrt{s_0}\,b_n).
\end{equation}
Also $M_{XX}=O_P(M_{XX,n})$, so
\begin{equation}\label{eq:test-S374}
 R_{n,22}\le M_{XX}\|Xv\|_n^2
  =O_P(M_{XX,n}s_0b_n^2)=O_P(\sqrt{s_0}\,b_n),
\end{equation}
where the last equality uses \eqref{eq:max-envelope-rate}. Therefore
\begin{equation}\label{eq:test-S375}
                         r_{\Sigma,n}=O_P(\sqrt{s_0}\,b_n).
\end{equation}

We next control centering. Set $q_n=\max_a|m_a|$ and
$\widetilde v=(v',0_p')'$. Since $e_{p+j}'\widetilde v=0$,
\begin{equation}\label{eq:test-S376}
\begin{split}
 m_j(\tau)
 &=\Theta_{p+j}(\tau)'\frac{\bm X(\tau)'U}{n}
  +\{\widehat\Theta_{p+j}(\tau)-\Theta_{p+j}(\tau)\}'
                                       \frac{\bm X(\tau)'U}{n}\\
 &\quad-\{\widehat\Theta_{p+j}(\tau)'\widehat{\bm\Sigma}(\tau)
                                                -e_{p+j}'\}\widetilde v.
\end{split}
\end{equation}
The standardized population scores have zero means and average second moments equal to one. In Lemma~\ref{concenine} their envelope satisfies
\[
 \left(E\max_{i,a}|\psi_{i,a}|^2\right)^{1/2}
 \le\left(\sum_iE\max_a|\psi_{i,a}|^4\right)^{1/4}
 \le n^{1/4}B_n.
\]
Since the index set is deterministic, that lemma applies directly and gives
\begin{equation}\label{eq:test-S377}
 \max_a\left|\frac1{\sqrt n}\sum_i\psi_{i,a}\right|
 =O_P(\sqrt{\ell_n}+B_nn^{-1/4}\ell_n)
 =O_P(\sqrt{\ell_n}),
\end{equation}
by \eqref{eq:max-cck-rate}. Use this bound and
Assumption~\ref{as:max-test}(i) for the first term in \eqref{eq:test-S376}. The nodewise $\ell_1$ and KKT bounds control the
remaining terms. Consequently,
\begin{align}
 q_n&=O_P\left\{\sqrt{\ell_n/n}
                    +(s_0+\bar s_\delta)\ell_n/n\right\},
                      \label{eq:test-S378}\\
 q_n^2&=O_P\left\{\ell_n/n
                    +(s_0+\bar s_\delta)^2\ell_n^2/n^2\right\}.
                      \label{eq:test-S379}
\end{align}

Write $\theta_j(\tau)=\Theta_{p+j}(\tau)$ and
$\widehat\theta_j(\tau)=\widehat\Theta_{p+j}(\tau)$ as column vectors.
Then
\[
 \widehat\sigma_{\delta,j}^2(\tau)
  =\widehat\theta_j(\tau)'\widehat{\bm\Sigma}_{xu}^{\mathrm{null}}(\tau)
                         \widehat\theta_j(\tau)-m_j(\tau)^2,
 \qquad
 \sigma_{\delta,j}^2(\tau)=\theta_j(\tau)'\bm\Sigma_{xu}(\tau)\theta_j(\tau).
\]
Hence, we obtain
\begin{equation}\label{eq:test-S380}
 |\widehat\sigma_{\delta,j}^2-\sigma_{\delta,j}^2|
 \le|\widehat\theta_j'\widehat{\bm\Sigma}_{xu}^{\mathrm{null}}
                   \widehat\theta_j-\theta_j'\bm\Sigma_{xu}\theta_j|+m_j^2.
\end{equation}
Suppressing $\tau$, expand the matrix difference as
\[
 (\widehat\theta_j-\theta_j)'
                  \widehat{\bm\Sigma}_{xu}^{\mathrm{null}}\widehat\theta_j
 +\theta_j'(\widehat{\bm\Sigma}_{xu}^{\mathrm{null}}-\bm\Sigma_{xu})
                                                   \widehat\theta_j
 +\theta_j'\bm\Sigma_{xu}(\widehat\theta_j-\theta_j).
\]
The nodewise results give
$\max_a|\widehat\theta_a-\theta_a|_1=O_P(\bar s_\delta b_n)$,
$\max_a|\widehat\theta_a|_1=O_P(\sqrt{\bar s_\delta})$, and
$\max_a|\theta_a|_1=O(\sqrt{\bar s_\delta})$.
Population covariance entries are uniformly bounded by the moment conditions, and \eqref{eq:test-S375} makes their estimated counterparts $O_P(1)$. Thus
\begin{equation}\label{eq:test-S381}
 \max_a|\widehat\theta_a'\widehat{\bm\Sigma}_{xu}^{\mathrm{null}}
                    \widehat\theta_a-\theta_a'\bm\Sigma_{xu}\theta_a|
 =O_P(\bar s_\delta^{3/2}b_n+\bar s_\delta\sqrt{s_0}\,b_n).
\end{equation}
It follows that
\begin{equation}\label{eq:test-S382}
 \max_a|\widehat\sigma_a^2-\sigma_a^2|
 =O_P\left\{\bar s_\delta^{3/2}b_n+\bar s_\delta\sqrt{s_0}\,b_n
             +\frac{\ell_n}{n}
             +\frac{(s_0+\bar s_\delta)^2\ell_n^2}{n^2}\right\}.
\end{equation}
The rates in Assumption~\ref{as:max-test} make this $o_P(\ell_n^{-1})$ (indeed $o_P(\ell_n^{-2})$), so
\begin{equation}\label{eq:test-S383}
                   \max_a|\widehat\sigma_a^2-\sigma_a^2|
                              =o_P(\ell_n^{-1}).
\end{equation}
Uniform nondegeneracy implies \eqref{eq:test-S368}, and positivity of all estimated variances on $\mathcal G_n^\circ$ with probability tending
to one. The grid equivalence gives positivity on $\mathcal G_n$ as well.

For studentization, set $S_{n,j}(\tau)=n^{-1/2}\sum_i\psi_{i,j}(\tau)$.
Lemma~\ref{lem:test-linear} gives
$\sqrt n\,\widetilde a_{\delta,j}(\tau)
=\sigma_{\delta,j}(\tau)S_{n,j}(\tau)\\+r_{n,j}(\tau)$, with
$\max_a|r_{n,a}|=o_P(\ell_n^{-1/2})$. Therefore, on the positive-variance event,
\begin{equation}\label{eq:test-S384}
 Z_{n,j}(\tau)-S_{n,j}(\tau)
 =\frac{r_{n,j}(\tau)}{\widehat\sigma_{\delta,j}(\tau)}
  +S_{n,j}(\tau)
       \left\{\frac{\sigma_{\delta,j}(\tau)}{\widehat\sigma_{\delta,j}(\tau)}-1\right\}.
\end{equation}
Equations~\eqref{eq:test-S368} and \eqref{eq:test-S377} bound the
second term uniformly by $O_P(\sqrt{\ell_n})o_P(\ell_n^{-1})$.
This proves \eqref{eq:test-S369}.
\end{proof}

%===============================================================

\begin{lem}[Uniform bootstrap covariance approximation]\label{lem:test-covariance}
For $a,b\in\mathcal I_n^\circ$, define
\[
 \Gamma_{ab}=\frac1n\sum_iE[\psi_{i,a}\psi_{i,b}],
 \qquad \widehat\Gamma_{ab}=\frac1n\sum_i\widehat\psi_{i,a}\widehat\psi_{i,b},
\]
with zero standardized scores for omitted coordinates. Let
\begin{align}
 \Delta_n&=\max_{a,b\in\mathcal I_n^\circ}
                           |\widehat\Gamma_{ab}-\Gamma_{ab}|,
                         \label{eq:test-S385}\\
 r_{\Delta,n}&=\bar s_\delta^{3/2}b_n+\bar s_\delta\sqrt{s_0}\,b_n
        +\frac{\ell_n}{n}+\frac{(s_0+\bar s_\delta)^2\ell_n^2}{n^2}.
                         \label{eq:test-S386}
\end{align}
Under the conditions of Lemma~\ref{lem:test-variance},
\begin{align}
 \Delta_n&=O_P(r_{\Delta,n}),\label{eq:test-S387}\\
 \Delta_n\{1\vee\log(2d_n^\circ/\Delta_n)\}^2&=o_P(1),
                         \label{eq:test-S388}
\end{align}
where the expression in \eqref{eq:test-S388} is zero when $\Delta_n=0$.
\end{lem}

\begin{proof}[Proof of Lemma~\ref{lem:test-covariance}]
For $a=(j,\tau)$ and $b=(k,\widetilde\tau)$, let
$\theta_a=\Theta_{p+j}(\tau)$, $\widehat\theta_a=\widehat\Theta_{p+j}(\tau)$,
and define $\theta_b,\widehat\theta_b$ analogously. Introduce
\[
\begin{split}
 \widehat{\bm\Sigma}_{xu}^{\mathrm{null}}(\tau,\widetilde\tau)
  &=\frac1n\sum_i\bm X_i(\tau)\bm X_i(\widetilde\tau)'
                                              \left(\widehat U_i^{\mathrm{null}}\right)^2,\\
 \bm\Sigma_{xu}(\tau,\widetilde\tau)
  &=\frac1n\sum_iE\left[\bm X_i(\tau)\bm X_i(\widetilde\tau)'U_i^2\right],
\end{split}
\]
and
\begin{equation}\label{eq:test-S389}
 r_{\Sigma,n}^{\times}
 =\max_{\tau,\widetilde\tau\in\mathcal G_n^\circ}
 \left|\widehat{\bm\Sigma}_{xu}^{\mathrm{null}}(\tau,\widetilde\tau)
                             -\bm\Sigma_{xu}(\tau,\widetilde\tau)\right|_\infty.
\end{equation}
Products of two threshold indicators equal
$\bm1\{Q_i<\tau\wedge\widetilde\tau\}$. The oracle covariance therefore has the same moment and envelope bounds as in \eqref{eq:test-S371}. Lemma~\ref{concenine} may be applied to at most
$4p^2(K_n^\circ)^2$ deterministic coordinates, whose logarithm remains $O(\ell_n)$. The residual-replacement bounds \eqref{eq:test-S373}-\eqref{eq:test-S374} are unchanged, since the
additional indicators are bounded by one. Consequently, we obtain
\begin{equation}\label{eq:test-S390}
                r_{\Sigma,n}^{\times}=O_P(\sqrt{s_0}\,b_n).
\end{equation}
Define
\[
\begin{split}
 V_{ab}&=\frac1n\sum_iE[m_{i,a}m_{i,b}]
          =\theta_a'\bm\Sigma_{xu}(\tau,\widetilde\tau)\theta_b,\\
 \widehat V_{ab}&=\frac1n\sum_i(\widehat m_{i,a}-m_a)(\widehat m_{i,b}-m_b)
          =\widehat\theta_a'\widehat{\bm\Sigma}_{xu}^{\mathrm{null}}
                      (\tau,\widetilde\tau)\widehat\theta_b-m_am_b,
 \qquad \eta_n=\max_{a,b}|\widehat V_{ab}-V_{ab}|.
\end{split}
\]
Suppressing the threshold arguments, expand
\begin{equation}\label{eq:test-S391}
\widehat\theta_a'\widehat{\bm\Sigma}_{xu}^{\mathrm{null}}\widehat\theta_b
                   -\theta_a'\bm\Sigma_{xu}\theta_b
 =(\widehat\theta_a-\theta_a)'
                      \widehat{\bm\Sigma}_{xu}^{\mathrm{null}}\widehat\theta_b
 \quad+\theta_a'(\widehat{\bm\Sigma}_{xu}^{\mathrm{null}}-\bm\Sigma_{xu})
                                                          \widehat\theta_b
       +\theta_a'\bm\Sigma_{xu}(\widehat\theta_b-\theta_b).
\end{equation}
The first and third terms are $O_P(\bar s_\delta^{3/2}b_n)$ uniformly while
the second is $O_P(\bar s_\delta\sqrt{s_0}\,b_n)$ by
\eqref{eq:test-S390}. Also $\max_{a,b}|m_am_b|\le q_n^2$, bounded in \eqref{eq:test-S379}. Thus
\begin{equation}\label{eq:test-S392}
                              \eta_n=O_P(r_{\Delta,n}).
\end{equation}
Since $V_{aa}=\sigma_a^2$ and $\widehat V_{aa}=\widehat\sigma_a^2$, uniform nondegeneracy gives $\max_a|\widehat\sigma_a-\sigma_a|=O_P(\eta_n)$.
Moreover, $|V_{ab}|\le\sigma_a\sigma_b\le C$, and
\[
 |\widehat\Gamma_{ab}-\Gamma_{ab}|
 \le\frac{|\widehat V_{ab}-V_{ab}|}{\widehat\sigma_a\widehat\sigma_b}
 +|V_{ab}|\left|\frac1{\widehat\sigma_a\widehat\sigma_b}
                                  -\frac1{\sigma_a\sigma_b}\right|
\]
on the positive variance event. This proves \eqref{eq:test-S387}.

Equations~\eqref{eq:max-linear-rate}-\eqref{eq:max-cck-rate} imply
\begin{equation}\label{eq:test-S393}
                              r_{\Delta,n}\ell_n^2\longrightarrow0.
\end{equation}
For the last two terms of $r_{\Delta,n}$, use $\ell_n^3/n\to0$ and
\[
 \frac{(s_0+\bar s_\delta)^2\ell_n^4}{n^2}
 =\left\{\frac{(s_0+\bar s_\delta)^2\ell_n^3}{n}\right\}
                                       \frac{\ell_n}{n}\longrightarrow0.
\]
Hence $\Delta_n\ell_n^2=o_P(1)$ and $\Delta_n=o_P(1)$.
Since $\log(2d_n^\circ)=O(\ell_n)$ and
$x\log^2(1/x)\to0$ as $x\downarrow0$, \eqref{eq:test-S388} follows.
\end{proof}

%%%%%%%%%%%%%%%%%%%%%%%%%%%%%%%%%%%%%%%%
\begin{proof}[Proof of Proposition~\ref{prop:threshold-size}]
The proof has four steps carried out on the deterministic grid $\mathcal G_n^\circ$ and then transferred to the implemented grid.

\noindent\textit{Step 1. Approximation by the population score maximum.}
Let
\[
 S_{n,a}=\frac1{\sqrt n}\sum_i\psi_{i,a},\qquad
 S_n=\max_{a\in\mathcal I_n^\circ}|S_{n,a}|.
\]
By Lemma~\ref{lem:test-variance}, we have
\begin{align}
 r_n:=\max_{a\in\mathcal I_n^\circ}|Z_{n,a}-S_{n,a}|
              &=o_P(\ell_n^{-1/2}),\label{eq:test-S394}\\
 |T_n^\circ-S_n|&\le r_n=o_P(\ell_n^{-1/2}).\label{eq:test-S395}
\end{align}
The vectors $((\psi_{i,a})_a',(-\psi_{i,a})_a')'$ of dimension $2d_n^\circ$ are independent across $i$, centered, and have average coordinate variances equal to one. Let $G=(G_a:a\in\mathcal I_n^\circ)$
be a centered Gaussian vector with covariance $\Gamma$, and set $G_n=\max_a|G_a|$. The polynomial-moment case of Proposition~2.1 in
\citet{da801549-3c84-30a7-aaa4-ce3f2dd2650d}, with exponent four, gives
\begin{equation}\label{eq:test-S396}
\begin{split}
 \delta_{1,n}:=\sup_t|P_{\alpha_0}(S_n\le t)-P(G_n\le t)|
 &\le C\left\{\left(\frac{B_n^2\ell_n^7}{n}\right)^{1/6}
       +\left(\frac{B_n^2\ell_n^3}{\sqrt n}\right)^{1/3}\right\}
 \longrightarrow0.
\end{split}
\end{equation}
Both terms vanish under \eqref{eq:max-cck-rate}, since $B_n\ge1$. This uses a finite-moment Gaussian approximation, without exponential tail assumptions or a polynomial rate of convergence in \eqref{eq:max-cck-rate}.

Uniformity in \eqref{eq:test-S394} yields a deterministic sequence $\varepsilon_n>0$ such that
\begin{equation}\label{eq:test-S397}
 \varepsilon_n\sqrt{\ell_n}\to0,\qquad
 \sup_{\alpha_0\in\mathcal A_{\ell_0}^{(1)}(s_0)}
                  P_{\alpha_0}(r_n>\varepsilon_n)\to0.
\end{equation}
Sandwiching $T_n^\circ$ between $S_n\pm\varepsilon_n$ gives
\begin{equation}\label{eq:test-S398}
\begin{split}
 \sup_t|P_{\alpha_0}(T_n^\circ\le t)-P(G_n\le t)|
 &\le2\delta_{1,n}+P_{\alpha_0}(r_n>\varepsilon_n)
            +\sup_tP(|G_n-t|\le\varepsilon_n).
\end{split}
\end{equation}
All Gaussian coordinates have variance one. The Gaussian anti-concentration inequality of Theorem~3 in \citet{chernozhukov2015comparison}, applied to the signed coordinates, bounds the last term by $C\varepsilon_n\sqrt{\log(2d_n^\circ)}=o(1)$. Therefore
\begin{equation}\label{eq:test-S399}
 \delta_{2,n}:=\sup_t|P_{\alpha_0}(T_n^\circ\le t)-P(G_n\le t)|\to0
 \quad\text{uniformly over the null class}.
\end{equation}

\noindent\textit{Step 2. Conditional Gaussian approximation of the bootstrap.}
Conditional on the sample, the vector
\[
 \left(\frac1{\sqrt n}\sum_i\iota_i\widehat\psi_{i,a}:
                                      a\in\mathcal I_n^\circ\right)
\]
is Gaussian with covariance $\widehat\Gamma$. On the positive variance event, $\widehat\Gamma_{aa}=1$. Lemma~\ref{lem:test-covariance} gives
\begin{equation}\label{eq:test-S3100}
        \Delta_n\{1\vee\log(2d_n^\circ/\Delta_n)\}^2=o_P(1).
\end{equation}
Gaussian comparison for the signed coordinates in Theorem~2 of \citet{chernozhukov2015comparison} now gives
\begin{equation}\label{eq:test-S3101}
 \delta_{3,n}:=\sup_t|P^*(T_n^{*,\circ}\le t)-P(G_n\le t)|=o_P(1)
 \quad\text{uniformly over the null class}.
\end{equation}
The covariance $\Gamma$ and the comparison distribution may depend on $n$ and the underlying law, but are deterministic for each such law.

\noindent\textit{Step 3. Bootstrap critical value approximation.}
Let $F_G(t)=P(G_n\le t)$ and
$c_G(q)=\inf\{t:F_G(t)\ge q\}$ for $q\in(0,1)$. Choose a deterministic
$\eta_n\downarrow0$ such that
\begin{equation}\label{eq:test-S3102}
 \sup_{\alpha_0\in\mathcal A_{\ell_0}^{(1)}(s_0)}
                  P_{\alpha_0}(\delta_{3,n}>\eta_n)\to0.
\end{equation}
On $E_n=\{\delta_{3,n}\le\eta_n\}$, for large $n$,
\begin{equation}\label{eq:test-S3103}
 c_G(1-\zeta-\eta_n)\le c_{1-\zeta}^{*,\circ}
                       \le c_G(1-\zeta+\eta_n).
\end{equation}
For example, $t<c_G(1-\zeta-\eta_n)$ implies
$P^*(T_n^{*,\circ}\le t)<1-\zeta$ on $E_n$. The upper bound follows by evaluating the bootstrap CDF at $c_G(1-\zeta+\eta_n)$.

\noindent\textit{Step 4. Uniform rejection probability.}
The quantile sandwich gives
\begin{equation}\label{eq:test-S3104}
\begin{split}
 P_{\alpha_0}\{T_n^\circ>c_G(1-\zeta+\eta_n)\}-P_{\alpha_0}(E_n^c)
 &\le P_{\alpha_0}(T_n^\circ>c_{1-\zeta}^{*,\circ})\\
 &\le P_{\alpha_0}\{T_n^\circ>c_G(1-\zeta-\eta_n)\}
                                      +P_{\alpha_0}(E_n^c).
\end{split}
\end{equation}
The Gaussian maximum has a continuous distribution, since each of its finitely many coordinates has nonzero variance. Thus $F_G\{c_G(q)\}=q$.
By \eqref{eq:test-S399}, the upper and lower bounds in \eqref{eq:test-S3104} equal $\zeta+o(1)$ uniformly. Finally,
\[
 \left|P_{\alpha_0}(T_n>c_{1-\zeta}^*)
       -P_{\alpha_0}(T_n^\circ>c_{1-\zeta}^{*,\circ})\right|
                \le P_{\alpha_0}(\mathcal E_n^c)=o(1)
\]
uniformly, which proves \eqref{eq:test-size}. This step also accounts for the randomness of $K_n$ under the observed value implementation.

For the Monte Carlo implementation, conditional empirical CDF consistency gives a uniform error $O_P(B^{-1/2})$, and the correction in \eqref{eq:test-mc-pvalue} is at most $1/(B+1)$. If $B\to\infty$, these errors can be included in the quantile sandwich above. The same asymptotic size conclusion follows.
\end{proof}
%%%%%%%%%%%%%%%%%%%%%%%%%%%%%%%%%%%%

\subsubsection{Proof of Proposition~\ref{prop:threshold-power}}
\begin{proof}[Proof of Proposition~\ref{prop:threshold-power}]
The proof proceeds in two steps. All probability bounds in this proof refer to the stated sequence of alternatives.

\noindent\textit{Step 1. Alternative expansion at the informative coordinate.}
For any threshold $t$, the algebraic identity is
\[
\begin{split}
 \widetilde a(t)-\alpha_0
 &=\widehat{\bm\Theta}(t)\bm X(t)'U/n
   +\widehat{\bm\Theta}(t)\bm X(t)'
                           \{\bm X(\tau_0)-\bm X(t)\}\alpha_0/n\\
 &\quad+\{I_{2p}-\widehat{\bm\Theta}(t)\widehat{\bm\Sigma}(t)\}
                                      \{\widehat\alpha(t)-\alpha_0\}.
\end{split}
\]
In the same partition case \eqref{eq:test-same-partition}, set $t_n=\tau_0$. On that event, all fitted quantities evaluated at $\tau_n$ coincide with
their counterparts at $t_n$, by the stipulated solution rules. In the deterministic grid case, set $t_n=\tau_n$. In both cases, $t_n$ is deterministic for the underlying model sequence. Writing
$r_n=p+j_n$, we obtain, with probability tending to one,
\begin{equation}\label{eq:test-S3105}
 \sqrt n\{\widetilde a_{\delta,j_n}(\tau_n)-\delta_{0,j_n}\}
                      =A_{1,n}+A_{2,n}+A_{3,n},
\end{equation}
where
\[
\begin{split}
 A_{1,n}&=\widehat\Theta_{r_n}(t_n)'\bm X(t_n)'U/\sqrt n,\\
 A_{2,n}&=\widehat\Theta_{r_n}(t_n)'\bm X(t_n)'
                       \{\bm X(\tau_0)-\bm X(t_n)\}\alpha_0/\sqrt n,\\
 A_{3,n}&=\sqrt n\,e_{r_n}'
       \{I_{2p}-\widehat{\bm\Theta}(t_n)\widehat{\bm\Sigma}(t_n)\}
                                      \{\widehat\alpha(t_n)-\alpha_0\}.
\end{split}
\]
For the first term, add and subtract $\Theta_{r_n}(t_n)$. The population
row score is centered and has variance
$\sigma_{\delta,j_n}^2(t_n)\le 2C_U/\kappa_{\mathrm{node}}^2$.
This bound follows from Assumption~\ref{as1}(iv) and the population
nodewise residual-scale bounds in Lemma~\ref{lem:S14-nodewise}:
$\bm\Sigma_{xu}(t)\preceq C_U\bm\Sigma(t)$ and
$\Theta_{p+j,p+j}(t)=A_{jj}(t)+B_{jj}(t)
\le 2/\kappa_{\mathrm{node}}^2$.
Since $t_n$ and $j_n$ are deterministic, Chebyshev's inequality shows that
the score is $O_P(1)$.
The estimated row error is bounded by
\[
 |\widehat\Theta_{r_n}(t_n)-\Theta_{r_n}(t_n)|_1
                      |\bm X(t_n)'U/\sqrt n|_\infty
         =O_P(\bar s_\delta\ell_n/\sqrt n)=o_P(1).
\]
Hence, we derive
\begin{equation}\label{eq:test-S3106}
                                  A_{1,n}=O_P(1).
\end{equation}
This argument does not invoke a fixed threshold CLT at a random argument.

In the same partition case, $A_{2,n}=0$. In the deterministic grid case, write
$R_i(t_n)=X_i'\delta_0\{\bm1\{Q_i<\tau_0\}-\bm1\{Q_i<t_n\}\}$
and $R(t_n)=(R_1(t_n),\ldots,R_n(t_n))'$. Then
\begin{equation}\label{eq:test-S3107}
 |A_{2,n}|\le\sqrt n\,|\widehat\Theta_{r_n}(t_n)|_1
                   |\bm X(t_n)'R(t_n)/n|_\infty.
\end{equation}
At $\eta_n=|t_n-\tau_0|\vee b_n^2$, the stated local design bound gives
\[
 |\bm X(t_n)'R(t_n)/n|_\infty
 \le |\delta_0|_1\max_{j,l}\frac1n\sum_i|X_i^{(j)}X_i^{(l)}|
       |\bm1\{Q_i<t_n\}-\bm1\{Q_i<\tau_0\}|
 =O_P(|\delta_0|_1\eta_n).
\]
Together with $|\widehat\Theta_{r_n}(t_n)|_1=O_P(\sqrt{\bar s_\delta})$
and \eqref{eq:test-grid-rate}, this yields
\begin{equation}\label{eq:test-S3108}
                                  A_{2,n}=o_P(1).
\end{equation}

For $A_{3,n}$, we verify the coefficient bound at the candidate itself. In the same partition case, the model is correctly specified at $t_n=\tau_0$,
so the fixed $t_n$ weighted Lasso gives
$|\widehat\alpha(t_n)-\alpha_0|_1=O_P(s_0b_n)$.
In the deterministic grid case,
$Y=\bm X(t_n)\alpha_0+U+R(t_n)$. Its coefficient basic inequality, with $v_n=\widehat\alpha(t_n)-\alpha_0$, is
\[
 \|\bm X(t_n)v_n\|_n^2
 +\lambda_n|\bm D(t_n)\widehat\alpha(t_n)|_1
 \le\frac2n\{U+R(t_n)\}'\bm X(t_n)v_n
       +\lambda_n|\bm D(t_n)\alpha_0|_1.
\]
The added score satisfies
$|\bm X(t_n)'R(t_n)/n|_\infty=O_P(|\delta_0|_1\eta_n)=o_P(b_n)$
by \eqref{eq:test-grid-rate}. With sufficiently large $A$, the simultaneous score event ensures a fixed margin in the weighted cone inequality. The positive lower bound on the weights and the empirical
restricted eigenvalue event therefore give the same coefficient rate $O_P(s_0b_n)$ at $t_n$. A bound
for the jointly selected $\widehat\tau$ is not applied at $\tau_n$. The nodewise KKT bound now gives
\begin{equation}\label{eq:test-S3109}
 |A_{3,n}|\le\sqrt n\,
 |\widehat\Theta_{r_n}(t_n)'\widehat{\bm\Sigma}(t_n)-e_{r_n}'|_\infty
 |\widehat\alpha(t_n)-\alpha_0|_1
       =O_P(s_0\ell_n/\sqrt n)=o_P(1).
\end{equation}
Combining these bounds,
\begin{equation}\label{eq:test-S3110}
 \sqrt n\,\widetilde a_{\delta,j_n}(\tau_n)
                =\sqrt n\,\delta_{0,j_n}+O_P(1).
\end{equation}
Neither threshold identification nor consistency of
$\widehat\beta^{\mathrm{null}}$ for $\beta_0$ under the alternative is used here. The assumption on the null-residual standard deviation at the informative coordinate is stated separately in the proposition.

\noindent\textit{Step 2. Comparison with the bootstrap critical value.}
For every retained coordinate $a\in\mathcal I_n$,
$W_{n,a}^*=n^{-1/2}\sum_i\iota_i\widehat\psi_{i,a}$ is conditionally $N(0,1)$ by construction. Omitted coordinates can be represented by
zero Gaussian variables. Regardless of their dependence, the union bound gives, for $x>0$,
\begin{equation}\label{eq:test-S3111}
            P^*(T_n^*>x)\le2d_n\exp(-x^2/2),
\end{equation}
which is a conditional statement and therefore also holds for random $d_n=pK_n$. It follows that
\begin{equation}\label{eq:test-S3112}
                   c_{1-\zeta}^*\le\sqrt{2\log(2d_n/\zeta)}.
\end{equation}
Let $R_n=\sqrt n\{\widetilde a_{\delta,j_n}(\tau_n)-\delta_{0,j_n}\}$. By \eqref{eq:test-S3110}, $R_n=O_P(1)$. On the event
$0<\widehat\sigma_{\delta,j_n}(\tau_n)\le\bar\sigma_n$,
\begin{equation}\label{eq:test-S3113}
 |Z_{n,j_n}(\tau_n)|
 \ge\frac{\sqrt n\,|\delta_{0,j_n}|-|R_n|}{\bar\sigma_n}.
\end{equation}
If the numerator on the right is negative, the inequality is automatic. Otherwise, it follows by bounding the positive denominator above.
The signal gap condition, \eqref{eq:test-S3112}, and
$\bar\sigma_n\ge c_\sigma>0$ imply
$|Z_{n,j_n}(\tau_n)|-c_{1-\zeta}^*\xrightarrow{P}+\infty$. Since $T_n\ge|Z_{n,j_n}(\tau_n)|$ whenever that coordinate is retained, the probability of rejection tends to one.
\end{proof}

%%%%%%%%%%%%%%%%%%%%%%%%%%

\section{Uniform Inference for Time Series Data}\label{timeseries}

This section extends the debiased threshold Lasso estimator construction to strictly stationary, weakly dependent time series under exact sparsity, with inference on contrasts. We impose joint strict stationarity within each row of a triangular array and
finite-moment near-epoch dependence (NED) on an underlying absolutely regular process. The observed process need not itself be mixing. These conditions accommodate non-sub-Gaussian covariates and conditionally heteroskedastic, serially correlated errors.

We keep the coefficient parameterization, empirical root mean square weights, and deterministic reference point of the main paper. Temporal dependence requires modified concentration arguments and estimation of the long-run covariance matrix. Moreover, a central limit theorem at a deterministic threshold does not alone justify inference at the estimated threshold. We therefore establish threshold-indexed concentration bounds and a random-index approximation that controls the effect of evaluating the statistic at the estimated threshold. %These arguments complement the linear-model NED results of \citet{ADAMEK20231114}. 
Finally, we illustrate the method using government spending and GDP data in Section~\ref{empiricaltwo}.

\subsection{Oracle Inequalities and Uniform Inference} \label{ts-theory}

Consider
\begin{equation}\label{ts-model}
 Y_i=\bm X_i(\tau_0)'\alpha_0+U_i,
 \qquad \bm X_i(\tau)=(X_i',X_i'1\{Q_i<\tau\})',
 \qquad i=1,\ldots,n.
\end{equation}
Here $\alpha_0=(\beta_0',\delta_0')'$, $\mathbb T=[t_0,t_1]\subset(0,1)$,
and $\tau_\ast\in(t_0,t_1)$ is deterministic. Let
\[
 \mathbb T_n=(\{Q_1,\ldots,Q_n\}\cap\mathbb T)
                 \cup\{\tau_\ast,t_1\}.
\]
Use the same finite grid estimator as in the main paper: 
\begin{equation}\label{ts-estimator}
 (\widehat\alpha,\widehat\tau)=
 \arg\min_{\alpha\in\mathcal B,\,\tau\in\mathbb T_n}
 \left\{\|Y-\bm X(\tau)\alpha\|_n^2
       +\lambda_n|D(\tau)\alpha|_1
       +\rho_n|\tau-\tau_\ast|\right\},
 \qquad \mathcal B=[-C_1,C_1]^{2p}.
\end{equation}
Every diagonal entry of $D(\tau)$ is the empirical root mean square of the corresponding design column. Select the smallest threshold if several thresholds minimize the profile criterion. For a scalar random variable, write
$\|V\|_{L^q}=(E|V|^q)^{1/q}$. The following conditions concern the entire time sequence, including $Q_i$.

\begin{assm}[NED and finite moments]\label{ts-as5}
Within each row of a triangular array, $(X_i',Q_i,U_i,s_{n,i}')'$,
$i\in\mathbb Z$, is jointly strictly stationary. 
(i) For fixed $d_{\rm mix}>5$ and $d_{\rm NED}>11/2$,
the underlying process satisfies
$\beta_s(k)\le C(1+k)^{-d_{\rm mix}}$.
Let $\mathcal F^s_{i,m}=\sigma(s_{n,i-m},\ldots,s_{n,i+m})$.
For every primitive coordinate
$V_i\in \{Q_i,U_i,X_i^{(1)},\ldots,X_i^{(p)}\}$,
\begin{equation}\label{ts-ned}
 \sup_{i,V}\|V_i-E(V_i\mid\mathcal F^s_{i,m})\|_{L^8}
       \le C(1+m)^{-d_{\rm NED}}.
\end{equation}
The mixing condition is on $s_{n,i}$, not on $(X_i',Q_i,U_i)'$.

\noindent (ii) $Q_i\sim\operatorname{Unif}[0,1]$ and
$E(U_i\mid X_i,Q_i)=0$. 

\noindent (iii)  The unconditional moments satisfy
\[
 \sup_{i,j}E|X_i^{(j)}|^{12}\le C_{X,12},\qquad
 \sup_iE|U_i|^8\le C_{U,8}.
\]
For a fixed neighborhood $\mathbb T^\epsilon\subset(0,1)$ of $\mathbb T$,
\begin{equation}\label{ts-local-moments}
 \underset{q\in\mathbb T^\epsilon}{\operatorname{sup}}\;
 \max_{j,l}
 E\!\left[(X_i^{(j)}X_i^{(l)})^2+(X_i^{(j)}U_i)^2
             \mid Q_i=q\right]\le C_{\rm loc,ts}.
\end{equation}
\end{assm}

\noindent The local condition \eqref{ts-local-moments} is additional regularity,
which does not follow from unconditional moment bounds alone. Independence between $X_i$ and $Q_i$ is not required.
Continuous marginal distributions under serial dependence do not exclude
ties across dates, so the arguments keep observations with the same $Q_i$
together. Population nodewise residuals satisfy linear projection
orthogonality and they are not assumed conditionally mean zero given $(X_i,Q_i)$.

Write $\ell_n=\log(2pn)$. For a fixed $\nu_p\ge1$, suppose $p\le Cn^{\nu_p}$.
Choose fixed exponents $\omega_{\rm ts},a_{\rm ts}$ satisfying
\begin{equation}\label{ts-exponents}
 \begin{gathered}
 p\le Cn^{\nu_p},\qquad \nu_p\ge1,\\
 a_{\rm ts}>\frac{1+\nu_p}{12},\qquad
 \omega_{\rm ts}>\frac1{1+d_{\rm mix}},\qquad
 \omega_{\rm ts}+2a_{\rm ts}<\frac12,\\
 \omega_{\rm ts}d_{\rm NED}>
       1-\omega_{\rm ts}-a_{\rm ts}+\frac{1+\nu_p}{8}.
 \end{gathered}
\end{equation}
For $\nu_p=1$, such exponents exist for every fixed
$d_{\rm NED}>11/2$ and $d_{\rm mix}>5$: take $\omega_{\rm ts}=1/6-\varepsilon$ and
$a_{\rm ts}=1/6+\varepsilon/4$ for a sufficiently small fixed $\varepsilon>0$. For larger $\nu_p$, the restrictions remain part of the assumptions.

Define
\begin{align}
 m_n&=\lceil n^{\omega_{\rm ts}}\rceil,\qquad
 T_n^{\rm tr}=n^{a_{\rm ts}},\qquad
 b_{n,\rm ts}=T_n^{\rm tr}\sqrt{m_n\ell_n/n},\label{ts-b-rate}\\
 \lambda_n&=A b_{n,\rm ts},\qquad
 \rho_n=c_\rho\lambda_n,\qquad
 \lambda_{\rm node}=A_{\rm node}b_{n,\rm ts},\label{ts-tuning}\\
 c_{4,n}&=m_n(T_n^{\rm tr})^4\ell_n/n.\label{ts-fourth-rate}
\end{align}
Here $A,A_{\rm node}>0$ are sufficiently large and $c_\rho\in(0,\infty)$
is fixed. These rates apply only to this section and are sufficient rather
than optimal. The main paper's $b_n=\sqrt{\ell_n/n}$ is not redefined.
The window radius $m_n$ and truncation level $T_n^{\rm tr}$ are used only
in the proof while estimation uses the original observations. The separate
rate $c_{4,n}=o(1)$ controls fourth order products in variance estimation.

As in the main paper,
\begin{equation}\label{ts-precision}
 \bm\Sigma(\tau)=E[\bm X_i(\tau)\bm X_i(\tau)'],\qquad
 \bm\Theta(\tau)=\bm\Sigma(\tau)^{-1}
 =\begin{pmatrix}B(\tau)&-B(\tau)\\
                 -B(\tau)&A(\tau)+B(\tau)\end{pmatrix}.
\end{equation}
For a deterministic contrast $g\in\mathbb R^{2p}$, let
$|g|_2=1$, $H=\{j:g_j\ne0\}$, and $h=|H|\le C_H$.
Write $H_A=\{j\le p:j+p\in H\}$,
$H_B=\{j\le p:j\in H\text{ or }j+p\in H\}$, and
$\bar s=\sup_{\tau\in\mathbb T}\max_{j\in H}|\Theta_j(\tau)|_0$.

For the proofs, define the deterministic reference $\tau^\circ=\tau_\ast$ when $\delta_0=0$ and $\tau^\circ=\tau_0$
otherwise. For a population matrix function $G$ and a nonempty threshold subset $\mathbb S\subseteq\mathbb T$, use the main paper's adaptive
restricted eigenvalue:
\[
 \kappa(s,c,\mathbb S,G)
 =\inf_{\tau\in\mathbb S}\min_{1\le|J|\le s}
   \inf_{\substack{v\ne0\\|v_{J^c}|_1\le c\sqrt{s}|v_J|_2}}
       \frac{\{v'G(\tau)v\}^{1/2}}{|v_J|_2}.
\]
The dimensions of $v$ and $J$ match those of $G$.

\begin{assm}\label{ts-as6}
(i) $\alpha_0\in\mathcal B$, $|\alpha_0|_0\le s_0$, $s_0\ge1$,
and $s_0b_{n,\rm ts}\to0$. The matrices $M(\tau)$ and $N(\tau)$
are nonsingular for every $\tau\in\mathbb T$, and there are constants
$0<w_-<w_+<\infty$ such that all their diagonal entries lie in
$[w_-^2,w_+^2]$. Fix the same $\mu\in(0,1)$ as in the oracle and nodewise bounds of the
main paper. Choose $c_0>3\sqrt2\,w_+/\{(1-\mu)w_-\}$ and assume
\[
 \kappa(s_0,c_0,\mathbb S,\bm\Sigma)\ge\underline\kappa>0.
\]
Here $\mathbb S=\mathbb T$ when $\delta_0=0$; when $\delta_0\ne0$,
$\mathbb S=\mathbb T\cap[\tau_0-\epsilon_{\rm RE},\tau_0+\epsilon_{\rm RE}]$
for a fixed common $\epsilon_{\rm RE}>0$.

\noindent (ii)  For inference, $\bar s\ge1$, $\bar s b_{n,\rm ts}\to0$, and, for a
fixed $c_{\rm node}>\max\{1,(1+\mu)w_+/[(1-\mu)w_-]\}$,
\[
 \kappa(\bar s,c_{\rm node},\mathbb T,M)
 \wedge\kappa(\bar s,c_{\rm node},\mathbb T,N)\ge\kappa_{\rm node}>0.
\]
The $\ell_1$ norms of the required rows of $A$ and $B$
are bounded by $C_\Theta$, uniformly in $\tau$.
\end{assm}

These restricted-eigenvalue conditions allow unrestricted minimum
eigenvalues to approach zero. Nonsingularity is assumed because the population inverse matrices are used in the inferential target. The global row norm bound in part (ii) supports uniform nodewise estimation and is assumed as an explicit design restriction, which does not follow from marginal moments alone. Assumption~\ref{ts-as5}(iii), the nodewise restricted eigenvalue,
and the required row norm and sparsity bounds imply
\begin{equation}\label{ts-inverse-modulus}
 \max_{j\in H}|\Theta_j(\tau)-\Theta_j(t)|_1
       \le C\bar s|\tau-t|,\qquad \tau,t\in\mathbb T.
\end{equation}
The proof is given at the beginning of Section~\ref{ts-proofs}. The factor $\bar s$ improves on the $\bar s^{3/2}$ bound in the cross-sectional comparison because part (ii) already bounds the required
regime inverse rows in $\ell_1$ by a common constant. The stated dimension and sparsity conditions yield the corresponding empirical restricted eigenvalue bounds.

%%%%%%%%%%%%%%%%%%%%%%%%%%%%%%%%%%%%%%%%
The following assumption is required only when $\delta_0\ne0$. 

\begin{assm}[Identification in the fixed-threshold class]\label{ts-as7}
There are common
constants $0<d_-\le d_+<\infty$ with
$d_-\le|\delta_0|_2\le d_+$ and $\tau_0\in\mathbb T$.
Let $D_{\delta,n}\ge1$ be a deterministic bound on $|\delta_0|_1$ over this class. For example, one may take $D_{\delta,n}=\max\{1,d_+\sqrt{s_0}\}$. Suppose that $s_0D_{\delta,n}\lambda_n\to0$.
For every sufficiently large fixed $L$, there are common constants
$c_I>0$ and $C_I<\infty$, possibly depending on $L$, such that, with
probability tending to one uniformly over the class,
\begin{equation}\label{ts-margin}
 \|\bm X(\tau)\alpha-\bm X(\tau_0)\alpha_0\|_n^2
 \ge c_I|\tau-\tau_0|
       -C_I\lambda_n|\alpha-\alpha_0|_1-c_Ib_{n,\rm ts}^2
\end{equation}
for all $\tau\in\mathbb T$ and $\alpha\in\mathcal B$ with
$|\alpha-\alpha_0|_1\le Ls_0$.
\end{assm}

Equation~\eqref{ts-margin} is an explicit empirical identification condition
on the coefficient localization region, analogous to the fixed-threshold
condition in the main analysis. 
We still keep the coefficient class notation of the main paper:
\[
 \mathcal A_{\ell_0}^{(1)}(s_0)
 =\{\alpha_0\in\mathbb R^{2p}:|\alpha_0|_\infty\le C_{\ell_0},
       |\alpha_0|_0\le s_0,\ \delta_0=0\},
\]
\[
 \mathcal A_{\ell_0}^{(2)}(s_0)
 =\{\alpha_0\in\mathbb R^{2p}:|\alpha_0|_\infty\le C_{\ell_0},
       |\alpha_0|_0\le s_0,\ \delta_0\ne0\},\qquad
 \mathcal B_{\ell_0}(s_0)=\mathcal A_{\ell_0}^{(1)}(s_0)
                         \cup\mathcal A_{\ell_0}^{(2)}(s_0),
\]
where $0<C_{\ell_0}\le C_1$. As in the main paper, $P_{\alpha_0}$ denotes
any qualifying data generating distribution with coefficient vector
$\alpha_0$.

\begin{thm}\label{ts-thm1}
Suppose that Assumptions~\ref{ts-as5} and \ref{ts-as6}(i),
\eqref{ts-exponents}, and
\eqref{ts-tuning} hold. Then the no-threshold class satisfies
\begin{equation}\label{ts-null-rates}
 \|\widehat f-f_0\|_n=O_P(\sqrt{s_0}\lambda_n),\quad
 |\widehat\alpha-\alpha_0|_1=O_P(s_0\lambda_n),\quad
 |\widehat\tau-\tau_\ast|
       =O_P(s_0\lambda_n^2/\rho_n)=O_P(s_0\lambda_n).
\end{equation}
Along with Assumption~\ref{ts-as7}, the fixed-threshold class satisfies
\begin{equation}\label{ts-fixed-rates}
 \|\widehat f-f_0\|_n=O_P(\sqrt{s_0}\lambda_n),\quad
 |\widehat\alpha-\alpha_0|_1=O_P(s_0\lambda_n),\quad
 |\widehat\tau-\tau_0|=O_P(s_0\lambda_n^2).
\end{equation}
\end{thm}

 %Without loss of generality, we assume that $X_i$ and $Q_i$ are independent in this section. The Gram matrix then simplifies to \[\bm{\Sigma}(\tau)= {\begin{bmatrix} \begin{array}{cccc}
	%\bm {M} &\bm {M}(\tau)\\
	%\bm {M}(\tau)&\bm {M}(\tau) \end{array} \end{bmatrix}}= {\begin{bmatrix} \begin{array}{cccc}
	%\bm {M} &\tau\bm {M}\\
	%\tau\bm {M}&\tau\bm {M} \end{array} \end{bmatrix}}.\]
%and the corresponding inverse of the Gram matrix is \[\bm{\Theta} (\tau) = \bm{\Sigma}(\tau)^{-1}= {\begin{bmatrix} \begin{array}{cccc}
	%\frac{1}{1-\tau}\bm {M}^{-1} &\frac{-1}{1-\tau}\bm {M}^{-1}\\
	%\frac{-1}{1-\tau}\bm {M}^{-1}&\frac{1}{\tau(1-\tau)}\bm {M}^{-1} \end{array} \end{bmatrix}}.\]

Following Section \ref{constructprecison} of the main paper, let $X_<(\tau)$ and
$X_\ge(\tau)$ denote the $n\times p$ matrices with $i$th rows
$X_i'1\{Q_i<\tau\}$ and $X_i'1\{Q_i\ge\tau\}$, respectively.
Thus $X_<(\tau)=X(\tau)$ and $X_\ge(\tau)=X-X(\tau)$ in the
notation of Section \ref{sec:est} of the main paper while $\bm X(\tau)=[X,X(\tau)]$
continues to denote the augmented $n\times2p$ design matrix.
Construct $\widehat{\bm\Theta}(\tau)$ using precisely the two weighted
regime nodewise regressions in Section \ref{constructprecison} of the main paper, with $\lambda_{\rm node}$ from~\eqref{ts-tuning}. Explicitly,
\[
 \widehat\gamma_j(\tau)\in\arg\min_{\gamma\in\mathbb R^{p-1}}
 \left\{\|X_<^{(j)}(\tau)-X_<^{(-j)}(\tau)\gamma\|_n^2
       +\lambda_{\rm node}|\widehat{\Gamma}_j(\tau)\gamma|_1\right\},
\]
where $\widehat{\Gamma}_j(\tau)$ is diagonal with entries
$\|X_<^{(l)}(\tau)\|_n$, $l\ne j$, and the upper-regime construction is analogous. In particular, the lower-regime
residual scale is
\[
 \widehat z_j^2(\tau)
 =\|X_<^{(j)}(\tau)-X_<^{(-j)}(\tau)\widehat\gamma_j(\tau)\|_n^2
  +\frac{\lambda_{\rm node}}2
       |\widehat{\Gamma}_j(\tau)\widehat\gamma_j(\tau)|_1,
\]
and the upper-regime construction is analogous. In the notation of Section \ref{supp:nodewise-construction}, these scales yield $\widehat A(\tau)=\widehat Z(\tau)^{-2}\widehat C(\tau)$ and
$\widehat B(\tau)=\widehat{\widetilde Z}(\tau)^{-2}
                  \widehat{\widetilde C}(\tau)$.
We then construct $\widehat{\bm\Theta}(\tau)$ using the block formula in
Section \ref{constructprecison}, and set
\begin{equation}\label{ts-debias}
 \widetilde a(\widehat\tau)
 =\widehat\alpha+
 \widehat{\bm\Theta}(\widehat\tau)
 \bm X(\widehat\tau)'\{Y-\bm X(\widehat\tau)\widehat\alpha\}/n.
\end{equation}

Define the structural score $\zeta_i(\tau)=\bm X_i(\tau)U_i$.
To distinguish serial dependence from the contemporaneous covariance
$\bm\Sigma_{xu}(\tau)$ in Section \ref{uniforminference}, write
\begin{align}
 \Gamma_{xu}(k,\tau)&=E[\zeta_i(\tau)\zeta_{i-k}(\tau)'],\notag\\
 \bm\Sigma_{xu}^{\rm LR}(\tau)
   &=\sum_{k\in\mathbb Z}\Gamma_{xu}(k,\tau),\qquad
 \mathcal V^{\rm LR}(\tau)
   =\bm\Theta(\tau)\bm\Sigma_{xu}^{\rm LR}(\tau)\bm\Theta(\tau)'.
 \label{ts-lrv}
\end{align}
The series converges entrywise absolutely under Assumption~\ref{ts-as5}.
At lag zero, $\Gamma_{xu}(0,\tau)=\bm\Sigma_{xu}(\tau)$ under stationarity.
If all nonzero-lag score covariances vanish, these long-run quantities reduce
to $\bm\Sigma_{xu}(\tau)$ and $\mathcal V(\tau)$ in the main paper.
Opposite regimes are disjoint at the same date, but their scores can have
nonzero cross-covariances at different dates, so these terms are kept.
Define $\sigma_g^2(t)=g'\mathcal V^{\rm LR}(t)g,$ and $ \sigma_g^2=\sigma_g^2(\tau^\circ).$

\begin{assm}\label{ts-as8}
Let $g\in\mathbb R^{2p}$ be a deterministic contrast vector
satisfying $|g|_2=1$, and let
$H=\{j:g_j\ne0\}$. Suppose that $|H|\le C_H$ for a fixed
constant $C_H<\infty$ and that
\[
 \inf_{\alpha_0\in\mathcal B_{\ell_0}(s_0)}
 \sigma_g^2 \ge c_V
\]
for a fixed constant $c_V>0$.

Let $D_{\delta,n}$ denote the deterministic bound on
$|\delta_0|_1$ specified in Assumption~\ref{ts-as7}.
The following rate conditions hold:
\begin{align}
 \sqrt n\{s_0+\bar s+s_0D_{\delta,n}\}
 b_{n,\mathrm{ts}}^2
 &\longrightarrow0,
 &
 s_0b_{n,\mathrm{ts}}m_n\ell_n
 &\longrightarrow0.
 \label{ts-clt-rates}
\end{align}
For feasible inference, let $k_n$ be a deterministic sequence
of positive integers satisfying
\begin{equation}\label{ts-hac-rates}
\begin{gathered}
 k_n\longrightarrow\infty,
 \qquad
 \frac{k_n^3}{n}\longrightarrow0,\\
 k_n\left\{
 (s_0+\bar s+s_0\bar s)b_{n,\mathrm{ts}}
 +(1+D_{\delta,n})
 \sqrt{s_0b_{n,\mathrm{ts}}+c_{4,n}}
 \right\}\longrightarrow0.
\end{gathered}
\end{equation}
\end{assm}

Under the preceding assumptions, the rate conditions in \eqref{ts-clt-rates} ensure that replacing the estimated
threshold $\widehat\tau$ by the deterministic reference
$\tau^\circ$ in the leading score contributes an $o_P(1)$
error. The argument combines the inverse row continuity bound \eqref{ts-inverse-modulus} with the uniform score bound \eqref{ts-global-score-sharp} and the local indicator score bound. For feasible inference, the bandwidth condition
\eqref{ts-hac-rates} controls the effect of estimating the influence scores on the long run variance estimator. In particular, the term $k_n s_0\bar s b_{n,\mathrm{ts}}$ accounts for evaluating the population inverse rows at
$\widehat\tau$ rather than $\tau^\circ$. These conditions provide a common sufficient framework for inference in
both the no-threshold and identified fixed-threshold cases.

\begin{thm}\label{ts-thm2}
Suppose that Assumptions~\ref{ts-as5}--\ref{ts-as8}, and \eqref{ts-exponents} hold. Then,
\begin{equation}\label{ts-linear}
 \sqrt n\,g'\{\widetilde a(\widehat\tau)-\alpha_0\}
 =\frac1{\sqrt n}\sum_{i=1}^n
       g'\bm\Theta(\tau^\circ)\bm X_i(\tau^\circ)U_i+o_P(1),
\end{equation}
uniformly over $\mathcal B_{\ell_0}(s_0)$, with Assumption~\ref{ts-as7}
applying only to its fixed-threshold component. Consequently,
\[
 \sup_{\alpha_0\in\mathcal B_{\ell_0}(s_0)}\sup_{t\in\mathbb R}
 \left|P_{\alpha_0}\left(\frac{\sqrt n\,g'\{\widetilde a(\widehat\tau)-\alpha_0\}}
                         {\sqrt{\sigma_g^2}}\le t\right)-\Phi(t)\right|\to0.
\]
\end{thm}

For each $\tau\in\mathbb T$, we keep the residual notation of the main
paper, $\widehat U_i(\tau)=Y_i-\bm X_i(\tau)'\widehat\alpha(\tau)$, and set
$\widehat\zeta_i(\tau)=\bm X_i(\tau)\widehat U_i(\tau)$ and
$\overline\zeta(\tau)=n^{-1}\sum_i\widehat\zeta_i(\tau)$.
For $0\le k\le k_n$, define
\begin{align}
 \widehat\Gamma_{xu}(k,\tau)
 &=\frac1n\sum_{i=k+1}^n
 [\widehat\zeta_i(\tau)-\overline\zeta(\tau)]
 [\widehat\zeta_{i-k}(\tau)-\overline\zeta(\tau)]',\notag\\
 \widehat{\bm\Sigma}_{xu}^{\rm LR}(\tau)
 &=\widehat\Gamma_{xu}(0,\tau)+
 \sum_{k=1}^{k_n}\left(1-\frac{k}{k_n+1}\right)
 [\widehat\Gamma_{xu}(k,\tau)+\widehat\Gamma_{xu}(k,\tau)'],\label{ts-hac}\\
 \widehat{\mathcal V}^{\rm LR}_n(\tau)
 &=\widehat{\bm\Theta}(\tau)\widehat{\bm\Sigma}_{xu}^{\rm LR}(\tau)
          \widehat{\bm\Theta}(\tau)'.\notag
\end{align}
The common denominator $n$
and Bartlett weights give a positive semidefinite HAC estimator when using these quantities at $\tau=\widehat\tau$.  The superscript $\mathrm{LR}$ distinguishes this estimator from
$\widehat{\mathcal V}_n(\tau)$ in the cross-sectional data analysis. Only the
required rows of $\widehat{\bm\Theta}$ are needed for a fixed collection
of contrasts. All lags refer to the original time order.

\begin{thm}\label{ts-thm3}
Suppose Assumptions~\ref{ts-as5}--\ref{ts-as8} hold, along with
\eqref{ts-exponents}. Then, for every $\epsilon>0$,
\begin{equation}\label{ts-var-consistency}
 \sup_{\alpha_0\in\mathcal B_{\ell_0}(s_0)}
 P_{\alpha_0}\left(\left|g'\widehat{\mathcal V}^{\rm LR}_n(\widehat\tau)g-\sigma_g^2\right|>
          \epsilon\right)\longrightarrow0.
\end{equation}
\end{thm}

\begin{thm}\label{ts-thm4}
Under the conditions of Theorem~\ref{ts-thm3},
\begin{equation}\label{ts-feasible}
 \sup_{\alpha_0\in\mathcal B_{\ell_0}(s_0)}\sup_{t\in\mathbb R}
 \left|P_{\alpha_0}\left(
 \frac{\sqrt n\,g'\{\widetilde a(\widehat\tau)-\alpha_0\}}
      {\{g'\widehat{\mathcal V}^{\rm LR}_n(\widehat\tau)g\}^{1/2}}\le t
                  \right)-\Phi(t)\right|\longrightarrow0.
\end{equation}
For any fixed $\alpha\in(0,1)$, the interval
\[
 g'\widetilde a(\widehat\tau)\ \pm\
 z_{1-\alpha/2}\{g'\widehat{\mathcal V}^{\rm LR}_n(\widehat\tau)g/n\}^{1/2}
\]
has coverage tending uniformly to $1-\alpha$ and length $O_P(n^{-1/2})$.
\end{thm}

For a deterministic coordinate $j$, the corresponding main paper notation is
\[
 \widehat\sigma_j^{\rm LR}(\tau)
  =\{e_j'\widehat{\mathcal V}^{\rm LR}_n(\tau)e_j\}^{1/2},\qquad
 I_{j,n}^{\rm LR}(1-\alpha)
  =\left[\widetilde a^{(j)}(\widehat\tau)
         \pm z_{1-\alpha/2}\widehat\sigma_j^{\rm LR}(\widehat\tau)/\sqrt n\right].
\]
Theorem~\ref{ts-thm4} with $g=e_j$ gives its uniform coverage and length
claims over the qualifying coefficient class.

Thus the same full sample anchored threshold Lasso estimator and HAC-based statistic permit
inference over no-threshold classes and identified fixed-threshold classes
without pretesting. This result is about slope contrasts. It does not extend
the independent multiplier threshold effect test of Section \ref{sec:threshold_test} to serially
dependent observations because such a test requires a separate dependent bootstrap or other simultaneous approximation argument. We leave this extension to future work.

\subsection{Local Projection Inference}\label{ts-lp}
Let $r$ range over the fixed finite set $ R=\{0,\ldots,r_{\max}\}$ of projection horizons while $h=|H|$ continues to denote contrast support size. For each horizon consider
\begin{equation}\label{ts-lp-model}
 Y_{i+r}=X_i'\beta_{0,r}
       +X_i'\delta_{0,r}1\{Q_i<\tau_{0,r}\}+U_{r,i}.
\end{equation}
A possible regressor vector is
$X_i=(x_i,w_{s,i}',z_{i-1}',\ldots,z_{i-K}')'$, where $x_i$ is the shock,
$w_{s,i}$ contains appropriately timed controls, and $z_{i-k}$ contains
lagged outcomes, shocks, and additional controls.
Assume $E(U_{r,i}\mid X_i,Q_i)=0$ at each horizon. %A structural impulse-response interpretation additionally requires an appropriate economic identification argument for $x_i$.

For each $r$, apply~\eqref{ts-estimator}, \eqref{ts-debias}, and
\eqref{ts-hac} using $Y_{i+r}$ and the corresponding usable sample size
$n_r$. Denote the results by $\widehat\tau_r$, $\widetilde a_r(\widehat\tau_r)$,
and $\widehat{\mathcal V}^{\rm LR}_{n_r,r}(\widehat\tau_r)$.
The penalty weights remain $D_r(\tau)$ as defined by the RMS design columns. Using the same weighted coefficient penalty as in the
main paper allows direct application of the preceding results.
\citet{hdlp} analyze a partially penalized linear LP
procedure. However, we require an additional argument to transfer that modification to an estimated threshold
objective.

Let $j_x$ denote the coordinate of $x_i$ in $X_i$, and for the ordering
displayed above, $j_x=1$. The upper  and lower regime shock responses are, respectively,
\[
 \varphi_r=e_{j_x}'\alpha_{0,r},\qquad
 \varphi_r+\delta_{r,x,0}=(e_{j_x}+e_{p+j_x})'\alpha_{0,r}.
\]
Their contrast vectors can be divided by their Euclidean norms when
applying the theorems because the interval formula is unchanged by this
normalization. The lower minus upper difference uses $e_{p+j_x}$, and
the upper minus lower difference uses $-e_{p+j_x}$. When $Q_i$ is
unemployment as in the application, these are low minus high and high minus low, respectively.
Both use the full estimated covariance matrix.

\begin{cor}\label{ts-thm5}
Suppose that, for each $r\in\mathcal R$, the process
$(X_i',Q_i,U_{r,i})'$ and the parameters in~\eqref{ts-lp-model} satisfy
Assumptions~\ref{ts-as5}--\ref{ts-as8} and \eqref{ts-exponents},
uniformly over $r$, with $n_r\to\infty$.
Then Theorem~\ref{ts-thm4} applies at each horizon to the two regime
responses and their difference. In particular, for a prespecified nonzero
contrast $c_r$ with bounded support,
\[
 \max_{r\in\mathcal R}\sup_{\alpha_0\in\mathcal B_{\ell_0}(s_0)}
 \sup_{t\in\mathbb R}
 \left|P_{\alpha_0}\left(
 \frac{\sqrt{n_r}\,c_r'\{\widetilde a_r(\widehat\tau_r)-\alpha_{0,r}\}}
      {\{c_r'\widehat{\mathcal V}^{\rm LR}_{n_r,r}(\widehat\tau_r)c_r\}^{1/2}}
             \le t\right)-\Phi(t)\right|\to0.
\]
For a fixed finite collection, Bonferroni intervals give asymptotic simultaneous
coverage at least the nominal level.
\end{cor}

Corollary ~\ref{ts-thm5} extends the uniform inference framework
to local projections over a fixed finite set of horizons. Under the stated assumptions, the studentized debiased estimators of the regime-specific shock responses and their
differences are asymptotically standard normal. The long-run
covariance estimator accounts for heteroskedasticity and
serial dependence, including dependence induced by overlapping
projection errors. These results provide asymptotically valid
confidence intervals for the responses at each horizon.
For a fixed collection of responses across horizons, Bonferroni-adjusted intervals provide simultaneous coverage without requiring independence between the horizon-specific estimators.

\subsection{Impulse Responses to Military Spending News}
\label{empiricaltwo}

We illustrate the proposed method using the US government spending application of \citet{ramey2018government}. We study how the responses of government spending and GDP to military spending news vary with labor-market conditions.  We estimate the unemployment threshold from the contemporaneous GDP equation and use the resulting state classification for both outcomes at all projection horizons, as in \citet{LLSS2020}.

The quarterly source data span 1889Q1--2015Q4. Real GDP and real government spending are divided by estimated potential real GDP. Military spending news is normalized by the product of lagged potential real GDP and the lagged GDP deflator. Motivated by \citet{BP2018} and \citet{Goemans2022}, we include additional controls for consumption, investment, interest rates, inflation, debt, financial conditions, and taxation. %
\footnote{The 14 series comprise military spending news, GDP,
government spending, and transformations of consumption, investment, the nominal interest rate, inflation, the average marginal personal income tax rate, private and public debt, stock prices, the credit spread, the real interest rate, and the federal revenue-to-GDP ratio.
The replication files contain the sources and transformations, including historical proxies and annual-to-quarterly interpolation for portions of the early sample.}
Including four lags of 14 series gives 56 lagged controls.
Allowing the intercept and all slopes to differ across states yields 116 coefficients, including two intercepts.
Regularization accommodates this large, potentially correlated control set and the unequal numbers of observations across states.

For outcome $a\in\{Y,G\}$, corresponding to GDP and government spending, respectively, we estimate
\[
Y_{i+r}^{(a)}
=
\mu_{a,r}+X_i'\beta_{a,r}
+\mathbf{1}\{Q_i<\tau\}
 \bigl(\gamma_{a,r}+X_i'\delta_{a,r}\bigr)
+U_{a,r,i},
\qquad r=0,\ldots,20,
\]
where $X_i=(x_i,W_i')'$, $x_i$ denotes normalized military spending news, and $W_i$ contains the 56 lagged controls. The threshold variable $Q_i$ is the unemployment rate in quarter $i-1$, expressed as a decimal.

The high-unemployment state is defined by $Q_i\geq\tau$, and the low-unemployment state by $Q_i<\tau$. Their corresponding responses to military spending news are
\[
\phi_{a,r}^{H}=\beta_{a,r}^{(1)},
\qquad
\phi_{a,r}^{L}
=\beta_{a,r}^{(1)}+\delta_{a,r}^{(1)},
\]
where the superscript $(1)$ denotes the coefficient on $x_i$.
State membership is determined by unemployment before
the shock, and subsequent transitions between states are
permitted over the projection horizon.

Threshold estimation uses 500 observations spanning 1891Q1 to
2015Q4. The local projections use a common sample of 480 shock origins, from 1891Q1 to 2010Q4, at every horizon, with outcomes observed through 2015Q4. We estimate the threshold using the anchored threshold Lasso, restricting the candidate grid to the empirical 10th--90th percentile range of lagged unemployment. With $\tau^*=0.065$ and $c_\rho=0.75$, the plug-in procedure in Section~\ref{lambdarulets} selects
$\widehat\lambda=0.026$ on the standardized
half-squared-loss scale and an unemployment threshold of $\widehat\tau=0.093$. For the subsequent local projections, we parameterize the coefficients by unemployment state and penalize the state-specific lagged-control coefficients separately. We leave the two news coefficients and the state intercept difference unpenalized, as in \citet{hdlp}, and remove the common intercept by centering. We keep the debiasing correction because regularization of the controls can affect the unpenalized news coefficients.

We test the joint null of no state differences in the intercept or slopes of the contemporaneous GDP equation using a blockwise Gaussian multiplier calibration based on \citet{zhangcheng2014}. With six-quarter blocks and 9,999 draws, the test statistic is 3.8715, compared with 95\% and 99\% critical values of 3.8663 and 4.2789, respectively. The bootstrap $p$-value is 0.0489, implying rejection at the nominal 5\% but not at 1\%.

Figure~\ref{fig:fiscal-irfs} compares the responses under the
fixed 6.5\% and estimated 9.273\% unemployment thresholds, holding the controls, LP estimator, and sample constant. Under the estimated threshold, high-state government spending
peaks at 1.90 in quarter 16, while GDP peaks at 1.62 in quarter 15. The corresponding peaks under 6.5\% are 0.85 and 0.64. Low-state responses under the estimated threshold are smaller, peaking at 0.39 for spending and 0.26 for GDP, and approach zero by quarter 20. At that horizon, the high-state GDP remains at 1.02, with a 95\% interval of $[0.22,1.81]$.
The responses therefore increase gradually following a military spending news shock and are more persistent in the high-unemployment state.

Figure~\ref{fig:fiscal-differences} reports direct high-minus-low response contrasts, accounting for the covariance between the two state estimates. Under the estimated threshold, the 95\% pointwise intervals exclude zero at $r=4,\ldots,20$ for GDP and $r=7,\ldots,19$ for government spending.

\begin{figure}[!tbp]
\centering
\includegraphics[width=\linewidth]{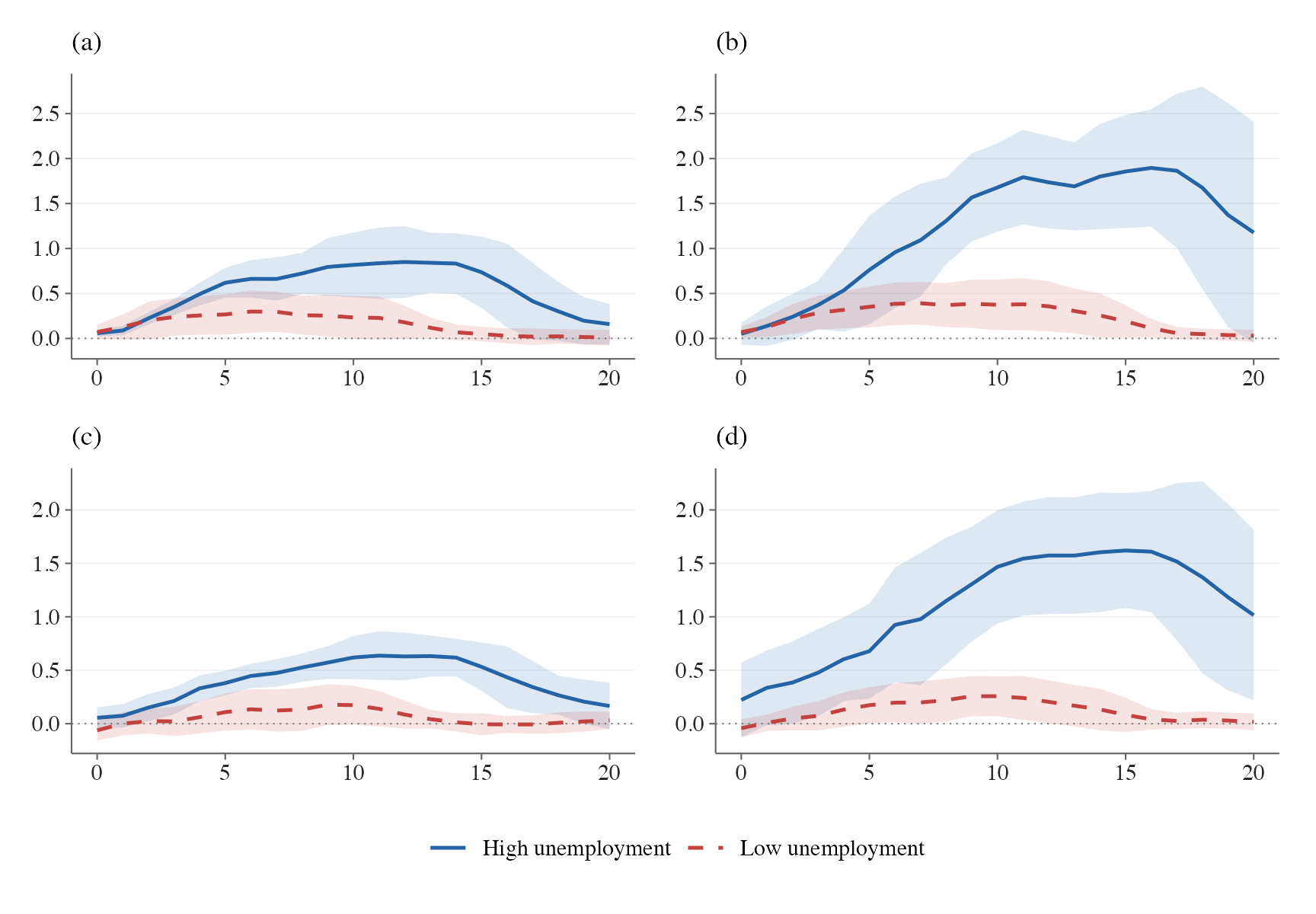}
\caption{Responses to military spending news under fixed and
estimated unemployment thresholds.}
\label{fig:fiscal-irfs}
\begin{minipage}{0.97\linewidth}
\footnotesize
\textit{Notes:} Panels (a)--(b) show government spending responses and panels (c)--(d) report GDP responses. The left column uses the fixed 6.5\% and the right column uses the estimated 9.273\%. Responses are expressed as a percentage of potential GDP following news equal to 1\% of lagged potential GDP. Horizons are measured in quarters.
Shaded regions denote 95\% pointwise HAC intervals conditional on the threshold. Vertical scales are common within each outcome.
\end{minipage}
\end{figure}

\begin{figure}[!tbp]
\centering
\includegraphics[width=\linewidth]{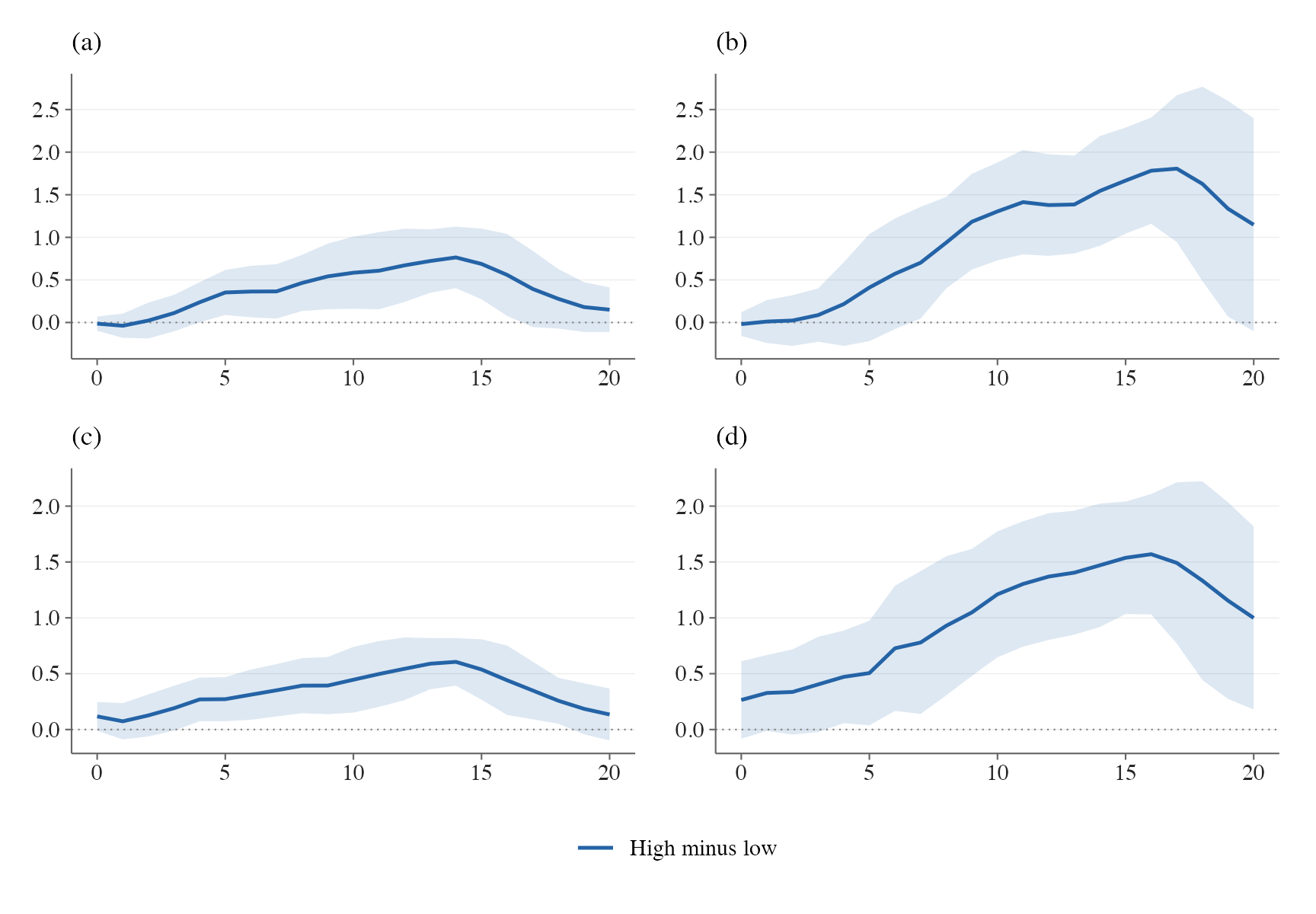}
\caption{High-minus-low responses to military spending news.}
\label{fig:fiscal-differences}
\begin{minipage}{0.97\linewidth}
\footnotesize
\textit{Notes:} Panel ordering and units follow
Figure~\ref{fig:fiscal-irfs}. Shaded regions denote 95\% pointwise intervals accounting for covariance between the high- and low-state response estimates.
\end{minipage}
\end{figure}

Let $\widetilde\phi_{a,r}^{s}$ denote the debiased response
of outcome $a\in\{Y,G\}$ in state $s\in\{L,H\}$.
Table~\ref{tab:fiscal-multipliers} reports cumulative fiscal
multipliers, defined as
\[
\widehat{\mathcal M}_s(\bar r)
=
\frac{\sum_{r=0}^{\bar r}\widetilde\phi_{Y,r}^{s}}
     {\sum_{r=0}^{\bar r}\widetilde\phi_{G,r}^{s}},
\qquad \bar r\in\{7,15\},
\]
corresponding to two- and four-year windows. Confidence intervals are constructed using the delta method and the joint HAC covariance across horizons, outcomes, and states.%
\footnote{The Bartlett HAC estimator uses uncentered score lag products with denominator $n-j$ at lag $j$
and an AR(1) plug-in bandwidth. For cumulative multipliers, the influence scores are summed over the relevant horizons and stacked across outcomes and states before the joint covariance is computed.}

Under the estimated threshold, the two- and four-year
high-state multipliers are 1.11 and 0.94, compared with
0.37 and 0.49 in the low state. The corresponding high-minus-low differences are 0.743 and 0.441, with 95\% intervals of $[0.198,1.287]$ and $[0.154,0.729]$. Under the fixed 6.5\% threshold, the differences are 0.439 and 0.347, and both intervals include zero. The estimated threshold therefore yields larger multiplier contrasts that are statistically
significant for both windows. However, the high-state multiplier intervals include one in both cases. These comparisons illustrate how estimating the state boundary
can affect fiscal responses, holding the control set and LP estimator fixed.

\begin{table}[!tbp]
\centering
\caption{Cumulative fiscal multipliers under alternative
unemployment thresholds.}
\label{tab:fiscal-multipliers}
\small
\setlength{\tabcolsep}{4pt}
\begin{tabular}{@{}llccc@{}}
\toprule
Specification & Window & High unemployment
& Low unemployment & High minus low \\
\midrule
Estimated $9.273\%$ & 2 years
& \shortstack{1.110\\$[0.652,1.569]$}
& \shortstack{0.368\\$[0.019,0.716]$}
& \shortstack{0.743\\$[0.198,1.287]$} \\[4pt]
Estimated $9.273\%$ & 4 years
& \shortstack{0.935\\$[0.742,1.129]$}
& \shortstack{0.494\\$[0.254,0.735]$}
& \shortstack{0.441\\$[0.154,0.729]$} \\[4pt]
Fixed $6.5\%$ & 2 years
& \shortstack{0.673\\$[0.476,0.871]$}
& \shortstack{0.234\\$[-0.377,0.845]$}
& \shortstack{0.439\\$[-0.173,1.052]$} \\[4pt]
Fixed $6.5\%$ & 4 years
& \shortstack{0.719\\$[0.599,0.839]$}
& \shortstack{0.372\\$[-0.104,0.849]$}
& \shortstack{0.347\\$[-0.126,0.819]$} \\[4pt]
\bottomrule
\end{tabular}
\par\smallskip
\begin{minipage}{0.97\linewidth}
\footnotesize
\textit{Notes:} Brackets report 95\% delta-method confidence
intervals. All specifications use the same 480 shock origins.
\end{minipage}
\end{table}

%%%%%%%%%%%%%%%%%%%%%%%%%%%%%%%%%%
\subsection{Plug-in Selection of Tuning Parameters}
\label{lambdarulets}

We adapt the residual-based plug-in procedure of
\citet[][Appendix~C.5]{ADAMEK20231114}
to our anchored threshold Lasso. For each prespecified $c_\rho>0$, the procedure iterates between calibrating a common coefficient penalty, searching over thresholds, and updating the residuals from the selected model.

Throughout this subsection, $y$ is centered and divided by its
empirical root mean square, and $Z(\tau)$ contains the centered
and standardized threshold regressors. Interactions are formed before centering. Let $\mathcal J_u$ and $\mathcal J_p$ index the unpenalized and penalized columns, respectively. In the application, $\mathcal J_u$ contains the news shock, its low-state interaction, and the low-state indicator. The penalized columns contain the baseline lagged controls and their low-state interactions. For a given $\lambda$, define
\[
\mathcal L_n(\vartheta,\tau;\lambda)
=
\frac{\|y-Z(\tau)\vartheta\|_2^2}{2n}
+\lambda\|\vartheta_{\mathcal J_p}\|_1.
\]
The anchor criterion adds $c_\rho\lambda|\tau-\tau^*|$ to this loss. For fixed $\tau$ and $\lambda$, this additional term is constant in $\vartheta$ and therefore leaves the coefficient fit unchanged.

To calibrate one penalty across the threshold grid $\mathbb T_n$, project out the unpenalized columns and stack the resulting designs:
\[
\begin{aligned}
\widetilde Z_{\mathcal J_p}(\tau)
&=
\bigl(I_n-P_{Z_{\mathcal J_u}(\tau)}\bigr)
Z_{\mathcal J_p}(\tau),\\
V
&=
\bigl[
\widetilde Z_{\mathcal J_p}(\tau_1),\ldots,
\widetilde Z_{\mathcal J_p}(\tau_{|\mathbb T_n|})
\bigr].
\end{aligned}
\]
Here $P_A$ is the orthogonal projection onto the column space
of $A$, and $v_t'$ denotes row $t$ of $V$.
Given residuals $\widehat u^{(k-1)}$, form
\[
\widehat a_t^{(k)}
=
v_t\widehat u_t^{(k-1)}
-\frac1n\sum_{j=1}^{n}
v_j\widehat u_j^{(k-1)}.
\]
The calibration covariance is
\begin{equation}\label{eq:ts-plugin-cov}
\widehat\Omega^{(k)}
=
\frac1n\sum_{t=1}^{n}\sum_{j=1}^{n}
\left(1-\frac{|t-j|}{L_k}\right)_+
\widehat a_t^{(k)}\widehat a_j^{(k)\prime},
\end{equation}
where $(z)_+=\max\{z,0\}$.
The bandwidth $L_k$ is selected using the AR(1) plug-in rule
of \citet{andrews1991heteroskedasticity}, applied to the
nondegenerate score coordinates after rescaling them to unit
empirical root mean square. Fitted AR coefficients are restricted
to $[-0.99,0.99]$, and $L_k$ to $[2,\lceil n/2\rceil]$.

\begin{algorithm}[tbp]
\caption{Plug-in tuning for anchored threshold Lasso}
\label{alg:threshold-plugin}
\algtext*{EndIf}
\begin{algorithmic}[1]
\ForAll{$c_\rho\in\mathcal C_\rho$}
    \State Initialize
    $\widehat u^{(0)}\gets y$ and
    $\lambda^{(0)}\gets\|V'y\|_\infty/n$.

    \For{$k=1,\ldots,K_{\mathrm{iter}}$}
        \State Compute the centered scores, bandwidth $L_k$,
        and covariance $\widehat\Omega^{(k)}$
        in \eqref{eq:ts-plugin-cov}.

        \For{$b=1,\ldots,B$}
            \State Draw
            $g^{(b,k)}\sim N(0,\widehat\Omega^{(k)})$
            and set $m_b^{(k)}=\|g^{(b,k)}\|_\infty$.
        \EndFor

        \State Let $q_{1-\alpha}^{(k)}$ be the empirical
        $(1-\alpha)$-quantile of
        $m_1^{(k)},\ldots,m_B^{(k)}$.

        \State Set
        $\lambda^{(k)}
        \gets c_\lambda q_{1-\alpha}^{(k)}/\sqrt n$
        and $\rho^{(k)}\gets c_\rho\lambda^{(k)}$.

        \ForAll{$\tau\in\mathbb T_n$}
            \State Obtain
            $\widehat\vartheta^{(k)}(\tau)
            \in\arg\min_{\vartheta}
            \mathcal L_n(\vartheta,\tau;\lambda^{(k)})$.
        \EndFor

        \State Select
        \[
        \widehat\tau^{(k)}
        \in\arg\min_{\tau\in\mathbb T_n}
        \left\{
        \mathcal L_n\bigl(
        \widehat\vartheta^{(k)}(\tau),\tau;\lambda^{(k)}
        \bigr)
        +\rho^{(k)}|\tau-\tau^*|
        \right\}.
        \]

        \State Update
        $\widehat u^{(k)}
        \gets y-Z(\widehat\tau^{(k)})
        \widehat\vartheta^{(k)}(\widehat\tau^{(k)})$.

        \If{the convergence criterion below is satisfied}
            \State Mark the run as converged and
            \textbf{break}.
        \EndIf
    \EndFor

    \State Record the final penalty, threshold,
    coefficient estimates, and convergence status.
\EndFor
\end{algorithmic}
\end{algorithm}

The implementation generates Gaussian draws without forming
the large covariance matrix, using correlated multipliers shared
across all score coordinates and candidate thresholds.
Specifically, for independent standard normal innovations
$e_j^{(b)}$, set
\[
\xi_t^{(b,k)}
=
L_k^{-1/2}
\sum_{j=t}^{t+L_k-1}e_j^{(b)},
\qquad
g^{(b,k)}
=
n^{-1/2}\sum_{t=1}^{n}
\widehat a_t^{(k)}\xi_t^{(b,k)}.
\]
Conditional on the data and current residuals, these draws have
covariance \eqref{eq:ts-plugin-cov}.

Convergence requires
\[
\frac{|\lambda^{(k)}-\lambda^{(k-1)}|}
     {\lambda^{(k-1)}}<\epsilon_\lambda,
\qquad
\frac{\|\widehat u^{(k)}-\widehat u^{(k-1)}\|_2}
     {\|\widehat u^{(k-1)}\|_2}<\epsilon_u,
\]
with unchanged threshold and bandwidth.
These conditions must hold for two consecutive iterations and
remain satisfied after one additional update.
Runs that reach the iteration limit or enter a numerical cycle
without satisfying the criterion are recorded as unconverged.
Exact ties in threshold selection are resolved by taking the
largest minimizing candidate.
All choices of $c_\rho$ use the same initialization, and the
simulation seed is held fixed across iterations and values of $c_\rho$.

In the application, we set
\[
\mathcal C_\rho=\{0.5,0.75,0.9,1\},\qquad
c_\lambda=0.4,\qquad
\alpha=0.05,\qquad
B=1000,
\]
with $K_{\mathrm{iter}}=60$ and
$\epsilon_\lambda=\epsilon_u=0.01$.
The empirical quantile is the 951st ordered maximum among
the 1,000 draws.
Unemployment is expressed as a decimal, with $\tau^*=0.065$.
The threshold grid contains the observed values between the
empirical 10th and 90th percentiles and includes the anchor.
Each value of $c_\rho$ is evaluated by running the complete
plug-in iteration separately.
All four runs converge to $\widehat\tau=0.093$, with final
penalties approximately equal to $0.026$.
We use $c_\rho=0.75$ as the baseline, and this run converges after
ten iterations with calibration bandwidth 11.%
\footnote{Expressing unemployment and $\tau^*$ in percentage
points requires $c_\rho=0.0075$ instead of $0.75$ to leave the
criterion unchanged.}

After threshold selection, each outcome and projection horizon
uses a separate residual-based plug-in calibration with the
threshold held fixed. These LP calibrations use scores from the full standardized regression design, and the required nodewise regressions have separately calibrated penalties.

\subsection{Proofs for the Time Series Extension}\label{ts-proofs}

For brevity write $b=b_{n,\rm ts}$, $T=T_n^{\rm tr}$, $m=m_n$, and
$\lambda=\lambda_n$ throughout the proofs. For a scalar
function $F_i$ write $P_nF=n^{-1}\sum_iF_i$ and $PF=EF_i$.

\begin{proof}[Derivation of \eqref{ts-inverse-modulus}]
By conditional Cauchy--Schwarz and \eqref{ts-local-moments}, we have
$|E[X_i^{(j)}X_i^{(l)}\mid Q_i=q]|\le C_{\rm loc,ts}^{1/2}$
for almost every $q\in\mathbb T$. Since $Q_i$ is uniform, 
\[
 \|M(\tau)-M(t)\|_\infty
 \vee\|N(\tau)-N(t)\|_\infty
 \le C_{\rm loc,ts}^{1/2}|\tau-t|.
\]
The block formula \eqref{ts-precision} implies that every required row of $A(\tau)$ or $B(\tau)$ has at most $\bar s$ nonzero entries. For a lower block row of $\bm\Theta$, the supports of $-B_j$ and $A_j+B_j$ together contain that of $A_j$, while an upper block row contains both $B_j$ and $-B_j$.

Consider a required row $A_j(\tau)'$ and write $v=A_j(\tau)-A_j(t)$. Its support has size at most $2\bar s$.
Take $J$ to contain the $\min\{\bar s,p\}$ largest coordinates of $v$ in absolute value. If $v\ne0$, then
$|v_{J^c}|_1\le\sqrt{\bar s}|v_J|_2$ and
$|v_J|_2^2\ge |v|_2^2/2$. Since $c_{\rm node}>1$,
Assumption~\ref{ts-as6}(ii) gives
\[
 v'M(\tau)v\ge\frac{\kappa_{\rm node}^2}{2}|v|_2^2.
\]
Symmetry and the inverse-row identities give
\[
 M(\tau)v=\{M(t)-M(\tau)\}A_j(t).
\]
Using the bound $|A_j(t)|_1\le C_\Theta$ in Assumption~\ref{ts-as6}(ii), we obtain
\[
 \frac{\kappa_{\rm node}^2}{2}|v|_2^2
 \le C_\Theta C_{\rm loc,ts}^{1/2}|v|_1|\tau-t|
 \le C_\Theta C_{\rm loc,ts}^{1/2}
        \sqrt{2\bar s}|v|_2|\tau-t|.
\]
The zero vector is trivial; otherwise, divide by $|v|_2$ and use the support bound again. This yields
\[
 |A_j(\tau)-A_j(t)|_1
 \le\frac{4C_\Theta C_{\rm loc,ts}^{1/2}}{\kappa_{\rm node}^2}
       \bar s|\tau-t|.
\]
The same argument applies to $B_j$. Therefore, the blocks give
\eqref{ts-inverse-modulus}, for example with
$C=12C_\Theta C_{\rm loc,ts}^{1/2}/\kappa_{\rm node}^2$.
We only use the retained sparse restricted eigenvalues, not a lower
bound for the unrestricted minimum eigenvalue. 
\end{proof}

%%%%%%%%%%%%%%%%%%%%%%%%%%
\begin{lem}[Threshold-indexed concentration under NED]\label{ts-lem26}
Suppose Assumption~\ref{ts-as5} holds along with \eqref{ts-exponents}; the following hold.

\textnormal{(i)}  Uniformly over $\tau\in\mathbb T$,
\[
 \max_{j\le2p}|P_n\bm X_i^{(j)}(\tau)U_i|=o_P(b),\qquad
 \|\widehat{\bm\Sigma}(\tau)-\bm\Sigma(\tau)\|_\infty=o_P(b).
\]
More precisely, the structural score satisfies
\begin{equation}\label{ts-global-score-sharp}
 \sup_{\tau\in\mathbb T}
 \left|n^{-1/2}\bm X(\tau)'U\right|_\infty
 =O_P(\sqrt{m_n\ell_n}).
\end{equation}
The uniform centered empirical bound for
$X_i^{(j)}X_i^{(l)}X_i^{(a)}X_i^{(c)}$ and
$X_i^{(j)}X_i^{(l)}U_i^2$, including their absolute versions and
multiplication by threshold indicators or their complements, is
$O_P(c_{4,n})$. Maxima are over all coordinate choices.

\textnormal{(ii)}   For any deterministic $t\in\mathbb T$, on an event of probability
tending to one, simultaneously for $b^2\le\eta\le1$,
\begin{align}
 \max_{j,l}\sup_{|\tau-t|\le\eta}
 P_n|X_i^{(j)}X_i^{(l)}|
 |1\{Q_i<\tau\}-1\{Q_i<t\}|
 &\le C\eta,\label{ts-local-design}\\
 \max_j\sup_{|\tau-t|\le\eta}
 |P_nX_i^{(j)}U_i\{1\{Q_i<\tau\}-1\{Q_i<t\}\}|
 &\le C\{\sqrt{\eta m\ell_n/n}+b^2\}.
 \label{ts-local-score}
\end{align}
All thresholds in these suprema belong to $\mathbb T$. 

\textnormal{(iii)}   For
$t_n=\min((\{Q_i:Q_i\ge t\}\cap\mathbb T)\cup\{t_1\})$,
\[
 \bm X(t_n)=\bm X(t),\qquad 0\le t_n-t\le b^2
\]
with probability tending to one, uniformly over deterministic $t\in\mathbb T$.
\end{lem}

\begin{proof}[Proof of Lemma~\ref{ts-lem26}]
We establish the threshold approximation explicitly. Write
$V_i^{[m]}=E(V_i\mid\mathcal F^s_{i,m})$ for each primitive coordinate.
More generally, if $\|Q_i-Q_i^{[m]}\|_{L^r}\le a_Q(m)$ and $Q_i$ has
a uniformly bounded density near $\mathbb T$, then, for every $\varepsilon>0$,
\[
 \sup_{\tau\in\mathbb T}
 P(1\{Q_i<\tau\}\ne1\{Q_i^{[m]}<\tau\})
 \le C\varepsilon+a_Q(m)^r\varepsilon^{-r}.
\]
The indicators can differ only within the $\varepsilon$-band around the
threshold or when the approximation error exceeds $\varepsilon$. Optimization
at $\varepsilon=a_Q(m)^{r/(r+1)}$ gives the bound
$Ca_Q(m)^{r/(r+1)}$.
If $\|H_i\|_{L^v}\le C$, $1\le u<v$, and an
$\mathcal F^s_{i,m}$-measurable $H_i^{[m]}$ satisfies
$\|H_i-H_i^{[m]}\|_{L^u}\le a_H(m)$, H\"older's inequality yields
\begin{align}
 &\sup_{\tau\in\mathbb T}
 \|H_i1\{Q_i<\tau\}
       -E[H_i1\{Q_i<\tau\}\mid\mathcal F^s_{i,m}]\|_{L^u}
 \notag\\
 &\hspace{15mm}\le
 C\{a_H(m)+a_Q(m)^{\frac r{r+1}(\frac1u-\frac1v)}\}.
 \label{ts-ned-product}
\end{align}
Use the approximation $H_i^{[m]}1\left\{Q_i^{[m]}<\tau\right\}$, decompose its error into a product error and an indicator error, and then apply conditional-expectation contraction, which costs a factor of at most two.
The same bound, without that factor, applies to this explicit approximation. By Assumption~\ref{ts-as5}(i), the threshold variable admits
an $L^8$-NED approximation. Thus, the preceding bound applies with moment order $q_Q=8$ and
$a_Q(m)=C(1+m)^{-d_{\mathrm{NED}}}$. Products and absolute products have the required approximation bounds by H\"older and the primitive $L^8$ bounds. For instance, a degree two product
has approximation error $O(m^{-d_{\rm NED}})$ in $L^u$ for $u\le4$,
and a degree-four product has this error for $u\le2$.

We now transfer the empirical process to functions of finite innovation
windows. Write $\omega=\omega_{\rm ts}$, $a=a_{\rm ts}$ and
$c_b=1-\omega-2a$. The exponent restrictions imply
$1/2<c_b<2/3$ and $\omega<1/6$. Choose fixed exponents
\[
 1-\omega-a<e_X<
      \omega d_{\rm NED}-(1+\nu_p)/8,\qquad
 e_Q=c_b+1/8.
\]
The first interval is nonempty by \eqref{ts-exponents}. Moreover,
$\omega d_{\rm NED}>c_b+a+(1+\nu_p)/8$ implies
$\omega d_{\rm NED}>e_Q+1/8$.
Markov's inequality and a union bound therefore give
\begin{align}
 P_{\alpha_0}\left(\max_{i,j}|X_i^{(j)}-X_i^{(j),[m]}|>n^{-e_X}
       \text{ or }\max_i|U_i-U_i^{[m]}|>n^{-e_X}\right)
 &\le Cn^{1+\nu_p+8e_X-8\omega d_{\rm NED}}=o(1),\notag\\
 P_{\alpha_0}\left(\max_i|Q_i-Q_i^{[m]}|>n^{-e_Q}\right)
 &\le Cn^{1+8e_Q-8\omega d_{\rm NED}}=o(1).
 \label{ts-approx-event}
\end{align}
Also,
\[
 P_{\alpha_0}\left(\max_{i,j}|X_i^{(j)}|>T
              \text{ or }\max_i|U_i|>T\right)
 \le CnpT^{-12}+CnT^{-8}=o(1).
\]
For each design or error coordinate $V_i$, define the bounded
approximation
\[
 V_i^{[m],T}
 =\mathcal W_{2T}(V_i^{[m]}),
 \qquad
 \mathcal W_{2T}(v)=\max\{-2T,\min(v,2T)\}.
\]
On the approximation and boundedness events established above, and outside the threshold boundary band of width $n^{-e_Q}$, the difference between an observed product of degree
$d\in\{2,4\}$ and the corresponding product of bounded approximations is at most $CT^{d-1}n^{-e_X}.$ The choice of $e_X$ ensures that this bound is $o(b^2)$ for degree-two products and $o(c_{4,n})$ for degree-four products.

Moreover, conditional expectation contraction preserves the uniform eighth moment bounds for the approximating coordinates. These bounds imply that replacing an approximating product by its bounded version introduces an absolute expectation error of at most $CT^{d-8}$. Consequently, the resulting
expectation errors are bounded by $CT^{-6}$ for degree-two
products and $CT^{-4}$ for degree-four products. Under the
stated exponent restrictions, these terms are respectively
$o(b^2)$ and $o(c_{4,n})$.

The required degree-two products have uniformly bounded $L^{24/5}$ norms,
and the required degree-four products have bounded $L^{12/5}$ norms.
Equation \eqref{ts-ned-product} at $u=1$ therefore bounds the population
centering errors by $Cm^{-19d_{\rm NED}/27}$ and
$Cm^{-14d_{\rm NED}/27}$, respectively. The exponent conditions make them $o(b^2)$ and $o(c_{4,n})$.
Interval indicators are differences of two threshold indicators, so the
same bounds suffice for local centering.

Partition time into blocks of length $L_n=\lceil4m\rceil$. The innovation
windows supporting consecutive odd blocks, and similarly consecutive even
blocks, are separated by at least a constant multiple of $m$. Independent block coupling for absolutely regular sequences gives total error bounded by
\[
 C(n/m)\beta_s(m)\le Cn^{1-\omega(1+d_{\rm mix})}=o(1).
\]
This couples entire innovation block vectors, simultaneously for every
coordinate and threshold function,  whose error is not multiplied by the
number of these functions. Conditional expectations on each finite window are measurable functions of that window.

A clipped summand of degree $d$ has envelope $CT^d$. Its second moment
is bounded uniformly for the product families under consideration.
A block sum variance bound $CL_n^2$ suffices. Bernstein's inequality on
the coupled blocks and a union bound over polynomially many functions give
\begin{equation}\label{ts-block-global}
 O_P\{\sqrt{m\ell_n/n}+mT^d\ell_n/n\}.
\end{equation}
The bounded-block term covers the incomplete last block. We do not assume independence within a block.

For a degree two product over an interval of length $\eta$,
\eqref{ts-local-moments} bounds its second moment by $C\eta$. For the finite window approximation, this bound holds with $\eta$ replaced by
$\eta+e_n$, where
\[
 e_n=C\{n^{-e_Q}+m^{-d_{\rm NED}}
             +(m^{-d_{\rm NED}}n^{e_Q})^{14/3}\}=o(b^2).
\]
To check this, enlarge the interval by $2n^{-e_Q}$ on
$|Q_i-Q_i^{[m]}|\le n^{-e_Q}$. On its complement, use the $L^{12/5}$
bound for the squared product and Markov's inequality for the eighth
moment of the approximation error. Differences of squared unthresholded
products have expectation $O(m^{-d_{\rm NED}})$.
The relation
$\omega d_{\rm NED}-e_Q>a+\nu_p/8>7/24$
verifies the small order.
Consequently the local version of \eqref{ts-block-global} is $
 C\{\sqrt{(\eta+e_n)m\ell_n/n}+mT^2\ell_n/n\}.$
Since $mT^2\ell_n/n=b^2$, this gives the local stochastic term in
\eqref{ts-local-score} for $\eta\ge b^2$.

Use a deterministic threshold grid of mesh at most $n^{-2}$, adding the
center $t$, and a dyadic grid of radii. There are polynomially many
endpoints and coordinate choices. Between grid points, bracket differences by positive absolute product sums over grid cells. To transfer from $Q_i^{[m]}$ to $Q_i$, enlarge these cells by $2n^{-e_Q}$.
The preceding local bound makes the degree two bracketing error $O_P(b^2)$. For degree four products, their expected mass on such cells is at most $Ce_n$, by \eqref{ts-local-moments} and the same H\"older bound, while the global block inequality makes their empirical contribution $O_P(c_{4,n})$. This proves uniformity without reordering the time sequence.

In addition, the structural score has mean zero. Positive local design products have expectation at most $C\eta$, and $\sqrt{\eta m\ell_n/n}\le\eta/T$ for $\eta\ge b^2$. These observations establish (ii).
For degree two, the global bound and the errors from finite-window approximation, truncation, and threshold discretization are $o_P(b)$.
For the structural score, whose expectation is zero, the same argument
more precisely gives
\[
 \sup_\tau|P_n\bm X_i(\tau)U_i|_\infty
 =O_P\{\sqrt{m\ell_n/n}+b^2\}.
\]
The approximation and differences between the population means of the original and approximating summands above are $o(b^2)$, while the error from extending the bounds from
the finite threshold grid to all $\tau\in\mathbb T$ is $O_P(b^2)$. Since
$b=T\sqrt{m\ell_n/n}$ and $Tb=\sqrt{c_{4,n}}=o(1)$,
$b^2=o(\sqrt{m\ell_n/n})$. Multiplication by $\sqrt n$ proves
\eqref{ts-global-score-sharp}.
For degree four, the dominant bound is
$mT^4\ell_n/n=c_{4,n}$; in particular
$\sqrt{m\ell_n/n}=o(c_{4,n})$ since $a>1/6$.
This proves (i).

For (iii), the distance bound is automatic if $t_1-t\le b^2$.
Otherwise apply the same approximation and block argument to the count
on $[t,t+b^2]$. Its mean is $b^2$, and its stochastic error relative to
that mean tends to zero: the square-root term has ratio $1/T$ and the
bounded-block term has ratio $1/T^2$. Boundary approximation errors are
also negligible. Thus the interval contains an observed value with
probability tending to one. There are no observations in $[t,t_n)$,
which proves equality of the threshold designs with the strict indicator.
\end{proof}

%%%%%%%%%%%%%%%%%%%%%%%%%%%%%%%%%%
\begin{proof}[Proof of Theorem~\ref{ts-thm1}]
Write $q_n=\|\widehat f-f_0\|_n$,
$v=\widehat\alpha-\alpha_0$, and $a_n=|v|_1$ in this proof only.
Assumption~\ref{ts-as6}(i) and Lemma~\ref{ts-lem26} bound all weights above
and away from zero. Use the fixed $\mu\in(0,1)$ from Assumption~\ref{ts-as6} and choose $A$ so that the uniform score
event has the usual weighted bound $\mu\lambda/2$.

We first establish preliminary prediction and coefficient bounds
that hold in both the identified threshold effect and no threshold
effect cases. In the identified threshold effect case, define
\[
\tau_{0,n}
=
\min\left[
\bigl(\{Q_i:Q_i\ge\tau_0\}\cap\mathbb T\bigr)
\cup\{t_1\}
\right].
\]
By the definition of $\mathbb T_n$, we have
$\tau_{0,n}\in\mathbb T_n$. Lemma~\ref{ts-lem26}(iii) gives
\[
\bm X(\tau_{0,n})=\bm X(\tau_0),
\qquad
0\le\tau_{0,n}-\tau_0\le b^2
\]
with probability tending to one. We therefore compare the value
of the joint objective at $(\widehat\alpha,\widehat\tau)$ with
its value at $(\alpha_0,\tau_{0,n})$.

In the no threshold effect case, we instead compare directly
with $(\alpha_0,\tau_\ast)$, since $\tau_\ast\in\mathbb T_n$
by construction and $\bm X(\tau_\ast)\alpha_0=f_0$.
In both cases, the uniform score bound, the bounded coefficient
space, and $\rho_n\asymp\lambda_n$ yield
\begin{equation}\label{ts-coarse}
 q_n^2+(1-\mu)\lambda
            |D(\widehat\tau)\widehat\alpha|_1
 \le C\lambda|\alpha_0|_1+C\rho_n
 \le Cs_0\lambda.
\end{equation}
Thus $a_n\le Ls_0$ for a fixed sufficiently large $L$. Under
\eqref{ts-margin}, $|\widehat\tau-\tau_0|=O_P(s_0\lambda)$.
This establishes the coefficient region to which the identification
assumption is applied, but it does not presume the sharper oracle inequality. In the fixed-threshold case it also gives
$P(\widehat\tau\in\mathbb S)\to1$, since $s_0\lambda\to0$.
All later uses of the main restricted eigenvalue condition are on this
localization event. In the null case $\mathbb S=\mathbb T$.

In the no threshold effect case, let $J_0=J(\alpha_0)$,
$p_v=|D(\widehat\tau)v|_1$, and
$z_v=|\{D(\widehat\tau)v\}_{J_0}|_1$.
Only the baseline coefficients of $\alpha_0$ are nonzero, so their penalty
weights do not depend on the threshold. Comparing with
$(\alpha_0,\tau_\ast)$ therefore yields the exact inequality
\begin{equation}\label{ts-null-basic}
 q_n^2+(1-\mu)\lambda p_v
            +\rho_n|\widehat\tau-\tau_\ast|
 \le2\lambda z_v.
\end{equation}
If $J_0$ is empty, the displayed inequality gives the claims directly.
Otherwise, the displayed inequality implies the weighted cone
condition
\[
|\{D(\widehat\tau)v\}_{J_0^c}|_1
\le \frac{1+\mu}{1-\mu}
|\{D(\widehat\tau)v\}_{J_0}|_1.
\]
Since the penalty weights are uniformly bounded above and
bounded away from zero, this inequality and $|J_0|\le s_0$
give
\[
a_n=|v|_1
\le C|v_{J_0}|_1
\le C\sqrt{s_0}\,|v_{J_0}|_2.
\]The population restricted eigenvalue
condition and
$s_0\sup_\tau\|\widehat{\bm\Sigma}(\tau)-\bm\Sigma(\tau)\|_\infty=o_P(1)$
then give the empirical restricted eigenvalue bound. Explicitly, for a vector in the adaptive cone and $G=\bm\Sigma$,
\[
 v'\widehat G(\tau)v
 \ge\{\underline\kappa^2-(1+c_0)^2s_0
                  \|\widehat G(\tau)-G(\tau)\|_\infty\}|v_{J_0}|_2^2.
\]
The term subtracted from $\underline\kappa^2$ is $o_P(1)$.
The choice of $c_0$ contains both coefficient cones used below, since
the limiting ratio of the largest to smallest main penalty weight is at
most $\sqrt2w_+/w_-$.
Hence
$z_v\le C\sqrt{s_0}q_n$ and $q_n\le C\sqrt{s_0}\lambda$.
Substitution into~\eqref{ts-null-basic} gives
$a_n\le Cs_0\lambda$ and
$\rho_n|\widehat\tau-\tau_\ast|\le Cs_0\lambda^2$.

Now suppose $\delta_0\ne0$. Write
$d_n=|\widehat\tau-\tau_0|$ and $D=|\delta_0|_1$.
The upper bound on $|\delta_0|_2$ implies $D^2\le d_+^2s_0$.
The grid comparator $\tau_{0,n}$ from Lemma~\ref{ts-lem26}(iii) has the
same sample design as $\tau_0$ and incurs additional anchoring cost at most $\rho_n b^2$. Apply the local bounds with radius $d_n\vee b^2$; the event then holds simultaneously for all such radii. The structural threshold noise term is bounded by
$C D\{b\sqrt{d_n\vee b^2}+b^2\}.$
The difference of the coefficient penalty weights at $\widehat\tau$ and
$\tau_0$ is at most $CD(d_n\vee b^2)$: this follows from
\eqref{ts-local-design} and the weights being bounded away from zero.
Finally, the difference of anchoring penalties is at most
$\rho_n(d_n+b^2)$. The basic inequality can consequently be written
\begin{equation}\label{ts-fixed-basic}
 q_n^2+(1-\mu)\lambda p_v
 \le2\lambda z_v+\lambda E_n,\qquad
 E_n=C\{D\sqrt{d_n}+(1+D)d_n+s_0\lambda\}.
\end{equation}
We can enlarge $C$ if necessary so that this is a nonnegative upper bound for all
remainder terms, including those at the radius floor $b^2$.

Since $z_v\le Ca_n$, \eqref{ts-margin} and
\eqref{ts-fixed-basic} imply
\[
 d_n\le C\lambda a_n+C\lambda D\sqrt{d_n}
             +C\lambda(1+D)d_n+Cs_0\lambda^2.
\]
Use $\lambda(1+D)=o(1)$ and Young's inequality
$C\lambda D\sqrt{d_n}\le d_n/4+C'\lambda^2D^2$.
As $D^2\le Cs_0$, we obtain
\begin{equation}\label{ts-d-a}
 d_n\le C\lambda a_n+Cs_0\lambda^2.
\end{equation}

We consider two cases for the coefficient bound. If $z_v\le E_n$,
\eqref{ts-fixed-basic} gives $a_n\le CE_n$. Substitution
of~\eqref{ts-d-a} and $D\le C\sqrt{s_0}$ yields
\[
 a_n\le C D\sqrt{\lambda a_n}
             +C(1+D)\lambda a_n+Cs_0\lambda.
\]
Absorb the second term and use Young's inequality on the first to obtain
$a_n\le Cs_0\lambda$.

If $z_v>E_n$, \eqref{ts-fixed-basic} gives
$q_n^2+(1-\mu)\lambda p_v\le3\lambda z_v$.
Thus $a_n\le Cz_v$ and the empirical restricted eigenvalue bound applies.
The identity
\[
 \widehat f-f_0=\bm X(\widehat\tau)v+
       \{\bm X(\widehat\tau)-\bm X(\tau_0)\}\alpha_0
\]
and~\eqref{ts-local-design} imply
\[
 \|\bm X(\widehat\tau)v\|_n^2
 \le q_n^2+CDa_n(d_n\vee b^2),\qquad
 z_v^2\le Cs_0\lambda z_v+Cs_0Da_n(d_n\vee b^2).
\]
Using $a_n\le Cz_v$ and~\eqref{ts-d-a},
\[
 z_v^2\le Cs_0\lambda z_v+Cs_0D\lambda z_v^2
                          +Cs_0^2D\lambda^2z_v.
\]
The restriction $s_0D_{\delta,n}\lambda\to0$ absorbs the quadratic term and bounds the last coefficient by a constant times $s_0\lambda$. Therefore $z_v\le Cs_0\lambda$, and again $a_n\le Cs_0\lambda$. Equation~\eqref{ts-d-a} now gives $d_n\le Cs_0\lambda^2$, and
\eqref{ts-fixed-basic} yields $q_n^2\le Cs_0\lambda^2$. The local stochastic bounds hold on a common event for all admissible radii and therefore apply at the random radius
$d_n\vee b^2$. By moving terms with sufficiently small coefficients to the left-hand side, the preceding inequalities establish the claimed rates in both subcases, $z_v\le E_n$
and $z_v>E_n$, without an iterative localization argument.
\end{proof}
%%%%%%%%%%%%%%%%%%%%%%%%%%%%%
\begin{lem}[Uniform nodewise bounds]\label{ts-lem27}
Suppose Assumptions~\ref{ts-as5}--\ref{ts-as6} hold along with
\eqref{ts-exponents}. Then we have
\begin{align}
 \max_{j\in H}\sup_{\tau\in\mathbb T}
 |\widehat\Theta_j(\tau)-\Theta_j(\tau)|_1
 &=O_P(\bar s b),\label{ts-node-l1}\\
 \max_{j\in H}\sup_{\tau\in\mathbb T}
 |\widehat\Theta_j(\tau)|_1&=O_P(1),\notag\\
 \max_{j\in H}\sup_{\tau\in\mathbb T}
 |\widehat\Theta_j(\tau)'\widehat{\bm\Sigma}(\tau)-e_j'|_\infty
 &=O_P(b).\label{ts-node-kkt}
\end{align}
\end{lem}

\begin{proof}
Consider a required row in the lower regime. Its population regression
coefficient satisfies
$\gamma_{0,j}=-A_{j,-j}/A_{jj}$ and its residual variance is $z_j^2=1/A_{jj}$. Positive definiteness gives $A_{jj}\ge1/M_{jj}\ge1/w_+^2$,
whereas the required row-norm bound gives $A_{jj}\le C_\Theta$. Thus $A_{jj}$ is bounded above and away from zero. Assumption~\ref{ts-as6}(ii) consequently bounds
$|\gamma_{0,j}|_1$ uniformly. The block identity~\eqref{ts-precision}
bounds the sparsities of all required regime rows by $\bar s$: a lower
block row contains both $-B_j$ and $A_j+B_j$, whose supports together
contain the support of $A_j$, while an upper block row contains $B_j$
and $-B_j$.

Population projection orthogonality and the Gram concentration bound give
\[
 \left|\frac1n X_<^{(-j)}(\tau)'
       \{X_<^{(j)}(\tau)-X_<^{(-j)}(\tau)\gamma_{0,j}(\tau)\}
       \right|_\infty
 \le C\|\widehat{M}(\tau)-M(\tau)\|_\infty
          \{1+|\gamma_{0,j}(\tau)|_1\}=O_P(b).
\]
No conditional mean zero assumption on this nodewise residual is used.
Choosing $A_{\rm node}$ sufficiently large gives the uniform weighted
Lasso cone with the same $\mu$. Its limiting weight ratio is at most
$w_+/w_-$, so the fixed choice of $c_{\rm node}$ in Assumption~\ref{ts-as6}(ii) contains that cone. The population nodewise restricted eigenvalue and $\bar s b=o(1)$
transfer it to an empirical restricted eigenvalue bound. The basic
inequality then gives, uniformly over the required rows and thresholds,
\[
 |\widehat\gamma_j-\gamma_{0,j}|_1=O_P(\bar s b),\qquad
 \|X_<^{(-j)}(\widehat\gamma_j-\gamma_{0,j})\|_n
       =O_P(\sqrt{\bar s}\,b).
\]
The KKT identity for the prescribed scale is
\[
 \widehat z_j^2(\tau)=\frac1n X_<^{(j)}(\tau)'
       \{X_<^{(j)}(\tau)-X_<^{(-j)}(\tau)\widehat\gamma_j(\tau)\}.
\]
Gram concentration and the preceding coefficient bound yield
$|\widehat z_j^2-z_j^2|=O_P(\bar s b)=o_P(1)$.
Normalizing $(1,-\widehat\gamma_j')$ gives a regime inverse row error
$O_P(\bar s b)$ and a bounded row $\ell_1$ norm. Its KKT inverse error
against $\widehat{M}$ is $O_P(b)$, with an exactly normalized diagonal.
Repeat for the upper regime. Construct the two blocks as in
\eqref{ts-precision} so that the exact sample block identities for
$\widehat{\bm\Sigma}$ give~\eqref{ts-node-l1}--\eqref{ts-node-kkt}.
\end{proof}

\begin{lem}[A deterministic-threshold CLT and a bound on local score variation]\label{ts-lem28}
Suppose Assumptions~\ref{ts-as5}--\ref{ts-as6} hold along with \eqref{ts-exponents}. Define
\[
 G_n(t)=n^{-1/2}\sum_{i=1}^n
             g'\bm\Theta(t)\bm X_i(t)U_i.
\]
At any deterministic $t\in\mathbb T$ whose scalar long-run variance is
bounded below by $c_V>0$, the variance normalized sum has a standard
normal limit uniformly over the class.
For the deterministic reference $t=\tau^\circ$ and deterministic
$\eta_n\ge b^2$ with $\eta_n\to0$,
\begin{equation}\label{ts-modulus}
 \sup_{|u-t|\le\eta_n}|G_n(u)-G_n(t)|
 =O_P\{\sqrt{\eta_n m_n\ell_n}+\sqrt n\,b^2
                         +\bar s\eta_n\sqrt{m_n\ell_n}\}.
\end{equation}
Under \eqref{ts-clt-rates} and Theorem~\ref{ts-thm1},
$G_n(\widehat\tau)-G_n(\tau^\circ)=o_P(1)$.
\end{lem}

\begin{proof}
Write $\theta(t)=\bm\Theta(t)'g$ and
$\xi_i(t)=\theta(t)'\bm X_i(t)U_i$.
The required population row norms and bounded support of $g$ give
$|\theta(t)|_1\le C$, hence
$\sup_t\|\xi_i(t)\|_{L^{24/5}}\le C$ and $E\xi_i(t)=0$.
The bound \eqref{ts-ned-product}, followed by the triangle inequality
over the bounded-$\ell_1$ coefficients, implies
\begin{align}
 \|\xi_i(t)-E[\xi_i(t)\mid\mathcal F^s_{i,m}]\|_{L^1}
 &\le Cm^{-19d_{\rm NED}/27},\notag\\
 \|\xi_i(t)-E[\xi_i(t)\mid\mathcal F^s_{i,m}]\|_{L^2}
 &\le Cm^{-7d_{\rm NED}/27}.\label{ts-score-ned}
\end{align}

\noindent Define the finite-window approximation $\xi_i^{[m]}(t)
=
E[\xi_i(t)\mid\mathcal F^s_{i,m}].$
The preceding approximation bound also gives
\[
\|\xi_i(t)-\xi_i^{[m]}(t)\|_{L^{24/19}}
\le Cm^{-14d_{\rm NED}/27},
\]
where
$\frac89(\frac{19}{24}-\frac5{24})=\frac{14}{27}$.
Thus, the approximation error decreases as the innovation window expands.

To bound the covariance between scores at dates $0$ and $k$, take $m=\lfloor k/3\rfloor$ for $k\ge3$. The approximating scores depend on innovations in the windows $[-m,m]$ and $[k-m,k+m]$, respectively. These windows are separated by at least $k/3$ observations. We control the covariance in two steps: first, we bound the error from replacing the
original scores by their finite window approximations; second, we bound the remaining covariance between the approximating scores.

Since $24/19$ and $24/5$ are conjugate exponents,
H\"older's inequality, the preceding approximation bound, and the uniformly bounded $L^{24/5}$ norms imply
\[
\left|
\operatorname{Cov}(\xi_0(t),\xi_k(t))
-
\operatorname{Cov}(\xi_0^{[m]}(t),\xi_k^{[m]}(t))
\right|
\le Cm^{-14d_{\rm NED}/27}.
\]
For the approximating scores, the mixing covariance inequality gives
\[
\left|
\operatorname{Cov}(\xi_0^{[m]}(t),\xi_k^{[m]}(t))
\right|
\le C\beta_s(k-2m)^{7/12}
\le C(1+k)^{-7d_{\rm mix}/12},
\]
where $7/12=1-2/(24/5)$. Combining these bounds and enlarging $C$ to cover the finitely many smaller lags yields
\begin{equation}\label{ts-score-cov}
|\operatorname{Cov}(\xi_0(t),\xi_k(t))|
\le C\left\{
(1+k)^{-14d_{\rm NED}/27}
+
(1+k)^{-7d_{\rm mix}/12}
\right\}.
\end{equation}

Under Assumption~\ref{ts-as5}(i), both decay exponents exceed two. Consequently,
\[
\sum_{k=0}^{\infty}
|\operatorname{Cov}(\xi_0(t),\xi_k(t))|
\le C,
\qquad
\sum_{k=1}^{\infty}
k\,|\operatorname{Cov}(\xi_0(t),\xi_k(t))|
\le C,
\]
uniformly over the relevant thresholds and model class. The first bound ensures that the long-run variance is finite, while the second controls the finite block variance corrections used below.

For the CLT, replace the scores by their conditional expectations on
windows of radius $m_n$. The normalized sum changes by $o_P(1)$ because
\[
 \sqrt n\,m_n^{-19d_{\rm NED}/27}\to0.
\]
Indeed, \eqref{ts-exponents} implies
$\omega_{\rm ts}d_{\rm NED}>11/12$.
Use large blocks of length $B_n=\lfloor n^{1/5}\rfloor$ separated by
gaps of length $4m_n$.
We have $m_n/B_n\to0$ since $\omega_{\rm ts}<1/6$.
Covariance summability bounds the variance over any deterministic index
set by a constant times its cardinality. Thus the normalized original score sum over the gaps and the last incomplete block is $o_P(1)$, and the projection error over those indices is also negligible by \eqref{ts-score-ned}.

After excluding the gaps and the final incomplete block, we retain the large blocks of finite window score
approximations
\[
\xi_i^{[m_n]}(t)
=
E[\xi_i(t)\mid\mathcal F^s_{i,m_n}].
\]
The innovation windows supporting successive retained blocks
are separated by at least $m_n$. Consequently, these blocks
can be coupled with mutually independent copies, each having
the same distribution as its corresponding block, such that
the probability of any disagreement is at most
\[
C(n/B_n)\beta_s(m_n)
\le Cn^{4/5-\omega_{\rm ts}d_{\rm mix}}
=o(1).
\]

For an original large block, \eqref{ts-score-cov} gives variance
$B_n\sigma_g^2(t)+O(1)$. Replacing it by its projection changes its variance, divided by $B_n$, by at most
\[
 C\{B_n^{1/2}m_n^{-7d_{\rm NED}/27}
          +B_nm_n^{-14d_{\rm NED}/27}\}=o(1).
\]
Consequently, the sum of the independent block variances divided by $n$ converges uniformly to $\sigma_g^2(t)$. Each projected block sum has third absolute moment at most $CB_n^3$. The independent block Berry--Esseen bound is therefore controlled by $CB_n^2/\sqrt n=o(1)$, using the positive lower variance bound.
This proves the uniform scalar CLT.

For the local score variation, decompose
\begin{align*}
 G_n(u)-G_n(t)
 &=\{\theta(u)-\theta(t)\}'n^{-1/2}\sum_i\bm X_i(u)U_i\\
 &\quad+\theta(t)'n^{-1/2}\sum_i
                 \{\bm X_i(u)-\bm X_i(t)\}U_i.
\end{align*}
At $t=\tau^\circ$, \eqref{ts-inverse-modulus}, bounded $|g|_1$,
and \eqref{ts-global-score-sharp} give order
$O_P(\bar s\eta_n\sqrt{m_n\ell_n})$ for the first term. The second is bounded by
$O_P(\sqrt{\eta_n m_n\ell_n}+\sqrt n\,b^2)$ using
\eqref{ts-local-score} and $|\theta(t)|_1\le C$. This proves \eqref{ts-modulus}.

Choose $M_n\uparrow\infty$ sufficiently slowly and set $\eta_n=M_ns_0b$. Theorem~\ref{ts-thm1} implies $P(|\widehat\tau-\tau^\circ|>\eta_n)=o(1)$ uniformly in either regime. The indicator term vanishes by $s_0bm_n\ell_n\to0$, and
$\sqrt n b^2\to0$. For the inverse row term, note that
\[
 \bar s s_0b\sqrt{m_n\ell_n}
 =\frac{\bar s}{T_n^{\rm tr}}\,s_0\sqrt n b^2.
\]
The first CLT rate implies both
$\bar s\sqrt n b^2\to0$ and $s_0\sqrt n b^2\to0$.
Furthermore,
\[
 \sqrt n T_n^{\rm tr}b^2
 \asymp n^{3a_{\rm ts}+\omega_{\rm ts}-1/2}\ell_n\longrightarrow\infty,
\]
because $a_{\rm ts}>(1+\nu_p)/12\ge1/6$ and
$\omega_{\rm ts}>0$. Hence $\bar s/T_n^{\rm tr}\to0$. Under the rate conditions in \eqref{ts-clt-rates}, choosing $M_n\to\infty$ sufficiently slowly makes every term in \eqref{ts-modulus} vanish. This proves the random index conclusion for a threshold selected from
the same sample. 
\end{proof}

%%%%%%%%%%%%%%%%%%%%%%%%%%%%%%%%%%%%%%%%%%%%
\begin{proof}[Proof of Theorem~\ref{ts-thm2}]
The exact algebraic decompositions in the main paper continue to hold under dependence. Their coefficient remainder is bounded by
\[
 \sqrt n\,|g|_1
 \max_{j\in H}|\widehat\Theta_j(\widehat\tau)'
                    \widehat{\bm\Sigma}(\widehat\tau)-e_j'|_\infty
 |\widehat\alpha-\alpha_0|_1
 =O_P(\sqrt n\,s_0b^2)=o_P(1).
\]
Replacing $\widehat{\bm\Theta}$ by $\bm\Theta$ in the leading score has
error at most
\[
 |g|_1\max_{j\in H}\sup_\tau
 |\widehat\Theta_j(\tau)-\Theta_j(\tau)|_1
 \sup_\tau|n^{-1/2}\bm X(\tau)'U|_\infty
 =O_P(\sqrt n\,\bar s b^2)=o_P(1).
\]
In the fixed-threshold case, the additional  term is bounded using~\eqref{ts-local-design}, the bounded row norms, and
$|\widehat\tau-\tau_0|=O_P(s_0b^2)$:
\begin{align*}
 &\left|g'\widehat{\bm\Theta}(\widehat\tau)
 \bm X(\widehat\tau)'\{\bm X(\tau_0)-\bm X(\widehat\tau)\}
                         \alpha_0/\sqrt n\right|\\
 &\hspace{25mm}=O_P(\sqrt n\,D_{\delta,n}s_0b^2)=o_P(1).
\end{align*}
It is identically zero under the no threshold effect case. Finally, Lemma~\ref{ts-lem28} replaces the population score at $\widehat\tau$ by that at $\tau^\circ$, and its CLT establishes the assertion. All remainder bounds and the normal approximation hold uniformly over the union of the no threshold effect and identified threshold effect classes.
\end{proof}

%%%%%%%%%%%%%%%%%%%%%%%%%%%%%%%%%%%%%%%%%
\begin{lem}[Estimated influence scores]\label{ts-lem29}
Suppose Assumptions~\ref{ts-as5}--\ref{ts-as7}, \eqref{ts-exponents}, and \eqref{ts-tuning} hold, with Assumption~\ref{ts-as7} needed only in the fixed-threshold class. Let
\[
 \widehat\xi_i=g'\widehat{\bm\Theta}(\widehat\tau)
                         \bm X_i(\widehat\tau)\widehat U_i(\widehat\tau),
 \qquad \xi_i=g'\bm\Theta(\tau^\circ)\bm X_i(\tau^\circ)U_i.
\]
Then
\begin{equation}\label{ts-score-l2}
 \|\widehat\xi-\xi\|_n=O_P(r_{\xi,n}),\qquad
 r_{\xi,n}=(s_0+\bar s+s_0\bar s)b
              +(1+D_{\delta,n})\sqrt{s_0b+c_{4,n}},\qquad
 \|\xi\|_n=O_P(1).
\end{equation}
\end{lem}

\begin{proof}
The degree four bound in Lemma~\ref{ts-lem26} implies
\[
 \max_j\|X^{(j)}U\|_n+
 \max_{j,l}\|X^{(j)}X^{(l)}\|_n=O_P(1).
\]
It also implies, for a deterministic center $t$ and any deterministic
$\eta\to0$,
\begin{align*}
 \max_j P_n(X_i^{(j)})^2U_i^2 1\{|Q_i-t|\le\eta\}
 &\le C\eta+O_P(c_{4,n}),\\
 \max_{j,l}P_n(X_i^{(j)})^2(X_i^{(l)})^2
                      1\{|Q_i-t|\le\eta\}
 &\le C\eta+O_P(c_{4,n}).
\end{align*}
Uniformity in thresholds permits application on the threshold localization events from Theorem~\ref{ts-thm1}.

First use $U_i$ in place of $\widehat U_i(\widehat\tau)$ in $\widehat\xi_i$.
Changing the estimated precision row to its population counterpart costs
$O_P(\bar s b)$ in empirical $L_2$ norm. Changing its population threshold
argument costs $O_P(\bar s|\widehat\tau-\tau^\circ|)$ by
\eqref{ts-inverse-modulus} and the empirical product moment bound on the localization event.
Changing the indicator in the design costs
$O_P(\sqrt{|\widehat\tau-\tau^\circ|+c_{4,n}})$ by the first local fourth moment bound. For any fixed $K>0$, we apply the local moment bounds on
the event $|\widehat\tau-\tau^\circ|\le Ks_0b.$
On this event, the random threshold lies within a
deterministic neighborhood of radius $Ks_0b$.
The uniform localization bound from Theorem~\ref{ts-thm1} ensures that, for sufficiently large $n$, the probability of the complementary event can be made arbitrarily small by choosing $K$ sufficiently large.
The resulting stochastic orders therefore hold uniformly over the admissible model class.

For the residual replacement, use
\[
 \widehat U_i(\widehat\tau)-U_i=-\bm X_i(\widehat\tau)'(\widehat\alpha-\alpha_0)
                  +\{\bm X_i(\tau_0)-\bm X_i(\widehat\tau)\}'\alpha_0.
\]
Multiplication by the bounded $ \ell_1$-estimated contrast row and the empirical fourth moment bounds show that the first term contributes $O_P(s_0b)$. The second is absent under the no threshold effect case and otherwise contributes $O_P(D_{\delta,n}\sqrt{|\widehat\tau-\tau_0|+c_{4,n}})$. The uniform bound $|\widehat\tau-\tau^\circ|=O_P(s_0b)$ therefore gives \eqref{ts-score-l2}. In the identified threshold effect case, the contribution from replacing $\bm\Theta(\widehat\tau)$ by
$\bm\Theta(\tau_0)$ is $O_P(s_0\bar s b^2)=o_P(s_0b)$, since $\bar s b\to0$. In the no threshold effect case, the available localization rate yields the larger bound
$O_P(s_0\bar s b)$ for the corresponding contribution. We therefore retain the term $s_0\bar s b$ in the common
rate bound covering both cases. Finally $E\xi_i^2\le C$ implies $\|\xi\|_n=O_P(1)$.
\end{proof}

\begin{lem}[Oracle Bartlett variance estimator]\label{ts-lem30}
Suppose Assumptions~\ref{ts-as5}--\ref{ts-as6} and \eqref{ts-exponents} hold. For the deterministic-reference score $\xi_i$ in Lemma~\ref{ts-lem29}, form the scalar, sample centered version of \eqref{ts-hac} using $\xi_i$
and denote it by $\widehat V_{\rm or}$.
If $k_n\to\infty$ and $k_n^3/n\to0$, then
$\widehat V_{\rm or}-\sigma_g^2=o_P(1)$ uniformly.
\end{lem}

\begin{proof}
Let $\gamma(k)=E(\xi_i\xi_{i-k})$. Since $E(\xi_i)=0$, the long-run variance is
\[
\sigma_g^2=\gamma(0)+2\sum_{k=1}^{\infty}\gamma(k).
\]
Equation~\eqref{ts-score-cov} implies
$\sum_{k\ge1}k|\gamma(k)|\le C$ uniformly over the admissible model class. The error from truncating the covariance series at lag $k_n$ and applying Bartlett weights therefore satisfies
\begin{align*}
&\left|
\gamma(0)
+2\sum_{k=1}^{k_n}
\left(1-\frac{k}{k_n+1}\right)\gamma(k)
-\sigma_g^2
\right|\\
&\quad\le
2\sum_{k>k_n}|\gamma(k)|
+\frac{2}{k_n+1}\sum_{k=1}^{k_n}k|\gamma(k)|\\
&\quad\le
\frac{2}{k_n+1}\sum_{k=1}^{\infty}k|\gamma(k)|
=O(k_n^{-1}).
\end{align*}
The first term in the initial upper bound accounts for the omitted lags, while the second accounts for the Bartlett weights on the retained lags. Thus, their combined contribution vanishes as $k_n\to\infty$. The finite-sample boundary and sample-centering
corrections are treated below.

The covariances of score products need not be summable. For $A_{i,k}=\xi_i\xi_{i-k}$ we have $\|A_{i,k}\|_{L^{12/5}}\le C$.
Equation \eqref{ts-ned-product} at $u=8/3$ gives a score approximation error $Cm^{-4d_{\rm NED}/27}$ in that norm.
H\"older with $1/(12/7)=1/(8/3)+1/(24/5)$ then gives the same error order in $L^{12/7}$ for $A_{i,k}$, using the union of the innovation
windows at its two dates. For date separation $r>k+3$, take window radii proportional to $r-k$.
The covariance approximation error is bounded using conjugate exponents $12/7$ and $12/5$, and the remaining mixing covariance has exponent
$1-2/(12/5)=1/6$. Hence
\begin{align}
 |\operatorname{Cov}(A_{0,k},A_{r,k})|
 &\le C\{(r-k)^{-4d_{\rm NED}/27}
                   +(r-k)^{-d_{\rm mix}/6}\}\notag\\
 &\le C(r-k)^{-4/5}.\label{ts-product-cov}
\end{align}
The final inequality follows from $d_{\rm NED}>11/2$ and $d_{\rm mix}>5$.
Use bounded second moments for the first $k+3$ separations.
It follows that
\[
 \operatorname{Var}\left[
   \frac1n\sum_{i=k+1}^n\{\xi_i\xi_{i-k}-\gamma(k)\}\right]
 \le C\{(k+1)/n+n^{-4/5}\}.
\]
By Cauchy--Schwarz, summing the absolute stochastic errors through lag
$k_n$ gives
\[
 O_P(k_n^{3/2}/\sqrt n+k_nn^{-2/5})=o_P(1).
\]
Both terms converge to zero under $k_n^3/n\to0$, since
\[
\frac{k_n^{3/2}}{\sqrt n}
=\left(\frac{k_n^3}{n}\right)^{1/2}\to0,
\qquad
k_n n^{-2/5}
=\left(\frac{k_n^3}{n}\right)^{1/3}n^{-1/15}\to0.
\]
Thus summability of the product covariance sequence is unnecessary.

The finite sample expectation has factor $(1-k/n)\gamma(k)$. Its total boundary error is $O(n^{-1})$ by lag weighted covariance
summability. Also $\bar\xi=O_P(n^{-1/2})$ and $\|\xi\|_n=O_P(1)$. Cauchy--Schwarz bounds the effect of sample centering over the included
lags by $O_P(k_n/\sqrt n)=o_P(1)$. All terms vanish uniformly under the stated bandwidth conditions.
\end{proof}
%%%%%%%%%%%%%%%%%%%%%%%%%%%%%%%%%%%%%%%%%%%%%
\begin{proof}[Proof of Theorem~\ref{ts-thm3}]
We compare the proposed variance estimator with the oracle Bartlett estimator $\widehat V_{\rm or}$ from Lemma~\ref{ts-lem30}. Recall the estimated and population scores from Lemma~\ref{ts-lem29}:
\[
\widehat\xi_i
=
g'\widehat{\bm\Theta}(\widehat\tau)
  \bm X_i(\widehat\tau)\widehat U_i(\widehat\tau),
\qquad
\xi_i
=
g'\bm\Theta(\tau^\circ)\bm X_i(\tau^\circ)U_i.
\]
The triangle inequality gives
\begin{align*}
&\left|
g'\widehat{\mathcal V}^{\rm LR}_n(\widehat\tau)g
-\sigma_g^2
\right|\\
&\quad\le
\left|
g'\widehat{\mathcal V}^{\rm LR}_n(\widehat\tau)g
-\widehat V_{\rm or}
\right|
+
\left|\widehat V_{\rm or}-\sigma_g^2\right|.
\end{align*}
Lemma~\ref{ts-lem30} shows that the second term is $o_P(1)$ uniformly. It therefore remains to show that replacing the population scores $\xi_i$ by the estimated scores $\widehat\xi_i$ changes the Bartlett estimator by $o_P(1)$.

For this purpose, define the centered scores
\[
\widehat\xi_i^c
=
\widehat\xi_i-\frac1n\sum_{l=1}^n\widehat\xi_l,
\qquad
\xi_i^c
=
\xi_i-\frac1n\sum_{l=1}^n\xi_l.
\]
For $0\le k\le k_n$, let
\[
\widehat\gamma(k)
=
\frac1n\sum_{i=k+1}^n
\widehat\xi_i^c\widehat\xi_{i-k}^c,
\qquad
\widehat\gamma_{\rm or}(k)
=
\frac1n\sum_{i=k+1}^n
\xi_i^c\xi_{i-k}^c.
\]
By the definition in \eqref{ts-hac},
\[
g'\widehat{\mathcal V}^{\rm LR}_n(\widehat\tau)g
=
\widehat\gamma(0)
+
2\sum_{k=1}^{k_n}
\left(1-\frac{k}{k_n+1}\right)\widehat\gamma(k).
\]
The oracle estimator $\widehat V_{\rm or}$ has the same expression with $\widehat\gamma_{\rm or}(k)$ in place of
$\widehat\gamma(k)$.

We next bound the difference between these covariance estimates at each lag. Using
$xy-uv=(x-u)y+u(y-v)$ and the Cauchy--Schwarz inequality, we obtain
\[
|\widehat\gamma(k)-\widehat\gamma_{\rm or}(k)|
\le
\|\widehat\xi^c-\xi^c\|_n
\bigl(\|\widehat\xi^c\|_n+\|\xi^c\|_n\bigr).
\]
This bound holds for every $0\le k\le k_n$, since the empirical norms over the indices in each lagged sum, with denominator $n$, are bounded by the corresponding full sample norms.

Subtracting the sample mean cannot increase an empirical $L_2$ norm. Hence,
\[
\|\widehat\xi^c-\xi^c\|_n
\le \|\widehat\xi-\xi\|_n,
\qquad
\|\xi^c\|_n\le\|\xi\|_n,
\]
and the triangle inequality also gives
\[
\|\widehat\xi^c\|_n
\le
\|\xi\|_n+\|\widehat\xi-\xi\|_n.
\]
Consequently,
\[
\max_{0\le k\le k_n}
|\widehat\gamma(k)-\widehat\gamma_{\rm or}(k)|
\le
\|\widehat\xi-\xi\|_n
\bigl(2\|\xi\|_n+\|\widehat\xi-\xi\|_n\bigr).
\]

Since the Bartlett weights lie between zero and one, summing over the included lags yields
\begin{equation*}
\left|
g'\widehat{\mathcal V}^{\rm LR}_n(\widehat\tau)g
-\widehat V_{\rm or}
\right|
\le
C(k_n+1)\|\widehat\xi-\xi\|_n
\bigl(\|\xi\|_n+\|\widehat\xi-\xi\|_n\bigr).
\end{equation*}
Lemma~\ref{ts-lem29} gives
\[
\|\widehat\xi-\xi\|_n=O_P(r_{\xi,n}),
\qquad
\|\xi\|_n=O_P(1).
\]
The bandwidth condition in \eqref{ts-hac-rates} implies
$k_nr_{\xi,n}\to0$ and, since $k_n\to\infty$, also
$r_{\xi,n}\to0$. Therefore,
\[
\left|
g'\widehat{\mathcal V}^{\rm LR}_n(\widehat\tau)g
-\widehat V_{\rm or}
\right|
=
O_P(k_nr_{\xi,n})
=
o_P(1).
\]
Combining this result with
$\widehat V_{\rm or}-\sigma_g^2=o_P(1)$ proves
\eqref{ts-var-consistency}. 

Finally, consider a fixed collection of contrasts. Apply the same argument to each contrast and to their pairwise sums and differences, imposing the row conditions on their combined support. Nonzero contrasts can be normalized and then rescaled. For any two contrasts
$g_a$ and $g_b$, symmetry gives
\begin{equation*}
g_a'\widehat{\mathcal V}^{\rm LR}_n(\widehat\tau)g_b
=
\frac14\Bigl\{
(g_a+g_b)'
\widehat{\mathcal V}^{\rm LR}_n(\widehat\tau)
(g_a+g_b)
- (g_a-g_b)'
\widehat{\mathcal V}^{\rm LR}_n(\widehat\tau)
(g_a-g_b)
\Bigr\}.
\end{equation*}
The corresponding identity holds for the population long-run covariance. Consistency of the quadratic forms therefore implies consistency of each cross covariance. This argument does not require a positive lower bound on the variance of the sums or differences.
Since the collection is fixed and finite, the entrywise conclusion follows uniformly.
\end{proof}

\begin{proof}[Proof of Theorem~\ref{ts-thm4}]
Write
\[
\widehat\sigma_g^2
=
g'\widehat{\mathcal V}^{\rm LR}_n(\widehat\tau)g,
\qquad
\sigma_g=(\sigma_g^2)^{1/2},
\qquad
\widehat\sigma_g=(\widehat\sigma_g^2)^{1/2}.
\]
The Bartlett estimator is positive semidefinite, so $\widehat\sigma_g^2\ge0$. 

We first show that the estimated standard deviation is close to the population standard deviation in relative terms. Theorem~\ref{ts-thm3} gives $\widehat\sigma_g^2-\sigma_g^2=o_P(1)$, while Assumption~\ref{ts-as8} ensures $\sigma_g^2\ge c_V>0$.
Consequently,
\[
\left|\frac{\widehat\sigma_g}{\sigma_g}-1\right|
=
\frac{|\widehat\sigma_g^2-\sigma_g^2|}
     {\sigma_g(\widehat\sigma_g+\sigma_g)}
\le
\frac{|\widehat\sigma_g^2-\sigma_g^2|}{c_V}
=o_P(1).
\]
In particular, the probability that
$\widehat\sigma_g=0$ tends to zero uniformly.

Define the statistics standardized by the population and estimated standard deviations, respectively:
\[
Z_n
=
\frac{\sqrt n\,g'\{\widetilde a(\widehat\tau)-\alpha_0\}}
     {\sigma_g},
\qquad
\widehat Z_n
=
\frac{\sqrt n\,g'\{\widetilde a(\widehat\tau)-\alpha_0\}}
     {\widehat\sigma_g}.
\]
On the event $\widehat\sigma_g=0$, assign
$\widehat Z_n$ an arbitrary fixed value.
Theorem~\ref{ts-thm2} establishes a uniform standard normal approximation for $Z_n$, and therefore $Z_n=O_P(1)$ uniformly.

We next verify that estimating the standard deviation changes the statistic by a negligible amount. For fixed $K>0$ and $0<\delta<1/2$, consider the event
\[
|Z_n|\le K,
\qquad
\left|\frac{\widehat\sigma_g}{\sigma_g}-1\right|
\le\delta.
\]
On this event, $\widehat\sigma_g>0$ and
\[
|\widehat Z_n-Z_n|
=
|Z_n|
\left|\frac{\sigma_g}{\widehat\sigma_g}-1\right|
\le
\frac{K\delta}{1-\delta}
\le 2K\delta.
\]
Uniform boundedness in probability of $Z_n$ allows $K$ to be chosen so that the probability of $|Z_n|>K$ is arbitrarily small for sufficiently large $n$. Relative consistency of $\widehat\sigma_g$ controls the other exceptional event. Choosing $K$ sufficiently large, then $\delta$ sufficiently small, and finally $n$ sufficiently large therefore
gives
\[
\widehat Z_n-Z_n=o_P(1)
\]
uniformly.

To establish the uniform normal approximation for
$\widehat Z_n$, fix $\epsilon>0$. For every
$x\in\mathbb R$,
\begin{align*}
&P_{\alpha_0}(Z_n\le x-\epsilon)
-
P_{\alpha_0}(|\widehat Z_n-Z_n|>\epsilon)\\
&\qquad\le
P_{\alpha_0}(\widehat Z_n\le x)\\
&\qquad\le
P_{\alpha_0}(Z_n\le x+\epsilon)
+
P_{\alpha_0}(|\widehat Z_n-Z_n|>\epsilon).
\end{align*}
Since the standard normal density is bounded by
$(2\pi)^{-1/2}$, these inequalities imply
\begin{align*}
&\sup_{\alpha_0\in\mathcal B_{\ell_0}(s_0)}
 \sup_{x\in\mathbb R}
 \left|P_{\alpha_0}(\widehat Z_n\le x)-\Phi(x)\right|\\
&\quad\le
\sup_{\alpha_0\in\mathcal B_{\ell_0}(s_0)}
 \sup_{x\in\mathbb R}
 \left|P_{\alpha_0}(Z_n\le x)-\Phi(x)\right|
+
\sup_{\alpha_0\in\mathcal B_{\ell_0}(s_0)}
 P_{\alpha_0}(|\widehat Z_n-Z_n|>\epsilon)
+\frac{\epsilon}{\sqrt{2\pi}}.
\end{align*}
The first two terms converge to zero by
Theorem~\ref{ts-thm2} and the preceding argument.
Letting $n\to\infty$ and then $\epsilon\downarrow0$ proves \eqref{ts-feasible}.

We now establish confidence interval coverage.
For a fixed $\alpha\in(0,1)$, on the event
$\widehat\sigma_g>0$, the interval in the theorem
contains $g'\alpha_0$ exactly when
$|\widehat Z_n|\le z_{1-\alpha/2}.$
The complementary event has probability tending to zero uniformly. The uniform normal approximation and continuity of the standard normal distribution therefore yield
\[
\sup_{\alpha_0\in\mathcal B_{\ell_0}(s_0)}
\left|
P_{\alpha_0}\left(
\left|g'\widetilde a(\widehat\tau)-g'\alpha_0\right|
\le
z_{1-\alpha/2}\frac{\widehat\sigma_g}{\sqrt n}
\right)
-(1-\alpha)
\right|
\longrightarrow0.
\]

Finally, the moment bounds and the covariance
summability implied by \eqref{ts-score-cov} give
$\sigma_g^2\le C$ uniformly.
Theorem~\ref{ts-thm3} consequently implies
$\widehat\sigma_g=O_P(1)$ uniformly.
The interval length is therefore
\[
2z_{1-\alpha/2}\frac{\widehat\sigma_g}{\sqrt n}
=
O_P(n^{-1/2}),
\]
which completes the proof.
\end{proof}

\begin{proof}[Proof of Corollary~\ref{ts-thm5}]
Fix a horizon $r\in\mathcal R$. The local projection equation \eqref{ts-lp-model} has the same threshold regression form as the model considered in Theorems~\ref{ts-thm1}--\ref{ts-thm4}, with $\alpha_0$, $U_i$, and $n$ replaced by $\alpha_{0,r}$, $U_{r,i}$, and $n_r$, respectively. The corollary imposes the required assumptions on the entire horizon-specific process $(X_i',Q_i,U_{r,i})'$. The preceding estimation and inference results therefore apply to the estimators computed at this horizon.

Let $c_r$ be a prespecified nonzero contrast and write $g_r=\frac{c_r}{|c_r|_2}.$ Then $|g_r|_2=1$, as required by Theorem~\ref{ts-thm4}. Replacing $c_r$ by $g_r$ divides both the numerator and denominator of the studentized statistic by the same positive constant
$|c_r|_2$. Hence, the statistic is unchanged by this normalization.

Define
\[
\widehat Z_r
=
\frac{
\sqrt{n_r}\,
c_r'\{\widetilde a_r(\widehat\tau_r)-\alpha_{0,r}\}
}{
\{c_r'\widehat{\mathcal V}^{\rm LR}_{n_r,r}
(\widehat\tau_r)c_r\}^{1/2}
},
\]
using the same convention for a zero denominator as in Theorem~\ref{ts-thm4}. Applying that theorem with the normalized contrast $g_r$ gives
\[
\sup_{\alpha_0\in\mathcal B_{\ell_0}(s_0)}
\sup_{t\in\mathbb R}
\left|
P_{\alpha_0}(\widehat Z_r\le t)-\Phi(t)
\right|
\longrightarrow0.
\]
The assumptions have common constants and rate
restrictions across horizons, and $n_r\to\infty$
for every $r\in\mathcal R$. Since $\mathcal R$ is
fixed and finite, the maximum of these approximation errors over $r\in\mathcal R$ also converges to zero. This proves the stated uniform normal approximation.

The regime-specific responses and their difference are particular choices of the contrast. Specifically, $c_r=e_{j_x}$ gives the upper-regime response, $c_r=e_{j_x}+e_{p+j_x}$ gives the lower-regime response, and $c_r=e_{p+j_x}$ gives the lower-minus-upper
difference. The opposite difference uses
$c_r=-e_{p+j_x}$. Thus, the preceding argument applies to all the stated responses, subject to the variance condition for each normalized contrast.

It remains to establish simultaneous coverage.
Consider a fixed collection of $J$ requested
horizon-contrast pairs $(r_j,c_j)$, $j=1,\ldots,J$, and write $ \theta_j=c_j'\alpha_{0,r_j}.$
For a desired simultaneous confidence level
$1-\alpha$, let $I_j$ be the interval from
Theorem~\ref{ts-thm4} for $\theta_j$, constructed
using the critical value $z_{1-\alpha/(2J)}$.
Each interval therefore has asymptotic marginal
coverage $1-\alpha/J$, uniformly over the
admissible model class. In particular,
\[
\sup_{\alpha_0\in\mathcal B_{\ell_0}(s_0)}
P_{\alpha_0}(\theta_j\notin I_j)
\le \frac{\alpha}{J}+o(1).
\]
By the union bound,
\begin{align*}
&\sup_{\alpha_0\in\mathcal B_{\ell_0}(s_0)}
P_{\alpha_0}\bigl(
\theta_j\notin I_j
\text{ for at least one }j
\bigr)\\
&\quad\le
\sum_{j=1}^{J}
\sup_{\alpha_0\in\mathcal B_{\ell_0}(s_0)}
P_{\alpha_0}(\theta_j\notin I_j)
\le \alpha+o(1),
\end{align*}
where the final remainder tends to zero because
$J$ is fixed. Consequently,
\[
\liminf_{n\to\infty}
\inf_{\alpha_0\in\mathcal B_{\ell_0}(s_0)}
P_{\alpha_0}\bigl(
\theta_j\in I_j
\text{ for every }j=1,\ldots,J
\bigr)
\ge 1-\alpha.
\]
The union bound applies regardless of dependence
across horizons, including dependence arising from
overlapping projection errors. This establishes
the simultaneous coverage claim.
\end{proof}

\clearpage
\printbibliography[title={References for the Appendix}]
\end{refsection}

\end{document}